%% file: main.tex
\documentclass[11pt]{article}

\usepackage{fullpage}

\usepackage[dvipsnames]{xcolor}
\usepackage[colorinlistoftodos,prependcaption,textsize=scriptsize,disable]{todonotes}

\newcommand{\revise}[2]{
#2%
}

\definecolor{uecolor}{HTML}{00BFFF}
\newcommand{\uezatocomment}[1]{\todo[linecolor=uecolor,backgroundcolor=uecolor!25,bordercolor=uecolor]{#1}}

\newcommand{\snote}[1]{\textcolor{blue}{(Soh: #1)}}

\usepackage{graphicx} 
\usepackage{wrapfig}

\usepackage{tikz}
\usetikzlibrary{automata,positioning,calc,intersections,arrows.meta}
\usepackage{tikz-cd}

\usepackage{amsmath,amssymb}
\usepackage{ascmac}
\usepackage{comment}
\usepackage{nicematrix}

\usepackage{amsthm}

\usepackage{thmtools}
\usepackage{thm-restate}

\newtheorem{theorem}{Theorem}
\newtheorem{lemma}[theorem]{Lemma}

\newtheorem{definition}[theorem]{Definition}
\newtheorem{proposition}[theorem]{Proposition}

\newtheorem*{hypothesis*}{Hypothesis}
\newtheorem*{problem*}{Problem}
\newtheorem*{theorem*}{Theorem}
\newtheorem*{claim*}{Claim}

\PassOptionsToPackage{only,fatsemi,llparenthesis,rrparenthesis,lbag,rbag,shortleftarrow,shortrightarrow,llbracket,rrbracket,Mapsto}{stmaryrd}
\usepackage{extpfeil} 

\usepackage{forest}

\newcommand{\hpath}{\pi} 

\usepackage{mathrsfs}
\newcommand{\light}{\mathscr{L}}
\newcommand{\heavy}{\mathscr{H}}

\usepackage{algpseudocode}
\usepackage[linesnumbered,ruled,vlined]{algorithm2e}
\SetKwProg{Procedure}{Procedure}{}{}
\SetKwInput{KwRequire}{Require}

\usepackage{enumitem}
\usepackage{soul}
\usepackage[tikz]{bclogo}

\usepackage[normalem]{ulem}

\newcommand{\set}[1]{\{ #1 \}}
\newcommand{\tuple}[1]{\langle #1 \rangle}
\newcommand{\readv}[1]{\mathord{\backslash #1}}
\newcommand{\lang}[1]{\mathbb{L}\left(#1\right)}
\newcommand{\sem}[1]{\llbracket #1 \rrbracket}
\newcommand{\comp}{\odot}

\newcommand{\regex}{{\textbf Regex}}
\newcommand{\rewb}{{\textbf REWB}}
\newcommand{\redos}{{\textbf ReDoS}}

\newcommand{\seth}{{\textbf SETH}}

\newcommand{\wone}{\textbf{W[1]}}

\newcommand{\transm}[1]{\mathcal{T}_{#1}}
\newcommand{\trans}{\delta}

\newcommand{\evalc}[1]{\Delta_C(\,w#1\,)}
\newcommand{\evalcc}[1]{\Delta_C(#1)}

\newcommand{\evald}[1]{\Delta_D(\,#1\,)}

\newcommand{\btw}{{\!\>..\!\>}}

\newcommand{\stree}{\ensuremath{\mathsf{ST}}}
\newcommand{\fforest}{\ensuremath{\mathsf{FF}}}

\newcommand{\rightenumds}{\textsc{RightEnumerator}}
\newcommand{\leftenumds}{\textsc{LeftEnumerator}}
\newcommand{\auxds}{\textsc{AuxiliaryEnumerator}}

\newcommand{\que}{Q}
\newcommand{\fna}{f}
\newcommand{\fnb}{g}

\newcommand{\mini}{\mathrm{minidx}}

\newcommand{\man}{m}
\newcommand{\eval}[1]{\Delta_C(\,w#1\,)}

\newcommand{\GetA}{\ensuremath{\textsc{Get}_{\mathrm{A}}}}
\newcommand{\ShrinkA}{\ensuremath{\textsc{Shrink}_{\mathrm{A}}}}
\newcommand{\AddA}{\ensuremath{\textsc{Add}_{\mathrm{A}}}}
\newcommand{\UpdateA}{\ensuremath{\textsc{Update}_{\mathrm{A}}}}
\newcommand{\UpdateIdemA}{\ensuremath{\textsc{UpdateIdem}_{\mathrm{A}}}}
\newcommand{\InitializeA}{\ensuremath{\textsc{Initialize}_{\mathrm{A}}}}

\newcommand{\GetL}{\textsc{Get}_{\mathrm{L}}}
\newcommand{\ShrinkL}{\textsc{Shrink}_{\mathrm{L}}}
\newcommand{\AddL}{\textsc{Add}_{\mathrm{L}}}

\newcommand{\lazyD}{\mathcal{D}^{\mathrm{L}}}
\newcommand{\lazymini}{\mathrm{minidx}^{\mathrm{L}}}

\SetKwProg{DefClass}{Class}{:}{}

\SetKwProg{DefMethod}{Method}{:}{}

\SetKwProg{DefInitialize}{Initialize}{:}{}

\SetKwProg{Initialize}{Initialize}{}{}

\SetKwProg{Method}{method}{}{}

\SetKwProg{Data}{(Shared) Data}{}{}

\SetKwProg{GlobalData}{Global Data}{}{}

\newcommand{\makearray}[1]{%
  \def\N{#1}%
  \foreach \k in {1,...,#1}{%
    \node[leaf] (a\k) at (\k,0) {}; 
  }%
  \node[anchor=east,font=\footnotesize] at (0.5, 0){$w$};
}

\newcommand{\putcell}[2]{%
  \node[inner sep=0pt] at (a#1.center) {#2};
}

\newcommand{\cover}[5]{
  \path coordinate (L#4) at ({#1-0.5},{#3})
        coordinate (R#4) at ({#2+0.5},{#3});
  \draw[box] ($(L#4)+(0.1,0)$) rectangle ($(R#4)+(-0.1,0.6)$) ;
  \node (#4) at ($ (L#4)!0.5!(R#4) + (0,0.3)$) {#5};
}

\newcommand{\colcover}[6]{
  \path coordinate (L#4) at ({#1-0.5},{#3})
        coordinate (R#4) at ({#2+0.5},{#3});
  \draw[box, fill=#6] ($(L#4)+(0.1,0)$) rectangle ($(R#4)+(-0.1,0.6)$) ;
  \node (#4) at ($ (L#4)!0.5!(R#4) + (0,0.3)$) {#5};
}

\newcommand{\northhit}[3][]{%
  \draw[#1] (#2.north) -- ($(#2.north |- #3.south)$);
}

\newcommand{\northhitCover}[3][]{%
  \path[name path=up-#2-#3]  (#2.north) -- (#2.north |- L#3.north);
  \path[name path=bot-#3]    (L#3) -- (R#3);
  \path[name intersections={of=up-#2-#3 and bot-#3, by=t-#2-#3}];
  \draw[#1] (#2.north) -- (t-#2-#3);
}

\usepackage[%
 setpagesize=false,%
 bookmarks=true,%
 bookmarksnumbered=true,%
 colorlinks=false,%
 pdftitle={},%
 pdfsubject={},%
 pdfauthor={},%
 pdfkeywords={}%
]{hyperref}

\usepackage[mathlines]{lineno} 

\usepackage{etoolbox} 

\newcommand*\linenomathpatch[1]{%
  \cspreto{#1}{\linenomath}%
  \cspreto{#1*}{\linenomath}%
  \csappto{end#1}{\endlinenomath}%
  \csappto{end#1*}{\endlinenomath}%
}
\newcommand*\linenomathpatchAMS[1]{%
  \cspreto{#1}{\linenomathAMS}%
  \cspreto{#1*}{\linenomathAMS}%
  \csappto{end#1}{\endlinenomath}%
  \csappto{end#1*}{\endlinenomath}%
}

\expandafter\ifx\linenomath\linenomathWithnumbers
  \let\linenomathAMS\linenomathWithnumbers
  \patchcmd\linenomathAMS{\advance\postdisplaypenalty\linenopenalty}{}{}{}
\else
  \let\linenomathAMS\linenomathNonumbers
\fi

\linenomathpatch{equation}
\linenomathpatchAMS{gather}
\linenomathpatchAMS{multline}
\linenomathpatchAMS{align}
\linenomathpatchAMS{alignat}
\linenomathpatchAMS{flalign}

\usepackage{authblk}

\title{On the Complexity of the Matching Problem of \\Regular Expressions with Backreferences\thanks{This is the full version of the paper accepted to ICALP 2026.}}

\author[1]{Soh Kumabe\thanks{\texttt{kumabe\_soh@cyberagent.co.jp}}}
\author[1,2]{Yuya Uezato\thanks{\texttt{uezato\_yuya@cyberagent.co.jp}, \texttt{uezato@nii.ac.jp}. \par~~Yuya Uezato was supported by JST, CREST Grant Number JPMJCR21M3.}}
\affil[1]{CyberAgent, Inc., Japan.}
\affil[2]{National Institute of Informatics, Japan.}

\date{}

\begin{document}

\maketitle

\begin{abstract}
Regular Expression Denial of Service (ReDoS) is a well-known type of algorithmic complexity attack, where an adversary supplies maliciously crafted strings to a regular expression matching engine,
aiming to exhaust computational resources of systems.
Even quadratic-time behavior in matching engines has been exploited in successful attacks,
as exemplified by major outages at Stack Overflow (2016) and Cloudflare (2019).
These incidents motivate a fundamental question: Is it possible to construct matching engines that are provably efficient, running in linear or near-linear time in the length of the input string?

For classical regular expressions (REGEX), Thompson's construction yields a linear-time algorithm for fixed expressions.
However, practical engines support powerful features such as backreferences, which allow capturing a substring and reusing it later.
This feature strictly extends the expressive power of REGEX but unfortunately increases the risk of ReDoS attacks.

This paper investigates the fine-grained complexity of the string matching problem for regular expressions with backreferences (REWBs).
Specifically, we consider $r$-use $k$-REWBs,
i.e., REWBs with $k$ variables such that, in any computation,
the total number of backreference executions is at most $r$.

On the hardness side, we show that the string matching problem for $k$-REWBs cannot be solved in $O(n^{2k-\epsilon})$ time for any $\epsilon > 0$ under the Strong Exponential Time Hypothesis (SETH), where $n$ is the length of the input string.
We also prove that this problem is \textbf{W[2]}-hard when parameterized by the length of the REWB expression, strengthening the previous \textbf{W[1]}-hardness result.
Moreover, we prove that this problem for $2$-use $2$-REWBs cannot be solved in $n^{1+o(1)}$ time unless the triangle detection problem can be solved in that time.

On the algorithmic side, we present an $O(n \log^2 n)$-time algorithm for $1$-use REWBs.
In particular, we focus on the \emph{ABCBD problem}, 
which is the REWB matching problem for the form $A\,(B)_x\,C\,\backslash x\,D$
where $A$, $B$, $C$, and $D$ are fixed REGEXes.
We also show that every $1$-use REWB can be transformed into this canonical form.
Our algorithm significantly improves upon the recent $O(n^2)$-time algorithm for the ABCBD problem by Nogami and Terauchi (MFCS, 2025).
Our algorithm is highly non-trivial and employs several techniques including suffix trees, transition monoids of REGEXes, factorization forest data structures, and periodicity of strings.
\end{abstract}




\newpage

\input{onerewb/introduction.tex}

\input{onerewb/tools.tex}

\input{onerewb/hardness.tex}

\input{onerewb/sanitize.tex}

\input{onerewb/normalization.tex}

\input{onerewb/xyyz.tex}


 









\input{onerewb/uezato-overview.tex}





\input{onerewb/new-right.tex}

\input{onerewb/new-right-ds.tex}

\input{onerewb/leftenumeration.tex}

\input{onerewb/new-near-transitions.tex}

\bibliographystyle{abbrvurl}
\bibliography{onerewb/bib}

\newpage
\appendix

\input{onerewb/ultimately-periodic.tex}

\input{onerewb/new-normalizing-1useREWB}

\input{onerewb/hidden-constant.tex}


\end{document}

%% file: onerewb/introduction.tex

\newcommand{\JwithA}[1]{J^{A}_{#1}}

\section{Introduction}

\newcommand{\fixO}[2]{\ensuremath{O_{#1}(#2)}}


\paragraph{Background.}
Regular expressions are widely used in web services to validate and filter user-provided text.
However, depending on the engine implementation and the specific expressions used,
the worst-case regular expression matching time can be prohibitive.
In most web-service settings, expressions are fixed by the application,
while the input string is provided by potentially untrusted users.
In such cases,
an adversary can craft inputs that trigger the worst-case behavior, potentially exhausting computational resources (CPU and memory) and causing service outages.

Attacks that exploit such worst-case algorithmic behavior are known as
\emph{algorithmic complexity attacks}~\cite{Crosby:2003:1}.
Those specifically involving regular expressions are called \textbf{\redos}{} (Regular Expression Denial of Service)~\cite{Crosby:2003:2}.
\redos{} is a major vulnerability and is cataloged in the Common Weakness Enumeration (CWE, the de facto standard for vulnerability classification) as CWE-1333~\cite{CWE1333}.

\revise{R1-1}{We highlight two well-documented service outages caused by \redos,
together with the real-world regular-expression that triggered the
worst-case behavior in the deployed engines.}%
\footnote{%
\revise{R1-1}{
These two examples are written in the syntax of the corresponding real-world
regex engines, not in the formal syntax introduced below.
For example, \texttt{\textbackslash s} is a whitespace character class and
\texttt{\textbackslash u200c} is the Unicode escape
for U+200C.
These examples are used only to demonstrate that backtracking engines can already cause quadratic-time \redos{} without backreferences.}
}
\begin{enumerate}
\item The 2016 Stack Overflow outage, caused by the expression \verb@^[\s\u200c]+|[\s\u200c]+$@~\cite{stackstatus:2016}.
\item The 2019 Cloudflare outage, caused by an expression whose core part is \verb@.*.*=.*@~\cite{Cloudflare:2019}.
\end{enumerate}
According to their reports~\cite{stackstatus:2016,Cloudflare:2019},
these outages occurred because the employed regular expression engines relied on \emph{backtracking} algorithms.
Consequently, the above expressions exhibited quadratic-time computation on crafted inputs.\footnote{
Readers can reproduce the quadratic-time behavior of these expressions using the online demos in~\cite{playstackstatus:2016,playCloudflare:2019}.}
It is well known that backtracking-based engines can cause catastrophic \emph{exponential} backtracking~~\cite{Jeffrey:2006,Cox:2007,Goyvaerts:2021},
and it is widely believed that avoiding exponential behavior would be sufficient.
However, these incidents were surprising in showing that even quadratic-time behavior can lead to a system outage. This is stated explicitly in the post-mortem~\cite{stackstatus:2016}, which we quote verbatim below:
\begin{quote}
\itshape
So the Regex engine has to perform a “character belongs to a certain character class” check (plus some additional things) 20,000+19,999+19,998+…+3+2+1 = 199,990,000 [sic] times, and that takes a while. This is not classic catastrophic backtracking (\textbf{performance is $O(n^2)$, not exponential, in length}), but it was enough.
\end{quote}

Motivated by this background, we propose the following criterion:
\begin{quote}
\textbf{A regular expression engine is \redos-\emph{safe} if, for each fixed expression, it runs in linear or near-linear time in the length of the input string.}
\end{quote}

In this paper, we address the following fundamental questions: {\itshape
Can we construct a \redos-safe regular expression engine in the above sense?
If so, which practical extensions can such an engine still support?}

\paragraph{Complexity of \regex{} matching problems.}
%
The term ``regular expression'' is used for several different formalisms.
In this paper, following classical formal language theory,
we use the standard (Kleene) regular expressions, denoted by \regex{}.
\begin{definition}[Syntax of \regex{}]
The syntax of \regex{} over an alphabet $\Sigma$ is given by the grammar:
\[
    E ::= \emptyset \mid \epsilon \mid \sigma \mid E + E \mid E \cdot E \mid E^* \qquad \text{where $\sigma\in\Sigma$}.
\]
\end{definition}
We give the formal semantics, including the language $\lang{\cdot}$, in Section~\ref{section:tools}.
For \regex{}, we consider the following computational problem.
\begin{itembox}[l]{\regex{} Matching Problem}
\noindent{}\textbf{Fixed Object}: A \regex{} $E$. \\
\noindent{}\textbf{Input}: A string $w$. \\
\noindent{}\textbf{Task}: Deciding if $w \in \lang{E}$ where $\lang{E}$ is the language of $E$.
\end{itembox}

%
%
%
Is it possible to design a \redos-safe engine for the \regex{} matching problem?
The answer is \textbf{Yes}.
Using Thompson's construction~\cite{Thompson:1968}
and the standard efficient NFA simulation~\cite{Cox:2007}\cite[Sec~3.7]{Aho:1986},
we can solve this problem in $O(|E|\,|w|)$ time,
which is linear in $|w|$ for each fixed $E$.
%



\paragraph{From \regex{} to \rewb.}
Does the linear-time algorithm for \regex{} matching fully resolve \redos{} \emph{in practice}?
The answer is \textbf{No}.
This is because most real-world regular expression engines support extensions beyond \regex{}. 
One of the most prominent extensions is backreferences.
We use \rewb{} to denote the class of \emph{regular expressions with backreferences}.
Unlike \regex, \rewb{} allows \emph{variables} that capture substrings matched by a subexpression and refer back to that captured string later.
Backreferences are widely supported in the standard regular expression libraries of major programming languages (e.g., C\texttt{++}, Java, JavaScript, Python, Ruby, Perl, PHP, and OCaml).
The syntax of \rewb{} is given as follows.
\begin{definition}[Syntax of \rewb{}]
The syntax of \rewb{} is given by the following grammar:
\[
    E ::= \emptyset \mid \epsilon \mid \sigma \mid E + E \mid E \cdot E \mid E^*  \mid (E)_x \mid \readv{x},
\]
where $\sigma \in \Sigma$ and $x$ is a variable.
The expression $(E)_x$ captures the substring matched by $E$ and stores it in $x$, while $\readv{x}$ refers to the content stored in $x$.
\end{definition}
\noindent{}As with \regex{}, we give the formal semantics of \rewb{} in Section~\ref{section:tools}.

We note that \rewb{} is strictly more expressive than \regex.
Consider the expression $E := ((0+1)^*)_x\:\readv{x}$.
Intuitively, the subexpression $((0+1)^*)_x$ matches an arbitrary string $w \in \set{0,1}^*$ and binds it to $x$.
Then, $\readv{x}$ matches the captured substring stored in $x$.
Hence, $E$ defines the non-regular language $ \set{ w w : w \in \set{0, 1}^* }$.

\smallskip

This paper focuses on the following problem.
\begin{itembox}[l]{\rewb{} Matching Problem}
\noindent{}\textbf{Fixed Object}: A \rewb{} $E$. \\
\noindent{}\textbf{Input}: A string $w$. \\
\noindent{}\textbf{Task}: Deciding if $w \in \lang{E}$.
\end{itembox}


This raises a natural question:
Is it possible to design a \redos-safe engine for the \rewb{} matching problem?
Unfortunately, \textbf{No}.
If we view $E$ as part of the input and parameterize the problem by $c$, the number of variable occurrences in $E$,
then the \rewb{} matching problem is \wone-hard~\cite{Stephan:2012,Fernau:2016}.
In particular, under the standard assumption $\mathbf{FPT}\neq \wone$,
there is no uniform algorithm running in $f(c) \cdot |w|^{O(1)}$ for all inputs $(w, E)$.
%
However, this worst-case hardness is too pessimistic for practice.
In practice, many real-world \rewb{} expressions use backreferences in restricted ways.
To identify which \rewb{} subclasses are \redos-safe and which remain vulnerable,
we need a fine-grained analysis of the \rewb{} matching problem.

\paragraph{Classification of \rewb{}.}

We write $k$-\rewb{} to denote the set of \rewb{s} that have at most $k$ distinct variables.
\revise{R2-1}{
We say that an expression $E$ is $r$-use if, for every input string $w$ and
for all computations of $E$ on $w$, the total number of executed backreferences
(i.e., executed occurrences of $\readv{x}$) is at most $r$.
Intuitively, a computation is one possible run of $E$ on $w$;
different runs may arise from different choices at the union and Kleene-star
subexpressions of $E$, and may therefore execute different numbers of backreferences.
}
If no such finite $r$ exists, we call $E$ \emph{$\omega$-use}.
We illustrate these notions with the following examples:
\begin{itemize}
    \item $(\Sigma^*)_x \readv{x}$ is a $1$-use $1$-\rewb{} because there is only one variable $x$, which is referenced once.
    \item $(\Sigma^*)_x\: \readv{x} \: \readv{x}$ is
    a $2$-use $1$-\rewb{} because there is only one variable $x$, which is referenced twice.
    \item $(\Sigma^*)_x\: \readv{x} \: (\Sigma^*)_x\: \readv{x}$ is
    also a $2$-use $1$-\rewb{} that reuses the single variable $x$.
    \item $(\Sigma^*)_x \: (\readv{x})^*$ is
    an \emph{$\omega$-use} $1$-\rewb{}, as $x$ can be referenced an unbounded number of times.
    \item $(\Sigma^*)_x\ (\Sigma^*)_y \readv{y} \readv{x}$ is
    a $2$-use $2$-\rewb{} because it contains two variables $x$ and $y$, which are referenced a total of two times.
\end{itemize}
\revise{R2-2}{
We allow both capture groups and backreferences to appear inside capture groups.
That is, the body of a capture group $(\cdots)_x$ may contain subexpressions
of the form $\readv{y}$ or $(\cdots)_y$, where $y$ may be equal to $x$.
}
For example, $(\readv{x}\,\readv{x})_x$ (``doubling the content of $x$'') and $(0\,\readv{x}\:0)_x$ (``wrapping $x$ with $0$s'') are well-formed.

\paragraph{Known algorithmic results for \rewb{} matching.}

The current best-known algorithm for $k$-\rewb{} matching runs in $O(|E| \cdot |w|^{1+2k})$ time~\cite[\S IX.~C]{Davis:2021}\cite{Nogami:2025} based on standard dynamic programming.
Recently, Nogami and Terauchi studied the \emph{ABCBD problem}, a subclass of the $1$-\rewb{} matching problem, and gave an $O(|E|^2|w|^2)$-time algorithm~\cite[Thm.~5]{Nogami:2025}.
Although their result improved upon the standard DP-based $O(|E|\,|w|^3)$-time algorithm,
it is not \redos-safe due to its quadratic dependence on $|w|$.

\subsection{Our Contribution}

We summarize our main hardness and algorithmic results.

\subsubsection{Hardness Results for the \rewb{} Matching Problem} 


We strengthen the previously known $\wone$-hardness results~\cite{Stephan:2012,Fernau:2016} by showing  \textbf{W[2]}-hardness.

\begin{restatable}{theorem}{Hardnesswtwo}\label{thm:hardness_w2}%
When the expression $E$ is part of the input,
the \rewb{} matching problem is $\mathbf{W[2]}$-hard parameterized by the expression size $|E|$.
\end{restatable}

Assuming that Triangle Detection on an $m$-edge graph does not admit an $m^{1+\delta-o(1)}$-time algorithm~\cite[Conjecture~3]{Abboud:2014:focs},
we derive a conditional lower bound even for the highly restricted class, \emph{$2$-use} 2-\rewb{}.

\begin{restatable}{theorem}{Hardnesstworewb}\label{thm:hardness_2rewb}%
Assuming the no-almost-linear-time hypothesis for Triangle Detection,
there exists a \emph{fixed} $2$-use $2$-\rewb{} expression $\Psi$ such that the matching problem for $\Psi$ cannot be solved in $|w|^{1+\delta-o(1)}$ time for some constant $\delta > 0$.
\end{restatable}

Under the \emph{$k$-Orthogonal Vectors hypothesis} (\textbf{$k$-OV})~\cite{Williams:2019}, which is a weaker assumption than the \emph{strong exponential time hypothesis} (\textbf{\seth{}}), we have the following result.

\begin{restatable}{theorem}{HardnessBykOV}\label{thm:REWB hardness by kOV}%
Assuming the \textbf{$k$-OV} hypothesis,
for each fixed integer $k \geq 1$,
there exists a fixed $\omega$-use $k$-\rewb{} expression $\Psi_k$ such that
the matching problem for $\Psi_k$ cannot be solved in $O(|w|^{2k - \epsilon})$ time for any $\epsilon > 0$.
\end{restatable}

%
%
Very recently and independently of our work,
Nogami and Terauchi~\cite[Theorem~1]{Nogami:2026:arxiv}  proved the following lower bound:
assuming the $2$-\textbf{OV} hypothesis, no algorithm solves the $1$-\rewb{} matching problem in $O(|w|^{2-\epsilon}\,\text{poly}(|E|))$ time.
Our result strengthens theirs in the following two ways:
(1) our hardness holds for general $k$; and
(2) as our hardness holds for a fixed expression $\Psi_k$, it rules out any speedup gained by allowing arbitrary dependence on $|E|$.



\subsubsection{Algorithmic Results}

We have shown that the \rewb{} matching problem remains hard even for $2$-use $2$-\rewb{s} (Theorem~\ref{thm:hardness_2rewb})
and for ($\omega$-use) $1$-\rewb{s} (Theorem~\ref{thm:REWB hardness by kOV} with $k=1$)
when seeking a near-linear-time algorithm.
This naturally raises the question: What is the complexity of $1$-use $1$-\rewb{} matching?
Our main algorithmic result is a near-linear time algorithm for the following special case:

\begin{itembox}[l]{ABCBD Problem}
\noindent{}\textbf{Fixed Objects}: \regex{es} $A$, $B$, $C$, and $D$. \\
\noindent{}\textbf{Input}: A string $w$. \\
\noindent{}\textbf{Task}: Deciding if we can decompose $w$ into $w = w_A\, w_B\, w_C\, w_B\, w_D$ so that $w_A\in \lang{A}$, $w_B\in \lang{B}$, $w_C\in \lang{C}$, and $w_D\in \lang{D}$.
Equivalently, decide whether $w \in \lang{A\,(B)_x\,C\,\readv{x}\,D}$.
\end{itembox}

At first glance, this syntax may appear highly restricted.
However, the ABCBD problem captures the general case of $1$-use \rewb{} in the following sense:
\begin{restatable}{theorem}{Singleuse}\label{theorem:abcbd-singleuse}
For any 1-use \rewb{} $E$,
there exists an equivalent \rewb{} of the following form:
\[
(A_1 \, (B_1)_{x} \, C_1 \, \readv{x} \, D_1) +
(A_2 \, (B_2)_{x} \, C_2 \, \readv{x} \, D_2) + \cdots +
(A_{K} \, (B_{K})_{x}\,  C_{K} \readv{x} \, D_{K}) + E'
\]
where $A_i$, $B_i$, $C_i$, $D_i$, $E'$ are \regex{es}. 
We can construct the above representation in $O(|E|^4)$ time.
Furthermore, the following holds for their size:
\[
K = O(|E|^2),\ 
|A_i| = O(|E|^2),\ 
|B_i| = O(|E|),\ 
|C_i| = O(|E|^2),\ 
|D_i| = O(|E|),\ 
|E'| = O(|E|^2).
\]
\end{restatable}

The main result of this paper is as follows.
\begin{restatable}{theorem}{ABCBD}\label{thm:ABCBD}
The ABCBD problem admits an $O(|w|\log^2 |w|)$-time algorithm. The hidden constant is $2^{O((|A|+|B|+|C|+|D|)^2)}$.
\end{restatable}
This significantly improves the $O(|w|^2)$-time algorithm of Nogami and Terauchi~\cite{Nogami:2025}.
Our algorithm is highly nontrivial; for a more detailed overview, see Section~\ref{sec:short_tech_overview}. 
%
Note that our algorithm can also be viewed as a \emph{near-linear-time fixed-parameter} algorithm for the ABCBD matching problem, parameterized by the (combined) size of \regex{es} $A$, $B$, $C$, and $D$.
Combining Theorems~\ref{theorem:abcbd-singleuse}~and~\ref{thm:ABCBD}, we obtain the following.
\begin{restatable}{theorem}{1use1rewb}\label{thm:oneuseonerewb}
The $1$-use \rewb{} matching problem admits an $O(|w|\log^2 |w|)$-time algorithm. The hidden constant is $2^{O(\,|E|^4\,)}$.
\end{restatable}

Note that the hidden constant $2^{O(\,|E|^4\,)}$ is larger than in Theorem~\ref{thm:ABCBD} because Theorem~\ref{theorem:abcbd-singleuse} may yield expressions of quadratic length.

\paragraph{Remark: Uniformity}
\revise{R2-3}{
The algorithms in Theorems~\ref{thm:ABCBD} and~\ref{thm:oneuseonerewb}
are uniform.
Namely, they take the relevant expressions
($A,B,C,D$ in Theorem~\ref{thm:ABCBD},
and $E$ in Theorem~\ref{thm:oneuseonerewb})
together with the input string $w$,
perform all preprocessing depending on these expressions, and then run the main procedures.
The fixed-object formulation only emphasizes the dependence on $|w|$;
no non-uniform advice depending on the fixed expressions is assumed.
}





\medskip
As a key subroutine for solving the ABCBD problem,
we have also developed an efficient algorithm for the following problem.

\begin{itembox}[l]{$XY^{\alpha}Z$ Problem}
\noindent{}\textbf{Fixed Objects}: An integer $\alpha\geq 2$ and \regex{es} $X$, $Y$, and $Z$. \\
\noindent{}\textbf{Input}: A string $w$. \\
\noindent{}\textbf{Task}: Deciding if we can decompose $w$ into $w = w_X\, \smash{\overbrace{w_Y\, \cdots\,  w_Y}^{\alpha \text{ times}}}\, w_Z$ so that $w_X\in \lang{X}$, $w_Y\in \lang{Y}$, and $w_Z\in \lang{Z}$.
Equivalently, decide whether $w \in \lang{X\,(Y)_x\, (\readv{x})^{\alpha-1}\,Z}$.
\end{itembox}

Our theorem for this problem is the following.
\begin{restatable}{theorem}{XYYZ}\label{thm:XYYYYYZ}
The $XY^\alpha Z$ problem admits an $O(|w|\log |w|)$-time algorithm.
The hidden constant is $2^{O(s^2)} + \alpha \cdot 2^{O(s)}$, where $s := |X|+|Y|+|Z|$.
Alternatively, at the cost of an additional $\log|w|$ factor (i.e., time $O(|w| \log^2 |w|)$),
we can achieve a polynomial dependence on $s + \alpha$.
\end{restatable}

\subsubsection{Open Problems}
Several problems remain to be investigated.
The first direction is to specify classes of \rewb{} that admit near-linear-time algorithms.
We have shown that $1$-use \rewb{s} admit a near-linear-time algorithm, whereas $2$-use $2$-\rewb{s} and $\omega$-use $1$-\rewb{s} are unlikely to admit such algorithms.
However, the matching complexity of $r$-use $1$-\rewb{} for any constant $r \ge 2$ remains open.

The second direction is to improve our algorithm for the $ABCBD$ problem.
The most direct question is whether we can eliminate the remaining logarithmic factor and obtain a truly linear-time algorithm: reducing $O(|w| \log^2 |w|)$ to $O(|w|)$.
Moreover, when \rewb{} $E$ is given as input, 
our algorithm has an exponential dependency on the expression length $|E|$.
Specifically, it is open whether it is possible to reduce the dependence on $|E|$ to polynomial while preserving near-linear dependence on $|w|$.
Note that, if we allow quadratic dependence on $|w|$, Nogami and Terauchi~\cite{Nogami:2025} already achieves $O(|E|^2 |w|^2)$-time.
Furthermore, as we describe in Section~\ref{sec:short_tech_overview}, our algorithm consists of many building blocks and is therefore quite complex.
Simplifying it is a meaningful task both for advancing theoretical understanding and for improving practical implementability.

The third direction is to broaden the scope.
Real-world regular expression engines support many features beyond backreferences, such as lookaround~\cite{Jeffrey:2006,Goyvaerts:lookaround},
greedy and lazy Kleene stars~\cite{Jeffrey:2006,Goyvaerts:kleenestar},
and atomic grouping~\cite{Jeffrey:2006,Goyvaerts:atomic}.
%
Therefore, investigating the complexity of matching regular expressions augmented with such features is important in practice.
We believe that these problems are also theoretically interesting, in the same way that the \rewb{} matching problem admits deep algorithmic study, as demonstrated in this paper.


\input{onerewb/relatedwork}
\subsection{Technical Overview of Theorem~\ref{thm:ABCBD}}\label{sec:short_tech_overview}

\input{onerewb/to_zu.tex} 

\subsection{Organization}

The rest of this paper is organized as follows.
%
In Section~\ref{section:tools}, we introduce the necessary preliminaries and algorithmic tools we use in this paper, while we give some proofs in Appendix~\ref{appendix:closure-of-UP} for the sake of self-containment.
%
In Section~\ref{section:hardness}, we establish our hardness results by proving Theorems~\ref{thm:hardness_w2},~\ref{thm:hardness_2rewb},~and~\ref{thm:REWB hardness by kOV}.
As discussed in Section~\ref{sec:short_tech_overview}, Sections~\ref{section:sanitize}--\ref{section:near-transitions2} concentrate on proving Theorem~\ref{thm:ABCBD}.
Theorem~\ref{thm:XYYYYYZ} is proved in Section~\ref{section:XYYZ problem} as a corollary.
In Appendix~\ref{appendix:1use1rewb-to-abcbd}, we prove Theorem~\ref{theorem:abcbd-singleuse}.
In Appendix~\ref{sec:hidden_constant}, we analyze the hidden constants of Theorems~\ref{thm:ABCBD},~\ref{thm:oneuseonerewb},~and~\ref{thm:XYYYYYZ}.

%% file: onerewb/relatedwork.tex
\subsection{Related Work}\label{section:related_work}

\subsubsection{\redos}

Regular Expression Denial of Service (\redos) has been recognized as an important instance of \emph{algorithmic complexity attacks}~\cite{Crosby:2003:1} since at least 2003~\cite{Crosby:2003:2}.
Outages at Stack Overflow (2016)~\cite{stackstatus:2016} and Cloudflare (2019)~\cite{Cloudflare:2019} are representative examples, and notably, both incidents were triggered by quadratic behavior of regular expression engines, highlighting that even \textbf{polynomial} complexity can result in severe real-world failures.
\redos{} is a common type of {\textbf DoS} attack that is officially cataloged as CWE-1333: Inefficient Regular Expression Complexity~\cite{CWE1333}.
Vulnerabilities continue to be discovered and reported in practice;
high-risk examples include \cite{CVE-2021-21240,CVE-2023-39663,CVE-2023-6159,CVE-2024-26142,CVE-2024-28865} from 2021 to 2024.
In 2025, notable entries in the CVE (Common Vulnerabilities and Exposures, which maintains a public record of disclosed security issues) include \cite{CVE-2025-25200,CVE-2025-29907,CVE-2025-5197}.

At its core, \redos{} arises from backtracking-based regex engines, which may exhibit catastrophic (possibly exponential) backtracking~\cite{Goyvaerts:2021}.
This behavior has long been known to automata theorists, and influential books~\cite{Jeffrey:2006} helped disseminate the problem more widely to practitioners.
For classical \regex{} matching, Cox~\cite{Cox:2007,Cox:2010} recommended that practitioners prefer Thompson NFA-based implementations over backtracking-based engines, as the former run in $O(|E|\,|w|)$ time.
In contrast, it is unlikely that \rewb{} matching admits such a ``silver bullet'' solution like the Thompson construction, as this problem is \textbf{NP}-hard in general~\cite{Aho:1991,BerglundM23}.
This has motivated complexity-theoretic research into the \rewb{} matching problem, including the identification of tractable classes~\cite{FreydenbergerS19,Nogami:2025,Schmid24}.

As an example of \redos{} involving expressions with backreferences, we focus on \texttt{CVE-2017-16114}~\cite{CVE-2017-16114} and \texttt{CVE-2019-25103}~\cite{CVE-2019-25103}.
Both cases are triggered by an expression that can be abstracted as $E := (a^*)_x\ b^*\,((a+b)^*\,b)\,b^*\ \readv{x}$, on which backtracking-based engines take cubic time in the worst case.
In \texttt{CVE-2019-25103}, the corresponding patch~\cite{simple-markdown-redos-patch}
rewrote $E$ to $E_1 := (a^*)_x\ ((a+b)^*\,b)\ \readv{x}$; it modified only the ``\regex{}-part'' without touching the ``backreference-part''.
Consequently,
although
the \redos{} vulnerability is mitigated,
potential vulnerabilities remain, as even the patched version exhibits quadratic time complexity; this issue seems to be an inevitable overhead of backtracking-based engines, as they try all possible prefixes that the variable $x$ captures.
In \texttt{CVE-2017-16114}, the corresponding patch~\cite{marked-redos-patch}
rewrote $E$ to $E_2 := (a^*)_x\,\bigl(b+(b(a+b)b)\bigr)^*\,\readv{x}$.
Although $E_2$ admits linear-time matching, it alters the accepting language: for example, $E$ and $E_1$ accept $a\,ab\,a$ but $E_2$ does not.
In contrast, the original expression belongs to our $ABCBD$ problem; thus, it can be evaluated in almost linear time by our algorithm without modifying the semantics.

\subsubsection{Hardness of Regular Language Matching Problems}

%
\paragraph{On \regex{}.}
While Thompson's construction~\cite{Thompson:1968} solves \regex{} matching in $O(\,|E||w|\,)$ time,
Backurs and Indyk~{\cite[Thm.3]{Backurs:2016} proved that $O(\,(|E||w|)^{1-\epsilon}\,)$ time is unachievable under \seth{}.

Recently, Bringmann et al.~introduced the NFA Acceptance Hypothesis~\cite{Bringmann:2024}:
intuitively, for a dense NFA $A$,
the membership problem $w \in \lang{A}$ cannot be solved in $O(\,(|A| |w|)^{1-\epsilon}\,)$ time for any $\epsilon > 0$.
Bille and G{\o}rtz~\cite{Bille:2024} introduced a parameter called \emph{density} $\Delta$ (the total number of active states during the simulation on the input $w$)
and proved that matching for sparse NFAs (or \regex{es}) also cannot be solved in $O(\Delta^{1- \epsilon})$ time under \seth.
These results suggest that improving upon the textbook $O(|E| |w|)$-time algorithm for \regex{} matching is  unlikely at present.


%

\paragraph{On \rewb{}.}
The \rewb{} matching problem is known to be \textbf{NP}-complete~\cite[Thm~6.2]{Aho:1991}
when the \rewb{} expression is a part of the input, even over a unary alphabet~\cite{BerglundM23}.
Moreover, when parameterized by the number of variables $k$, the $k$-\rewb{} matching problem is \wone-hard~\cite[Thm.~1]{Fernau:2016}.
%
In fact, \wone-hardness already holds for a more restricted fragment of \rewb{} called \emph{pattern languages}~\cite[Thm.~1]{Fernau:2016}\cite[Thm.~1]{Stephan:2012}.
We briefly recall pattern languages below to relate them to \rewb{} and summarize known algorithmic results.


\subsubsection{Algorithms for Regular Expression Matching Problems}

\paragraph{On \rewb{}.}
The current best-known algorithm for the $k$-\rewb{} matching problem is based on dynamic programming and runs in $O(|E| \cdot |w|^{1+2k})$ time~\cite[\S IX.\,C]{Davis:2021}.
Applying this algorithm to the ABCBD problem yields an $O(|E|\cdot |w|^3)$-time algorithm.
Nogami and Terauchi gave an $O(|E|^2 |w|^2)$-time algorithm, and thereby improved the dependence on $|w|$ from cubic to quadratic.
As noted earlier, our Theorem~\ref{thm:ABCBD} further improves the dependence on $|w|$ to near-linear.

Freydenberger and Schmid~\cite{FreydenbergerS19} considered the matching problem for the class of \emph{deterministic} \rewb{s}, which is a subclass of \rewb{s} that can deterministically match strings when reading from the left. They gave an $O(k|w|)$-time algorithm for $k$-\rewb{}.
Schmid~\cite{Schmid24} extended this result to the wider class of \emph{memory-deterministic} \rewb{s} and gave an $O(|E|^c|w|)$-time algorithm for some constant $c$.
They also refined the time complexity of the \rewb{} matching problem to $O(|E|\cdot |w|^{O(\alpha)})$ for the cases of small $\alpha$, where $\alpha\,(\leq k)$ denotes the \emph{active variable degree}, which represents the maximum number of variables that are ``active'' at the same time.

\vspace{-5pt}
\paragraph{Pattern Language and Word Equation.}
%

A \emph{pattern language} is a class of languages generated by \emph{patterns}, that is, strings that may also contain variables that can be matched to strings.
For example, the pattern $P := 01 X 10 X$ over $\Sigma = \set{ 0, 1 }$ with a variable $X$ generates the language $\set{\,01 w 10 w : w \in \Sigma^* }$.
This concept was introduced by Angluin in the late 1970s in the context of learning theory and has since been studied extensively~\cite{Angluin:1979,Angluin:1980,Mateescu:1997,Shinohara:1983}.

Pattern languages can be seen as a very restricted subclass of \rewb. Specifically, a pattern language can be expressed as a \rewb{} that:
(1) uses a capture group $(\ldots)_x$ for the first occurrence of each variable $x$,
(2) uses a backreference $\readv{x}$ for all subsequent occurrences, and
(3) only allows concatenation as the outermost operator (no union or Kleene star over captures/backreferences).
For example, the above pattern $P$ equals the \rewb{} $0 1 (\Sigma^*)_X 1 0 \readv{X}$.



The matching problem of pattern languages is NP-complete~\cite{Angluin:1979,Angluin:1980,Ehrenfeucht:1979,Jiang:1994} even for the case with alphabet size $2$, each variable occurs at most twice, and each variable can be replaced with a string of length in $\{0,1\}$ or $\{1,2,3\}$~\cite{FernauS15}.
Ibarra et al.~gave an $O\left(\frac{(|P|+|w|)|w|^{k-1}}{(k-1)!}\right)$-time algorithm for the membership problem of $k$-variable pattern languages, where $P$ is the given pattern expression~\cite{Ibarra:1995}.
In particular, it runs in linear time when $k=1$.
Fernau et al.~\cite{Fernau:2015} gave algorithms for several classes of pattern languages with good properties. Particularly, they showed that for the case of \emph{non-crossing} (i.e., scopes of different variables do not overlap) pattern languages, the matching problem admits an $O(m |w| \log |w|)$-time algorithm, where $m$ is the number of one-variable blocks occurring in a given pattern~\cite[Thm~8]{Fernau:2015}\cite[Thm~4.16]{Fernau:2020}.
Observe that our $XY^{\alpha}Z$ problem generalizes the case $m=1$ by allowing arbitrary \regex{es} for $X$, $Y$, and $Z$.

Given two pattern languages (possibly sharing variables), the \emph{word equation} problem asks whether there is a string that matches both of them.
The case that one of the languages contains no variables is equal to the pattern language matching problem.
In general, the word equation problem is known to be in PSPACE~\cite{Plandowsk:2004}.
For the $1$-variable case,
D\k{a}browski and Plandowski~\cite{Dabrowski:2011} gave an $O(m \: + \: \#_X \log m)$-time algorithm where $m$ is the total size of given patterns and $\#_X$ is the total number of variable occurrences. 
Later, in the RAM-model setting,
Je\.{z}~\cite{Jez:2016:onevariable} gave an $O(m)$-time algorithm using the \emph{recompression} technique introduced by Je\.{z}~\cite{Jez:2016}.
D\k{a}browski and Plandowski~\cite{Dabrowski:2004} gave an $O(m^5)$-time algorithm for the $2$-variable case.
In contrast, the problem becomes significantly harder for $k \geq 3$ variables; even for the $3$-variable case, it is not even known whether the problem is in NP~\cite{Jez:2020}.
Schulz~\cite{Schulz90} considered a variant of the word equation problem that each variable has a \regex{} constraint.
They proved this variant is decidable;
it was later proved to be PSPACE-complete by Diekert et al.~\cite{DiekertGH05}.

\paragraph{On \regex{} with Lookaround.}
\emph{Lookaround} is one of the features that real-world regular expression engines support, which checks whether a specified pattern exists around the current position, without consuming characters.
It is folklore that the language class of \regex{} equipped with lookaround is regular\footnote{We can convert it to a \emph{two-way alternating finite automaton} (2AFA)~\cite{Chandra:1981}, whose language class is regular~\cite{Ladner:1984}.}; thus, it admits $O(|w|)$ time matching by Thompson's construction when we only consider the dependence on $|w|$.
Therefore, research concerns the dependence on the expression length $|E|$, and particularly, focuses on $O(|E|\cdot |w|)$ time algorithms.

Mamouras and Chattopadhyay~\cite{Mamouras:2024} provided an $O(|E|\cdot |w|)$ time membership algorithm.
Chattopadhyay~et~al.~\cite{Chattopadhyay:2025} mechanized an algorithm of the same complexity based on~\cite{Mamouras:2024} and verified its correctness using the proof assistant \emph{Rocq} (formerly known as Coq).
Barri\`{e}re and Pit-Claudel~\cite{Barriere:2024} provided a matching algorithm of the same complexity that is faithful to JavaScript \regex{} matching semantics.
Fujinami~and~Hasuo \cite{Fujinami:2024} provided a matching algorithm of the same complexity, which also supports \emph{atomic} grouping, which is a practical extension that does not enlarge the language expressiveness~\cite{Jeffrey:2006, Goyvaerts:atomic}.
Uezato~\cite{uezato:2024} proved that the membership problem of \rewb{} equipped with \emph{lookahead} (lookaround that only examines after the current position) is PSPACE-complete.
}

\paragraph{Gapped Patterns.}
\revise{R1-3}{
Another well-studied regular-language formalism, motivated by practical pattern matching applications (rather than \rewb),
is that of \emph{gapped patterns}.
A typical gapped pattern $G$ has the form
\[
  G \equiv p_0\,\Sigma^{[\alpha_1,\,\beta_1]}\,p_1
  \Sigma^{[\alpha_2,\,\beta_2]}\,\cdots\,
  \Sigma^{[\alpha_m,\,\beta_m]}\,p_m,
\]
where each $p_i \in \Sigma^*$ is a fixed string and
$\Sigma^{[\alpha,\beta]} := \set{ w \in \Sigma^* : \alpha \leq |w| \leq \beta }$ is the set of substring whose length lies between $\alpha$ and $\beta$.
Thus every gapped pattern denotes a regular language and can be expressed as a
classical \regex.
Matching such patterns, including dictionary variants in which a
finite set of gapped patterns is matched against a text, has been studied
extensively; see, e.g.,
\cite{NavarroRaffinot03,FredrikssonGrabowski08,BilleGVW12,HaapasaloSSSS11,HonLSTTY18}
and the survey~\cite{LevyS22}.
Particularly relevant from a fine-grained perspective is the work of Amir~et~al.~\cite{AmirKLPPS19} on dictionary matching with one gap.
As the uniformly bounded case,
a dictionary $\mathcal{D} = \set{ G_1, G_2, \ldots, G_n}$ consist of patterns $G_i$ forms $p^i_1\,\Sigma^{[\alpha,\,\beta]}\,p^i_2$ with common bounds $\alpha, \beta$.
The task is to report their occurrences for a input text $w$.
Under the 3SUM conjecture~\cite{GajentaanO95,Patrascu10,KopelowitzPP16},
for suitable dictionaries,
they proved a lower bound of
$\Omega(|w|(\beta-\alpha)^{1-o(1)} + \mathit{op})$
where $\mathit{op}$ is the output size~\cite[Thm.~5 and Table~2]{AmirKLPPS19}.
Complementing this, they also gave an online reporting algorithm whose running time depends on structural parameters of dictionaries,
in particular the degeneracy of the associated bipartite graph~\cite[Thm.~9]{AmirKLPPS19}.
}

\newcommand{\poslk}{?_=}
\newcommand{\neglk}{?_!}
\newcommand{\reglk}{\textsc{RegLK}}

%% file: onerewb/to_zu.tex
We provide an overview of our algorithm for Theorem~\ref{thm:ABCBD},
which finds a decomposition $w = w_A\,w_B\,w_C\,w_B\,w_D$ with $w_\xi \in \lang{\xi}$ for $\xi \in \set{A, B, C, D}$.
The following figure summarizes the overall structure.

\begin{center}
{\tikzset{>=latex}
\hspace{-30pt}
\scalebox{0.95}{
\begin{tikzpicture}
\node[draw, rounded corners, inner sep=4pt] {
\begin{forest} for tree={align=center,edge=->}
[ABCBD Problem~\S\ref{section:sanitize}
    [ABCBD Problem (Normalized)~\S\ref{section:two_subproblems}
        [XYYZ Problem~\S\ref{section:XYYZ problem}
            [Maximal Local Power \\ Enumeration~\S\ref{section:MR-decomposition}]
            [Ultimately Periodic \\ Constraint \\ Solving~\S\ref{section:solving semi-linear constraints}]
        ]  
        [Branching ABCBD Problem~\S\ref{section:branchingABCBD-to-twosubproblems}
            [\textsc{Right-$B$} \\ \textsc{Transition} \\ \textsc{Enumeration}~\S\ref{section:solve:rightbenum}
                [\textsc{Right Enumeration}\\ Data Structure~\S\ref{sec:data_structures_right}]
            ]
            [\textsc{Left-$B$} \\ \textsc{Transition} \\ \textsc{Enumeration}~\S\ref{new:section:left-B-overview}
                [\textsc{Left Enumeration}\\ Data Structure~\S\ref{section:implementing-left-B}
                    [Auxiliary\\ Data Structure~\S\ref{sec:auxdatastructure}]
                    [Near Transition\\ Enumeration~\S\ref{section:near-transitions2}]
                ]
            ]
        ]
    ]
]
\end{forest}
};
\end{tikzpicture}
}}
\end{center}

To streamline our algorithm,
we first normalize the input $w$ and expressions so that $\epsilon \not\in \lang{D}$ and the alphabet is binary $\Sigma = \set{0, 1}$ in Section~\ref{section:sanitize}.
Then, we can assume $w_D \neq \epsilon$.
%
In Section~\ref{section:two_subproblems}, we reduce the ABCBD problem to two subproblems: the \emph{XYYZ problem} (Section~\ref{section:XYYZ problem}) and the \emph{branching ABCBD problem} (Sections~\ref{section:branchingABCBD-to-twosubproblems}--\ref{section:near-transitions2}).
In the former problem, we consider the case $w_C = \epsilon$, i.e., $w = w_A\,w_B\ w_B\,w_D$. 
In the latter problem, we consider the case where
two suffixes $w_C\,w_B\,w_D$ and $w_D$ satisfy a \emph{branching condition}; that is, their first letters are distinct.
To show that this decomposition covers all cases,
we rely on the closure properties of \regex{} and the equivalence with \emph{nondeterministic finite automata} (NFAs), as detailed in Section~\ref{subsection:ABCBD to XYYZ and Branching ABCBD}.

The key structure of the former (XYYZ) problem (our $X\,Y^{\alpha}\,Z$ problem with $\alpha = 2$, Section~\ref{section:XYYZ problem}) is that, since $w_C = \epsilon$, two occurrences of $w_B$ appear consecutively in $w$ as $w = \ldots w_B\,w_B \ldots$.
We crucially exploit the periodicity of $w$ induced by this immediate repetition $w_B w_B$.
We use Crochemore's maximal local power enumeration~\cite[Sec.~9.2]{Crochemore:2007}\cite{Crochemore:1981} (Section~\ref{section:MR-decomposition}), in combination with the periodic properties of \regex{} (Section~\ref{section:solving semi-linear constraints}).
This part also yields Theorem~\ref{thm:XYYYYYZ}.

To solve the latter (branching ABCBD) problem (Sections~\ref{section:branchingABCBD-to-twosubproblems}--\ref{section:near-transitions2}), we work with the \emph{suffix tree} of $w$, $\stree_w$, or simply $\stree$ (see Section~\ref{section:tools} for the definition).
Each leaf $\ell$ has an index $i_{\ell}$ of $w$ and corresponds to the suffix $w[i_{\ell}\,\btw\,|w|]$
($w = w[1]\,w[2]\,\ldots,w[|w|]$.)
Each inner node $v$ represents the substring given by the path-label from the root.
For each $v$, we maintain the set $\mathscr{I}_v$ of all indices $i_{\ell}$ of leaves $\ell$ in the subtree rooted at $v$.
We note that our suffix tree $\stree$ is binary because $\Sigma = \set{0, 1}$.

The key property of the problem is the branching condition.
It means that
the two paths from the root to $\ell_{w_B w_C w_B w_D}$ and to $\ell_{w_B w_D}$ \emph{branch} at the node $v$ corresponding to $w_B$ (as illustrated in the below figure).
Using this property, it suffices to consider the following task:\\
$(\clubsuit)$: {\itshape For each node $v$, find a pair of suffixes branching at $v$ that induces a feasible decomposition.}

\begin{wrapfigure}{l}{0.31\textwidth}%
\hspace{-10pt}
  \centering
  \begin{tikzpicture}[
      >=Latex,
      node distance=5mm,
      every node/.style={font=\small},
      edge/.style={-Latex, line width=0.4pt},
      inner/.style={inner sep=1.4pt},
      leaf/.style={draw, rounded corners, inner sep=2.2pt},      
    ]

    \node[inner] (root) {root};
    \node[inner, below=of root] (x) {$v$};
    \node[leaf, below left=of x] (xy) {$w_B\,w_D$};
    \node[leaf, below right=of x] (xz) {$w_B\,w_C\,w_B\,w_D$};

    \draw[edge] (root) -- node[right, inner sep=1pt] {$w_B$} (x);
    \draw[edge] (x) -- node[left, inner sep=1pt, yshift=2pt, xshift=-2pt] {$w_D$} (xy);
    \draw[edge] (x) -- node[right, inner sep=1pt, yshift=2pt, xshift=1pt] {$w_C\,w_B\,w_D$} (xz);

    \node[overlay] at (root) [xshift=50pt, yshift=0pt] {\scriptsize \begin{tabular}{c}Illustration \\ of $(\clubsuit)$
    \end{tabular}};
  \end{tikzpicture}
  \vspace{-5pt} 
\end{wrapfigure}

Let $\mathscr{H}_v := \mathscr{I}_u$, where $u$ is a child of $v$ maximizing the size $|\mathscr{I}_u|$, and define
$\mathscr{L}_v := \mathscr{I}_v \setminus \mathscr{H}_v$.
We employ a classical \emph{small-to-large}
(or heavy-light decomposition) strategy:
for each $i \in \mathscr{L}_v$
we determine whether there exists
$j \in \mathscr{H}_v$ such that the pair $(i,j)$ induces a feasible decomposition.
%
Since $\sum_{v : \text{node}} |\mathscr{L}_v| = O(|w| \log |w|)$,
the remaining task of $(\clubsuit)$ is to find such an index $j \in \mathscr{H}_v$ in $O(\log |w|)$ time for each $i \in \mathscr{L}_v$.
We then split the task $(\clubsuit)$ into the two cases \emph{Right-$B$} and \emph{Left-$B$} in Section~\ref{section:branchingABCBD-to-twosubproblems}.

In \emph{Right-$B$}, for each $i \in \mathscr{L}_v$,
we consider the case $i < j$: i.e., the shorter suffix $w_B w_D$ starts at $j$.
Namely, we search for $j \in \mathscr{H}_v$ such that $w[j \btw |w|] = w_B w_D$, illustrated as follows:
\[
(\clubsuit_R) \quad 
    w=\underbrace{w[1 \btw i)}_{w_A}\, \underbrace{w[i \btw i+|w_B|)}_{w_B}\, \underbrace{w[i+|w_B| \btw j)}_{w_C}\, \underbrace{w[j \btw j+|w_B|)}_{w_B}\, \underbrace{w[j+|w_B| \btw |w|]}_{w_D}.
\]
In \emph{Left-$B$}, we consider the case $j < i$ where $j$ is the starting position of the longer suffix $w_B\,w_C\,w_B\,w_D$.
We search for $j \in \mathscr{H}_v$ such that $w[j \btw |w|] = w_B w_C w_B w_D$, illustrated as follows:
\[
(\clubsuit_L) \quad
    w=\underbrace{w[1 \btw j)}_{w_A}\, \underbrace{w[j \btw j+|w_B|)}_{w_B}\, \underbrace{w[j+|w_B| \btw i)}_{w_C}\, \underbrace{w[i \btw i+|w_B|)}_{w_B}\, \underbrace{w[i+|w_B| \btw |w|]}_{w_D}.
\]

To handle both cases, we use the same high-level approach: traverse $\stree$ bottom-up and maintain the sets $\mathscr{H}_v$ using suitable data structures.
Since the concrete algorithms and data structures differ between the two cases, we present them separately.

\paragraph{Right-$B$ Case (Sections~\ref{section:solve:rightbenum}~and~\ref{sec:data_structures_right}).}
Let $v$ be a node of $\stree$ and
$w_B$ be the string represented by $v$.
Let $i\in \light_v$, where $w[i \btw |w|]$ corresponds to $w_B w_C w_B w_D$. 
We assume $w[1 \btw i) \in \lang{A}$ and $w_B \in \lang{B}$ because otherwise there is no feasible solution.
Our task $(\clubsuit_R)$ is to find an index $j \in \heavy_v$ such that $w[j \btw |w|]$ corresponds to $w_B w_D$.
To treat the constraint $w_C\in \lang{C}$,
we manage $\heavy_v$ using a data structure that answers whether there is an index $j\in \heavy_v$ with $i < j$
such that $w[i+|w_B|\,\btw\,j)\in \lang{C}$.
We call this structure \rightenumds{}~(Section~\ref{sec:data_structures_right}).

However, this is not yet sufficient, as we must also enforce the constraint $w_D \in \lang{D}$.
Crucially, whether a fixed index $j$ satisfies this constraint depends on the current node $v$.
This is because the substring corresponding to $w_D$ begins at $j + |w_B|$, and this starting position shifts as the length $|w_B|$ varies with $v$.
Consequently, we must handle the constraints for both $C$ and $D$ simultaneously. To address this, we partition $\mathcal{H}_v$ into subsets such that all indices in the same subset share the same validity regarding $D$ (Section~\ref{section:solve:rightbenum}).
Leveraging the fact that \regex{} admits a finite transition monoid (see Section~\ref{section:tools} for definition),
we can ensure that we need to maintain only a constant number of subsets.
By maintaining a separate instance of \rightenumds{} for each subset and merging them using a disjoint set union data structure, we can resolve the issue and handle this case.
A simple way of constructing \rightenumds{} is to use a \emph{segment tree}, where each node is equipped with the set of elements of transition monoids (Section~\ref{sec:data_structures_right}).
This yields an $O(|w| \log^3 |w|)$-time algorithm.
In our algorithm, employing a \emph{factorization forest} (see Section~\ref{section:tools} for definition),
we further eliminate one logarithmic factor and achieve $O(|w| \log^2 |w|)$ time.

\vspace{-10pt}
\paragraph{Left-$B$ Case (Sections~\ref{new:section:left-B-overview}--\ref{section:near-transitions2}).}

Let $v$ be a node of $\stree$ and $w_B$ be the substring represented by $v$.
Let $i \in \light_v$, where $w[i \btw |w|]$ corresponds to $w_B w_D$. 
We assume $w_B \in \lang{B}$ and $w_D \in \lang{D}$ because otherwise there is no feasible solution.
Our task $(\clubsuit_L)$ is to find an index $j \in \heavy_v$ with $j < i$ such that $w[j \btw |w|]$ corresponds to $w_B w_C w_B w_D$.
Moreover, we can handle the constraint for $A$ simply by managing only indices $j\in \heavy_v$ with $w[1\btw j)\in \lang{A}$
(recall that, unlike the Right-$B$ case, it is independent of $v$).
The remaining task is to find $j$ satisfying $j + |w_B| < i$ 
and $w[j+|w_B|\,\btw\,i)\in \lang{C}$.
We treat it as a data structure problem (Section~\ref{section:implementing-left-B}).
We note that \rightenumds{} is insufficient for our purpose,
since the constraint $j + |w_B| < i$ depends on $|w_B|$,
which varies with the current node $v$.

To address the problem, we consider the following two subcases separately:
\[
    \text{\textbf{near-$j$:}}~~\text{$j + |w_B| \leq i < j+2|w_B|$}, \qquad
    \text{\textbf{far-$j$:}}~~\text{$j+2|w_B| \leq i$}.
\]


In the near-$j$ case (Section~\ref{section:near-transitions2}), we derive a periodicity for $w_B$ because, for any two indices $j,j'$ in the near-$j$ case, the substrings $w[j\,\btw\,j+|w_B|)$ and $w[j' \,\btw\, j'+|w_B|)$ overlap.
With the periodicity inherent to \regex{}, we can efficiently solve the near-$j$ case in $O(\log |w|)$ time.
On the other hand, in the far-$j$ case, the index $j$ is ``sufficiently far'' from $i$.
Technically, the condition $j + 2|w_B| \leq i$ implies that
the distance between $j+|w_B|$ and $i$ is at least $|w_B|$; we use this margin for the design of our data structure.

To solve the far-$j$ case, we adopt the following strategy.
We first build a data structure \auxds{} (Section~\ref{sec:auxdatastructure}), 
that maintains intervals $[j \btw k)$ as the item $[j,\,k)$ while ensuring $j + |w_B| \leq k$.
Given $i$, it answers whether there is an interval $[j, k)$ such that $w[j+|w_B|\,\btw\,i) \in \lang{C}$ and $k \leq i$.
For the condition $j + 2|w_B| \leq i$,
it suffices to ensure $k \leq j+2|w_B|$ holds at any moment. 
%
%
The issue is that as we traverse $\stree$, the node $v$ changes (thus, $|w_B|$ decreases), so the condition $k \leq j + 2|w_B|$ will be violated.
To resolve this issue, when the violation occurs
(the margin $k - (j+|w_B|)$ becomes at least $|w_B|$),
we ``re-insert'' the new interval $[j,\,k')$ to the data structure, where $k' := j + |w_B|$.
The total overhead cost is $O(\log |w|)$ time because ``re-insertion'' happens only after $|w_B|$ is halved;
this is ensured by the ``far-$j$'' condition.

The remaining task is to construct \auxds{} (Section~\ref{sec:auxdatastructure}).
This is similar to \rightenumds{}, but requires additional mechanisms to quickly answer queries of $w[j+|w_B|\,\btw\,i) \in \lang{C}$,
which depends on the current $v$ (and thus, $w_B$).
We adopt a \emph{lazy update} strategy for this purpose.
As in \rightenumds{}, using the segment tree leads to the time complexity $O(|w|\log ^3 |w|)$, which can be accelerated to $O(|w|\log^2 |w|)$ using the factorization forest.

%% file: onerewb/tools.tex
\newcommand{\integers}{\mathbb{Z}}
\newcommand{\naturals}{\mathbb{Z}_{\geq 0}}

\section{Tools and Preliminaries}\label{section:tools}


\paragraph{Integers.}

We write $\integers$ for the set of all integers and
write $\naturals$ for $\set{ i \in \integers : i \geq 0 }$.
For an integer $n \geq 1$, we write $[n]$ to denote the set $\set{1, 2, \ldots, n}$.

\paragraph{Strings.}
Let $\Sigma$ be a finite alphabet.
We write $\Sigma^*$ for the set of all finite strings over $\Sigma$ and $\epsilon$ for the empty string.
For $w \in \Sigma^*$, $|w|$ is its length.
We use 1-based indexing: for an index $i \in \set{1, \ldots, |w|}$, $w[i]$ is the $i$-th character of $w$.

For indices $i, j \in \set{1, \ldots, |w|}$ with $i \leq j$, $w[i \btw j] := w[i]\,w[i+1] \ldots w[j]$ is the substring from $i$ to $j$, inclusive.
We define $w[i \btw j] := \epsilon$ if $i > j$ for convenience.
For example, if $w = abcd$, then $|w| = 4$, $w[1] = a$, $w[4] = d$, $w[2 \btw 2] = b$, $w[2 \btw 4] = bcd$, and $w[2 \btw 1] = \epsilon$.
We also write $w[i \btw j)$ for the shorthand for $w[i\,\btw\,j-1]$. It is convenient to extract the $\ell$-length substring starting from $i$.
For example, if $w = abcd$, then $w[1 \btw 1+3) = abc$, $w[2 \btw 2+3) = bcd$, and $w[1 \btw 1+0) = \epsilon$.
We similarly define $w(i\btw j]$ for $w[i+1\btw j]$.

For a string $u$ and an integer $k \in \naturals$, we write $u^k$ for the string obtained by $k$-power of $u$.
For example, $(abaa)^2 = abaa\,abaa$.
We naturally extend this notation to $u^{r/|u|}$ with a rational $r/|u| \geq 0$.
For example, $(abaa)^{\frac{3}{4}} = aba$ and $(abaa)^{2+\frac{2}{4}} = abaa\,abaa\,ab$.

If $w[i] = w[p+i]$ for all $i \in \set{1, 2, \ldots, |w|-p}$, then we say that $w$ has a \emph{period} $p$.
Alternatively, if a prefix $u$ of $w$ satisfies $w = u^r$ for some $r \geq 1$, $w$ also has a period $|u|$.
For example, $w = abaa\,abaa\,ab$ has three periods:
(i) $4$ for $w = (abaa)^{2+\frac{2}{4}}$;
(ii) $8$ for $w = (abaaabaa)^{1+\frac{2}{8}}$;
(iii) $10 = |w|$ for $w = w^1$.

The following theorem is well-known and quite useful property on periodic strings.
It is sometimes called Fine and Wilf's theorem~\cite{FineWilf:1965}.
\begin{lemma}[{\cite[Cor.~6.1]{Choffrut:1997}}]\label{lemma:FineWilf}
Let $w$ be a non-empty word.
If $w$ has two periods $p$ and $q$ and $|w| \geq p + q - \mathit{gcd}(p, q)$,
then $w$ also has a period $\mathit{gcd}(p, q)$.
\end{lemma}

\paragraph{\regex: Classical Regular Expressions.}

Let $\Sigma$ be a fixed finite alphabet.

The calculus of (classical) regular expressions, \regex, is defined by the following grammar:
\[
E ::= \emptyset \mid \epsilon \mid \sigma \mid E_1 + E_2 \mid E_1 \cdot E_2 \mid E^* .
\]
where $\sigma \in \Sigma$.


The language of \regex{}, $\lang{\cdot}$, is defined as follows:
\[
\begin{array}{lcl}
\lang{\emptyset} & = & \emptyset, \\
\lang{e} & = & \set{ e }  \quad e \in \set{\epsilon, \sigma}, \\
\lang{ E_1 + E_2 } & = & \lang{E_1} \cup \lang{E_2}, \\
\lang{ E_1 \cdot E_2} & = & \lang{E_1} \cdot \lang{E_2}, \qquad (\text{\small where $X \cdot Y = \set{ x y : x \in X, y \in Y}$}) \\
\lang{ E^* } & = & \bigcup_{i \geq 0} (\lang{E})^i
\end{array}
\]
where $L^0 = \set{ \epsilon }$ and $L^{k+1} = L^k \cdot L$ for $k \geq 0$; i.e., $L^k$ is the $k$-th power of $L$.

We define the size of a regular expression $E$, denoted $|E|$, as the total number of occurrences of operators (such as union, concatenation) and operands (alphabet symbols, $\epsilon$, and $\emptyset$) that appear in $E$.
The length of the string representation of $E$ is then clearly bounded by $O(|E|)$.

\paragraph{\rewb: Regular Expressions with Backreferences.}
The calculus of extended regular expressions \rewb{} is an extension of \regex{} with \emph{backreferences}.
Let $\mathcal{V}$ be a finite set of \emph{variables}.
Then, the grammar of \rewb{} over $\mathcal{V}$ is the following:
\[
E ::= \emptyset \mid \epsilon \mid \sigma \mid E_1 + E_2 \mid E_1 \cdot E_2 \mid E^* \mid ( E )_x \mid \readv{x}
\]
where $x \in \mathcal{V}$.
We write $k$-\rewb{} to denote the set of \rewb{s} that have $k$ variables.

We define the semantics of \rewb{} via the semantic function $\sem{E}_{\lambda}$ and then define the language.
For a \rewb{} $E$ and a valuation $\lambda$ from $\mathcal{V}$ to $\Sigma^*$,
$\sem{E}_\lambda$ returns a set of pairs $\tuple{w, \lambda'}$ consisting of a matched string and an updated valuation:
\[
\begin{array}{lcl}
\sem{ \epsilon }_{\lambda} & := & \set{\,\tuple{\epsilon, \lambda}\,}, \\
\sem{ \sigma }_{\lambda} & := & \set{\,\tuple{\sigma, \lambda}\,}, \\[3pt]
\sem{ \emptyset }_{\lambda} & := & \emptyset, \\[3pt]
\sem{ E_1 \cdot E_2 }_{\lambda} & := & \set{ \tuple{w_1 w_2, \nu} : \tuple{w_1, \mu} \in \sem{E_1}_{\lambda}, \tuple{w_2, \nu} \in \sem{ E_2 }_{\mu} }, \\[3pt]
\sem{ E_1 + E_2 }_{\lambda} & := & \sem{ E_1 }_{\lambda} \cup \sem{ E_2 }_{\lambda}, \\[3pt]
\sem{ E^* }_{\lambda} & := & \bigcup_{k\ge 0} \sem{E^k}_\lambda, \qquad \text{where}\ E^0 := \epsilon,\; E^{k+1}:=E^k E, \\[3pt]
\sem{\,(E)_x\,}_{\lambda} & := & \set{ \tuple{w, \mu[x := w]} : \tuple{w, \mu} \in \sem{E}_{\lambda} }, \\[3pt]
\sem{ \readv{x} }_{\lambda} & := & \set{ \tuple{\lambda(x), \lambda} }.
\end{array}
\]
This function is well-founded for the height of Kleene stars.
This definition follows the formalization in~\cite{uezato:2024}.

Let $\iota$ be the initialized valuation that maps every variable to the empty string $\epsilon$: i.e., $\iota(v) = \emptyset$ for all $v \in \mathcal{V}$.
We simply write $\sem{E}$ to denote $\sem{E}_{\iota}$.

We define the language of a \rewb{} $E$, $\lang{E}$, as follows:
\[
\lang{E} := \set{ w : \tuple{w, \lambda} \in \sem{E} }.
\]

If every computation generated by an expression $E$ reads from variables at most $R$ times in total,
we say $E$ is \emph{$R$-use}.


\smallskip
\noindent{}\emph{\rewb{} Example.}
For example, the following expression $E_1$ is a 1-\rewb{} because there only exists a variable $x$:
\[
E_1 = (\Sigma^*)_x \ \# \ \Sigma^* \ \readv{x} \ \Sigma^* .
\]
The language of $E_1$ is $\lang{E_1} = \set{ w_{\text{pat}} \# w_1 w_{\text{pat}} w_2 : w_{\text{pat}}, w_1, w_2 \in \Sigma^* }$, and it is called a \emph{pattern searching language} because, if $w_{\text{pat}} \# w$ is matched with $E_1$, then $w_{\text{pat}}$ appears in $w$.
The above $E_1$ is $1$-use because every computation read $x$ once by $\readv{x}$.

\newcommand{\Epsilon}{\mathcal{E}}

\paragraph{$\epsilon$-NFA: Nondeterministic finite automata with $\epsilon$-transitions.}
An $\epsilon$-NFA $A$ is a tuple $A = (Q, \Sigma, \Delta, \Epsilon, q_0, F)$ where
$Q$ is a finite set of states,
$\Sigma$ is an input alphabet,
$\Delta \subseteq Q \times \Sigma \times Q$ is a set of transition rules, 
$\Epsilon \subseteq Q \times Q$ is a set of $\epsilon$ transitions,
$q_0 \in Q$ is the initial state,
and $F \subseteq Q$ is a set of final (accepting) states.
We write $p \xrightarrow{\sigma} q$ for each transition element $(p, \sigma, q)$ of $\Delta$ and $p \xrightarrow{\epsilon} q$ for $(p, q) \in \mathcal{E}$.
A rule $p \xrightarrow{\sigma} q$ means that we can move to $q$ from $p$ by consuming a letter $\sigma$. 
We write $p \xRightarrow{\epsilon} q$ if we can reach $q$ from $p$ only using $\epsilon$-transitions $p \xrightarrow{\epsilon} \cdots \xrightarrow{\epsilon} q$.
We also write $p \xRightarrow{\sigma} q$ if we can move to $q$ from $p$ by only consuming a letter $\sigma$: i.e., we have a sequence of transitions of the form
$p \xRightarrow{\epsilon} p' \xrightarrow{\sigma} p''  \xRightarrow{\epsilon} q$.

If $\Epsilon = \emptyset$, we call $A$ an $\epsilon$-free NFA.
We write $\Delta(p, \sigma)$ to denote the set of states reachable from $p$ by $\sigma$: i.e., $\Delta(p, \sigma) :=  \set{ q : p \xrightarrow{\sigma} q }$.
If $\Epsilon = \emptyset$ and
$|\Delta(p, \sigma)| = 1$ for all $(p, \sigma) \in Q \times \Sigma$,
we call $A$ a \emph{deterministic} automaton (DFA).

The language of $A$, $\lang{A}$, is the set of words accepted by $A$ and formally defined as follows:
\[
\lang{A} := \set{ \sigma_1 \sigma_2 \ldots \sigma_n : q_0 \xRightarrow{\sigma_1} q_1 \xRightarrow{\sigma_2} \cdots \xRightarrow{\sigma_n} q_n, \ q_n \in F }.
\]
If $q_0 \xRightarrow{\epsilon} q_F$, the empty string $\epsilon$ belongs to $\lang{A}$.

By the well-known Thompson's construction,
every regular expression $E$ can be translated in linear time into a sparse $\epsilon$-NFA $N_E$.
\begin{proposition}[Thompson's construction{~\cite{Thompson:1968}}{\cite[Sec.~5]{Sakarovitch:2009}}]\label{prop:Thompson}
Let $E$ be a \regex{} over an alphabet $\Sigma$.
We can translate $E$ to a corresponding sparse $\epsilon$-NFA $N_E = (Q, \Sigma, \Delta, \Epsilon, q_0, F)$ in $O(|E|)$-time where $\lang{E} = \lang{N_E}$ and $|Q| = |\Delta| = |\Epsilon| = O(|E|)$.
\end{proposition}

\begin{proposition}[Product Construction of $\epsilon$-NFA]
\label{prop:epsilon-nfa:product}
Let
$A = (Q_A, \Sigma, \Delta_A, \Epsilon_A, q^A_0, F_A)$ and
$B = (Q_B, \Sigma, \Delta_B, \Epsilon_B, q^B_0, F_B)$ be $\epsilon$-NFAs, where $\Delta_A, \Epsilon_A, \Delta_B, \Epsilon_B$ are represented by adjacency lists.
Then, there exists an $\epsilon$-NFA $C$ such that
$\lang{C} = \lang{A} \cap \lang{B}$.
In particular, we can construct $C$ in $O(|Q_A| \cdot |Q_B| + |\Delta_A| \cdot |\Delta_B| + |\Epsilon_A|\cdot|Q_B|+ |Q_A|\cdot|\Epsilon_B|)$ time.
\end{proposition}
%
\begin{proof}
This follows from the classical product construction for DFA and ($\epsilon$-free) NFA (c.f. {\cite[Sec.~5]{Sakarovitch:2009}}).
For completeness, we define the construction for our product $\epsilon$-NFA $C = (Q_C, \Sigma, \Delta_C, \Epsilon_C, q^C_0, F_C)$ as follows:
\[
\begin{array}{lcl}
Q_C & := & Q_A \times Q_B, \\
\Delta_C & := & \set{ (p, q) \xrightarrow{\sigma} (p', q') : p \xrightarrow{\sigma} p',\ q \xrightarrow{\sigma} q' }, \\
\Epsilon_C & := & \set{ (p, q) \xrightarrow{\epsilon} (p', q) : p \xrightarrow{\epsilon} p' } \cup
\set{ (p, q) \xrightarrow{\epsilon} (p, q') : q \xrightarrow{\epsilon} q' }, \\
q^C_0 & := & (q^A_0,\ q^B_0), \\
F_C & := & F_A \times F_B.
\end{array}
\]
By the construction, $\lang{C} = \lang{A} \cap \lang{B}$ trivially holds.
Note that
adding $\set{ (p, q) \xrightarrow{\epsilon} (p', q') : p \xrightarrow{\epsilon} p',\ q \xrightarrow{\epsilon} q' }$ to $\Epsilon_C$
does not change $\lang{C}$, since such a move is simulated by two $\epsilon$-steps.

Regarding the time complexity, we can construct $Q_C$ in $O(|Q_A| \cdot |Q_B|)$ time.
Since the transitions are represented by adjacency lists,
we can construct $\Delta_C$ in $O(|\Delta_A| \cdot |\Delta_B|)$ time by iterating over matching transitions.
Similarly, we can build $\Epsilon_C$ in $O(|\Epsilon_A|\cdot|Q_B| + |\Epsilon_B|\cdot|Q_A|)$ time.
\end{proof}

\begin{proposition}[Removal of the $\epsilon$ word from $\epsilon$-NFA's language]
\label{prop:epsilon-nfa:epsilon-remove}
Let
$A = (Q_A, \Sigma, \Delta_A, \Epsilon_A, q^A_0, F_A)$ be an $\epsilon$-NFA.
Then, there exists an $\epsilon$-NFA $B$ such that $\lang{B} = \lang{A} \setminus \set{\,\epsilon\,}$.
In particular, we can construct $B$ in $O(|Q_A| + |\Delta_A| + |\Epsilon_A|)$ time.
\end{proposition}
\begin{proof}
We construct an $\epsilon$-NFA $B = (Q_B, \Sigma, \Delta_B, \Epsilon_B, q^B_0, F_B)$ step-by-step.
First, we define the set $Q_B$ as follows:
\[
\begin{array}{lcl}
Q_B & := & \set{ (q, 0),\ (q, 1) : q \in Q_A}, 
\end{array}
\]
We define $\Delta_B$ and $\Epsilon_B$ to satisfy the following condition:
\begin{itemize}
\item A state $(q, 0)$ represents that the automaton has reached $q$ from the initial state by using only $\epsilon$-transitions.
\item A state $(q, 1)$ means that we can reach $q$ from $p$ by using some non-$\epsilon$-transitions.
\end{itemize}
Therefore,
\[
\begin{array}{lcl}
\Delta_B & := & \set{ (p, i) \xrightarrow{\sigma} (q, 1) : p \xrightarrow{\sigma} q,\ i \in \set{0, 1} }, \\
\Epsilon_B & := & \set{ (p, i) \xrightarrow{\epsilon} (q, i) : p \xrightarrow{\epsilon} q,\ i \in \set{0, 1} }.
\end{array}
\]
We define $q^B_0 := (q^A_0, 0)$ and $F^B := F^A \times \set{ 1 }$.

If $w \in \lang{A}$ and $w = \sigma w'$, then we have a decomposition of the following form:
\[
q^A_0 \xRightarrow{\epsilon} q \xrightarrow{\sigma} q' \xRightarrow{w'} r \in F^A.
\]
This computation is simulated in $B$ as follows:
\[
q^B_0 = (q^A_0, 0) \xRightarrow{\epsilon} (q, 0) \xrightarrow{\sigma} (q', 1) \xRightarrow{w'} (r, 1) \in F^B.
\]
By this observation, $\lang{A} \setminus \set{ \epsilon } \subseteq \lang{B}$ clearly holds.
Conversely, if $B$ accepts $w$, then the accepting run must reach some state in $F_A \times \set{ 1 }$,
hence it uses at least one labeled transition; projecting the run to the first component yields an accepting run of $A$ on $w$, so $w \neq \epsilon$ and $w \in \lang{A}$.

Regarding the time complexity, we can construct $Q_B$ in $O(|Q_A|$ time.
Since the transitions are represented by adjacency lists, 
we can construct $\Delta_B$ in $O(|\Delta_A|)$ time and $\Epsilon_B$ in $O(|\Epsilon_A|)$ time.
\end{proof}

\noindent{}\textbf{Remark:}
The standard $\epsilon$-elimination procedure first computes $\epsilon$-closures (the reflexive transitive closure) of $\Epsilon$ and then uses them to redirect labeled transitions and accepting states so as to remove $\epsilon$-transitions.
When $\Epsilon$ is given as adjacency lists,
computing all $\epsilon$-closures by running a BFS graph search from every state takes
$O(|Q|\,(|Q|+|\Epsilon|))$ time in the worst case.
Alternatively, representing $\Epsilon$ as an $|Q| \times |Q|$ Boolean adjacency matrix allows one to compute its transitive closure in $O(|Q|^3)$ time via the Floyd-Warshall algorithm,
or in $O(|Q|^\omega \log|Q|)$ time via fast Boolean matrix multiplication ~\cite{Furman:1970,Munro:1971,Fischer:1971}.
In contrast, Proposition~\ref{prop:epsilon-nfa:epsilon-remove} addresses the weaker task of removing only the empty word $\epsilon$
from the language (without eliminating $\epsilon$-transitions),
and it runs in linear $O(|Q|+|\Delta|+|\Epsilon|)$ time without computing the $\epsilon$-closure.

\paragraph{Ultimately Periodic Set.}

A set $P$ of non-negative integers is \emph{ultimately periodic}
if there is 4-tuple $(\mu, M, \lambda, T)$ where
\begin{itemize}
\item $\mu$ is a non-negative integer and $M \subseteq \set{0, 1, \ldots, \mu-1 }$; and
\item $\lambda$ is a positive integer and $T\subseteq \set{ 0, 1, \ldots, \lambda-1 }$
\end{itemize}
and the following holds:
\[
n \in P \iff
\begin{cases}
n \in M & \text{if } n < \mu, \\
(n \bmod \lambda) \in T & \text{otherwise}.
\end{cases}
\]


The class of ultimately periodic sets enjoys several closure properties.
In this paper, we will use the following ones in~\rm{\cite{Ginsburg:1966}}, for which we give a proof in Appendix~\ref{appendix:closure-of-UP} for self-containedness.
\begin{restatable}{lemma}{PropertiesOfUltimatelyPeriodic}\label{lem:ult_operations}
Let $P$ and $Q$ be two ultimately periodic sets. Then, we have the following.
\begin{itemize}
    \item The union $P \cup Q$ is ultimately periodic.
    \item The Minkowski sum $P + Q := \set{ p + q : p \in P, q \in Q}$ is ultimately periodic.
    \item For any positive integer $e$,
the set $e \cdot P := \set{ e p : p \in P }$ is ultimately periodic.
\end{itemize}
\end{restatable}

\paragraph{Monoid.}

A \emph{monoid} $\mathcal{M}$ is a tuple $(M, \odot, 1_M)$ where $M$ is a base set, $\odot : M \times M \to M$ is an associative binary function --- $(m_1 \odot m_2) \odot m_3 = m_1 \odot (m_2 \odot m_3)$ for any $m_1, m_2, m_3 \in M$ --- and $1_M$ is the identity element --- $m \odot 1_M = 1_M \odot m = m$ for any $m \in M$.
An element $e\in M$ is \emph{idempotent} if $e\comp e = e$ holds. 
By definition, the identity element $1_M$ is idempotent, but the opposite direction does not necessarily hold.

The following objects are typical examples of monoids:
(1) The set of all strings of $\Sigma^*$ forms an infinite monoid $(\Sigma^*, \bullet, \epsilon)$ where$\_\bullet\_$ is concatenation: i.e.,  $w_1 \bullet w_2 = w_1 w_2$.
(2) The set of all mappings $X^X = \set{ f : X \to X }$ on a set $X$ forms a monoid $(X^X, \circ, \text{id}_X)$ where $\circ$ is the function composition and $\text{id}_X$ is the identity function.
(3) The set of non-negative integers with addition $(\naturals, +, 0)$ forms an infinite monoid.

Let $\mathcal{M} = (M, \comp, 1_M)$ be a monoid, $m \in M$ be an element of $M$, and $X, Y \subseteq M$ be subsets of $M$.
We naturally extend the monoid operator $\comp$ to sets as follows:
\[
m \comp X := \set{ m \comp x : x \in X}, \quad
X \comp m := \set{ x \comp m : x \in X},  \quad
X \comp Y := \set{ x \comp y : x \in X, y \in Y}.
\]

Let $\mathcal{M}_1 = (M_1, \comp_1, 1_{M_1})$ and $\mathcal{M}_2 = (M_2, \comp_2, 1_{M_2})$ be monoids.
If a function $\varphi : \mathcal{M}_1 \to \mathcal{M}_2$ has the following property, it is a \emph{monoid homomorphism}:
\[
\varphi(a \comp_1 b) = \varphi(a) \comp_2 \varphi(b).
\]

\noindent{}\emph{Example.}
Let $\mathcal{M_1}$ be the monoid formed by the set of all strings $\Sigma^*$ with concatenation and $\mathcal{M}_2$ be the infinite monoid of the non-negative integers $\naturals$ with addition $+$.
Then, the length function $\ell : \mathcal{M}_1 (= \Sigma^*) \to \mathcal{M}_2 (= \naturals)$ is a monoid homomorphism, as $\ell(w_1 w_2) = \ell(w_1) + \ell(w_2)$ for any string $w_1$ and $w_2$.

\newcommand{\mhom}{\alpha}

\paragraph{Transition Monoids.}
Let $\mathcal{M} = (M, \odot, 1_M)$ be a finite monoid,
$S\,(\subseteq M)$ be a subset of $M$, and $\mhom: \Sigma^* \to \mathcal{M}$ be a monoid homomorphism.
A language $L \subseteq \Sigma^*$ is \emph{recognized} by the triple $(\mathcal{M}, S, \mhom)$ if $L = \alpha^{-1}(S) = \set{ w : \alpha(w) \in S }$.
In this case, we refer to $\mathcal{M}$ as a \emph{transition monoid} of $L$.
\footnote{
Traditionally, a transition monoid is a transformation monoid $T \subseteq X^X$ (under function composition) on some set $X$.
$T$ is typically generated by the transition maps of an automaton (particularly, of a DFA).
Given a recognizing triple $(\mathcal{M}, S, \mhom)$,
we obtain a transformation monoid on $X = M$ via the right action $\rho_w(m) := m \odot \mhom(w)$.
Indeed, $w \in L \iff \alpha(w) \in S \iff \rho_w(1_M) \in S$.
Thus, we use ``transition monoid'' for $\mathcal{M}$ as a mild abuse of terminology.
}
The following classical result connecting regular languages and
recognizability by finite monoids is fundamental in automata theory.
\begin{proposition}[{\cite{Pin:1997,Pin:1995}\cite[Chap.~II]{Sakarovitch:2009}}]
Let $L \subseteq \Sigma^*$ be a language.
The following conditions are equivalent:
\begin{enumerate}
\item $L$ is regular: i.e., there exists a \regex{} $A$ or an $\epsilon$-NFA $A$ such that $L = \lang{A}$.
\item $L$ is recognized by some triple $(\mathcal{M}, S, \mhom)$.
\end{enumerate}
\end{proposition}

For a \regex{} $E$, the standard construction---translating $E$ to an NFA, determinizing it to a DFA, and then taking the transition monoid---yields a transition monoid of size at most $N^N$ with $N = 2^{O(|E|)}$, and hence at most $ 2^{2^{O(|E|)}}$.

The following explicit construction yields a smaller transition monoid of size $2^{O(\,|E|^2\,)}$.
\begin{lemma}[\rm{\cite{Bojanczyk:2010,Holzer:2023}}]
\label{lemma:NFA-transition-monoid}
Let $E$ be a \regex{} over an alphabet $\Sigma$.
There exists a triple $(\mathcal{M}, S, \mhom)$ that recognizes $\lang{E}$.
In particular, the size of $\mathcal{M}$ is $2^{O(|E|^2)}$ and we can construct it in $2^{O(|E|^2)}$ time.
\end{lemma}
\begin{proof}
%
By Proposition~\ref{prop:Thompson}, we can construct an $\epsilon$-NFA $A$ for $E$
with $O(|E|)$ states and transitions in time $O(|E|)$.
Moreover, by computing the $\epsilon$-closure of $A$,
we can construct an $\epsilon$-free NFA $B$ such that $\lang{E} = \lang{B}$ in time polynomial in the size of $A$.
In particular, $A$ and $B$ have the same set of states.

From the $\epsilon$-free NFA $B$,
we construct our Boolean matrix monoid $\mathcal{M}$ as follows.
We write $B = (Q, \Sigma, \Delta, \Epsilon = \emptyset, q_0, F)$.
For each $\sigma\in \Sigma$, let $M^{\sigma}$ be the $|Q| \times |Q|$ Boolean matrix defined by $M^{\sigma}_{p,q} = 1$ iff $(p, \sigma, q) \in \Delta$.
Let $\mathcal{M}$ be the submonoid of $\set{0, 1}^{Q \times Q}$
generated by $\set{ M^{\sigma} : \sigma \in \Sigma }$ under Boolean matrix multiplication --- using the Boolean OR $\vee$ as addition and the Boolean AND $\wedge$ as multiplication --- with the identity element given by the identity matrix $\mathbf{I}$.

Now we define our monoid homomorphism $\alpha : \Sigma^* \to \mathcal{M}$ as follows:
\[
\alpha(\sigma) := M^{\sigma}.
\]
Observe that $w \in \lang{B}$ iff there exists $q_f \in F$
with $(\alpha(w))_{q_0, q_f} = 1$.
Therefore, letting
\[
S= \set{ X \in \mathcal{M} : \exists q \in F.\, X_{q_0, q} = 1 },
\]
the triple $(\mathcal{M}, S, \alpha)$ recognizes $\lang{B}$ and thus $\lang{E}$.

Finally, since $\mathcal{M} \subseteq \set{0,1}^{Q\times Q}$,
we have $|\mathcal{M}|\leq 2^{|Q|^2}$.
Since $|Q|=O(|E|)$, this yields $|\mathcal{M}| = 2^{O(|E|^2)}$,
and $\mathcal M$ can be constructed by
closing $\set{ \mathbf{I} } \cup \set{ M^\sigma : \sigma \in \Sigma }$ under multiplication, in time polynomial in $|\mathcal{M}|$ and $|Q|$.
\end{proof}



\paragraph{Relational Transition Monoid $\transm{E}$ and Associated Homomorphism $\Delta_E$ of \regex{} $E$.}\mbox{}

Let $E$ be a \regex{}, and let $(\mathcal{M}, S, \alpha)$ be the recognizing triple constructed in Lemma~\ref{lemma:NFA-transition-monoid} for $E$.

Hereafter, we refer to the Boolean matrix monoid $\mathcal{M}$ as the \emph{(relational) transition monoid} of $E$
and write $\transm{E}$ to denote $\mathcal{M}$.
Moreover, we write $\Delta_{E}$ to denote the homomorphism $\mhom$ and set $S_E := S$.
With this notation, we have the following:
\[
w \in \mathbb{L}(E) \iff \Delta_E(w) \in S_E.
\]

\paragraph{Factorization Forest.}


Let $E$ be a \regex{}. Consider the following range evaluation problem over a fixed domain.
\begin{enumerate}
\item We receive a domain (string) $w$ and perform some pre-processing for $w$.
\item For each query $(i, j)$ with $1 \leq i < j \leq |w|$, we decide if $w[i \btw j] \in_? \lang{E}$ in $O(1)$ time for fixed $E$.
\end{enumerate}
To answer every query in $O(1)$ time for fixed $E$, we employ \emph{factorization forests}~\cite{Bojanczyk:2020,Colcombet:2021,Simon:1990}.

Let $\Sigma$ be a finite alphabet, $\mathcal{M}$ be a finite monoid, and $\varphi : \Sigma^* \to \mathcal{M}$ be a homomorphism from the set of strings $\Sigma^*$.
Let $w = \sigma_1 \sigma_2 \ldots \sigma_N \in \Sigma^+$ be a non-empty string.
A $\varphi$-factorization forest of $w$ (Ramseyan factorization tree), denoted by $\fforest_\varphi(w)$, is an ordered rooted tree in which a letter is assigned to each leaf, and reading those letters on leaves from left to right yields $w$.
Each node $v$ is annotated by a value $\nu(v) \in \mathcal{M}$ and corresponds to a contiguous factor of $w$.
The annotation $\nu(\cdot)$ satisfies the following condition:
\begin{description}
\item[Leaf] If $v$ is a leaf with letter $\sigma_i$, then $\nu(v) = \varphi(\sigma_i)$.
\item[Binary] If $v$ has two children $c_1$ and $c_2$, then $\nu(v) = c_1 \odot c_2$.
\item[Idempotent] If $v$ has at least three children $c_1, c_2, \ldots, c_m$, then $\nu(c_1) = \nu(c_2) = \cdots = \nu(c_m) = e$ is an idempotent element.
It also holds that $\nu(v) = e$ due to the idempotency of $e$.
\end{description}

For each node $v$ of the factorization forest,
we write $[\ell_v, r_v)$ to denote the interval exactly covered by $v$, $\set{\ell_v, \ell_v + 1, \ldots, r_v - 1}$,
where $\ell_v$ refers to the leftmost index managed by the node $v$ and
$r_v$ refers to the index immediately to the right of the rightmost index managed by the node $v$.

\medskip
\noindent{}\emph{Example of Factorization Forest.}
Consider the following NFA $N$ over $\Sigma = \set{ a, b }$.
\begin{center}
\scalebox{0.8}{
\begin{tikzpicture}[
  ->, >=Stealth,
  node distance=2.2cm,
  stsize/.store in=\stsize, stsize=11mm,
  every state/.style={circle, draw, thick, fill=gray!10, minimum size=\stsize, inner sep=0pt},
  accept ring/.style={draw, double, circle, minimum size=\stsize, inner sep=0pt, fill=none},
  initial text={},
  shorten >=1pt, shorten <=1pt
]

\node[state, initial]   (q0) {$q_0$};
\node[state]            (q1) [right of=q0] {$q_1$};
\node[state, accept ring] (q2) [right of=q1] {$q_2$};
\node[state]            (qd) [right of=q2] {$\bot$};

\path[every loop/.style={looseness=2}]
(q0) edge[loop above] node{$b$} ()
     edge[bend left]  node[auto]{$a$} (q1)
(q1) edge[loop above] node{$b$} ()
     edge[bend left]  node[auto]{$a$} (q2)
(q2) edge[loop above] node{$b$} ()
     edge[bend left]  node[auto]{$a$} (qd)
(qd) edge[loop above] node{$a,b$} ();

\end{tikzpicture}
}
\end{center}

This automaton is actually a DFA and has four states $q_0$, $q_1$, $q_2$, and $\bot$.
This DFA recognizes the language $b^* a b^* a b^*$: i.e., the set of words that contain exactly two occurrences of $a$.

There are two base monoid elements (matrices) $\trans_a$ and $\trans_b$ of the relational transition monoid of $N$, which correspond to the sets of transitions by the letters $a$ and $b$:
\[
\delta_a =
\scalebox{0.8}{$
\begin{pNiceMatrix}[first-row,first-col]
& q_0 & q_1 & q_2 & \bot \\
q_0  & 0 & 1 & 0 & 0 \\
q_1  & 0 & 0 & 1 & 0 \\
q_2  & 0 & 0 & 0 & 1 \\
\bot & 0 & 0 & 0 & 1 
\end{pNiceMatrix}
$}, \qquad
\delta_b =
\scalebox{0.8}{$
\begin{pNiceMatrix}[first-row,first-col]
& q_0 & q_1 & q_2 & \bot \\
q_0  & 1 & 0 & 0 & 0 \\
q_1  & 0 & 1 & 0 & 0 \\
q_2  & 0 & 0 & 1 & 0 \\
\bot & 0 & 0 & 0 & 1
\end{pNiceMatrix}
$}.
\]
By composing these two matrices by using the Boolean matrix multiplication $(\odot)$,
we obtain the following two elements:
\[
\delta_2 =
\scalebox{0.8}{$
\begin{pNiceMatrix}[first-row,first-col]
& q_0 & q_1 & q_2 & \bot \\
q_0  & 0 & 0 & 1 & 0 \\
q_1  & 0 & 0 & 0 & 1 \\
q_2  & 0 & 0 & 0 & 1 \\
\bot & 0 & 0 & 0 & 1 
\end{pNiceMatrix}
$}, \qquad
\delta_3 =
\scalebox{0.8}{$
\begin{pNiceMatrix}[first-row,first-col]
& q_0 & q_1 & q_2 & \bot \\
q_0  & 0 & 0 & 0 & 1 \\
q_1  & 0 & 0 & 0 & 1 \\
q_2  & 0 & 0 & 0 & 1 \\
\bot & 0 & 0 & 0 & 1
\end{pNiceMatrix}
$}.
\]
For example, $\trans_2 = \trans_a \comp \trans_b$ and $\trans_3 = \trans_2 \comp \trans_a$.
There are no other matrices in the transition monoid; therefore,
$\transm{N} = (\set{ \trans_a, \trans_b=\mathbf{1}, \trans_2, \trans_3}, \odot, \mathbf{1})$ where $\mathbf{1}$ is the identity matrix.

Let us consider the following factorization forest over a string $aaababaabaaabbbbb$:
\begin{center}
\scalebox{0.8}{
\input{onerewb/ffzu2.tex}
}
\end{center}
For simplicity, we do not display the values for the leaf nodes. Assume every leaf (i.e., letter $\sigma$) has $\nu(\sigma) = \delta_\sigma$.

Then, we can easily check
(1) that every binary node $n$ has the correct value $\nu(c_1) \comp \nu(c_2)$ using its child $c_1$ and $c_2$;
and also
(2) for the two idempotent nodes, $\delta_3 = \delta_3 \comp \delta_3 \comp \delta_3$ holds and also $\delta_b = \delta_b \comp \delta_b \comp \delta_b \comp \delta_b$ holds.

The (red) idempotent node $v$ labeled with $\delta_3$ covers the interval $[1, 2, \ldots, 12]$. Thus, we write $[\ell_v, r_v) = [1, 13)$.
Similarly, the idempotent node $v'$ labeled with $\delta_b$ covers the interval $[14, 15, 16, 17]$; thus, $[\ell_{v'}, r_{v'}) = [14, 18)$.
Every leaf node $\ell_i$ of the index $i$ exactly covers $[i]$; thus, $[\ell_i, r_i) = [i, i+1)$.

There are multiple factorization forests for a single string.
For every monoid homomorphism $\varphi : \Sigma^* \to \mathcal{M}$ from $\Sigma^*$ to a \textbf{finite} monoid $\mathcal{M}$,
the following result ensures that we can effectively compute a factorization forest $\fforest_\varphi(w)$ for a given string $w$ in linear time $O(|w|)$.
Moreover, the height of the factorization forest does not depend on $w$; indeed, it just depends on $\mathcal{M}$.

\begin{lemma}[Factorization Forest{\cite{Bojanczyk:2012,Colcombet:2021,Kufleitner:2007,Kufleitner:2008,Simon:1990}}]
Let $\mathcal{M}=(M, \comp, 1_{M})$ be a finite monoid and
$\lambda$ be the computational cost of performing $\comp$ once.
For every word $w \in \Sigma^*$, there is a $\fforest_{\varphi}(w)$ of height $\leq 3|M|$, which depends only on $M$, but not on $w$ or $\varphi$.
Moreover, we can compute $\fforest_{\varphi}(w)$ in $O(|M|^3 \cdot \lambda + |M||w|)$ time\footnote{
Although Kufleitner does not mention the detailed time complexity in~\cite{Kufleitner:2007,Kufleitner:2008},
we can obtain the time complexity $O(|M|^3 \cdot \lambda + |M| \cdot |w|)$ as follows.
We first compute the data structure for $\mathcal{J}$-class (Green's relation), which is a basic tool to build factorization forests, in $O(|M|^3 \cdot \lambda)$.
Then we can build $\fforest_\varphi(w)$ in $O(|M| \cdot |w|)$ time using the recursive procedure in~\cite{Kufleitner:2007}.}.
\end{lemma}



We also refer to the following application of factorization forests.
\begin{lemma}[Range Evaluation{\cite{Bojanczyk:2012,Colcombet:2021}}]\label{lemma:constant-eval-by-ff}
%

Let $w \in \Sigma^*$ be a string, $\varphi : \Sigma^* \to \mathcal{M}$ be a monoid homomorphism, and $\fforest_\varphi(w)$ be a factorization forest of $w$ and $\varphi$.
Then, we can compute $\varphi(w[i \btw j])$ in $O(h \cdot \lambda)$-time, where $h$ is the height of $\fforest_\varphi(w)$ and $\lambda$ is the computational complexity of the monoid operation cost of $\mathcal{M}$.
\end{lemma}
\begin{proof}[(Proof Sketch)]
We search nodes in $\fforest_\varphi(w)$ top-down to exactly cover $w[i \btw j]$.
Efficiently using idempotent nodes, we can achieve $O(|M|)$ monoid operations; therefore, we totally need $O(|M| \cdot \lambda)$-time.
For more detailed construction, refer to~\cite[Section~4]{Bojanczyk:2012}~and~\cite{Colcombet:2021}.
\end{proof}

\paragraph{Suffix Tree.}

\newcommand{\nword}[1]{\Gamma(#1)}
\newcommand{\nlen}[1]{\gamma(#1)}
\def\ndepth{\mathit{depth}}
\newcommand{\isSuffix}[1]{\mathit{isSuffix}(#1)}
\newcommand{\Occ}[1]{\mathit{Occ}(#1)}

The \emph{suffix tree}~\cite{Crochemore:2007,Jewels:Book,Gusfield:Book} $\stree_w$ of a string $w$ is an edge-labeled rooted tree such that 
\begin{itemize}
    \item each edge is equipped with a substring of $w$, and
    \item for each index $i\in [|w|]$, there is a node $v$ such that the string obtained by concatenating the strings on the edges on the root-$v$ path (in this order) coincides with the suffix $w[i\btw |w|]$.
\end{itemize}
For a given string $w$, $\stree_w$ can be computed in linear time $O(|w|)$, where labels on edges are represented by intervals.\todo{定義が書いてなかったので書き換えました}

We use the following notation:
\begin{itemize}
\item $\nword{v}$: The path label of $v$: i.e., the string obtained by concatenating the edge labels on the path from the root to $v$.
%
%
\item $\isSuffix{v}$:
A binary flag indicating whether $\nword{v}$ is a suffix of $w$.
Specifically, $\isSuffix{v} = 1$ if and only if $\nword{v}$ is a suffix of $w$.
\end{itemize}
For a suffix node $v$ (i.e., $\isSuffix{v}=1$),
we define $\Occ{v}$ as the starting position of the suffix $\nword{v}$ in $w$: $\Occ{v} := |w| - |\nword{v}| + 1$.
For a node $v$, we write $\stree_w[v]$ for the subtree of $\stree_w$ rooted at $v$.

\smallskip
\noindent{}\textbf{Remark.}
Throughout this paper, we consider suffix trees without appending the special end-marker symbol \$.
In particular, some suffixes may end at internal nodes; we mark such nodes with a Boolean flag $\isSuffix{v}$, which is 1 iff the path label $\nword{v}$ is a suffix of $w$.

\medskip
\noindent{}\emph{Example of Suffix Tree.}
Let us consider the suffix tree for the following string of length 10:
\[
w =
\begin{array}{ll}
\text{index} & 1\,2\,3\,4\,5\,6\,7\,8\,9\,0 \\\hline
\text{letter} & a\,a\,b\,a\,a\,b\,a\,b\,b\,a
\end{array}
\]
\begin{center}
    \input{onerewb/streezu2.tex}
\end{center}

This figure puts an integer $i$ to a suffix node $v$ (i.e., $\isSuffix{v} = 1$) where $i = \Occ{v}$.
As the figure shows, every leaf node is labeled with an integer $i$ where $1 \leq i \leq |w|$.
A leaf labeled with $i$ corresponds to the suffix of $w$ starting at position $i$.
For example, the leaf $\ell_5$ labeled 5 represents the suffix starting from $5$, i.e., $w[5 \btw |w|] = ababba$,
and its path-label is $\nword{\ell_5} = a\,b\,a\,bba$.

In addition to leaves, some inner nodes are also labeled with integers.  
For example, the nodes labeled $9$ and $10$ in the figure correspond to suffixes of $w$.
Let $n_9$ and $n_{10}$ denote these nodes. 
Then $\nword{n_9} = ba$ and $\nword{n_{10}} = a$, both of which are suffixes of $w$.  
It follows that $\Occ{n_9} = 9$ and $\Occ{n_{10}} = 10$.

The subtree $\stree_w[n_{10}]$ contains all suffixes of $w$ that start with the letter $a = \nword{n_{10}}$.
Similarly, the subtree $\stree_w[n_{9}]$ contains all suffixes of $w$ that start with the string $ba = \nword{n_{9}}$

\smallskip
\noindent{}\textbf{Remark:} Unlike the above illustration, every actual edge from a node $v$ to its child $u$, denoted $v \xrightarrow{(l, r)} u$, is labeled a pair of integers $1 \leq l \leq r \leq |w|$.
The pair $(l, r)$ represents the substring $w[l \btw r]$.
This path compression is crucial to build the suffix tree in linear time.
It is noteworthy that on the paths from a node $v$ to two of its children, $u_1$ and $u_2$, 
where $v \xrightarrow{(l_1, r_1)} u_1$ and $v \xrightarrow{(l_2, r_2)} u_w$,
the first characters of the edge labels must be different, i.e., $w[l_1] \neq w[l_2]$ holds.

\paragraph{Heavy-Light Decomposition (of Suffix Tree).}

\newcommand{\subt}[1]{\stree[#1]}
\newcommand{\heavysym}{\sigma_{\!\heavy}}

Let $\stree$ be the suffix tree of a string $w$.
We classify every edge as either heavy or light as follows.  
For each non-leaf node $u \in \stree$ with a left child $\ell$ and a right child $r$:
\begin{itemize}
\item If $|\,\subt{\ell}\,| \geq |\,\subt{r}\,|$ (i.e., the number of nodes in the left subtree $\geq$ the number of nodes in the right subtree),
then the edge  $u \to \ell$ is \emph{heavy} and the edge $u \to r$ is \emph{light}.
In case of equality, the left edge is chosen as heavy by convention.
\item Otherwise, the edge $u \to \ell$ is \emph{light} and the edge $u \to r$ is \emph{heavy}.
\end{itemize}
For each non-leaf node $u \in \stree$ with a single child $v$, the edge $u \to v$ is a \emph{heavy} edge.
A heavy path $\hpath$ is a maximal path such that every edge is heavy.
%
%
%
The collection of all heavy paths forms a partition of the nodes of $\stree$.
We refer to this partition as the \emph{heavy-light decomposition} of $\stree$.

\todo{使ってない定義を消しました（もし使ってたら復活させておいてください）}

\medskip
\noindent{}\emph{Example of Heavy Light Decomposition.}
The left-hand and right-hand sides of the following figure illustrate the heavy-light decomposition of the suffix trees for $aabaababba$ and $babbabbba$, respectively.

\begin{center}
\includegraphics[scale=0.45]{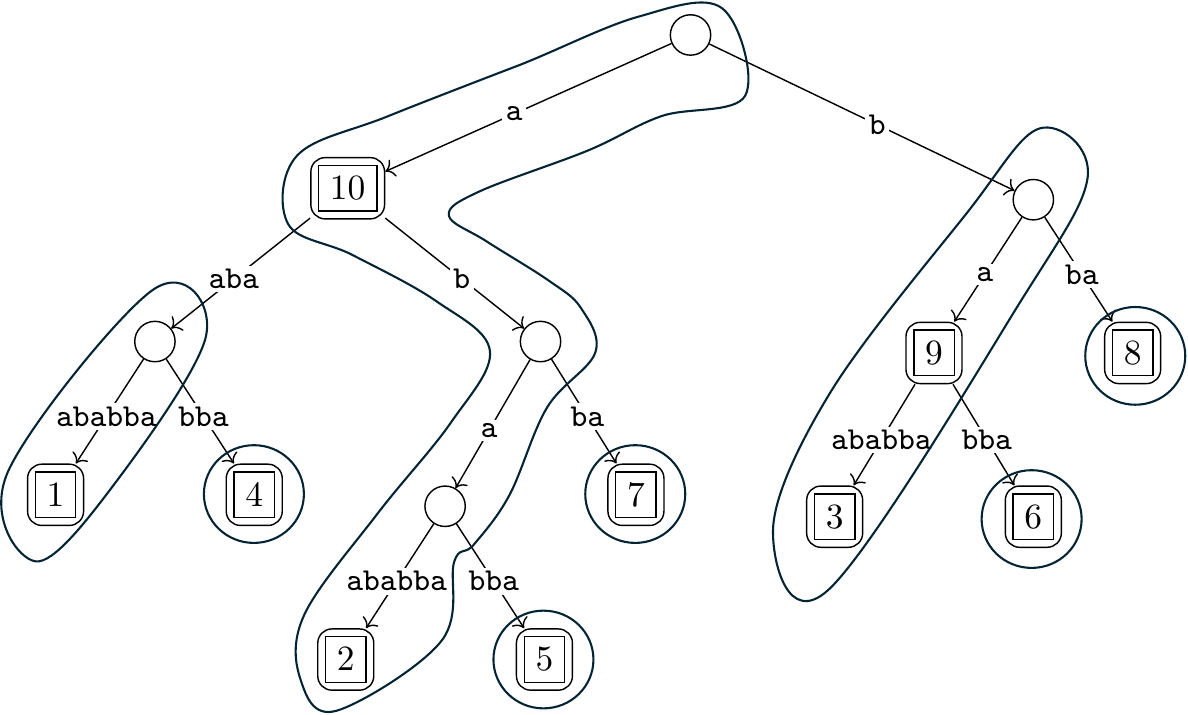}%
\includegraphics[scale=0.45]{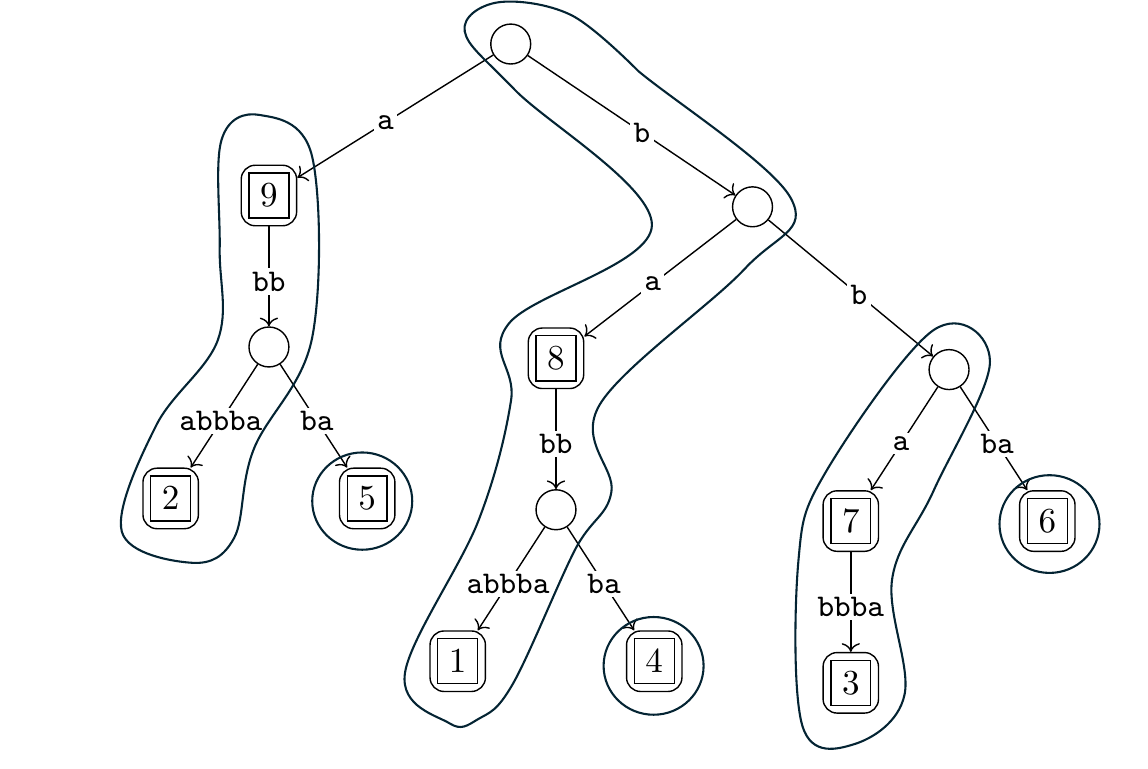}
\end{center}

As the figures illustrate, a leaf node may form a heavy path consisting of a single node. All leaf nodes are the bottommost node of some heavy-path and vice versa. Furthermore, if an inner node has only one child, the edge to that child is heavy by definition.

%

The heavy-light decomposition guarantees the following basic and well-known property.
\begin{proposition}
Let $\stree_w$ be a suffix tree of a string $w$, and consider its heavy-light path decomposition.  
For any node $v \in \stree$,
the path from the root to $v$ contains $O(\log |w|)$ light edges.
\end{proposition}

From this property, the following useful fact can be derived.
Suppose that, for every inner node $v \in \stree_w$, we collect all vertices in the subtree $\stree_w[u]$ rooted at the light child $u$ of $v$ where $v \to u$ is a light edge.
Then the total cost of this collecting operation, summed over all nodes $v$, is \(O(|w| \log |w|)\).

%% file: onerewb/ffzu2.tex
\begin{tikzpicture}[x=0.95cm,y=1.05cm,
  leaf/.style={draw,rectangle,minimum width=0.95cm,minimum height=0.65cm,
               inner sep=1pt,align=center},
  box/.style={draw,rounded corners=4pt,fill=gray!12},
  guide/.style={thin},
]

\makearray{17} 

\putcell{1}{\small $a$}
\putcell{2}{\small $a$}
\putcell{3}{\small $a$}
\putcell{4}{\small $b$}

\putcell{5}{\small $a$}
\putcell{6}{\small $b$}
\putcell{7}{\small $a$}
\putcell{8}{\small $a$}

\putcell{9}{\small $b$}
\putcell{10}{\small $a$}
\putcell{11}{\small $a$}
\putcell{12}{\small $a$}

\putcell{13}{\small $b$}
\putcell{14}{\small $b$}
\putcell{15}{\small $b$}
\putcell{16}{\small $b$}
\putcell{17}{\small $b$}

\def\ha{0.7}
\def\hb{1.7}
\def\hc{2.7}
\def\hd{3.7}
\def\he{4.7}

\cover{1}{2}{\ha}{n1}{$\trans_2$}
\cover{3}{4}{\ha}{n2}{$\trans_1$}

\cover{6}{7}{\ha}{n3}{$\trans_1$}
\cover{5}{7}{\hb}{n4}{$\trans_2$}

\northhitCover{a5}{n4}
\northhitCover{a6}{n3}
\northhitCover{a7}{n3}
\northhitCover{n3}{n4}

\cover{9}{10}{\ha}{n5}{$\trans_1$}
\northhitCover{a9}{n5}
\northhitCover{a10}{n5}

\cover{11}{12}{\ha}{n6}{$\trans_2$}
\northhitCover{a11}{n6}
\northhitCover{a12}{n6}

\cover{1}{4}{\hb}{n7}{$\trans_3$}
\northhitCover{n1}{n7}
\northhitCover{n2}{n7}

\cover{5}{8}{\hc}{n8}{$\trans_3$}
\northhitCover{a8}{n8}
\northhitCover{n4}{n8}

\cover{9}{12}{\hb}{n9}{$\trans_3$}
\northhitCover{n5}{n9}
\northhitCover{n6}{n9}

\colcover{14}{17}{\ha}{n10}{$\trans_b$}{red!20}
\northhitCover{a14}{n10}
\northhitCover{a15}{n10}
\northhitCover{a16}{n10}
\northhitCover{a17}{n10}

\cover{13}{17}{\hb}{n11}{$\trans_b$}
\northhitCover{a13}{n11}
\northhitCover{n10}{n11}



\colcover{1}{12}{\hd}{n12}{$\trans_3$}{red!20}

\northhitCover[-]{n7}{n12}
\northhitCover[-]{n9}{n12}
\northhitCover[-]{n8}{n12}

\cover{1}{17}{\he}{n13}{$\trans_3$}

\northhitCover[-]{n11}{n13}

\northhitCover[-]{a1}{n1}
\northhitCover[-]{a2}{n1}

\northhitCover[-]{a3}{n2}
\northhitCover[-]{a4}{n2}

\northhitCover{n12}{n13}

\end{tikzpicture}

%% file: onerewb/streezu2.tex
\tikzset{
  stlabelover/.style={
    pos=.55,
    inner sep=1pt, outer sep=0pt,
    font=\ttfamily\small,
    fill=white, draw=none, text=black
  },
  stlabelabove/.style={stlabelover, above}
}

\forestset{
  circlenode/.style={
    circle, draw, minimum size=4mm, inner sep=0pt
  },
  rootnode/.style={
    circle, draw, minimum size=4mm, inner sep=0pt,
  }
}

\begin{forest}
for tree={
  draw, rounded corners, inner sep=2pt,
  edge=->,
  s sep=14mm, l sep=10mm
}
[ , rootnode 
  [\boxed{10}, edge label={node[stlabelover]{\texttt{a}}}
    [ , circlenode, edge label={node[stlabelover]{\texttt{aba}}}
      [\boxed{1}, edge label={node[stlabelover]{\texttt{ababba}}}]
      [\boxed{4}, edge label={node[stlabelover]{\texttt{bba}}}]
    ]
    [ , circlenode, edge label={node[stlabelover]{\texttt{b}}}
      [ , circlenode, edge label={node[stlabelover]{\texttt{a}}}
        [\boxed{2}, edge label={node[stlabelover]{\texttt{ababba}}}]
        [\boxed{5}, edge label={node[stlabelover]{\texttt{bba}}}]
      ]
      [\boxed{7}, edge label={node[stlabelover]{\texttt{ba}}}]
    ]
  ]
  [ , circlenode, edge label={node[stlabelover]{\texttt{b}}}
    [\boxed{9}, edge label={node[stlabelover]{\texttt{a}}}
      [\boxed{3}, edge label={node[stlabelover]{\texttt{ababba}}}]
      [\boxed{6}, edge label={node[stlabelover]{\texttt{bba}}}]
    ]
    [\boxed{8}, edge label={node[stlabelover]{\texttt{ba}}}]
  ]
]
\end{forest}

%% file: onerewb/hardness.tex

\newcommand{\vertices}{\mathcal{V}}
\newcommand{\edges}{\mathcal{E}}
\newcommand{\vconv}{\widetilde}
\newcommand{\adj}{\alpha}
\newcommand{\openv}[1]{\langle #1 \rangle}
\newcommand{\closev}[1]{\langle /#1 \rangle}
\newcommand{\graphr}[1]{\widetilde{#1}}
\newcommand{\rev}[1]{\overleftarrow{#1}}
\newcommand{\graphrr}[1]{\widehat{#1}}

\section{Hardness}\label{section:hardness}

We show our three hardness results in this section:
\begin{enumerate}
\item The \rewb{} matching problem is \textbf{W[2]}-hard (Theorem~\ref{thm:hardness_w2}) in Section~\ref{subsection:REWB-W2-hardness}.
%
\item Under a hypothesis for the triangle detection problem, 2-use 2-\rewb{} matching problem cannot be solved in almost linear time (Theorem~\ref{thm:hardness_2rewb}) in Section~\ref{sec:triangle-by-2use2rewb}.
\item Under \seth{} (more technically, \textbf{k-OV} Hypothesis),
$k$-\rewb{} matching problem cannot be solved in $O(n^{2k - \epsilon})$ time (Theorem~\ref{thm:REWB hardness by kOV}) in Section~\ref{section:REWB hardness by kOV}.
\end{enumerate}

In Section~\ref{subsection:REWB-W2-hardness}~and~\ref{sec:triangle-by-2use2rewb}, we use the following notation.
Let $G = (\vertices, \edges)$ be an undirected graph without self-loops where $\vertices$ is a set of vertices and $\edges$ is a set of edges.
For our constructions, we assume that each vertex $v$ is uniquely encoded as a binary string of length $O(\log n)$; that is, $v \in (0 + 1)^+$, where $E^+ := E\,E^*$ denotes one or more repetitions of a \rewb{} $E$.

For readability, in this Section~\ref{section:hardness}, we sometimes write $[E]_x$ instead of $(E)_x$.

\subsection{W[2]-hardness for \rewb{}: Reduction from $k$-Dominating Set Problem}
\label{subsection:REWB-W2-hardness}

\newcommand{\kDSP}{\ensuremath{k\text{-}\textbf{DSP}}}

We prove Theorem~\ref{thm:hardness_w2} by reduction from the following $k$-dominating set problem (\kDSP). 

\begin{itembox}[l]{\kDSP{}: $k$-Dominating Set Problem} 
\textbf{Input}: An undirected graph $G = (\vertices, \edges)$ and an integer $k$.\\
\textbf{Task}: Deciding if there is a $k$-dominating set: i.e., an (exactly) $k$-vertex subset of $\mathcal{D}$ that satisfies the following condition:
\[
\forall v \in \vertices.\: ((v \in \mathcal{D}) \lor (\exists u \in \mathcal{D}. (u, v) \in \edges)).
\]
\end{itembox}
We note that, since $G$ is undirected, $(u, v) \in E \iff (v, u) \in E$.
It is known that this problem is $\mathbf{W[2]}$-complete~\cite{Downey:1995:1} when parameterized by $k$.
Furthermore, the following is known.
\begin{lemma}[\cite{Patrascu:2010}{\cite[Thm 14.42]{Cygan:2015}}]\label{lem:dominating_set_hypothesis}
Assuming \seth, there is no algorithm running in $O(n^{k - \epsilon})$ time solving \kDSP{} for any $\epsilon > 0$.
\end{lemma}



Our main task is to provide the following reduction from \kDSP{} to $k$-\rewb{} matching problem.

\begin{lemma}[Reduction from \kDSP{} to $k$-\rewb{} Matching Problem]
\label{lemma:reduction:dominating to rewb}
Let $k \geq 1$ be an integer.
In time $O(k)$, we can construct a $k$-\rewb{} $\Psi_k$ of length $O(k)$ satisfying the following.

Let $G = (\vertices, \edges)$ be an undirected graph.
There exists a string $\widetilde{G}$ such that the following holds:
\[
\widetilde{G} \in \lang{\Psi_k} \iff \text{ there is a $k$-Dominating Set $\mathcal{D}$ of $G$}.
\]
In particular, we can construct the string $\widetilde{G}$ in $O((n + m) \log n)$ time with $|\widetilde{G}| = O((n + m) \log n)$
where $n := |\vertices|$ is the number of vertices and $m := |\edges|$ is the number of edges.
\end{lemma}

From this lemma, we obtain the following \textbf{W[2]}-hardness result.
\Hardnesswtwo*
\begin{proof}
We give an FPT-reduction from \kDSP.

Given an instance $(G, k)$,
we output the matching instance $(\Psi_k, \widetilde{G})$,
where $\Psi_k$ is constructed in $O(k)$ time
and $\widetilde{G}$ is constructed in $O((n+m)\log n)$ time as in Lemma~\ref{lemma:reduction:dominating to rewb},
such that $\widetilde{G}\in \lang{\Psi_k}$ iff $G$ has a $k$-dominating set of size exactly $k$.

Moreover, the parameter of the output instance satisfies $|\Psi_k|=O(k)$.
Hence, this is an FPT-reduction, and the claim follows
since \kDSP{} is \textbf{W[2]}-hard.
\end{proof}
We note that this result is stronger than \textbf{W[2]}-hardness parameterized by the number of variables $k$,
since the number of variables $k$ is at most the expression length $|\Psi_k|$.

\subsubsection{Proof of Lemma~\ref{lemma:reduction:dominating to rewb}}

To reduce the \kDSP{} to $k$-\rewb{} matching problem, we first encode the given graph $G = (\vertices, \edges)$ into a string.
For $\vertices =\set{v_1,\ldots,v_n}$, we assume that each vertex $v_i \in \set{0,1}^+$ is a distinct $\lceil\, \log n \rceil$-bit encoding.

For a vertex $v$ with its adjacency list $\text{Adj}(v) = \set{ u \in \vertices :  (v, u) \in \edges } =\set{ u_1, \ldots, u_t }$,
we define
\[
    \adj(v):= \openv{v}\:u_1\,\#\,u_2\,\#\,\cdots\,\#\,u_t\,\#\:\closev{v}.
\]
For example, if $\text{Adj}(10) = \set{ 1111, 00, 01, 100 }$, $\adj(10) = \openv{10}\,1111\,\#\,00\,\#\,01\,\#\,100\#\,\closev{10}$.
It is clear that, for each vertex $v$, $|\alpha(v)| = O(e_v \log n)$ where $e_v := |\text{Adj}(v)|$ is the number of edges connected to $v$.

Then, we define $\graphr{G}$ as a string over an alphabet $\Sigma = \set{ 0,1,\#,\natural, @, \openv{,},/ }$ as follows:
\[
w_{\vertices} :=  v_1 \,\#\, v_2 \,\#\, \cdots \,\#\, v_n \,\#,
\qquad
w_{\edges} := \adj(v_1)\,\natural\,\adj(v_2)\,\natural \cdots \natural\,\adj(v_n)\,\natural,
\qquad 
\graphr{G} := w_{\vertices}\:@\:w_{\edges},
\]
where $\vertices = \set{ v_1, v_2, \ldots, v_n }$ with $v_i \in (0+1)^+$.
It is clear that $|\graphr{G}| = O(n \log n) + \sum_{v \in \vertices} |\alpha(v)| = O(n \log n + m \log n)$, where $n$ is the number of vertices and $m$ is the number of edges.

Hereafter, we translate the dominating set constraint into a \rewb{} expression.
To this end, we introduce our building blocks (expressions) $\psi_{\text{pick}}, \psi_{\text{itself}}, \psi_{\text{neigh}}$.

We first define $\psi_{\text{pick}}$ as follows:
\[
\psi_{\text{pick}} :=
\psi_{\text{skip}}\:
(\psi_{\vertices})_{x_1}\:\#\:
\psi_{\text{skip}}\:
(\psi_{\vertices})_{x_2}\:\#\:\cdots\:\#
\psi_{\text{skip}}\:
(\psi_{\vertices})_{x_k}\:\#\:
\psi_{\text{skip}}
\]
where $\psi_{\vertices} := (0 + 1)^+$ and $\psi_{\text{skip}} := (\psi_{\vertices}\,\#)^*$.
Each vertex code is delimited by $\#$,
hence $\psi_{\text{skip}}$ can skip any number of vertices before binding the next one.

The expression $\psi_{\text{pick}}$ is applied to $w_{\vertices}$ to nondeterministically select $k$ vertices.
For example, we consider the following $w_{\vertices}$:
\[
w_{\vertices} = 0 \# 1 \# 01 \# 10 \# 11\# 100 \#
\]
When $k = 2$, applying $\psi_{\text{pick}}$ generates the following computation, for example:
\[
\overbracket{\phantom{ho}}^{\psi_{\text{skip}}}
\overbracket{\, 0 \,}^{(\psi_{\vertices})_{x_1}}
\overbracket{\,\#\,}^{\#}\,
\overbracket{1 \#\,01 \#\, 10 \#}^{\psi_{\text{skip}}}
\overbracket{\,11\,}^{(\psi_{\vertices})_{x_2}} \overbracket{\,\#\,}^{\#}\,
\overbracket{\,100 \#\,}^{\psi_{\text{skip}}}.
\]
In particular, after this computation, $x_1$ is bound to $0$ and $x_2$ to $11$.

From the construction and this explanation,
it is clear that, by applying $\psi_{\text{pick}}$ to $w_{\vertices}$,
we can nondeterministically select all $k$-vertex  combinations and store each combination into $k$ variables $x_1, x_2, \ldots, x_k$.
This property of $\psi_{\text{pick}}$ is formalized as follows.
\begin{proposition}\label{lem:dominating_pick}
For each $k$-subset $U = \set{u_1,\dots, u_k} \subseteq \vertices$,
there exists a match of $w_{\vertices}$ with $\psi_{\text{pick}}$ such that the variables $x_1, \ldots, x_k$ are bound to the elements of $U$ in the order they appear in $w_{\vertices}$.

Conversely, in any match of $w_{\vertices}$ with $\psi_{\text{pick}}$,
the variables $x_1, \ldots, x_k$ are bound to $k$ distinct vertices, hence represent a $k$-subset of $\vertices$.
\end{proposition}
%
In particular, in any match of $w_{\vertices}$ with $\psi_{\text{pick}}$, the string stored in each variable $x_i$ is a vertex in $V$.
Hereafter, we write $v_{x_i}$ to denote the vertex captured by $x_i$.

We next define the following two expressions for each $i\in \{1,\dots, k\}$:
\[
    \psi_{\text{itself},i} := 
    \openv{\readv{x_i}}\, 
    \psi_{\text{skip}}\,
    \closev{\readv{x_i}}  ,
    \qquad
    \psi_{\text{neigh},i} := \openv{\psi_{\vertices}}\, 
    \psi_{\text{skip}}\,
    \readv{x_i}\,\#\,
    \psi_{\text{skip}}\,
    \closev{\psi_{\vertices}}.
\]
To explain these two expressions,
suppose that the input head is positioned at the first symbol of $\alpha(v_i)$ for some vertex $v_i$ within $w_{\edges} = \cdots \natural\  \alpha(v_i)\ \natural \cdots$.\todo{Are we use $\triangledown$ in the remaining part?}
Applying $\psi_{\text{itself}, j}$ in this situation succeeds iff $v_i = v_{x_j}$; otherwise the match fails.
Similarly, applying $\psi_{\text{neigh}, j}$ succeeds iff $v_{x_j}$ is adjacent to $v_i$.

We formalize the above argument for $\psi_{\text{itself}}$ and $\psi_{\text{neigh}}$ as the following proposition.
\begin{proposition}\label{lem:dominating_set_inc_dom}
Let $i\in \{1,\dots, k\}$ and $v\in \vertices$.
Let $u_i$ be the vertex stored in the variable $x_i$.
Then, 
\begin{itemize}
    \item $\adj(v)$ matches $\psi_{\text{itself},i}$ if and only if $u_i=v$, and
    \item $\adj(v)$ matches $\psi_{\text{neigh},i}$ if and only if $(u_i,v)\in \edges$.
\end{itemize}
\end{proposition}

Using $\psi_{\text{itself}}$ and $\psi_{\text{neigh}}$, we define the following expression $\psi_{\text{dominated}}$, which checks whether the selected vertices
$\set{ v_{x_1}, v_{x_2}, \ldots, v_{x_k} }$
form a dominating set, when applied to $w_{\edges}$:
\todo{Is $\diamondsuit$ necessary?}
\[
    \psi_{\text{dominated}} := \left(\,(\psi_{\text{itself},1} +
    \cdots +
    \psi_{\text{itself},k} + \psi_{\text{neigh},1} +
    \cdots + \psi_{\text{neigh},k})\,\natural\right)^*~.
\]

We show the following property of $\psi_{\text{dominated}}$.
\begin{proposition}\label{prop:psi-dominated}
$w_{\edges}$ matches $\psi_{\text{dominated}}$ if and only if $\set{ v_{x_1}, \ldots, v_{x_k} }$ is a $k$-dominating set of $G$.
\end{proposition}
\begin{proof}
($\Rightarrow$)
Assume that $\set{ v_{x_1}, \ldots, v_{x_k} }$ is a $k$-dominating set of $G$.
Then, for all $v \in \vertices$,
either $v = v_{x_i}$ for some $i$ or $(v, v_{x_i}) \in \edges$ for some $i$.
Therefore, from Proposition~\ref{lem:dominating_set_inc_dom}, for each $v\in \vertices$,
$\adj(v)$ matches
$\psi_{\text{itself},1} +\cdots + \psi_{\text{itself},k} + \psi_{\text{neigh},1} +\cdots + \psi_{\text{neigh},k}$, which is a subexpression of $\psi_{\text{dominated}}$.
Combining such matching for all $v\in \vertices$ yields a matching of $w_{\edges}$ to $\psi_{\text{dominated}}$.

($\Leftarrow$)
Conversely, assume $w_{\edges}$ matches $\psi_{\text{dominated}}$.
The assumption means that,
for each $v\in \vertices$, the string $\adj(v)$ in $w_{\edges}$ is consumed by  $\psi_{\text{itself},i}$ or $\psi_{\text{neigh},i}$ for some $i$.
%
From Proposition~\ref{lem:dominating_set_inc_dom}, $v = v_{x_i}$ holds if $\adj(v)$ matches $\psi_{\text{itself},i}$ and
$(v, v_{x_i}) \in \edges$ holds if $\adj(v)$ matches $\psi_{\text{neigh}, i}$.
Since it holds for all $v \in \vertices$, $\set{v_{x_1}, \ldots, v_{x_k} }$ is a $k$-dominating set of $G$.
\end{proof}

Finally, we combine $\psi_{\text{pick}}$ and $\psi_{\text{dominated}}$ and obtain our main expression $\Psi$ as follows:
\[
\Psi_k := \psi_{\text{pick}}\ @\ \psi_{\text{dominated}} .
\]
Note that $\Psi_k$ is a fixed \rewb{} expression depending only on $k$.
The following property of $\Psi$ clearly holds.
\begin{proposition}\label{lem:dominating_at}
In any match of $\graphr{G}$ against $\Psi_k$,
the prefix $w_{\vertices}$ is matched by $\psi_{\text{pick}}$
and the suffix $w_{\edges}$ is matched by $\psi_{\text{dominated}}$.
\end{proposition}
\begin{proof}
The string $\graphr{G} = w_{\vertices}\,@\,w_{\edges}$ contains exactly one occurrence of $@$.
Hence, in any match with $\Psi_k$,
this symbol must be consumed by the $@$ in $\Psi_k$,
which separates the part matched by $\psi_{\text{pick}}$ from the part matched by $\psi_{\text{dominated}}$.
\end{proof}

Combining Propositions~\ref{lem:dominating_pick},~\ref{prop:psi-dominated},~and~\ref{lem:dominating_at} directly yields the following property, which concludes Lemma \ref{lemma:reduction:dominating to rewb}:
\begin{quote}
$G$ has a $k$-dominating set iff $\graphr{G}$ matches the $k$-\rewb{} $\Psi_k$ (equivalently, $\graphr{G} \in \lang{\Psi_k}$).
\end{quote}

\subsection{Hardness for $2$-use $2$-\rewb{}: Reduction from Triangle Detection Problem}

\label{sec:triangle-by-2use2rewb}

We prove Theorem~\ref{thm:hardness_2rewb} by reduction from the following triangle detection problem.
\begin{itembox}[l]{Triangle Detection Problem}
\noindent{}\textbf{Input}: An undirected graph $G = (\vertices, \edges)$ with $n := |\vertices|$ vertices and $m := |\edges|$ edges. \\
\noindent{}\textbf{Task}: Deciding whether $G$ contains a triangle, that is, whether there is a set of three distinct vertices $\set{ a, b, c }$ with $(a, b), (b, c), (c, a) \in \edges$.
\end{itembox}
The following is conjectured about the triangle detection problem.
\begin{hypothesis*}[\textbf{No-Almost-Linear-Time Hypothesis}~{\cite{Abboud:2014,Abboud:2014:focs}}]
There is a constant $\delta > 0$,
such that in the Word RAM model with words of $O(\log n)$ bits,
any algorithm requires $m^{1 + \delta - o(1)}$ time in expectation to detect whether a graph with $m$ edges graph contains a triangle.
\end{hypothesis*}

In this section, our main task is giving the following reduction from the triangle detection problem to the $2$-use $2$-\rewb{} matching problem.

\begin{lemma}[Reduction from Triangle Detection Problem to $2$-use $2$-\rewb{} Matching Problem]
\label{lemma:reduction:triangle detection to rewb}
\mbox{}

There exists a fixed $2$-use $2$-\rewb{} expression $\Psi$ that solves the triangle detection problem in the following sense:

Let $G = (\vertices, \edges)$ be an undirected graph with $n$ vertices and $m$ edges.
There exists a string $\widetilde{G}$ such that the following holds:
\[
\graphrr{G} \in \lang{\Psi} \iff \text{ there is a triangle in $G$}.
\]
In particular, we can construct the string $\graphrr{G}$ in $O(m \log n)$ time, and the size of $\graphrr{G}$ is $|\graphrr{G}| = O(m \log n)$.
\end{lemma}

We first show our hardness result from Lemma~\ref{lemma:reduction:triangle detection to rewb}.
While the proof follows a standard reduction argument, we provide the details below for completeness.

\Hardnesstworewb*
\begin{proof}
Assume the No-Almost-Linear-Time Hypothesis holds.
Let $\delta_\Delta > 0$ be the constant specified in the hypothesis; that is, any algorithm requires time $m^{1+\delta_\Delta-o(1)}$.

Let $\Psi$ be the fixed $2$-use $2$-\rewb{} expression from Lemma~\ref{lemma:reduction:dominating to rewb}.
For every graph $G$ with $n$ vertices and $m$ edges, we can construct $\widehat{G}$ such that $\widehat{G} \in \mathbb{L}(\Psi) \iff G \text{ has a triangle}$.
We construct a triangle detection algorithm $\mathcal{T}$ as follows: Given $G$, remove isolated vertices so that $n \le 2m$, construct $\widehat{G}$, and run the matching algorithm $\mathcal{A}$ on $\widehat{G}$.

Suppose for contradiction that there is an algorithm $\mathcal{A}$ for $\Psi$ running in time $|w|^{1+\delta-o(1)}$ for every $\delta > 0$.
Specifically, fix $\gamma := \delta_\Delta / 2$ and assume $\mathcal{A}$ runs in time $O(|w|^{1+\gamma})$.

The construction of $\widehat{G}$ takes $O(m \log n)$ time and yields $|\widehat{G}| = O(m \log n)$. Since $n \le 2m$, we have $|\widehat{G}| = O(m \log m)$.
Running $\mathcal{A}$ on $\widehat{G}$ takes the following time $T_{\mathcal{A}}(\widehat{G})$:
\[
    T_{\mathcal{A}}(\widehat{G}) \le O\left( (m \log m)^{1+\gamma} \right) 
    = O\left( m^{1+\gamma} (\log m)^{1+\gamma} \right).
\]
Since $(\log m)^{1+\gamma} = m^{o(1)}$, the total running time of $\mathcal{T}$ is:
\[
    T_{\mathcal{T}}(G) = O(m \log m) + m^{1+\gamma + o(1)} = m^{1+\gamma+o(1)}.
\]
Substituting $\gamma = \delta_\Delta / 2$, we get:
\[
    T_{\mathcal{T}}(G) = m^{1 + \frac{\delta_\Delta}{2} + o(1)}.
\]
For sufficiently large $m$, $1 + \frac{\delta_\Delta}{2} + o(1) < 1 + \delta_\Delta - o(1)$.
Thus, $\mathcal{T}$ solves the triangle detection problem strictly faster than the hypothesized lower bound $m^{1+\delta_\Delta-o(1)}$, which is a contradiction.
\end{proof}

\subsubsection{Proof of Lemma~\ref{lemma:reduction:triangle detection to rewb}}

We here encode the given graph $G$ into the following string $\graphrr{G}$ over an alphabet $\Sigma = \set{ 0, 1, \#, \natural, @, \langle, \rangle, / }$:
\[
\graphrr{G} := w_{\edges}\ @\ w_{\edges}~,
\]
where the string representation of the adjacency lists of $G$, $w_{\edges}$, is defined in the previous section.
Therefore, it is clear that the size of $|\graphrr{G}|$ satisfies $|\graphrr{G}| = O(m \log n)$ where $n$ and $m$ are the number of vertices and edges of $G$, respectively.
Moreover, we can construct $\graphrr{G}$ in $O(m \log n)$ time.

If $G$ has a triangle $\set{a, b, c}$,
then the following pattern must appear in $\graphrr{G}$:
\[
\cdots \openv{a} \cdots \# b \# \cdots \# c \# \cdots \closev{a} \cdots\ @ \cdots \openv{b} \cdots \# c \# \cdots \closev{b} \cdots~~.
\]
Conversely, if this pattern appears in $\graphrr{G}$, then $G$ has a triangle $\set{a, b, c}$.

We now prepare two expressions $\varphi_{\text{pick}}$ and $\varphi_{\text{check}}$ where
\begin{itemize}
\item $\varphi_{\text{pick}}$ nondeterministically chooses one vertex $v$ and selects two neighbors and save them to the variables $x$ and $y$.
\item $\varphi_{\text{find}}$ searches the adjacency list of the vertex stored in the variable $x$ and then searches the vertex stored in the variable $y$.
\end{itemize}

We first define $\varphi_{\text{pick}}$ as follows:
\[
\varphi_{\text{pick}}
 :=
 \Upsilon^*\:
 \langle \psi_{\vertices} \rangle \ 
 \psi_{\text{skip}}\ 
 (\psi_\vertices)_x\ \#\ 
 \psi_{\text{skip}}\ 
 (\psi_\vertices)_y\ \#\ 
 \psi_{\text{skip}}\ 
 \langle / \psi_{\vertices} \rangle\ 
 \Upsilon^*\:@
\]
where
$\psi_{\vertices} := (0 + 1)^+$, $\psi_{\text{skip}} := (\psi_{\vertices}\,\#)^*$ (as with Section~\ref{subsection:REWB-W2-hardness}),
and $\Upsilon := \Sigma \setminus \set{ @ }$.

We explain how this expression $\varphi_{\text{pick}}$ nondeterministically selects and saves two vertices in an adjacency list step-by-step.
The part $\Upsilon^*\,\openv{\psi_\vertices}$ nondeterministically consumes a prefix and proceeds to the following part
\begin{flalign*}
& \qquad
\overbracket{\cdots}^{\Upsilon^*} \overbracket{\openv{a}}^{\openv{\psi_\vertices}} \cdots \# b \# \cdots \# c \# \cdots \closev{a} \cdots\ @ \cdots \openv{b} \cdots \# c \# \cdots \closev{b} \cdots~~.
&&
\end{flalign*}
Then the following subexpression $\psi_{\text{skip}}$ can consume the substring before $b$ as follows:
\begin{flalign*}
& \qquad
\overbracket{\cdots}^{\Upsilon^*}\,
\overbracket{\openv{a}}^{\openv{\psi_\vertices}}\ 
\overbracket{\cdots \#}^{\psi_{\text{skip}}}\ 
b \# \cdots \# c \# \cdots \closev{a} \cdots\ @ \cdots \openv{b} \cdots \# c \# \cdots \closev{b} \cdots~~.
&&
\end{flalign*}
Now the subexpression $(\psi_\vertices)_x \#$ captures $b$ to $x$ successfully:
\begin{flalign*}
& \qquad
\overbracket{\cdots}^{\Upsilon^*}
\overbracket{\openv{a}}^{\openv{\psi_\vertices}}\ 
\overbracket{\cdots \#}^{\psi_{\text{skip}}}\ 
\overbracket{\,b\,}^{(\psi_\vertices)_x}
\overbracket{\#}^{\#} \cdots \# c \# \cdots \closev{a} \cdots\ @ \cdots \openv{b} \cdots \# c \# \cdots \closev{b} \cdots~~.
&&
\end{flalign*}
Similarly, the subexpression $\psi_{\text{skip}}\,(\psi_\vertices)_y\,\# \psi_{\text{skip}}\,\closev{\psi_\vertices}$ can capture the vertex $c$ to $y$ and successfully reach to the end of the adjacency list $\closev{a}$ as follows:
\begin{flalign*}
& \qquad
\overbracket{\cdots}^{\Upsilon^*}
\overbracket{\openv{a}}^{\openv{\psi_\vertices}}\ 
\overbracket{\cdots \#}^{\psi_{\text{skip}}}\ 
\overbracket{\,b\,}^{(\psi_\vertices)_x}\ 
\overbracket{\#}^{\#}\ 
\overbracket{\cdots \#}^{\psi_{\text{skip}}}\ 
\overbracket{\,c\,}^{(\psi_\vertices)_y} 
\overbracket{\#}^{\#}\ 
\overbracket{\cdots}^{\psi_{\text{skip}}}\ 
\overbracket{\closev{a}}^{\closev{\psi_\vertices}}
\cdots\ @ \cdots \openv{b} \cdots \# c \# \cdots \closev{b} \cdots~~.
&&
\end{flalign*}
Finally, the part $\Upsilon^* @$ successfully consumes the first $w_\edges$ as follows:
\begin{flalign*}
& \qquad
\overbracket{\cdots}^{\Upsilon^*}
\overbracket{\openv{a}}^{\openv{\psi_\vertices}}\ 
\overbracket{\cdots \#}^{\psi_{\text{skip}}}\ 
\overbracket{\,b\,}^{(\psi_\vertices)_x}\ 
\overbracket{\#}^{\#}\ 
\overbracket{\cdots \#}^{\psi_{\text{skip}}}\ 
\overbracket{\,c\,}^{(\psi_\vertices)_y} 
\overbracket{\#}^{\#}\ 
\overbracket{\cdots}^{\psi_{\text{skip}}}\ 
\overbracket{\closev{a}}^{\closev{\psi_\vertices}}
\overbracket{\cdots}^{\Upsilon^*}\ 
\overbracket{@}^{@}
\cdots \openv{b} \cdots \# c \# \cdots \closev{b} \cdots~~.
&&
\end{flalign*}
$\Upsilon^*$ also never goes beyond $@$ because $@ \notin \Upsilon$.

This argument immediately leads to the following proposition about $\varphi_{\text{pick}}$.
\begin{proposition}\label{prop:pickbc-triangle}
Applying $\varphi_{\text{pick}}$ to $\graphrr{G} = w_{\edges } \,@\, w_{\edges}$,
it consumes the prefix including $@$ (i.e., $w_\edges @$)
and save two vertices $b$ and $c$ to $x$ and $y$, respectively, where they are neighbors of some vertex $a$: that is, $a{-}b$ and $a{-}c$.
Moreover, $\varphi_{\text{pick}}$ can select all such triples $(a, b, c)$.
\end{proposition}

Our remaining task is designing a subexpression that can find the following pattern when $b$ and $c$ are stored in the variables:
\[
 \cdots \openv{b} \cdots \# c \# \cdots \closev{b} \cdots ~~~.
\] 
To this end, we define the following expression $\varphi_{\text{find}}$:
\[
  \varphi_{\text{find}} := 
  \Upsilon^*\ 
  \langle \readv{x} \rangle\ 
  \psi_{\text{skip}}\ 
  \readv{y}\ \#\ \Sigma^*~.
\]
\todo{Is there $\diamondsuit$?}
We apply this expression to the rest of the input consumed by $\varphi_{\text{pick}}$.
The first part $\Upsilon^* \openv{\readv{x}}$ successfully finds the adjacency list of the vertex stored in the variable $x$ as follows:
\begin{flalign*}
& \qquad
\overbracket{\cdots}^{\Upsilon^*}
\overbracket{\openv{a}}^{\openv{\psi_\vertices}}\ 
\overbracket{\cdots \#}^{\psi_{\text{skip}}}\ 
\overbracket{\,b\,}^{(\psi_\vertices)_x}\ 
\overbracket{\#}^{\#}\ 
\overbracket{\cdots \#}^{\psi_{\text{skip}}}\ 
\overbracket{\,c\,}^{(\psi_\vertices)_y}\ 
\overbracket{\#}^{\#}\ 
\overbracket{\cdots}^{\psi_{\text{skip}}}\ 
\overbracket{\closev{a}}^{\closev{\psi_\vertices}}\ 
\overbracket{\cdots}^{\Upsilon^*}\ 
\overbracket{@}^{@}~~~
\underbracket{\cdots}_{\Upsilon^*}
\underbracket{\openv{b}}_{\openv{\readv{x}}}
\cdots \# c \# \cdots \closev{b} \cdots~~
&&
\end{flalign*}
where the variable $x$ now has the vertex $b$.
After entering the adjacency list of $b$,
the part $\psi_{\text{skip}}\ \readv{y}\ \#$ successfully finds the neighbor stored in the variable $y$ as follows:
\begin{flalign*}
& \qquad
\overbracket{\cdots}^{\Upsilon^*}
\overbracket{\openv{a}}^{\openv{\psi_\vertices}}\ 
\overbracket{\cdots \#}^{\psi_{\text{skip}}}\ 
\overbracket{\,b\,}^{(\psi_\vertices)_x}\ 
\overbracket{\#}^{\#}\ 
\overbracket{\cdots \#}^{\psi_{\text{skip}}}\ 
\overbracket{\,c\,}^{(\psi_\vertices)_y}\ 
\overbracket{\#}^{\#}\ 
\overbracket{\cdots}^{\psi_{\text{skip}}}\ 
\overbracket{\closev{a}}^{\closev{\psi_\vertices}}\ 
\overbracket{\cdots}^{\Upsilon^*}\ 
\overbracket{@}^{@}~~~
\underbracket{\cdots}_{\Upsilon^*}
\underbracket{\openv{b}}_{\openv{\readv{x}}}
\underbracket{\cdots \#}_{\psi_{\text{skip}}}\ 
\underbracket{~c~}_{\readv{y}}
\underbracket{\#}_{\#}
\cdots \closev{b} \cdots~~ &&
\end{flalign*}
The success of applying the last part means that we find the triangle $\set{a, b, c}$; thus, we finally consume the rest of the input by $\Sigma^*$ as follows:
\begin{flalign*}
& \qquad
\overbracket{\cdots}^{\Upsilon^*}
\overbracket{\openv{a}}^{\openv{\psi_\vertices}}\ 
\overbracket{\cdots \#}^{\psi_{\text{skip}}}\ 
\overbracket{\,b\,}^{(\psi_\vertices)_x}\ 
\overbracket{\#}^{\#}\ 
\overbracket{\cdots \#}^{\psi_{\text{skip}}}\ 
\overbracket{\,c\,}^{(\psi_\vertices)_y}\ 
\overbracket{\#}^{\#}\ 
\overbracket{\cdots}^{\psi_{\text{skip}}}\ 
\overbracket{\closev{a}}^{\closev{\psi_\vertices}}\ 
\overbracket{\cdots}^{\Upsilon^*}\ 
\overbracket{@}^{@}~~~
\underbracket{\cdots}_{\Upsilon^*}
\underbracket{\openv{b}}_{\openv{\readv{x}}}
\underbracket{\cdots \#}_{\psi_{\text{skip}}}\ 
\underbracket{~c~}_{\readv{y}}
\underbracket{\#}_{\#}
\underbracket{\cdots \closev{b} \cdots}_{\Sigma^*}~. &&
\end{flalign*}

This argument immediately leads to the following proposition about $\varphi_{\text{find}}$.
\begin{proposition}\label{prop:findbc-triangle}
Assume that the variable $x$ and $y$ stores a vertex $b$ and $c$, respectively.
On the application of $\varphi_{\text{find}}$ to $w_\edges$,
the following conditions are equivalent:
\begin{itemize}
\item the application results in success with consuming all the input
\item the vertices $b$ and $c$ are neighbors
\end{itemize}
\end{proposition}

Combining $\varphi_{\text{pick}}$ and $\varphi_{\text{find}}$, we define our expression $\Psi$ as follows:
\[
\Psi := \varphi_{\text{pick}}\ \varphi_{\text{find}}.
\]

Propositions~\ref{prop:pickbc-triangle}~and~\ref{prop:findbc-triangle} lead to our main property, which concludes Lemma~\ref{lemma:reduction:triangle detection to rewb}:
\[
w_\edges\ @\ w_\edges \in \lang{\Psi}
\iff
\text{there is a triangle in $G$}.
\]

\subsection{Hardness from $k$-Orthogonal Vectors Problem Hypothesis}

\label{section:REWB hardness by kOV}
\newcommand{\ip}[1]{\mathbf{IP}(#1)}
\newcommand{\kip}[2]{\mathbf{IP}_{#1}(#2)}
\newcommand{\kOV}{\ensuremath{\textbf{$k$-OV}}}
\newcommand{\OV}{\textbf{OV}}

We here show our $O(n^{2q-\epsilon})$-time hardness for the $q$-\rewb{} matching problem where $n$ is the length of input strings under \seth{}, or more technically, under the $k$-orthogonal vectors (\kOV) Hypothesis~\cite{Duraj:2019,Williams:2019}, which is weaker than \seth{} and indeed is implied by \seth{}.
To use $k$ for the \kOV{} and its hypothesis, we use ``$q$'' to denote the number of variables of \rewb{} in this section.



Let $v_1, v_2, \ldots, v_k \in \set{0, 1}^d$ be $d$-dimensional 0-1 vectors over the ring $\mathbb{Z}$.
The $k$ inner-product is defined as follows:
\[
\kip{k}{ v_1, v_2, \ldots, v_k } := \sum^d_{i = 1} \left(\ \prod^k_{j = 1} v_j[i]\,\right). 
\]
These vectors are \emph{$k$-orthogonal} if $\kip{k}{ v_1, v_2, \ldots, v_k } = 0$.
In our construction, we use the following equivalent condition:
$\displaystyle\bigwedge^d_{i = 1} \bigvee^k_{j = 1} (v_j[i] = 0)$.


We now consider the $\kOV$ and the $\kOV$ Hypothesis.
\begin{itembox}[l]{$\kOV_{N, d}$ Problem}
\noindent{}\textbf{Input:}
$k$ sets $A_1, \ldots, A_k \subseteq \set{0, 1}^d$
where each set $A_i$ has size $|A_i| = N$.

\noindent{}\textbf{Task:}
Deciding if the following holds:
\[
\exists v_1 \in A_1.\,\exists \:v_2 \in A_2.\:\ldots\:\exists v_k \in A_k.\ \kip{k}{v_1, v_2, \ldots, v_k} = 0.
\]
\end{itembox}

\begin{hypothesis*}[\textbf{$k$-OV} Hypothesis~{\cite[Hypothesis 2]{Duraj:2019},\cite[Hypothesis 4]{Williams:2019}}]
Let $k \geq 2$ be an integer.
For any $\epsilon > 0$ and
every dimension function $d(N) = \omega(\log N)$
(e.g., $d(N) = \lceil\log^2 N\rceil$),
$\kOV_{\!N,\,d(N)}$ cannot be solved in $O(N^{k - \epsilon})$ time.
\end{hypothesis*}

\noindent{}\textbf{Remark:} It is known that $\kOV$ Hypothesis is weaker than \seth. Actually, \seth{} implies $\kOV$ Hypothesis~\cite[Theorem~3.1]{Williams:2019}~\cite{Duraj:2019,Williams:2005}


\newcommand{\strfun}{\mathcal{S}}

We provide an efficient reduction from $2k\text{-}\OV$ problem to $k$-\rewb{} matching problem.
\begin{restatable}{lemma}{kOVRewbReduction}
\label{lemma:reduction:kOV to REWB}
Let $k \geq 1$ be an integer.
There exists a (fixed) $k$-\rewb{} $\Psi_k$ that solves $2k\text{-}\OV$ in the following sense:
Let $N \geq 1$ be an integer and $A_1, \ldots, A_{2k} \subseteq \set{0, 1}^d$ be an instance of $2k$-\OV{} with $|A_1| = \cdots = |A_{2k}| = N$.
There is a string $\strfun(A_1, \ldots, A_{2k})$ such that
\[
\strfun(A_1, \ldots, A_{2k}) \in \lang{\Psi_k} \iff \exists v_1 \in A_1. \exists v_2 \in A_2.\ \cdots. \exists v_{2k} \in A_{2k}.\ \kip{2k}{v_1, v_2, \ldots, v_{2k}} = 0.
\]
In particular, we can construct $\strfun(A_1, \ldots, A_{2k})$ in $O(k \cdot N \cdot d^2)$ time.
\end{restatable}

Before providing the proof of this lemma (shown in Section~\ref{section:kOV-to-REWB reduction}), we here prove Theorem~\ref{thm:REWB hardness by kOV} using this lemma.
While the proof follows a standard reduction argument, we provide the details below for completeness.
\begin{theorem*}[Reformulation of Theorem~\ref{thm:REWB hardness by kOV}]
The
$q$-\rewb{} matching
problem cannot be solved in $O(|w|^{2q - \epsilon})$ time for any $\epsilon > 0$
under $\kOV$ Hypothesis.
\end{theorem*}
\begin{proof}
%
We proceed by contradiction. 
Suppose there exists an algorithm $\mathcal{A}$ solving the problem in time $O(|w|^{2q-\epsilon})$ for some constant $\epsilon > 0$.
We construct an efficient algorithm to solve $2q\text{-}\OV_{\!N,\,d}$ as follows:

\begin{enumerate}
    \item Given an instance of $2q\text{-}\OV_{\!N,\,d}$ consisting of $2q$ sets $\set{A_1, A_2, \ldots, A_{2q}}$ of $N$ vectors in $\set{0,1}^d$, we first transform it into a string $\mathcal{S} := \mathcal{S}(A_1, \ldots, A_{2q})$ by our reduction.
    By the assumption, this reduction takes $T(N) = O(N d^2)$ time, and $n_{\mathcal{S}} := |\mathcal{S}| \leq c \cdot N d^2$ for some constant $c$ depending on $q$.
    \item Next, we run the algorithm $\mathcal{A}$ on $\mathcal{S}$. The running time is
    \[
    T_{\mathcal{A}}(n_{\mathcal{S}}) = O(n_{\mathcal{S}}^{2q-\epsilon}) = O\bigl( (N d^2)^{2q-\epsilon} \bigr).
    \]
\end{enumerate}

Fix the dimension function $d(N) := \lceil\,(\log N)^2\,\rceil$.
Since $d(N) = \omega(\log N)$ holds, the $\kOV$ Hypothesis (with $k = 2q$) applies also to $2q\text{-}\OV_{\!N,\,d(N)}$.

By Step 2 we obtain a $2q\text{-}\OV$ algorithm with running time :
\[
T_{\text{total}}(N)
= O(N d^2) + O\big((N d^2)^{2q-\epsilon}\big)
= O(N d^2) + O\big(N^{2q-\epsilon}\cdot d^{4q-2\epsilon}\big)
\quad \text{where}\ d = d(N).
\]
Since $d = \lceil\,(\log N)^2\,\rceil$,
we have
$d^{4q-2\epsilon} \leq
\bigl((\log{N})^2 + 1\bigr)^{4q - 2\epsilon} =
(\log{N})^{O(1)} = N^{o(1)}$.
Thus for any $\delta > 0$ and sufficiently large $N$,
\[
d^{4q-2\varepsilon} \leq N^\delta .
\]
Choose $\delta = \epsilon/2$.
Then
\[
T_{\text{total}}(N)
= O(N^{1+o(1)}) + O(N^{(2q-\epsilon)+(\epsilon/2)})
= O(N^{2q-\epsilon/2}),
\]
contradicting the $\kOV$ Hypothesis with $k = 2q$.

\end{proof}

\subsubsection{Proof of Lemma~{\ref{lemma:reduction:kOV to REWB}}}

\label{section:kOV-to-REWB reduction}

To explain our idea, we first give a reduction from $2\text{-}\OV$ to $1$-\rewb{} matching problem.
The reduction can be easily extended to a reduction from $2k\text{-}\OV$ to $k$-\rewb{}.
\begin{lemma}[Reduction for $k=1$]\label{lemma:reduction from 2OV to 1REWB}
There exists a (fixed) $1$-\rewb{} $\Psi_1$ that solves $2\text{-}\OV$ in the following sense:
Let $N \geq 1$ be an integer and $A_1, A_2 \subseteq \set{0, 1}^d$ be an instance of $2$-\OV{} with $|A_1| = |A_2| = N$.
There is a string $\strfun(A_1, A_2)$ such that
\[
\strfun(A_1, A_2) \in \lang{\Psi_1} \iff \exists v_1 \in A_1. \exists v_2 \in A_2.\ \ip{v_1, v_2} = 0.
\]
In particular, we can construct $\strfun(A_1, A_2)$ in $O(N \cdot d^2)$ time.
\end{lemma}
\begin{proof}
\newcommand{\va}{v}
\newcommand{\vb}{u}

We note that the proof presented here can be viewed as the \rewb-version of a similar result and construction for 1-turn bounded-height ($\leq d$) pushdown automata~\cite[Theorem~4]{Hansen:2021}.

Let $A_1 = \set{ \va_1, \va_2, \ldots, \va_ N}$ and
$A_2 = \set{ \vb_1, \vb_2, \ldots, \vb_N }$. 
We set $\Sigma := \set{ 0, 1, \#, @, \natural }$ and
define the following string $\mathcal{I}$, which serves as the arena for our construction:
\[
\mathcal{I} := \#\,\rev{\va_1}\,\#\,\rev{\va_2}\,\#\ \cdots \#\,\rev{\va_N}\,\# @ \#\,\vb_1\,\#\,\vb_2\,\# \cdots \#\,\vb_N\,\#\,\natural
\]
where $\rev{\va}$ means the reversed version of $\va$: i.e., $\rev{11001} = 10011$.
It should be noted that $|\mathcal{I}| = O(d \cdot N)$.

A key of our reduction is representing a pair of vectors $(\va_i, \vb_j)$ with $1 \leq i, j \leq N$ by a \emph{single} substring of $\mathcal{I}$.
A substring $\pi$ of $\mathcal{I}$ \emph{represents} $(\va_i, \vb_j)$ if the following holds:
\[
\mathcal{I} = \cdots \#\, \rev{\va_i}\ \overbracket{\# \cdots \#\,@\,\# \cdots \#}^{\text{\large $\pi$}}\ \vb_j \, \# \cdots \natural.
\]
More formally,
let $\pi = \pi(i, j)$ be the factor (substring) of $\mathcal{I}$ that starts right after the occurrence of $\rev{\va_i}$ and ends right before the occurrence of $\vb_j$.
Explicitly given by:
\[
\pi(i, j) := (\# \overleftarrow{\va_{i+1}} \# \cdots \# \overleftarrow{\va_N}) \, \# @ \# \, (\vb_1 \# \cdots \# \vb_{j-1} \#)
\]
If $i=N$, the left part is empty; if $j=1$, the right part is empty.

For example,
if $A_1 = \set{ \va_1 = 011, \va_2 = 001}$ and $A_2 = \set{ \vb_1 = 101, \vb_2 = 100 }$, then our string $\mathcal{I}$ takes the following form:
\[
\mathcal{I} = \#\,110\,\#\,100\,\#\,@\,\#\,101\,\#\,100\,\#\natural.
\]\todo{below we write... なのでこの段階では $\mathcal{I}$ にします}
In this setting, the substring $\pi = \#\,100\,\#\,@\,\#\,101\,\#$ represents $(\va_1, \vb_2)$ because $\mathcal{I} = \va_1\,\pi\,w_2$.
Below we write $\tilde{\mathcal{I}}$ to denote $\mathcal{I}$ for this example.

We define the following building block expression $\psi$ and $\phi_1, \phi_2$ where
\begin{itemize}
\item we use $\psi$ by applying it to a copy of $\mathcal{I}$ to nondeterministically select and represent vectors $\va_*$ from $A_1$ and $\vb_*$ from $A_2$; and
\item we use $\phi_1$ (resp. $\phi_2$) by applying it to a copy of $\mathcal{I}$ to check $\va_*[\ell] = 0$ (resp. $\vb_*[\ell] = 0$) for each index $\ell \in [1 \btw d]$.
\end{itemize}

We first define our expression $\psi$ as follows:
\[
\psi := \Gamma^*\,(\#\,\Gamma^*\,@\,\Gamma^* \#)_x\,\Gamma^*\,\natural
\]
where $\Gamma := (0 + 1 + \#)$ is the \regex. 
We note that $\Gamma$ cannot consume the symbols $@$ and $\natural$.

We apply $\psi$ to $\mathcal{I}$ to nondeterministically select vectors $v_*$ from $A_1$ and $\vb_*$ from $A_2$.
For our example $A_1 = \set{ \va_1 = 011, \va_2 = 001}$ and $A_2 = \set{ \vb_1 = 101, \vb_2 = 100 }$,
applying $\psi$ to $\tilde{\mathcal{I}}$ can generate the following match, in which $\va_* = \va_1$ and $\vb_* = \vb_2$ are selected:
\[
\overbracket{\#110}^{\Gamma^*}\ 
\overbracket{\#100\,\#@\#\,101\#}^{(\# \Gamma^* @ \Gamma^* \#)_x}\ 
\overbracket{100\#\natural}^{\Gamma^* \natural}
\quad
\text{where $x$ has $\#\,100\,\#@\#\,101\,\#$.}
\]
Another run can select $\va_* = \va_2$ and $\vb_* = \vb_1$ as follows:
\[
\overbracket{\#\,110\,\#\,100}^{\Gamma^*}\ 
\overbracket{\#@\#}^{(\# \Gamma^* @ \Gamma^* \#)_x}\ 
\overbracket{101\,\#\,100\#\natural}^{\Gamma^* \natural}.
\quad
\text{where $x$ has $\#\,@\,\#$}.
\]

By the above observation, the following is clear.
\begin{proposition}
Applying $\psi$ to $\mathcal{I}$,
it completely consumes $\mathcal{I}$ and saves a substring into the variable $x$, which represents $(\va_i, \vb_j)$ for some $1 \leq i, j \leq N$.
Moreover, any pair $(\va_i, \vb_j)$ can be captured and represented via $x$ by applying $\psi$ to $\mathcal{I}$.
\end{proposition}

Next we define our expressions $\varphi_{\va}$ and $\varphi_{\vb}$ as follows:
\[
\varphi_{\va} := \Gamma^* \# \mathcal{B}^*\, (0\,\readv{x}\,\mathcal{B})_{x}\,\mathcal{B}^* \# \Gamma^*\,\natural, \qquad
\varphi_{\vb} := \Gamma^* \# \mathcal{B}^*\, (\mathcal{B}\,\readv{x}\,0)_{x}\,\mathcal{B}^* \# \Gamma^*\,\natural .
\]
where $\mathcal{B} := (0 + 1)$.

The former $\varphi_{\va}$ requires $\va_*[\ell] = 0$ and the latter $\varphi_{\vb}$ requires $\vb_*[\ell] = 0$ for some index $\ell \in [1 \btw d]$.
For example,
applying $\varphi_{\va}$ to another copy of $\tilde{\mathcal{I}}$ under $x = \#100@\#101\#$,
it works as follows:
\[
\overbracket{\phantom{01}}^{\Gamma^*}\ 
\overbracket{\#}^{\#}\ 
\overbracket{11}^{\mathcal{B}^*}\ 
\overbracket{0\,\#100\,\#@\#\,101\#\,1}^{(0\,\readv{x} \,\mathcal{B})_x}\ 
\overbracket{00}^{\mathcal{B}^*}\ 
\overbracket{\#}^{\#}\ 
\overbracket{\,\natural\,}^{\Gamma^* \natural}.
\]
This application also updates $x$ as follows:
\[
x = \#100\,\#@\#\,101\# \leadsto 0\,\#100\,\#@\#\,101\#\,1~.
\]

Since one application of $\varphi_{\va}$ and $\varphi_{\vb}$ extends $x$ with a single-letter,
we can check $\displaystyle\bigwedge^d_{\ell = 1} (\va_*[\ell] = 0 \lor \vb_*[\ell] = 0)$ by repeatedly applying the expression $(\varphi_{\va} + \varphi_{\vb})$ .
For example, applying $\varphi_{\vb}$ to another copy of $\tilde{\mathcal{I}}$ under the updated $x$,
it checks if $w_*[2] = 0$ and updates $x$ as follows:
\[
\overbracket{\phantom{01}}^{\Gamma^*}\ 
\overbracket{\#}^{\#}\ 
\overbracket{1}^{\mathcal{B}^*}\ 
\overbracket{1\,0\#100\,\#@\#\,101\#1\,0}^{(\mathcal{B}\, \readv{x}\,0)_x}\ 
\overbracket{0}^{\mathcal{B}^*}\ 
\overbracket{\#}^{\#}\ 
\overbracket{\,\natural\,}^{\Gamma^* \natural},
\quad
x = 0\#100\,\#@\#\,101\#1 \leadsto 1\,0\#100\,\#@\#\,101\#1\,0~.
\]

By the above observation, the following is clear.
\begin{proposition}
Assume the content of $x$ represents $\va_i[\ell]$ and $\vb_j[\ell]$ for some $1 \leq \ell \leq d$.
Then, applying $(\varphi_{\va} + \varphi_{\vb})$ to $\mathcal{I}$, the following holds:
if the application succeeds, then $\va_i[\ell] \cdot \vb_j[\ell] = 0$ and $x$ is updated to represent $\va_i[\ell+1]$ and $\vb_j[\ell+1]$ (when $\ell \leq d$).
\end{proposition}

In summary, we employ the following gadgets:
\[
\strfun(A_1, A_2) := \overbrace{\mathcal{I}\,\mathcal{I}\,\cdots\,\mathcal{I}}^{\text{$(1+d)${-}copies}},
\qquad
\Psi_1 := \psi\ (\varphi_{\va} + \varphi_{\vb})^*
\]
It should be noted that the expression $\Psi_1$ does not depend on the instance $A_1$ and $A_2$; this $\Psi_1$ can be used for any $2\text{-}\OV$ instances.\todo{微調整}

By the above argument,
it is clear that $\strfun(A_1, A_2) \in \lang{\Psi_1}$ iff there are orthogonal vectors in $A_1$ and $A_2$.
We note that $|\strfun(A_1, A_2)| = O(N \cdot d^2)$ because $|\mathcal{I}| = O(N \cdot d)$, and $\strfun(A_1, A_2)$ can be constructed in $O(N \cdot d^2)$ time.
\end{proof}

We can generalize the above reduction to one from $2k\text{-}\OV$ to $k$-\rewb.

\kOVRewbReduction*

\begin{proof}
\newcommand{\va}{v}
\newcommand{\vb}{u}

We extend the construction for Lemma~\ref{lemma:reduction from 2OV to 1REWB} to the case $k \geq 2$.
We represent $2 \cdot k$ vectors using $k$ variables $x_1, x_2, \ldots, x_k$.
We assume the following representation:
\[
\begin{array}{l}
A_1 = \set{ \va^1_1, \ldots, \va^1_N },\ A_2 = \set{ \vb^1_1, \ldots, \vb^1_N }, \\[5pt]
A_3 = \set{ \va^2_1, \ldots, \va^2_N },\
A_4 = \set{ \vb^2_1, \ldots, \vb^2_N }, \\
~\vdots \\
A_{2k-1} = \set{ \va^{k}_1, \ldots, \va^{k}_N },\ A_{2k} = \set{ \vb^{k}_1, \ldots, \vb^{k}_N }.
\end{array}
\]
We change the definition of $\mathcal{I}$ as follows:
\[
\mathcal{I} := \mathcal{I}_1\ \mathcal{I}_2\ \cdots \mathcal{I}_k\ ,
\]
where
\[
\mathcal{I}_r := \#\,\rev{\va^r_1}\,\# \cdots \,\#\rev{\va^r_N} \,\# @ \#\,\vb^r_1\,\# \cdots \# \vb^r_{N} \#\,\natural\ .
\]
For example, if $k = 2$,
\[
\mathcal{I} = \# \rev{\va^1_1} \#  \cdots \# \rev{\va^1_N} \,\# @ \#\, \vb^1_1 \# \cdots \# \vb^1_n \# \natural\ \ \# 
\rev{\va^2_1} \# \cdots \# \rev{\va^2_N} @ \# \vb^2_1 \# \cdots \# \vb^2_N \# \natural\ .
\]

We use the variable $x_r$ ($1 \leq r \leq k$) to represent $(\va^r_{i}, \vb^r_{j})$ with $1 \leq i, j \leq N$.
For example, if $k = 2$,
we use $x_1$ and $x_2$ as follows:
\[
\# \rev{\va^1_1} \# \rev{\va^1_2} \overbracket{\# \cdots \# \rev{\va^1_N} \,\#@\#}^{x_1} \vb^1_1 \# \vb^1_2 \# \cdots \# \vb^1_N \# \natural\ \ 
\# \rev{\va^2_1} \# \rev{\va^2_2} \overbracket{\# \cdots \# \rev{\va^2_N}\,\#@ \# \vb^2_1 \#}^{x_2} \vb^2_2 \# \cdots \# \vb^2_N \# \natural
\]

To accommodate multiple variables,
we below modify $\psi$, $\varphi_{\va}$, and $\varphi_{\vb}$ in Lemma~\ref{lemma:reduction from 2OV to 1REWB} (the case $k=1$).
\todo{なんかまだ間違ってる気がするので直しました。これで大丈夫か確認お願いします}

We first define $\psi$ to nondeterministically select vectors from $A_1, A_2, \ldots, A_{2k}$ and bind them into $x_1,\ldots x_k$ as follows:
\[
\psi := \psi_1 \psi_2 \cdots \psi_k,
\qquad
\text{ where }\ 
\psi_{r} := \Gamma^*\,(\# \Gamma^* @ \Gamma^* \#)_{x_r}\,\Gamma^* \natural, \quad \Gamma := (0 + 1 + \#).
\]
As with $\psi$ in Lemma~\ref{lemma:reduction from 2OV to 1REWB}, the following holds for $\psi$.
\begin{proposition}
Applying $\psi$ to $\mathcal{I}$,
it completely consumes $\mathcal{I}$ and saves substrings to the variables $x_1, x_2, \ldots, x_{k}$
where each $x_r$ represents $(\va^r_i, \vb^r_j)$ for some $1 \leq i, j \leq N$.
Moreover, for each $1 \leq r \leq k$,
any pairs $(\va^r_i, \vb^r_j)$ can be captured and represented via $x_r$ by applying $\psi$ to $\mathcal{I}$.
\end{proposition}

We next define $\varphi_r$ to check $\va_i^r[\ell] = 0$ or $\vb_i^r[\ell] = 0$ ($1 \leq \ell \leq d$) as follows:
\[
\varphi_r := \varphi_{r}^{1}\ \varphi_{r}^{2}\ \cdots \ \varphi_{r}^{k} \qquad (1 \leq r \leq k),
\]
where
\[
\varphi_{r}^{s} := \begin{cases}
        \varphi_{s,\top} & (r\neq s),\\
        \varphi_{s, \va} + \varphi_{s, \vb} & (r=s),
    \end{cases}
\qquad
\begin{array}{lcl}
\varphi_{s, \va} & := & \Gamma^* \# \mathcal{B}^*\ (0\,\readv{x_s}\, \mathcal{B})_{x_s}\ \mathcal{B}^* \# \Gamma^* \, \natural\ ,\\
\varphi_{s, \vb} & := & \Gamma^* \# \mathcal{B}^*\ (\mathcal{B}\,\readv{x_s}\, 0)_{x_s}\ \mathcal{B}^* \# \Gamma^*\, \natural\ , \\
\varphi_{s, \top} & := & \Gamma^* \# \mathcal{B}^*\ (\mathcal{B} \readv{x_s}\, \mathcal{B})_{x_s}\ \mathcal{B}^* \# \Gamma^*\, \natural,\,
\end{array}
\quad \mathcal{B} = (0 + 1).
\]

\begin{proposition}
Assume that, for an index $\ell \in [1..d]$, every $x_r$ represents $\va^r_i[\ell]$ and $\vb^r_j[\ell]$ for $1 \leq r \leq k$.
If the application of $\varphi_r$ to $\mathcal{I}$ succeeds and consumes $\mathcal{I}$,
then either $\va^{r}_i[\ell] = 0$ or $\vb^{r}_j[\ell] = 0$ holds.
Moreover, every $x_r$ is updated to represent $\va^r_i[\ell+1]$ and $\vb^r_j[\ell+1]$ (when $\ell \leq d$).

Conversely, if $\va^{r}_i[\ell] = 0$ or $\vb^{r}_j[\ell] = 0$ holds,
then applying $\varphi_i$ to $\mathcal{I}$ must succeed.
\end{proposition}

Using $\varphi_i$, we further define $\varphi$ to check the $2k$-orthogonality condition:
\[
\varphi := \varphi_1 + \varphi_2 + \cdots + \varphi_k.
\]
We observe that applying $\varphi$ to $\mathcal{I}$ succeeds iff $\bigl(\prod^k_{r=1} (\va^r_i[\ell] \cdot \vb^r_j[\ell])\bigr) = 0$ when each $x_r$ represents $\va_i^r[\ell]$ and $\vb_j^r[\ell]$.
Then, we finally define $\mathcal{S}(A_1, A_2, \ldots, A_{2k-1}, A_{2k})$ and $\Psi_k$ as follows:
\[
\mathcal{S}(A_1, A_2, \ldots, A_{2k-1}, A_{2k}) :=
\overbrace{\,\mathcal{I}\ \mathcal{I}\ \ldots\ \mathcal{I}\,}^{(1 + d)\text{-copies}},
\qquad
\Psi_k := \psi\ \varphi^*~.
\]

From the above properties, we observe the following, which completes our proof:
\[
\mathcal{S}(A_1, A_2, \ldots, A_{2k-1}, A_{2k}) \in \lang{\Psi_k}
\iff
\text{the $2k$-\OV{} instance contains $2k$-orthogonal vectors}.
\]
\end{proof}

%% file: onerewb/sanitize.tex
\section{Input Normalization}
\label{section:sanitize}

Here, to simplify subsequent arguments, we convert the input to an equivalent, more structured form.
Let $w$ be the input string over an alphabet $\Sigma$ and $A$, $B$, $C$, and $D$ be \regex{es}. 

Let $\ell = \lceil \log |\Sigma| \rceil$.
For each $\sigma \in \Sigma$, we assign a unique binary string $\mathrm{bin}(\sigma)$ of length $\ell$.
Then, we define the string $\mathrm{san}(w)$ by replacing each letter $\sigma$ of $w$ with the string $\mathrm{bin}(\sigma)\ 0$ and appending a letter $1$ at the end $(\ell+1)$ times; that is,
\[
    \mathrm{san}(w) = 
    \underbrace{\mathrm{bin}(w[1])\ 0}_{w[1]}\, 
    \underbrace{\mathrm{bin}(w[2])\ 0}_{w[2]}\,
    \cdots \,
    \underbrace{\mathrm{bin}(w[|w|])\ 0}_{w[|w|]}\,
    1^{\ell+1}.
\]
For example, if $\Sigma = \{a, b, c\}$, $w = abacaba$, $\mathrm{bin}(a) = 00$,
$\mathrm{bin}(b) = 01$, and $\mathrm{bin}(c) = 10$,
\[
    \mathrm{san}(w) = \underbrace{00\, 0}_{a}\ \underbrace{01\, 0}_{b}\ \underbrace{00\, 0}_{a}\ \underbrace{10\, 0}_{c}\ \underbrace{00\, 0}_{a}\ \underbrace{01\, 0}_{b}\ \underbrace{00\, 0}_{a}\ 111.
\]
Clearly, $w$ consists only of binary letters $0$ and $1$.
Moreover, after the transformation, we have the following, which is the purpose of considering this transformation.
\begin{lemma}\label{lemma:sanitize-noprefix}
Let $1\leq p < q\leq |\mathrm{san}(w)|-\ell$.
Then, $\mathrm{san}(w)[q \btw |\mathrm{san}(w)|]$ is not a prefix of $\mathrm{san}(w)[p \btw |\mathrm{san}(w)|]$.
\end{lemma}
\begin{proof}
Since $q\leq |\mathrm{san}(w)|-\ell$, $\mathrm{san}(w)[q \btw |\mathrm{san}(w)|]$ ends with $\mathrm{san}(w)[|\mathrm{san}(w)|-\ell \btw |\mathrm{san}(w)|]=1^{\ell+1}$.
However, by definition, $\mathrm{san}(w)$ contains no substring of the form $1^{\ell+1}$ other than $\mathrm{san}(w)[|\mathrm{san}(w)|-\ell \btw |\mathrm{san}(w)|]$.
Therefore, $\mathrm{san}(w)[p \btw |\mathrm{san}(w)|]$ cannot contain $1^{\ell+1}$ other than at its end, and thus, cannot have $\mathrm{san}(w)[q \btw |\mathrm{san}(w)|]$ as a prefix. 
\end{proof}

Moreover, for a \regex{} $E$, we define the \regex{} $\mathrm{bin}(E)$ as the \regex{} obtained by replacing every letter symbol $\sigma$ occurring in $E$ by the string $\mathrm{bin}(\sigma)\ 0$.
For example, if $\Sigma = \{a, b, c\}$, $E=(a+b)^*(b+c)^*$, $\mathrm{bin}(a) = 00$,
$\mathrm{bin}(b) = 01$, and $\mathrm{bin}(c) = 10$,
\[
    \mathrm{bin}(E) = (\underbrace{00\, 0}_{a}+ \underbrace{01\, 0}_{b})^*\ (\underbrace{01\, 0}_{b} + \underbrace{10\, 0}_{c})^*.
\]
Furthermore, we define
\[
    \mathrm{san}(A):= \mathrm{bin}(A),\ 
    \mathrm{san}(B):= \mathrm{bin}(B),\ 
    \mathrm{san}(C):= \mathrm{bin}(C),\ \text{and}\ 
    \mathrm{san}(D):= \mathrm{bin}(D)\ 1^{\ell+1}.
\]

The following lemma is clear from the definition.
\begin{lemma}
For a \regex{} $E$ and a string $w$, 
$\lang{\mathrm{bin}(E)}=\{\mathrm{san}(w): w\in \lang{E}\}$.
\end{lemma}

The following lemma ensures that this transformation indeed yields an equivalent problem, which is also clear from the definition.
\begin{lemma}\label{lemma:sanitize-equiv}
The following two statements are equivalent.
\begin{itemize}
    \item $w$ can be decomposed as $w = w_A\, w_B\, w_C\, w_B\, w_D$ with $w_A \in \lang{A}$, $w_B \in \lang{B}$, $w_C \in \lang{C}$, and $w_D \in \lang{D}$.
    \item $\mathrm{san}(w)$ can be decomposed as $\mathrm{san}(w) = w'_A\, w'_B\, w'_C\, w'_B\, w'_D$ with $w'_A \in \lang{\mathrm{san}(A)}$, $w'_B \in \lang{\mathrm{san}(B)}$, $w'_C \in \lang{\mathrm{san}(C)}$, and $w'_D \in \lang{\mathrm{san}(D)}$.
\end{itemize}
\end{lemma}

The following lemma is clear from the definition of $\mathrm{san}(D)$, but critical for our further use.
\begin{lemma}\label{lemma:sanitize-longD}
$\lang{\mathrm{san}(D)}$ consists of only strings with length at least $\ell+1$.
\end{lemma}

Hereafter, throughout this paper, we refer to the string $\mathrm{san}(w)$ and the \regex{}es $\mathrm{san}(A)$, $\mathrm{san}(B)$, $\mathrm{san}(C)$, and $\mathrm{san}(D)$ simply as $w$, $A$, $B$, $C$, and $D$, respectively.
In particular, we assume that $\Sigma$ is binary and there exists an integer $\ell$ such that
\begin{itemize}
    \item for each $1 \leq p < q \leq |w| - \ell$, $w[q, |w|]$ is not a prefix of $w[p \btw |w|]$ (Lemma~\ref{lemma:sanitize-noprefix}), and
    \item $\lang{D}$ consists only of strings of length at least $\ell + 1$ (Lemma~\ref{lemma:sanitize-longD}).
\end{itemize}
By Lemma~\ref{lemma:sanitize-equiv}, this rewriting is an equivalent transformation.

%% file: onerewb/normalization.tex

\section{Decomposing ABCBD to Two Subproblems}
\label{section:two_subproblems}


The goal of this section is to decompose the ABCBD problem into the following two subproblems, namely, the \emph{$XYYZ$ problem} and the \emph{branching $ABCBD$ problem}.


\begin{itembox}[l]{$XYYZ$ Problem}
\noindent{}\textbf{Fixed Objects}: \regex{es} $X$, $Y$, and $Z$. \\
\noindent{}\textbf{Input}: A binary string $w$. \\
\noindent{}\textbf{Task}: Deciding if we can decompose $w$ into $w = w_X\, w_Y\, w_Y\, w_Z$ so that $w_X\in \lang{X}$, $w_Y\in \lang{Y}$, and $w_Z\in \lang{Z}$.
\end{itembox}
\begin{itembox}[l]{Branching $ABCBD$ Problem}
\noindent{}\textbf{Fixed Objects}: \regex{es} $A$, $B$, $C$, and $D$. \\
\noindent{}\textbf{Condition on $D$}: Following the argument of Section~\ref{section:sanitize}, we assume $\epsilon \notin \lang{D}$. \\
\noindent{}\textbf{Condition on $B$}: Following the argument of Section~\ref{section:removing-emptystring-from-B}, we assume $\epsilon \notin \lang{B}$. \\
\noindent{}\textbf{Input}: A binary string $w$. \\
\noindent{}\textbf{Task}: Deciding if we can decompose $w$ into $w = w_A\, w_B\, w_C\, w_B\, w_D$ so that $w_A\in \lang{A}$, $w_B\in \lang{B}$, $w_C\in \lang{C}$, $w_D\in \lang{D}$, and 
\begin{itemize}
    \item $\sigma\neq \sigma'$, where $\sigma$ and $\sigma'$ are the first letters that appear immediately after the first and second occurrence of $w_B$ in this decomposition, respectively.
\end{itemize}
\end{itembox}
Let us illustrate the difference between the Branching $ABCBD$ problem and the (vanilla) $ABCBD$ problem with examples. 
Consider $w = 0110110$, $A = C = \Sigma^*$, and $B=D=\Sigma^+:=\Sigma \Sigma^*$.

\noindent{}\emph{Example 1}.
The decomposition $w = \underbrace{0}_{w_A}\underbrace{11}_{w_B}\underbrace{\underline{0}}_{w_C}\underbrace{11}_{w_B}\underbrace{\underline{0}}_{w_D}$ is a feasible decomposition for the vanilla $ABCBD$ problem, but not for the branching version, because the characters appearing immediately after both occurrences of $w_B$ are the same letter $\underline{0}$.

\noindent{}\emph{Example 2}.
The decomposition $w = \underbrace{01}_{w_A}\underbrace{1}_{w_B}\underbrace{\underline{0}}_{w_C}\underbrace{1}_{w_B}\underbrace{\underline{1}0}_{w_D}$ is feasible for both problems, since the character appearing immediately after the first $w_B$ is $\underline{0}$ and the character after the second $w_B$ is $\underline{1}$.

\noindent{}\emph{Example 3}.
The decomposition $w = \underbrace{0}_{w_A}\underbrace{1}_{w_B}\underbrace{\epsilon}_{w_C}\underbrace{\underline{1}}_{w_B}\underbrace{\underline{0}110}_{w_D}$ is also feasible for both problems, since the character appearing immediately after the first $w_B$ is $\underline{1}$ and the character after the second $w_B$ is $\underline{0}$.


\subsection{Removing the empty string $\epsilon$ from $B$}
\label{section:removing-emptystring-from-B}

Let $A$, $B$, $C$, and $D$ be classical \regex{es} (without backreferences) over a finite alphabet $\Sigma$.
Recall that
\[
\lang{ A\;(B)_x\;C\;\readv{x}\;D } = \set{ w_A w_B w_C w_B w_D : w_A \in A, w_B \in B, w_C \in C, w_D \in D }.
\]

For the \regex{} $B$,
by Proposition~\ref{prop:epsilon-nfa:epsilon-remove}, there exists $B_{-\epsilon}$ whose language is $\lang{B} \setminus \set{ \epsilon }$.
The following trivial case analysis is the first step in our decomposition.
\begin{proposition}
If $\epsilon \in \lang{B}$, we have
\[
    \lang{A \, (B)_x \, C \, \readv{x} \, D} = \lang{A \, C \, D} \cup \lang{A \, (B_{-\epsilon})_x \, C \, \readv{x} \, D}.
\]
Otherwise, we have
\[
    \lang{A \, (B)_x \, C \, \readv{x} \, D} = \lang{A \, (B_{-\epsilon})_x \, C \, \readv{x} \, D}.
\]
\end{proposition}

Since $A C D$ is just a \regex{} rather than \rewb{},
we can solve the string-matching problem for the expression in linear time.
We hereafter focus on the expression $A \, (B_{-\epsilon})_x \, C \, \readv{x} \, D$.
From now on, for simplicity, we write $B$ for $B_{-\epsilon}$: i.e., $\epsilon \not\in \lang{B}$.

\subsection{The Decomposition}
\label{subsection:ABCBD to XYYZ and Branching ABCBD}



For an $\epsilon$-NFA $N = (Q, \Sigma, \Delta, \mathcal{E}, q^0, F)$,
we define $N^{q' \to F'}$ as the NFA $(Q, \Sigma, \Delta, \mathcal{E}, q', F')$ obtained
by replacing the initial state of $N$ with $q'$ and the set of accepting states with $F'$.

Let $N_C = (Q_C, \Sigma, \Delta_C, \mathcal{E}_C, q^0_C, F_C)$ be a NFA that accepts $\lang{C}$.
We similarly define $N_D = (Q_D, \Sigma, \Delta_D, \mathcal{E}_D, q^0_D, F_D)$.
For the main lemma of this section, we define our building block.
For states $q_C\in Q_C$ and $q_D\in Q_D$ of $N_C$ and $N_D$,
by Proposition~\ref{prop:epsilon-nfa:product} , there exists a \regex{} $B_{q_C,\:q_D}$ that satisfies the following:
\[
\lang{B_{q_C,\:q_D}} = \lang{B} \cdot \left(\lang{C^{q^0_C \to \set{ q_C }}} \cap \lang{D^{q^0_D \to \set{ q_D }}}\right).
\]

We now present the main decomposition lemma of this section.
\begin{lemma}\label{lem:two_subcases}
Let $w$ be an input string. The following two conditions (1) and (2) are equivalent:
\begin{enumerate}
    \item[(1)] We have $w  \: \in A \: (B)_x \: C \: \readv{x} \: D$.
    \item[(2)] We have either of the following.
    \begin{enumerate}[align=left]
        \item[\textbf{$XYYZ$ Case:}] For some $f_C\in F_C$ and $q_D\in Q_D$, $w$ is matched with the following \rewb:
        \[
            A \: (B_{f_C,\,q_D})_x \ \readv{x}\ D^{q_D \to F_D}.
        \] 
        \item[\textbf{Branching Case:}] For some $q_C\in Q_C$ and $q_D\in Q_D$, $w$ is matched with the following \rewb:
        \[
            A \: (B_{q_C,\,q_D})_x\ C^{q_C \to F_C}\ \readv{x}\ D^{q_D \to F_D},
        \]
        along with the next letter of the substring matched with $(B_{q_C,\,q_D})_x$, $\sigma$, and the next letter of the substring matched with $\readv{x}$, $\sigma'$, are different ($\sigma \neq \sigma'$).
    \end{enumerate}
\end{enumerate}
\end{lemma}

\begin{proof}
%
\textbf{(1) $\Rightarrow$ (2):}
We decompose $w$ as $w = w_A w_B w_C w_B w_D$ where $w_A \in \lang{A}$, $w_B \in \lang{B}$, $w_C \in \lang{C}$, and $w_D \in \lang{D}$.
Using the longest common prefix $v$ of $w_C$ and $w_D$, we can decompose $w_C$ and $w_D$ as follows:
\[
w_C = v\, w'_C, \quad w_D = v\, w'_D.
\]
We note that $w'_D\neq \epsilon$ holds because of the following reason:
From Lemma~\ref{lemma:sanitize-longD}, we have $|w_D|\geq \ell+1$. 
Therefore, from Lemma~\ref{lemma:sanitize-noprefix} and the fact $w_B\neq \epsilon$, $w_D$ is not a prefix of the suffix $w_C\, w_B\, w_D$.
Thus, we have $v\neq w_D$, and thus, $w'_D\neq \epsilon$.


\paragraph*{Case1: $w'_C = \epsilon$.}
It means that $w_C = v$ and $w = w_A \, w_B \, w_C \, w_B \, \underbrace{w_C \, w'_D}_{w_D}$.

\begin{itemize}
\item On the NFA $N_C$, let $N_C(w_C)$ be the set of possible states after reading $w_C$. $N_C(w_C)$ contains an accepting state $f_C \in F_C$ because $w_C \in \lang{C}$.
\item On the NFA $N_D$, let $N_D(v)$ be the set of possible states after reading $v = w_C$.
Since $w_D = v\, w'_D \in \lang{D}$, there is a state $q_D\in N_D(v)$ such that $w'_D \in \lang{D^{q_D \to F_D}}$.
\end{itemize}

Since $w_B w_C \in \lang{B_{f_C,\,q_D}}$, we have
\[
w = w_A \ \, w_B w_C \ \, w_B w_C \ \, w'_D \in 
\lang{A \ (B_{f_C,\,q_D})_x\ \readv{x}\ D^{q_D \to F_D}},
\]
where $w_A \in \lang{A}$, $w_B w_C \in \lang{B_{f_C,\,q_D}}$, and $w'_D \in \lang{D^{q_D \to F_D}}$.
Thus, this case matches the $XYYZ$ case.

\paragraph*{Case2: $w'_C \neq \epsilon$.}
It means that $w = w_A\ w_B\ \underbrace{v\ w'_C}_{w_C}\ w_B\ \underbrace{v\ w'_D}_{w_D}$. Moreover, from the fact that $v$ is the longest common prefix of $w_C$ and $w_D$, the first letters of $w'_C$ and $w'_D$ are different.

\begin{itemize}
\item On the NFA $N_C$, let $N_C(v)$ be the set of possible states after reading $v$. 
Since $w_C = v\, w'_C \in \lang{C}$, there is a state $q_C\in N_C(v)$ such that $w'_C \in \lang{C^{q_C \to F_C}}$.
\item On the NFA $N_D$, let $N_D(v)$ be the set of possible states after reading $v$.
Since $w_D = v\, w'_D \in \lang{D}$, there is a state $q_D\in N_D(v)$ such that $w'_D \in \lang{D^{q_D \to F_D}}$.
\end{itemize}
Then, the string $w = w_A (w_B v) w'_C (w_B v) w'_D$ matches the following expression:
\[
A\ (B_{q_C,\ q_D})_x\ ( C^{q_C \to F_C}) \readv{x}\  (D^{q_D \to F_D})
\]
with $w_B v \in \lang{B_{q_C, q_D}}$, $w'_C \in \lang{C^{q_C \to F_C}}$, and $w'_D \in \lang{D^{q_D \to F_D}}$.
Moreover, in such a matching, the letters that appear immediately after the two substrings $w_B v$ are different.
Thus, this case matches the branching case.

\textbf{(2) $\Rightarrow$ (1):}
\paragraph{$XYYZ$ case:}
Assume
\[
w \in \lang{A\ (B_{f_C, q_D})_x\ \readv{x}\ D^{q_D \to F_D}},
\]
where $f_C \in F_C$ and $q_D \in Q_D$.
We decompose $w $ as
\[
w  = w_A w_B w_B w_D,
\] 
where $w_A \in \lang{A}$, $w_B \in \lang{B_{f_C, q_D}}$, and $w_D \in \lang{D^{q_D \to F_D}}$.

By the definition of $B_{f_C, q_D}$,
we can further decompose $w_B$ as $w_B = w_1 w_2$, where $w_1 \in \lang{B}$ and
$w_2 \in \lang{C^{q^0_C \to \set{f_C}}} \cap \lang{D^{q^0_D \to \set{q_D}}}$.
Particularly, $w_2 \in \lang{C^{q^0_C \to \set{ f_C}}} \subseteq \lang{C}$ and $w_2 \, w_D \in \lang{D^{q^0_D \to \set{q_D}}}\cdot \lang{D^{\{q_D\} \to F_D}}\subseteq \lang{D}$.
Therefore, we have
\[
w = w_A\ \underbrace{w_1\ w_2}_{w_B}\ \underbrace{w_1\ w_2}_{w_B} w_D \in \lang{A (B)_x\ C\ \readv{x}\ D}
\]
with $w_A \in \lang{A}$, $w_1 \in \lang{B}$, $w_2 \in \lang{C}$, and $w_2 w_D \in \lang{D}$.

\paragraph{Branching Case:}

Assume
\[
w \in \lang{A \: (B_{q_C,\,q_D})_x\ C^{q_C \to F_C} \ \readv{x}\ D^{q_D \to F_D}},
\]
where $q_C \in Q_C$ and $q_D \in Q_D$ is a state of $C$ and $D$, respectively.
We decompose $w$ as
\[
w = w_A\,w_B\,w_C\,w_B\,w_D,
\]
where $w_A \in \lang{A}$, $w_B \in \lang{B_{q_C, q_D}}$, $w_C \in \lang{C^{q_C \to F_C}}$, and $w_D \in \lang{D^{q_D \to F_D}}$.

By the definition of $B_{q_C, q_D}$, we can further decompose $w_B$ as $w_B = w_1 w_2$, where $w_1 \in \lang{B}$ and
$w_2 \in \lang{C^{q^0_C \to \set{q_C}}} \cap \lang{D^{q^0_D \to \set{q_D}}}$.
Particularly, $w_2 \, w_C \in \lang{C^{q^0_C \to \set{q_C}}}\cdot \lang{C^{\{q_C\} \to F_C}}\subseteq \lang{C}$ and $w_2 \, w_D \in \lang{D^{q^0_D \to \set{q_D}}}\cdot \lang{D^{\{q_D\} \to F_D}}\subseteq \lang{D}$.
Therefore, we have
\[
w= w_A\ \underbrace{w_1\ w_2}_{w_B} w_C\ \underbrace{w_1\ w_2}_{w_B}  w_D \in \lang{A (B)_x\ C\ \readv{x}\ D}
\]
with $w_A \in \lang{A}$, $w_1 \in \lang{B}$, $w_2 w_C \in \lang{C}$, and $w_2 w_D \in \lang{D}$.
\end{proof}

Now, we reduce the $ABCBD$ problem into the $XYYZ$ problem and the branching $ABCBD$ problem.
Let $w$ be the given string.
If $\epsilon\in B$, we check whether $w$ matches the \regex{} $A\,C\,D$ and replace $B$ by $B_{-\epsilon}$.
If it matches, we return $\mathsf{yes}$.
Otherwise, Lemma~\ref{lem:two_subcases} ensures that $w  \: \in A \: (B)_x \: C \: \readv{x} \: D$ if and only if 
\begin{itemize}
    \item $w$ is a $\mathsf{yes}$-instance of the $XYYZ$ problem for $X:=A$, $Y:=B_{f_C,q_D}$, and $Z:=D^{q_D\to F_D}$ for some $f_C\in F_C$ and $q_D\in Q_D$, or
    \item $w$ is a $\mathsf{yes}$-instance of the branching $ABCBD$ problem for $A:=A$, $B:=B_{q_C,q_D}$, $C:=C^{q_C\to F_C}$, and $D:=D^{q_D\to F_D}$ for some $q_C\in Q_C$ and $q_D\in Q_D$.
\end{itemize}
This algorithm solves the $ABCBD$ problem by calling oracles for the $XYYZ$ problem and the branching $ABCBD$ problem.
Moreover, it internally solves the $XYYZ$ problem $|F_C||Q_D|$ times and the branching $ABCBD$ problem $|Q_C||Q_D|$ times, both of which are constant.
Therefore, we have the following.
\begin{lemma}
Let $f$ be a function and assume both of the $XYYZ$ problem and the branching $ABCBD$ problem admit $O(f(|w|))$-time algorithms. Then, the $ABCBD$ problem admits an $O(f(|w|))$-time algorithm.
\end{lemma}



%% file: onerewb/xyyz.tex

\section{$(X\,YY\,Z)$-Problem}
\label{section:XYYZ problem}

We first restate our subproblem, the $XYYZ$ problem.
\begin{itembox}[l]{$XYYZ$ Problem}
\noindent{}\textbf{Fixed Objects}: \regex{es} $X$, $Y$, and $Z$. \\
\noindent{}\textbf{Input}: A binary string $w$. \\
\noindent{}\textbf{Task}: Deciding if we can decompose $w$ into $w = w_X\, w_Y\, w_Y\, w_Z$ so that $w_X\in \lang{X}$, $w_Y\in \lang{Y}$, and $w_Z\in \lang{Z}$.
\end{itembox}

\subsection{Crochemore's Maximal Local Power Algorithm}
\label{section:MR-decomposition}

%
A string $x$ is \emph{repetitive} or \emph{power} if it is represented as $x=y^e$ for some string $y$ and an integer $e\geq 2$.
If $x$ is not repetitive, $x$ is \emph{primitive}.

For example,
(1) a string $abababab$ is repetitive because $ab ab ab ab = (ab)^4$;
(2) a string $abababa = (ab)^3 a$ is \emph{primitive}.

As in $(ab)^4 = (abab)^2$, there may be multiple representations for one repetitive string.
For the unique representation, we consider the decomposition by primitive words.
More precisely, for a repetitive string $x$, we denote it as $(p, k)$ where $x = p^k$ and $p$ is \emph{primitive}.
\begin{proposition}[{\cite[Chapter~9]{Crochemore:2007}}]
If a string $x$ is repetitive, there exists just one representation $(p, k)$.
\end{proposition}
\begin{proof}
Since $x$ is repetitive, for some string $v$, $x = v^i$ with $i \geq 2$. If $v$ is not primitive,
we (repeatedly) decompose $v$ and obtain a primitive factor $p$ such as $x = p^k$ .
If $x = p^k = q^r$ held for different primitives $p$ and $q$, then $x = y^e$ for $y = x[1 .. \mathit{gcd}(|p|, |q|)]$ by Fine and Wilf's theorem (Lemma~\ref{lemma:FineWilf}).
Since $|p| \bmod |y| = |q| \bmod |y| = 0$, it contradicts the assumption that $p$ and $q$ are primitives.
\end{proof}

\def\btw{{\,..\,}}

Let $w$ be a string.
A substring $x$ of $w$ is \emph{local power} at an index $p$
if $x$ is repetitive and $x$ starts from the index $p$,  $x = w[p \btw p + |x|)$.
For a primitive $\theta$, $1 \leq p \leq |w|$, and $e \geq 2$, we write $(p, \theta, e)$ to denote the local power $\theta^e$ at the index $p$.
A local power $(p, \theta, e)$ is (two-sided) \emph{maximal} if
(1) it is not extensible to the left: i.e., $w[p - |\theta| \,..\, p) \neq \theta$ or $p-|\theta| < 1$; and
(2) it is not extensible to the right: i.e., $\theta \neq w[p + e|\theta| \, .. \, p + (e+1)|\theta|)$ or $p+(e+1)|\theta| > n+1$.

\paragraph{Example.}

For a string $x = aabaabaabba$,
there are two kinds of maximal powers:
\begin{itemize}
\item $|\theta|=1$: $(1, a, 2)$, $(4, a, 2)$, $(7, a, 2)$, and $(9, b, 2)$.
\item $|\theta|=3$: $(1, aab, 3)$, $(2, aba, 2)$, and $(3, baa, 2)$.
\end{itemize}
It is worth noting that $(4, aab, 2)$ is a local power but not maximal because it is covered by $(1, aab, 3)$.
The following result is shown by Crochemore~\cite{Crochemore:1981} and plays a critical role in our algorithm.
\begin{theorem}[{\cite{Crochemore:1981}},{\cite[Chapter~9]{Crochemore:2007}}]
Let $w$ be a string and $n := |w|$.
We can compute the set of all maximal local powers of $w$ in $O(n \log n)$-time.
The size of this set is also $O(n \log n)$.
\end{theorem}

We show that we can limit the search space for solutions of the $XYYZ$ problem as follows.
\begin{lemma}\label{lem:max_local_power_reduce}
The following conditions are equivalent:
\begin{enumerate}
\item[(1)] We have a desired decomposition $w = x \: y\,y \: z$ where $x \in \lang{X}$, $y \in \lang{Y}$, and $z \in \lang{Z}$.
\item[(2)] There is a maximal local power $(p, \theta, e)$ such that
\[
w = w_1\ \underbrace{\theta^i\ \theta^j\,\theta^j\,\theta^k}_{(p, \theta, e)}\ w_2,
\qquad
\text{where} \quad (w_1 \theta^i) \in \lang{X}, \theta^j \in \lang{Y}, (\theta^k w_2) \in \lang{Z}.
\]
\end{enumerate}
\end{lemma}
\begin{proof}
\textbf{Trivial Case $(2) \Rightarrow (1)$:}
If $(2)$ holds, we just set $x := w_1 \theta^i$, $y := \theta^j$, and $z := \theta^k w_2$.

\noindent{}\textbf{Case $(1) \Rightarrow (2)$:} Consider a decomposition $w = x \, \underline{y \, y} \, z$.
Since the string $y\,y$ is a power, there is a maximal local power $(p,\theta, e)$ that covers the substring $\underline{y\,y}$; that is, $y=\theta^j$ for some $j$ and the substring $\underline{y\,y}$ begins at the index $p+i|\theta|$ for some $0\leq i \leq e-2j$.
By taking this maximal local power $(p, \theta, e)$ and integers $i,j$, together with setting $k=e-i-2j$, we have $x=w_1\, \theta^i$, $y = \theta^{j}$, and $z=\theta^k w_2$, where $w_1=w[1\,..\,p)$ and $w_2=w[p+e|\theta|\,..\,|w|]$, respectively.
\end{proof}


\subsection{Our Tool on Ultimately Periodic Set}

The following lemma is fundamental.
\begin{lemma}
For a finite monoid $\mathcal{M}$ and its elements $\delta$ and $\delta'$, $\{i\geq 0\colon \delta^i=\delta'\}$ is an ultimately periodic set.  
\end{lemma}
\begin{proof}
Consider the sequence $(\delta^0,\delta^1,\dots )$. Since $\mathcal{M}$ is finite, there is an element $\delta''$ that appears at least twice in this sequence. 
Let $\delta''=\delta^{k_1}=\delta^{k_2}$ for $k_1 < k_2$. Then, for $k\geq k_2$, we have $\delta^{k}=\delta^{k-k_2}\delta^{k_2}=\delta^{k-k_2}\delta^{k_1}=\delta^{k-(k_2-k_1)}$. Thus, the set $\{i\geq 0\colon \delta^i=\delta'\}$ is periodic for $i\geq k_1$.
\end{proof}
To give an algorithm for the $X YY Z$ problem, we repeatedly use the following lemma.
\begin{lemma}\label{lem:wtw_is_ult_periodic}
Let $E$ be a \regex{}, $\delta\in \transm{E}$, and $w_1, \theta, w_2$ be strings.
Then, $\{i\geq 0\colon \transm{E}(w_1\, \theta^i\,w_2)=\delta\}$ is ultimately periodic.
Moreover, the parameters $\mu$ and $\lambda$ of that set are bounded by a constant that depends only on $E$.
\end{lemma}
\begin{proof}
Let $\delta_1=\transm{E}(w_1)$, $\delta_{\theta}=\transm{E}(\theta)$, and $\delta_2=\transm{E}(w_2)$.
Let $\mathcal{D}=\{\delta'\in \transm{E}\colon \delta_1\comp \delta'\comp \delta_2=\delta\}$.
Then, the desired set is $\bigcup_{\delta'\in \mathcal{D}}\{i\geq 0\colon \delta_{\theta}^i=\delta'\}$, which is ultimately periodic because it is the union of a finite number of ultimately periodic sets.
Moreover, since that set depends only on transitions $\delta_1$, $\delta_{\theta}$, and $\delta_2$, the parameters $\mu$ and $\lambda$ of that set are bounded by a constant that depends only on $E$.
\end{proof}

\subsection{Checking each maximal local power in $O(1)$-time}
\label{section:solving semi-linear constraints}

Using Lemma~\ref{lem:max_local_power_reduce},
for a given maximal local power $(p, \theta, e)$,
we hereafter focus on the problem of whether we can decompose $w$ as
\[
w = \underbracket{\,w_1 \: \theta^i}_{\in X}\  \underbracket{\,\theta^j\,}_{\in Y} \underbracket{\,\theta^j\,}_{\in Y} \underbracket{\,\theta^k \, w_2\:}_{\in Z}
\qquad i + 2j + k = e, \tag{$\star$}
\]
where the prefix $w_1=[1\,..\,p)$ and the suffix $w_2=[p+e|\theta|\,..\,|w|]$ are uniquely determined from $(p,\theta, e)$.

To solve the above problem ($\star$) in $O(1)$-time, we introduce some sets of exponents that are candidates of integers $i$, $j$, and $k$.
Let
\[
\mathcal{S}_X := \set{ i \geq 0 : X(w_1\,\theta^i)\text{ is accepting} } = \bigcup_{\substack{\delta\in \mathcal{M}(X)\\ \delta \text{ is accepting}}}\set{i\geq 0 : \transm{X}(w_1\,\theta^i)=\delta},
\]
\[
\mathcal{S}_Y := \set{ j \geq 0 : Y(\theta^j)\text{ is accepting} } = \bigcup_{\substack{\delta\in \mathcal{M}(Y)\\ \delta \text{ is accepting}}}\set{j\geq 0 : \transm{Y}(\theta^j)=\delta},
\]
\[
\mathcal{S}_Z := \set{ k \geq 0 : Z(\theta^k\,w_2)\text{ is accepting} } = \bigcup_{\substack{\delta\in \mathcal{M}(Z)\\ \delta \text{ is accepting}}}\set{k\geq 0 : \transm{Z}(\theta^k\,w_2)=\delta},
\]
From Lemmas~\ref{lem:ult_operations}~and~\ref{lem:wtw_is_ult_periodic}, all of $\mathcal{S}_X$, $\mathcal{S}_Y$, and $\mathcal{S}_Z$ are ultimately periodic whose parameters $\mu$ and $\lambda$ are bounded by a constant that depends only on $X$, $Y$, and $Z$, respectively. 

Now, we define
\[
\mathcal{S} := \mathcal{S}_X + (2\cdot \mathcal{S}_Y) + \mathcal{S}_Z.
\]
From Lemma~\ref{lem:ult_operations}, $\mathcal{S}$ is also ultimately periodic whose parameters $\mu$ and $\lambda$ are bounded by a constant that depends only on $X$, $Y$, and $Z$.
We have the following.
\begin{lemma}
There is a way to decompose $w$ as in ($\star$) if and only if $e\in \mathcal{S}$. 
\end{lemma}
\begin{proof}
Assume there is a way to decompose $w$ as in ($\star$).
Let $i$, $j$, and $k$ be the exponents in that decomposition.
Since $w_1\,\theta^i$, $\theta^j$, and $\theta^k\, w_2$ are accepted by $X$, $Y$, and $Z$, respectively, we have $i\in \mathcal{S}_X$, $j\in \mathcal{S}_Y$, and $k\in \mathcal{S}_Z$.
Thus, from the definition of the Minkowski sum and the $2$-multiplication, we have $e=i+2j+k\in \mathcal{S}$.

Conversely, assume $e\in \mathcal{S}$.
Then, there is a tuple of integers $i,j,k$ with $i\in \mathcal{S}_X$, $j\in \mathcal{S}_Y$, $k\in \mathcal{S}_Z$ such that $i+2j+k=e$.
From the definition of $\mathcal{S}_X$, $\mathcal{S}_Y$, and $\mathcal{S}_Z$, we have $w_1\,\theta^i\in X$, $\theta^j\in Y$, and $\theta^k\, w_2\in Z$. Thus, the lemma is proved.
\end{proof}

Clearly, for $\mathcal{S}=(\mu,M,\lambda,T)$ and $e\leq |w|$, whether $e\in \mathcal{S}$ or not can be determined in $O(\mu + \lambda)=O(1)$-time.
Moreover, $\mathcal{S}$ itself can be computed in constant time from the transitions $\transm{X}(w_1)$, $\transm{X}(\theta)$, $\transm{Y}(\theta)$, $\transm{Z}(\theta)$, and $\transm{Z}(w_2)$.
Furthermore, those transitions can also be computed in constant time using the factorization forest data structure $\fforest_{X}$, $\fforest_{Y}$, and $\fforest_{Z}$, which can be constructed in linear time.
Thus, with linear time pre-computation, we can answer the problem defined in ($\star$) in constant time, which yields the following.
\begin{lemma}
The XYYZ problem can be solved in $O(|w|\log |w|)$ time.
\end{lemma}

The above discussion still works when we replace the specific constant $2$ with any constant integer $\alpha\geq 2$.
Thus, we have Theorem~\ref{thm:XYYYYYZ}.
\XYYZ*









%% file: onerewb/uezato-overview.tex

\section{Algorithm Overview for the Branching $ABCBD$ Problem}

\label{section:branchingABCBD-to-twosubproblems}

\newcommand{\entire}{w}

\subsection{Notation}
\label{section:hevy-and-light-indices}

Recall that $\stree_w$ is the suffix tree of $w$.
We simply write $\stree$ for $\stree_w$.

Let $v$ be a node of $\stree$.
We write $\mathscr{I}_v$ for the set of indices of suffix nodes of the subtree rooted at $v$ and it is formally defined as follows:
\[
\mathscr{I}_v := \set{ \Occ{u} : u \in \stree[v], \isSuffix{u} = 1}.
\]

We also write $\heavy_v$ and $\light_v$ for the set of indices of suffix nodes in heavy and light subtrees of $v$, respectively.
More formally, we define $\heavy_v$ and $\light_v$ as follows.
Let $u_H$ (resp. $u_L$) be the child of $v$ where $v \to u_H$ is a heavy edge (resp. $v \to u_L$ is a light edge). If there is no such child, we set $u_H=\bot$ (resp. $u_L=\bot$).
Then, we define
\[
\begin{array}{lcl}
  \heavy_v & := & \set{ \Occ{u} : u \in \stree[u_H],\ \isSuffix{u} = 1 }, \\
  \light_v & := & \set{ \Occ{u} : u \in \stree[u_L],\  \isSuffix{u} = 1 } \cup
  \begin{cases}
  \set{ v } & \text{if } \isSuffix{v} = 1 \\
  \emptyset & \text{otherwise}
  \end{cases}
\end{array}  
\]
where we set $\stree[\bot] = \emptyset$.

Note that $v$ itself is in $\light_v$ when $v$ is a suffix node.
Moreover, for each non-leaf node $v$, 
$\mathscr{I}_v = \heavy_{v} \cup \light_v$ and
$\heavy_{v} = \heavy_{u_H}\cup \light_{u_H}$ clearly hold.

\subsection{From Branching ABCBD to Two Enumeration Problems}
\label{sec:two-enumeration-problems}


In this section, we introduce two subproblems \textsc{Left-$B$ Transition Enumeration} and \textsc{Right-$B$ Transition Enumeration} for decomposing the branching ABCBD problem.
We define both of those problems on a heavy-path $\hpath$ of the suffix tree $\stree_w$,
and utilize them to determine whether there is a branching decomposition of $w=w_A\ w_B\ w_C\ w_B\ w_D$ such that $w_B=\nword{v}$ for some $v \in \hpath$.

We consider the following branching decomposition of $w$:
\[
w = w_A\ w_B\ w_C\ w_B\ w_D,
\]
where $w_B = \nword{v}$ for some $v \in \pi$.
We denote $m_v := |w_B|$.

Let $p$ and $q$ be the indices of the left ends of the first and the second occurrences of $w_B$:
i.e., $w[1 \btw p) = w_A$, $w[p \btw p+m_v) = w_B$, $w[p+m_v \btw q) = w_C$, $w[q \btw q+m_v) = w_B$, and $w[q+m_v \btw |w|] = w_D$.
It is clear that $p, q \in \mathscr{I}_v = \heavy_v \cup \light_v$.
We illustrate this situation as follows:
\[
w = \underbrace{w[1\btw p)}_{w_A} \ \underbrace{w[p\btw p+m_v)}_{w_B} \ \underbrace{w[p+m_v\btw q)}_{w_C} \ \underbrace{w[q\btw q+m_v)}_{w_B} \ \underbrace{w[q+m_v\btw |w|]}_{w_D}.
\]

\newcommand{\cand}{\mathcal{K}}
\newcommand{\candp}{\mathit{Cand}}

Deciding if there is a branching decomposition for a node $v$ of $\stree$, we first check $\nword{v} = w_B \in \lang{B}$.
This check is done in $O(1)$-time by using the factorization forest $\fforest$ and Lemma~\ref{lemma:constant-eval-by-ff}.
Then, we consider the following set of candidate transitions of $C$:
\[
\cand =
\set{ \delta \in \transm{C} : \exists p, q\in \mathscr{I}_v.\ \cand(p, q, \delta, m_v) \land q+m_v \leq |w| \land w[p+m_v] \neq w[q+m_v]
},
\]
where $\cand(p, q, \delta, m)$ is a predicate defined as follows:
\[
\cand(p, q, \delta, m) :=
  w[1 \btw p) \in \lang{A} \land
  p + m \leq q \land
  \delta = \evalc{ [ p+m \btw q) } \land
  w[q+m \btw |w|] \in \lang{D}.
\]

Now, we rephrase our branching ABCBD problem using the above set $\cand$.

First, for every node $v$,
the following two conditions are equivalent:
\begin{itemize}
\item There is a branching decomposition $w = w_A\ w_B\ w_C\ w_B\ w_D$ with $w_B = \nword{v}$.
\item There is a transition $\trans$ in the set $\cand$ where $\trans$ is accepting on $C$.
\end{itemize}
This correspondence is clear from our condition on the branching ABCBD problem; that is, we are seeking the decomposition that satisfies the following condition:
\begin{itemize}
    \item the letters immediately after two occurrences of $w_B$'s are different.
\end{itemize}
It should be noted that, from the argument of Section~\ref{section:sanitize}, we assume $\epsilon \notin \lang{D}$.

Next, on a node $v$, assume $\trans \in \cand$ for some accepting transition $\delta$.
It means that there exist $p$ and $q$ such that $\cand(p, q, \delta, m_v)$ and $q + m_v \leq |w| \land w[p+m_v] \neq w[q+m_v]$.
Then, it holds that one of $\set{p, q}$ must belong to the heavy side $\heavy_v$ and the other must belong to the light side $\light_v$. It is confirmed as follows:
\begin{itemize}
\item Recall $p, q \in \mathscr{I}_v = \heavy_v \cup \light_v$.
\item The condition $q + m_v \leq |w| \land w[p+m_v] \neq w[q+m_v]$ equals the ``branching by next letter'' condition of the suffix tree.
\item Therefore, one of $\set{p, q}$ must belong to the heavy-side subtree $\stree[u_H]$ of $v$ and the other must belong to the light-side subtree $\stree[u_L]$ of $v$.
\end{itemize}

Finally, we can compute $\cand$ by solving the following two enumeration problems:
\begin{itemize}
\item Enumerating transitions $\delta$ where $p$ is in the heavy side $\heavy_v$ and $q$ is in the light side $\light_v$.
\item Enumerating transitions $\delta$ where $p$ is in the light side $\light_v$ and $q$ is in the heavy side $\heavy_v$.
\end{itemize}
More formally, the above two tasks correspond to computing the following two sets:
\[
\begin{array}{rll}
\cand_L & := & \set{ \delta \in \transm{C} :  \exists p \in \heavy_v.\, \exists q \in \light_v.\ \cand(p, q, \delta, m_v) }, \\
\cand_R & := & \set{ \delta \in \transm{C} :  \exists p \in \light_v.\, \exists q \in \heavy_v.\ \cand(p, q, \delta, m_v) }.
\end{array}
\]
For every node $v \in \stree$, it is clear that $\cand = \cand_L \cup \cand_R$.

\subsection{Left-B and Right-B} \label{sec:leftb-and-rightb-tasks}

According to the discussion in the previous section, we define the following problems:
\begin{itembox}[l]{\textsc{Left-$B$ Transition Enumeration}}
\noindent{}\textbf{Fixed Objects}: \regex{es} $A$, $B$, $C$, and $D$ and a string $w$. \\
\textbf{Input}: A heavy-path $\hpath$ of $\stree_w$. \\
\textbf{Task:} Compute the union of $\cand_L(v)$ for all $v \in \hpath$. More formally, enumerate all $\trans \in \transm{C}$ such that there is a node $v\in \pi$, $i\in \light_v$, and $j\in \heavy_v$ with
    \begin{itemize}
        \item $j + m_v \leq i$,
        \item $w[1 \btw j) \in \lang{A}$, $\nword{v} \in \lang{B}$, and $w[i + m_v \btw |w|] \in \lang{D}$, and
        \item $\delta = \transm{C}(\,w[ j + m_v \btw i)\,)$.
    \end{itemize}
Recall that $\nword{v} = w[j\btw j+m_v) =w[i\btw i+m_v)$ holds by definition of $m_v$.
\end{itembox}
\begin{itembox}[l]{\textsc{Right-$B$ Transition Enumeration}}
\noindent{}\textbf{Fixed Objects}: \regex{es} $A$, $B$, $C$, and $D$ and a string $w$. \\
\textbf{Input}: A heavy-path $\hpath$ of $\stree_w$. \\
\textbf{Task:} Compute the union of $\cand_R(v)$ for all $v \in \hpath$. More formally, enumerate all $\trans \in \transm{C}$ such that
there is a node $v \in \hpath$, $i \in \light_v$, and $j \in \heavy_v$ with
    \begin{itemize}
        \item $i + m_v \leq j$
        \item $w[1 \btw i) \in \lang{A}$, $\nword{v} \in \lang{B}$, and $w[j + m_v \btw |w|] \in \lang{D}$, and
        \item $\delta = \transm{C}(\,w[ i + m_v \btw j)\,)$.
    \end{itemize}
Recall that $\nword{v} = w[i\btw i+m_v) =w[j\btw j+m_v)$ holds by definition of $m_v$.
\end{itembox}

These problems are identical except for the positional relation; that is, which interval $[i, i+m_v)$ or $[j, j+m_v)$ corresponding to the string $w_B$ is located to the left of the other.

\begin{algorithm}[t!]
\caption{Algorithm Overview.}\label{new:alg:overview_all}
\Procedure{\emph{\Call{branching-ABCBD}{$w$, $A$, $B$, $C$, $D$}}}{
    Compute the suffix tree  $\stree_w$ and its heavy-light decomposition\;
    Compute the factorization forests $\fforest_A$, $\fforest_B$, $\fforest_C$, and $\fforest_D$ for $A$, $B$, $C$, and $D$, respectively\;
    \tcp{in $O(|w|)$-time}
    \For{Heavy-path $\hpath$ of $\stree_w$}{
        $\cand_L \gets $\textsc{Left-$B$ Transition Enumeration}($\hpath, w, A, B, C, D$)\;
        $\cand_R \gets $\textsc{Right-$B$ Transition Enumeration}($\hpath, w, A, B, C, D$)\;
        \If{There is an \emph{accepting} transition in $\cand_L \cup \cand_R$}{
            \textbf{return} $\mathsf{yes}$\;
        }
    }
    \textbf{return} $\mathsf{no}$\;
}
\end{algorithm}

The remaining part of this paper concentrates on solving those problems in $O(|\stree[\pi]|\log |w|)$ time, where we define $\stree[\pi] := \stree[u]$ using the topmost vertex $u$ of $\pi$.
Using oracles for those problems as sub-routines, we solve the branching $ABCBD$ problem as in Algorithm~\ref{new:alg:overview_all}.
Our algorithm first computes the suffix tree $\stree_w$ and its heavy-light decomposition, and factorization forests $\fforest_A$, $\fforest_B$, $\fforest_C$, and $\fforest_D$ for further use.
Then, for each heavy-path $\pi$, we compute the union of all $\cand_{L}$ and $\cand_{R}$ using our oracles.
Then, the answer is $\mathsf{yes}$ if there is an accepting transition in at least one of such $\cand_{L}$'s or $\cand_{R}$'s.
The following lemma proves that, using our oracles, Algorithm~\ref{new:alg:overview_all} correctly solves the branching $ABCBD$ problem efficiently.
\begin{lemma}
Algorithm~\ref{new:alg:overview_all} solves the branching ABCBD problem for $w$ in $O(|w| \log^2 |w|)$-time.
\end{lemma}
\begin{proof}
From above discussion, it is clear that Algorithm~\ref{new:alg:overview_all} correctly solves the branching ABCBD problem.
The remaining task is to bound the time complexity.

Computation of the suffix tree, heavy-light decomposition, and factorization forests takes $O(|w|)$-time. 
As we will show, \textsc{Left-$B$ Transition Enumeration} and \textsc{Right-$B$ Transition Enumeration} can be solved in $O(|\stree[\pi]|\log |w|)$-time.
Moreover, the sum of $|\stree[\pi]|$ over all heavy-paths $\pi$ is $O(|w| \log |w|)$, since the number of vertices in $\stree$ is $O(|w|)$ and each vertex of $\stree$ belongs to at most $O(\log |w|)$ different $\pi$'s.
Thus, the time complexity of Algorithm~\ref{new:alg:overview_all} is $O(|w|\log^2 |w|)$.
\end{proof}

For both two enumeration problems, we employ a common strategy to solve them: we process $v \in \hpath$ in bottom-up order, and for each $i \in \light_v$, we efficiently enumerate transitions corresponding to $j \in \heavy_v$ that satisfy the conditions using queries to certain data structures.
However, the data structures and algorithms for those two problems are different.
From the next section, we provide efficient algorithms for those two enumeration problems.

%% file: onerewb/new-right.tex

\section{Right-$B$ Transition Enumeration}

In this section, we solve \textsc{Right-$B$ Transition Enumeration}, while the description of the data structure used internally is postponed to Section~\ref{sec:data_structures_right}.

\label{section:solve:rightbenum}

\newcommand{\AddR}{\textsc{Add}_{\textrm{R}}}
\newcommand{\InitR}{\textsc{Initialize}_{\textrm{R}}}
\newcommand{\GetR}{\textsc{Get}_{\textrm{R}}}
\newcommand{\dfbox}[1]{%
  \tikz[baseline=(X.base)]\node[draw,dashed,inner sep=2pt](X){\strut #1};}
\newcommand{\bag}[1]{\lbag #1 \rbag}
\newcommand{\disjs}{\mathcal{B}}
\newcommand{\dsu}{\text{DSU}}
\newcommand{\members}{\text{members}}
\newcommand{\accum}{\textsc{Accum}}

\newcommand{\origboxw}[4]{%
  \node[draw,minimum width=#2,minimum height=0.8cm,inner sep=2pt,#3] (#1) {#4};}

\newcommand{\dashedboxw}[4]{%
  \node[draw,dashed,minimum width=#2,minimum height=0.8cm,inner sep=2pt,#3] (#1) {#4};}


\subsection{Overall Strategy}

For $v\in \pi$ and $i\in \light_v$, we define the set $\mathcal{R}(v,i)$ by
\[
\mathcal{R}(v, i) :=
    \set{ \evalc{ [i + m_v \btw j) } : j \in \heavy_v,\ i+m_v\leq j,\ w[j+m_v \btw |w|]\in \lang{D}}.
\]
Then, by definition, the desired output for \textsc{Right-$B$ Transition Enumeration} is
\[
    \bigcup_{\substack{v\in \pi\\ \nword{v}\in \lang{B}}}\ 
    \bigcup_{\substack{i\in \light_v\\ w[1 \btw i)\in \lang{A}}}\mathcal{R}(v,i).
\]
Our algorithm iterates over all $v\in \pi$ in bottom-up order, 
and then for each $v$ with $\nword{v}\in \lang{B}$, iterates over all $i\in \light_v$ with $w[1 \btw i)\in \lang{A}$ to compute $\mathcal{R}(v,i)$, and outputs their union.

Now, we focus on computing $\mathcal{R}(v,i)$.
The main obstacle here is the constraint $w[j + m_v \btw |w|] \in \lang{D}$.
To tackle this, we partition $\heavy_v$ into the following disjoint sets, which we call \emph{bags}, indexed by a transition $\trans$ of $D$:
\[
\bag{ \trans }_v := \set{ j\in \heavy_v : \evald{ w[j + m_v \btw |w|] } = \trans }.
\]
Then, $\heavy_v$ can be represented as
\[
\heavy_v = \bigcup_{\delta\in \transm{D}}\bag{ \trans }_v.
\]
Then, we can compute $\mathcal{R}(v, i)$ by computing
\[
    \set{ \evalc{ [i + m_v \btw j) } : j \in \bag{ \trans }_v,\ i+m_v\leq j},
\]
for each accepting transition $\trans$ of $D$ and taking their union.

\subsection{Disjoint-Set Union}

Since the partition $\heavy_v=\bigcup_{\delta \in \transm{D}}\bag{ \delta }_v$ depends on $v$, we should update it when we traverse upward along $\pi$.
Let $u$ be the parent of $v$.
The next task is to compute the partition $\heavy_u=\bigcup_{\delta \in \transm{D}}\bag{ \delta }_u$ from the partition $\heavy_v=\bigcup_{\delta \in \transm{D}}\bag{ \delta }_v$.

Let 
\[
    \tau := \evald{ \nword{v}(m_u \btw m_v] }~~(= \evald{ \nword{v}[m_u+1\,\btw\,m_v] }.
\]
Then, $\tau = \evald{ w[j + m_u \btw j + m_v) }$ holds for all $j\in \heavy_v\cup \light_v=\heavy_u$
because $\nword{v}[1\,\btw\,m_v]=w[j \btw j+m_v)$.
(Recall that we employ the $1$-origin indexing, $w = w[1] w[2] \ldots w[|w|]$.)

Thus, for each $\delta\in \transm{D}$, we have
\begin{align*}
    \bag{ \trans }_u 
    &= \set{ j\in \heavy_u : \evald{ w[j + m_u \btw |w|] } = \trans }\\
    &= \set{ j\in \heavy_v\cup \light_v : \evald{ w[j + m_u \btw j + m_v) }\comp \evald{ w[j + m_v \btw |w|] } = \trans }\\
    &= \set{ j\in \heavy_v\cup \light_v : \tau \comp \evald{ w[j + m_v \btw |w|] } = \trans }\\
    &= \bigcup_{
      \delta'\colon \tau \comp \delta'=\delta
      } \set{ j\in \heavy_v\cup \light_v : \evald{ w[j + m_v \btw |w|] } = \trans' }. \qquad (\natural)
\end{align*}
Now, the partition $\heavy_u=\bigcup_{\delta \in \transm{D}}\bag{ \delta }_u$ can be computed from $\heavy_v=\bigcup_{\delta \in \transm{D}}\bag{ \delta }_v$ as follows:
First, for each $i\in \light_v$, we add $i$ to the bag $\bag{\evald{w[i+m_v\btw |w|)]}}_v$.
Then, $\bag{\delta}_v$ is updated to the set
\[
    \set{ j\in \heavy_v\cup \light_v : \evald{ w[j + m_v \btw |w|] } = \trans }.
\]
Thus, the set $\bag{ \delta }_u$ is obtained by 
\[
    \bigcup_{
      \delta'\colon \tau \comp \delta'=\delta
      } \bag{\delta'}_v.
\]

Now, to compute $\bag{ \trans }_u$, it suffices to merge the bags $\bag{ \trans' }_v$.
We manage the partition by the following data structure defined as follows, which is similar to the disjoint-union structure in~\cite[Sec.~21]{Cormen:2022}.
\begin{itembox}[l]{Disjoint Set Union (DSU)}
\noindent{}\textbf{Data to Manage}: Disjoint sets $\{\disjs_{\lambda}\}_{x\in \Lambda}$ indexed by a label set $\Lambda$.

\noindent{}\textbf{Queries}: 
\begin{itemize}[align=left]
    \item[\Call{MakeBag}{$\lambda$}:] Register a new empty bag $\disjs_{\lambda}$ and add $\lambda$ to $\Lambda$.
    \item[\Call{Insert}{$\lambda, j$}:] Add an element $j$ to the bag $\disjs_{\lambda}$.
    \item[\Call{Merge}{$X$}:] Merge all bags $\disjs_{\lambda}$ with $\lambda \in X$ into a single bag $\disjs_{\lambda^*}$ with $\lambda^*\in X$ and return a pair $(\lambda^*, \mathit{moves})$, where $\mathit{moves}=\bigcup_{\lambda\in X\setminus \{\lambda^*\}}\disjs_{\lambda}$. All labels $\lambda$ in $X\setminus \{\lambda^*\}$ is removed from $\Lambda$ together with the corresponding $\disjs_{\lambda}$. 
\end{itemize}
\end{itembox}

If we adopt the small-to-large strategy in \Call{Merge}{$\cdot$}, it is easy to see that 
\begin{itemize}
    \item \Call{MakeBag}{$\cdot$} and \Call{Insert}{$\cdot$} queries run in $O(1)$ time, and
    \item \Call{Merge}{$\cdot$} query runs in amortized $O(\log s)$ time, where $s:=|\Lambda|+|\bigcup_{\lambda\in \Lambda}\disjs_{\lambda}|$.
\end{itemize}

\subsection{Necessary Queries}

The remaining task is to give an algorithm to compute 
\[
    \set{ \evalc{ [i + m_v \btw j) } : j \in \bag{ \trans }_v,\ i+m_v\leq j},
\]
for each accepting transition $\trans$.
To do this, we define the following data structure, $\rightenumds$.
\begin{itembox}[l]{\rightenumds}
\begin{description}
    \item[Depending Global Objects] \regex{} $C$, the entire string $w$, and factorization forest of $C$.
    \item[Shared Data] \mbox{}
    \begin{itemize}
      \item A set of indices $J$. 
      \item An index $\widehat{j}$.
    \end{itemize}
\end{description}
\begin{description}
    \item[method $\InitR(\widehat{j})$]: Set $\widehat{w}\leftarrow [\widehat{j},|w|]$ and $J\leftarrow \emptyset$.
    \item[method $\AddR(j)$]: $J\leftarrow J\cup \{j\}$.
    \item[method $\GetR(i)$]: Return the set $\set{ \evalc{ [i \btw j) } \: : \: j \in J,\ i \leq j }$.
\end{description}
\end{itembox}

We defer to Section~\ref{sec:data_structures_right} the construction of the data structure, which processes $\InitR(\widehat{j})$, $\AddR(j)$, and $\GetR(i)$ in $O(1)$ time.

\newcommand{\mapping}{\mathcal{M}}

\begin{algorithm}[p]
\caption{Algorithm Overview for \textsc{Right-$B$ Transition Enumeration}.}\label{alg:overview_right}
\Procedure{\emph{\Call{RightBTransitionEnumeration}{$\hpath$}}}{
    $m \gets |\widehat{w}|$\;
    $\Lambda \gets \emptyset$; \tcp{Set of labels for bags}
    $\textsc{Accum} \gets \emptyset$\;
    $\text{Initialize DSU}$\;
    \ForEach{$v \in \pi$ in bottom-up (leaf-to-root) order}{
        $\delta_{\text{shrink}} \gets \evald{ \widehat{w}(m_v \btw m] }$\;
        $m \gets m_v$\;
        {\SetAlgoNoLine\SetAlgoNoEnd
        \ForEach{$\lambda \in \Lambda$}{\label{line:right_updatedelta}
            $\mapping[\lambda] \gets \delta_{\text{shrink}} \comp \mapping[\lambda]$\;
        }}

        \ForEach{$\delta \in \transm{D}$}{\label{line:right_merge}
            $\Lambda_\trans \gets \set{ \lambda \in \Lambda : \mapping[\lambda] = \trans }$\;
            \If{$\Lambda_\trans \neq \emptyset$}{
                $\lambda^*, \mathit{moved} \gets \text{DSU.merge}(\Lambda_\trans)$\; \label{line:merge-bags}
               {\SetAlgoNoLine\SetAlgoNoEnd
                \ForEach{$j \in \mathit{moved}$}{
                    $\mathcal{S}[\lambda^*].\AddR(j)$\;
                }
                \ForEach{$\lambda \in \Lambda_\trans \setminus \set{ \lambda^* }$}{
                    Remove $\lambda$ from $\Lambda$\;
                }}
            }
        }

        \If{$\nword{v} \in \lang{B}$}{\label{line:right_get}
            \ForEach{$i \in \light_v$ with $w[1 \btw i) \in \lang{A}$}{
                \ForEach{$\lambda \in \Lambda$}{
                   {\SetAlgoNoLine\SetAlgoNoEnd
                    \If{$\mapping[\lambda]$ is an accepting transition of $D$}{\label{line:right_accepting_lambda}
                      $\textsc{Accum} \gets \textsc{Accum}\ \cup\ \mathcal{S}[\lambda].\GetR(i + m)$\;\label{line:right_get_add_ans}
                    }}
                }
            }
        }

        \ForEach{$i \in \light_v$}{\label{line:right_add}
            $\trans \gets \evald{ w[i+m \btw |w|) }$\;
            $\Lambda_\trans \gets \set{ \lambda \in \Lambda : \mapping[\lambda] = \trans }$\;
            \If{$\Lambda_\trans = \emptyset$}{\label{line:right_create_criterion}
                $\lambda \gets \text{choose} (\set{ 1, 2, \ldots, |\transm{D}|} \setminus \Lambda)$; \tcp{because $|\Lambda| < |\transm{D}|$ must hold}
                $\text{DSU}.\text{makeBag}(\lambda)$\;
                $\text{DSU}.\text{insert}(\lambda, i)$\;
                Add $\lambda$ to $\Lambda$\;  \label{line:right_add_to_X}
                $\mapping[\lambda] \gets \trans$\;
                $\mathcal{S}[\lambda].\InitR(i)$; \tcp{(Re)initialize our data structure in $\mathcal{S}[\lambda]$} \label{line:right_create_sk}
            }
            \Else{
              $\lambda \gets \text{take the unique element of $\Lambda_\trans$}$\;
              $\text{DSU}.\text{insert}(\lambda, i)$\;
            }
            $\mathcal{S}[\lambda].\AddR(i)$\;\label{line:right_add_to_sk}
        }
    }
    \textbf{return} $\accum$\;
}
\end{algorithm}


\subsection{Algorithm Description}

Our implementation of the algorithm for the \textsc{Right-$B$ Transition Enumeration} is given in Algorithm~\ref{alg:overview_right}.
Algorithm iterates over $v\in \pi$ in bottom-up order while maintaining 
\begin{itemize}
    \item an integer $m$, which is the length of current string $\Gamma(v)$,
    \item a DSU data structure and set $\Lambda$ of labels, which is initially empty,
    \item a mapping $\mapping$ from bag labels $\lambda$ to transitions $\trans$, where $\mapping[\lambda]=\delta$ represents that the bag $\disjs_{\lambda}$ in DSU data structure manages the bag $\bag{\delta}_v$, and
    \item \rightenumds{} data structure $\mathcal{S}[\lambda]$ for each $\lambda\in \Lambda$.
\end{itemize}
Our algorithm manages the partition using two data structures, DSU and \rightenumds{}.
More precisely, we manage
\begin{itemize}
    \item the partition $\bigcup_{\lambda\in \Lambda}\disjs_{\lambda}=\bigcup_{\lambda\in \Lambda}\bag{\mapping[\lambda]}_v$ itself using DSU, and
    \item for each $\lambda$, the bag $\bag{\mapping[\lambda]}_v$ using the \rightenumds{} data structure $\mathcal{S}[\lambda]$ to enables us to use $\AddR(j)$ and $\GetR(j)$ queries on $\bag{\mapping[\lambda]}_v$.
\end{itemize}


To simulate our equation $(\natural)$ for updating, we first update transitions related to bags in the first loop (line~\ref{line:right_updatedelta}).
Then, we merge the bags designated by $\text{DSU.merge}(\Lambda_\delta)$ in line~\ref{line:merge-bags}.
We also update $\mathcal{S}[\lambda^*]$ and $\Lambda$ to reflect moves by the merging.

In the block from line~\ref{line:right_get}, we compute the set of adequate transitions corresponding to the case $v \in \pi$ using $\GetR$.

Finally, in the loop defined in line~\ref{line:right_add}, we add each $i \in \light_v$ to an appropriate bag $\bag{ \trans }$ of a label $\lambda$ with $\trans = \evald{ w[i+m_v \btw |w|) }$.
If such a bag does not exist, the algorithm adds a new bag consisting only of $i$.

The following lemma explicitly shows our loop invariants, which are somewhat straightforward from above discussion.
We give a formal proof for the sake of completeness.



\begin{lemma}\label{lem:right_invariants}
After merging required by moving to the current $v$, at the beginning of line~\ref{line:right_get},
the following holds:
\begin{itemize}
    \item[\rm{(i)}] $\set{ \disjs_\lambda }_{\lambda \in \Lambda}$ is a partition of $\heavy_v$.
    %
    %
    \item[\rm{(ii)}] For each $\lambda \in \Lambda$ and $j \in \disjs_\lambda$, $\evald{ w[j+m_v \btw |w|) } = \mapping[\lambda]$.
    \item[\rm{(iii)}] For each $\delta\in \transm{D}$, there is at most one $\lambda\in \Lambda$ with $\mapping[\lambda]=\delta$.
\end{itemize}
\end{lemma}
\begin{proof}
We prove by induction.

\noindent{}\textbf{Base Case:}
When $v$ is the bottom-most vertex of $\pi$, the conditions are satisfied because $\heavy_v = \emptyset$ and $\Lambda = \emptyset$.

\noindent{}\textbf{Induction Step:}
As induction step, consider the step that we move-up to $u$ from $v$ and check the conditions on $u$.
By induction hypothesis, at the beginning of line~\ref{line:right_get}, those three conditions hold for $v$.
The loop from line~\ref{line:right_add} for $v$ adds $\set{ i \in \light_v : \evald{ w[i+m_v \btw |w|] }= \mapping[\lambda] }$ to $\disjs_{\lambda}$ for each $\lambda$.
Therefore, after the entire iteration for $v$ (or equivalently, the beginning of the iteration for $u$),
the following holds for every $\lambda \in \Lambda$:
\[
\begin{array}{lcl}
    \disjs_\lambda
    & = & \set{ j \in \heavy_v : \evald{ w[j+m_v \btw |w|] } = \mapping[\lambda] } \cup \set{ i \in \light_v : \evald{ w[i+m_v \btw |w|] } = \mapping[\lambda] } \\
    & =& \set{ j \in \heavy_u : \evald{ w[j+m_v \btw |w| ] } = \mapping[\lambda] },
\end{array}
\]
since $u{-}v$ is the heavy edge of $u$ and thus $\heavy_u = \heavy_v \cup \light_v$.
It makes the condition~(i) satisfied.

In the iteration for $u$, after processing the loop from line~\ref{line:right_updatedelta},
the following holds for every $x \in X$ and $j \in \disjs_\lambda$:
\[
\begin{array}{lcl}
  \mapping[\lambda]
    & = & \evald{ \widehat{w}(m_u \btw m_v] } \comp \evald{ w[j+m_v \btw |w|] } \\
    & = & \evald{ w[j+m_u \btw j+m_v) } \comp \evald{ w[j+m_v \btw |w|] } \\
    & = & \evald{ w[j+m_u \btw |w|] },
\end{array}
\]
where the second equality holds from $\widehat{w}[1, m_v] = w[j, j+m_v)$ for every $j \in \heavy_u$.
It makes the condition~(ii) satisfied while keeping the condition~(i).

Furthermore, for each $\trans \in \transm{C}$, the loop from line~\ref{line:right_merge} merges all the bags labeled by $\Lambda_\trans$.
Since this makes condition~(iii) satisfied while keeping the conditions~(i)~and~(ii),
all the conditions hold for $u$ at the beginning of the loop defined in line~\ref{line:right_get}.
\end{proof}

The following lemma proves the correctness of the algorithm.
\begin{lemma}
Algorithm~\ref{alg:overview_right} correctly solves \textsc{Right-$B$ Transition Enumeration}.
More precisely, it returns the following set of transitions $\trans$ of $\transm{C}$:
\[
\left\{ \trans = \evalc{ [i+m_v \btw j) } :
\begin{array}{l}
v \in \hpath,\ i \in \light_v,\ j \in \heavy_v,\ i+m_v \leq j, \\
w[0 \btw i) \in \lang{A}, \\
w[i \btw i+m_v) = w[j \btw j+m_v) = \nword{v} \in \lang{B}, \\
w[j+m_v \btw |w|] \in \lang{D}
\end{array}
\right\}.
\]
\end{lemma}
\begin{proof}
From the definition of \Call{Get$_{\mathrm{R}}$}{$\cdot$}, line~\ref{line:right_get_add_ans}
adds transitions $\trans \in \transm{C}$ to $\accum$ if and only if there is a $j \in \disjs_{\lambda}$ with $i+m_v\leq j$ and $\evalc{ [i+m_v \btw j) } = \trans$.
This index $j$ satisfies $w[j \btw j+m_v)=w[i \btw i+m_v)=\nword{v}$ because $j \in \heavy_v$.
It also satisfies $w[j+m_v \btw |w|) \in \lang{M_D}$ because of the condition (ii) of Lemma~\ref{lem:right_invariants}, $j \in \disjs_{\lambda}$, and the fact that $\mapping[\lambda]$ is an accepting transition of $D$ (line~\ref{line:right_accepting_lambda}).
Thus, all transitions the algorithm adds to $\accum$ satisfy the condition of this lemma.

Conversely, if a vertex $v$ and indices $i$ and $j$ satisfy the condition of the lemma,
the algorithm adds $\trans=\evalc{ [i+m_v \btw j) }$ to $\accum$ at line~\ref{line:right_get_add_ans} at the loop corresponding to $v$, $i$, and $\lambda$,
where $\lambda \in \Lambda$ is the unique bag label such that $j \in \disjs_{\lambda}$, which exists because of condition~(i) of Lemma~\ref{lem:right_invariants}.
\end{proof}

\paragraph{Time Complexity Analysis.}

Now we analyze the time complexity.
First, we claim the following, which in particular implies that operations iterating over all $\lambda \in \Lambda$ cost only a constant factor to the time complexity.
\begin{lemma}
At any moment of the Algorithm~\ref{alg:overview_right}, we have $|\Lambda| \leq |\transm{D}|$.
\end{lemma}
\begin{proof}
$|\Lambda|$ only increases at line~\ref{line:right_add_to_X}.
From condition~(iii) of Lemma~\ref{lem:right_invariants} and the criterion of line~\ref{line:right_create_criterion},
after adding a bag label to $\Lambda$, the value (transition) $\mapping[\lambda]$ differs for all different labels $\lambda$.
Thus, we have $|\Lambda| \leq |\transm{D}|$.
\end{proof}

On the assumption that the \textsc{Right-Enumeration} data structure answers each query in constant time, the following holds.

\begin{lemma}
The time complexity of Algorithm~\ref{alg:overview_right} is bounded by $O(N \log N)$
where $N = |\stree[\hpath]|$. (Recall that $\stree[\hpath]:=\stree[u]$ for the topmost vertex of $\pi$.)
\end{lemma}

\begin{proof}

We bound the total number of queries on \rightenumds{} data structure the algorithm calls, which are the bottleneck of our algorithm.

\noindent{}\textbf{Loop of line~\ref{line:right_merge}.}
Since we call $\AddR(j)$ for $j \in \mathit{moved}$ obtained by $\dsu.\text{merge}$,
the total number of $\AddR$ called by this loop is bound by the total number of moving throughout merging bags.
Such number is clearly bound by $O(N \log N)$
due to the two facts that
1) the total number of elements managed by $\dsu$ is bounded by $N$;
2) each element is moved at most $\log N$ times due to the small-to-large merging.

\noindent{}\textbf{Loop of line~\ref{line:right_get}.}
We call $\GetR$ queries at most $|\Lambda|\leq |\transm{D}|$ times for each $i \in \light_v$.
Thus, the number of $\GetR$ queries in the algorithm is $|\transm{D}|$ times the sum of $|\light_v|$ for all $v \in \hpath$, which is $|\transm{D}| \cdot N \leq O(N)$.

\noindent{}\textbf{Loop of line~\ref{line:right_add}.}
The total number of $\AddR$ called by this loop is bounded by $\sum_{v\in \pi} |\light_v|=N$.

These arguments show that the time complexity of Algorithm~\ref{alg:overview_right} is bounded by $O(N \log N)$.
\end{proof}

%% file: onerewb/new-right-ds.tex

\newcommand{\maxidx}{\mathrm{maxidx}}

\section{Data Structure for Right-$B$ Transition Enumeration}
\label{sec:data_structures_right}

In this section, we give an implementation of the \rightenumds{} data structure that processes $\InitR$, $\AddR$, and $\GetR$ in constant time.
For the sake of readability, we here restate the function of \rightenumds{}.

\begin{itembox}[l]{\rightenumds}
\begin{description}
    \item[Depending Global Objects] \regex{} $C$, the entire string $w$, and factorization forest of $C$.
    \item[Shared Data] \mbox{}
    \begin{itemize}
      \item A set of indices $J$. 
      \item An index $\widehat{j}$.
    \end{itemize}
\end{description}
\begin{description}
    \item[method $\InitR(\widehat{j})$]: Set $\widehat{w}\leftarrow [\widehat{j},|w|]$ and $J\leftarrow \emptyset$.
    \item[method $\AddR(j)$]: $J\leftarrow J\cup \{j\}$.
    \item[method $\GetR(i)$]: Return the set $\set{ \evalc{ [i \btw j) } \: : \: j \in J,\ i \leq j }$.
\end{description}
\end{itembox}



For each node $\fna$ of the factorization forest, we denote the corresponding interval by $[\ell_v,r_v)$.
Moreover, if $f$ is an idempotent node, we denote the corresponding transition by $e_f$.

\subsection{Our Ideas for $\rightenumds$} \label{section:idea-of-rightenumerator}


\paragraph{Our Basic Idea.}

Our data structure utilizes the factorization forest $\fforest_{C, w}$ (or simply $\fforest$)
on the transition monoid $\transm{C}$ and the string $w$, and maintains a set $J$ of indices of $w$.
Our main strategy towards processing queries in constant time is to maintain the following set of transitions for each node $\fna \in \fforest$:
\[
\mathcal{D}_\fna := \set{ \evalc{ [\ell_\fna\btw j) } : j \in J,\ j \in [\ell_\fna, r_\fna) }.
\]
Consider answering $\GetR(i)$ query using $\mathcal{D}_\fna$.
Let $X_i$ be the set of nodes such that the (disjoint) union of intervals $[\ell_\fna, r_\fna)$ over all $\fna \in X_i$ is $[i, |w|]$: i.e., $[i, |w|] = \bigcup_{\fna \in X_i} [\ell_\fna, r_\fna)$.

Then, the set of transitions that $\GetR(i)$ should return can be written as follows:
\[
\everymath={\displaystyle}
\begin{array}{lcl}
    \set{ \evalc{ [i \btw j) } : j \in J,\ i \leq j }
     & = & \bigcup_{\fna \in X_i} \set{ \evalc{ [i \btw j) } : j \in J, j \in [\ell_\fna, r_\fna) } \\
     & = & \bigcup_{\fna \in X_i} \set{ \evalc{ [i \btw \ell_\fna) } \comp \evalc{ [\ell_\fna \btw j) } : j \in J, j \in [\ell_\fna, r_\fna) } \\
     & = & \bigcup_{\fna \in X_i} \set{ \evalc{ [i \btw \ell_\fna) } \comp \trans : \trans \in \mathcal{D}_\fna } = \bigcup_{\fna \in X_i} \evalc{ [i \btw \ell_\fna) } \comp \mathcal{D}_{\fna}
\end{array}
\]
which can be computed just by iterating each $\fna \in X_i$.

By the identity above, the target set can be obtained simply by iterating over the elements of $X_i$.
Hence, we need to have a small family $X_i$ that exactly covers the suffix interval $[i \btw |w|]$.
Using the famous \emph{segment tree} instead of the factorization forest yields such family $X_i$ with $|X_i|=O(\log |w|)$, which is from the fact that the segment tree on a $|w|$-length array has a height $O(\log |w|)$.
This yields the desired \rightenumds{} data structure that processes each query in $O(\log |w|)$ time, which is suitable for our purpose if we allow an additional logarithmic factor on the time complexity of the whole algorithm.

%
For further acceleration, we employ the factorization forest, which has a constant height.
The issue here is that, since some of the idempotent nodes may have an unbounded number of children, the size of $X_i$ is unbounded.
For example, consider the following factorization forest:
\begin{center}
\input{onerewb/ffzu1.tex}
\end{center}
In this case, $X_i$ should have size $|w| - i + 1=O(|w|)$ because we cannot use any of the non-leaf nodes to exactly cover $[i \btw |w|]$.
Thus, we cannot directly iterate over all $\fna \in X_i$ for the query $\GetR(i)$.

\paragraph{Using Idempotency.}
The idea here is to process simultaneously almost all $\fna \in X_i$ that are children of the same idempotent node. 
Let $\fnb$ be an idempotent node with $ i\in [\ell_\fnb, r_\fnb)$,
and let $\fna_1, \ldots, \fna_p, \dots, \fna_q$
be the children of $\fnb$
in that order, where $X_i$ contains $\fna_p, \dots, \fna_q$~\footnote{Since $X_i$ covers $[i \btw |w|]$, if $\fna_p \in X_i$, then $\fna_{p'} \in X_i$ for every $p' > p$.}.
Let $e_g$ be the idempotent transition corresponding to the node $g$; i.e., $\evalc{[\ell_{\fna_t},r_{\fna_t})}=e_g$ for all $t\in \{1,\dots, q\}$.
We simultaneously process $f_{p+1},\dots, f_{q}\in X_i$ using the following observation.
\begin{lemma}\label{lem:right_ds_simul}
We have
\[
    \bigcup_{t=p+1}^{q} \evalc{ [i \btw \ell_{\fna_t} ) } \comp\mathcal{D}_{\fna_t} = \evalc{ [i \btw \ell_{\fna_{p}}) } \comp \bigcup_{t=p+1}^{q} \left(e_\fnb \comp \mathcal{D}_{\fna_t}\right).
\]
\end{lemma}
\begin{proof}
We have the following equation:
\begin{align*}
    & \bigcup_{t=p+1}^{q} \evalc{ [i \btw \ell_{\fna_t} ) } \comp\mathcal{D}_{\fna_t} \\
    = & \bigcup_{t=p+1}^{q} \evalc{ [i \btw r_{\fna_{t-1}}) } \comp \mathcal{D}_{\fna_t} \qquad (\because r_{\fna_{t-1}} = \ell_{\fna_t}) \\
    = & \bigcup_{t=p+1}^{q} \bigl(\evalc{ [i \btw \ell_{\fna_{p}}) } \comp \evalc{ [\ell_{\fna_{p}} \btw r_{\fna_{p}}) } \comp \cdots \comp \evalc{ [\ell_{\fna_{t-1}} \btw r_{\fna_{t-1}}) }\ \bigr) \comp \mathcal{D}_{\fna_t} \\
    & \qquad (\because w[i \btw \ell_{\fna_p}) w[\ell_{\fna_p} \btw r_{\fna_p}) \cdots w[\ell_{\fna_{t-1}} \btw r_{\fna_{t-1}}) = w[i \btw r_{\fna_{t-1}})) \\[0pt]
    = & \bigcup_{t=p+1}^{q} \evalc{ [i \btw \ell_{\fna_{p}}) } \comp \bigl(\ \underbrace{e_\fnb \comp \cdots \comp e_\fnb}_{\text{$t-p$}}\ \bigr) \comp \mathcal{D}_{\fna_t} \\
    = & \left\{ \evalc{ [i \btw \ell_{\fna_{p}}) } \comp (e_\fnb \comp \Delta) \colon \Delta\in \bigcup_{t=p+1}^{q}\mathcal{D}_{\fna_t}\right\} \\
    = & \evalc{ [i \btw \ell_{\fna_{p}}) } \comp \bigcup_{t=p+1}^{q} \left(e_\fnb \comp \mathcal{D}_{\fna_t}\right).\qedhere
\end{align*}
\end{proof}
The remaining task is to compute the set $\displaystyle\bigcup_{t=p+1}^{q} \left(e_\fnb \comp \mathcal{D}_{\fna_t}\right)$ efficiently.
To do so, we define the following $\maxidx$:
\[
\begin{array}{lcl}
\maxidx_{\fnb}[\delta]
 & := & \max \set{ t : 1 \leq t \leq q, \trans \in e_\fnb \comp \mathcal{D}_{\fna_t} }.
 \end{array}
\]
where we define the max of the empty set is $-\infty$: i.e., $\max\emptyset = -\infty$.
Then, we have the following.
\begin{lemma}\label{lemma:node-transition-exchange}
We have
\[
\bigcup_{t=p+1}^{q} \left(e_{\fnb} \comp \mathcal{D}_{\fna_t}\right) =
  \set{\trans : \trans \in \transm{C},\ \maxidx_\fnb[\delta] > p}.
\]
\end{lemma}
\begin{proof}
($\subseteq$):
Assume $\delta\in \displaystyle\bigcup_{t=p+1}^{q} \left(e_{\fnb} \comp \mathcal{D}_{\fna_t}\right)$ and let $t$ be an index with $\delta\in e_{\fnb} \comp \mathcal{D}_{\fna_t}$.
Then, we have $\maxidx_\fnb[\delta]\geq t>p$.

($\supseteq$):
Assume $t:=\maxidx_\fnb[\delta]>p$.
Then, $\delta\in e_{\fnb} \comp \mathcal{D}_{\fna_t}\subseteq \displaystyle\bigcup_{t=p+1}^{q} \left(e_{\fnb} \comp \mathcal{D}_{\fna_t}\right)$.
\end{proof}

Combining Lemmas~\ref{lem:right_ds_simul}~and~\ref{lemma:node-transition-exchange} immediately yields the following.
\begin{lemma}\label{lemma:maxidx-is-crucial}
We have
\[
\bigcup_{t=p+1}^{q} \evalc{ [i \btw \ell_{\fna_t} ) }\comp \mathcal{D}_{\fna_t}
 =
\set{ \evalc{ [i \btw \ell_{\fna_p}) } \comp \trans : \trans \in \transm{C}, \maxidx_{\fnb}[\delta] > p }.
\]
\end{lemma}
Using $\maxidx$, the right hand side of the formula on Lemma~\ref{lemma:maxidx-is-crucial} can be computed just by iterating over all $\delta \in \transm{C}$. Thus, we can compute the desired left hand side in constant time.

\subsection{Implementation of \rightenumds}

\noindent{}\textbf{Remark.}
To simplify our pseudocodes, we use \emph{class} or \emph{module} features supported by many programming languages.
We employ the object-oriented (Java) style; see Sections~1.1~and~1.2 of~\cite{Sedgewick:2011}.
Particularly, in this section, we write
\begin{itemize}
    \item $S.\mathrm{add}(x)$ for the operation of adding the element $x$ to the set $S$,
    \item $f.\mathrm{parent}$ to represent the parent of the node $f$,
    \item $g.\mathrm{numOfChildren}$ to represent the number of children of the node $g$,
    \item $g.\mathrm{childOf}(p)$ to represent the $p$-th child (from the left) of the node $g$ ($1$-origin), and
    \item $g.\mathrm{childIdx}(f)$ to represent, for two nodes $f$ and $g$ with $g=f.\mathrm{parent}$, an integer $p$ such that $f$ is the $p$-th child of $g$ (from the left).
\end{itemize}
%


\paragraph{The Implementation.}
Our data structure is given in Algorithm~\ref{alg:ds_right},
which manages the sets $\mathcal{D}_\fna$ for each node $\fna$ and the value $\maxidx_\fnb[\trans]$ for each idempotent node $\fnb$ and $\trans \in \transm{C}$.

The constructor $\InitR$ initializes $\mathcal{D}_\fna$ and $\maxidx_\fnb[\trans]$, which seems to consume $O(|w|)$ time.
This issue can be bypassed as follows. Instead of initializing all nodes at once at the beginning, we create and initialize each node only when the data structure first accesses it. 
Then, \Call{Initialize$_\mathrm{R}$}{$\cdot$} can be implemented to work in constant time without affecting the time complexity of other procedures.

\begin{algorithm}[t!]
\caption{Implementation of the \textsc{Right-Enumeration} Data Structure.}\label{alg:ds_right}
\Method{\emph{\Call{Initialize$_\mathrm{R}$}{$E_C, \mathcal{F}$}}}{
    \KwIn{A \regex{} $C$ and its factorization forest $\fforest$, both are given by references or pointers}
    $\mathcal{D}_\fna \gets \emptyset$ for each node $\fna$ of $\fforest$\;
    $\maxidx_{\fnb}[\trans] \gets -\infty$ for each idempotent node $\fnb$ of $\fforest$ and $\trans \in \transm{C}$\;
}
\Method{\emph{\Call{Add$_\mathrm{R}$}{$j$}}}{
    \KwIn{$j\in [|w|]$}
    $\fna \gets \text{the leaf node of $\fforest$ corresponding to the index } j$\;
    \While{$\fna \text{ is not root}$}{\label{line:ds_right_while}
        $\fnb \gets \fna\text{.parent}$\;    
        $\trans \gets \evalc{ [\ell_\fna\btw j ) }$\;
        $\mathcal{D}_\fna.\text{add}(\trans)$\; \label{line:ds_right_add_d}
        \If{$\fnb$ is an idempotent node}{
            $p \gets \fnb.\text{childIdx}(\fna)$ \tcp*{ $\fna$ is $p$-th child of $\fnb$ }
            $\maxidx_{\fnb}[e_\fnb \comp \trans] \gets \max(p, \maxidx_{\fnb}[e_\fnb \comp \trans])$\; \label{line:ds_right_updatemax}
        }
        $\fna \gets \fnb$ \tcp*{ move up the factorization forest }
    }
}
\Method{\emph{\Call{Get$_\mathrm{R}$}{$i$}}}{
    \KwIn{$i\in [|w|]$}
    $\fna \gets \text{the leaf node corresponding to the index } i$\;
    $\accum \gets \mathcal{D}_{\fna}$\; \label{line:ds_right_lengthzero}
    \While{$\fna \text{ is not root}$}{\label{line:ds_right_get_while}
        $\fnb \gets \fna\text{.parent}$\;
        $p \gets \fnb.\text{childIdx}(\fna)$ \tcp*{$\fna$ is the $p$-th child of $\fnb$}
        \If(\tcp*[f]{$\fna$ is not the rightmost child of $\fnb$}){$p < \fnb.\mathrm{numOfChildren}$}{ 
            $\fna^+ \gets \fnb.\text{childOf}(p+1)$ \tcp*{$\fna^+$ is the immediate right sibling of $\fna$}
            \ForEach{$\delta \in \mathcal{D}_{\fna^+}$}{
                $\accum.\text{add}( \evalc{ [i \btw \ell_{\fna^+}) } \comp \trans )$\;\label{line:ds_right_addnext}
            }
            \If{$\fnb$ is an idempotent node}{
                \ForEach{$\trans \in \transm{C}$}{
                    \If{$p+1 < \maxidx_{\fnb}[\trans]$}{
                        $\accum.\text{add}( \evalc{ [i \btw \ell_{\fna^+}) } \comp \trans)$\;\label{line:ds_right_addidem}
                    }
                }
            }
        }
        $\fna \gets \fnb$ \tcp*{ move up the factorization forest }
    }
    \textbf{return} $\accum$\;
}
\end{algorithm}

Now we prove the correctness of our implementation.
The following lemma ensures that the procedure $\AddR$ correctly updates $\mathcal{D}_{\fna}$.
\begin{lemma}\label{lem:ds_right_d_lem}
Let $j \in [|w|]$.
Then, the query $\AddR(j)$ adds $\evalc{ [\ell_\fna \btw j)}$ to $\mathcal{D}_{\fna}$ for all nodes $\fna$ with $j \in [\ell_\fna, r_\fna)$ (and only for those nodes).
\end{lemma}
\begin{proof}
The loop defined in line~\ref{line:ds_right_while} iterates over all nodes $\fna$ with $j \in [\ell_\fna, r_\fna)$.
For each of such $\fna$, line~\ref{line:ds_right_add_d} adds the transition $\evalc{ [\ell_\fna \btw j) }$ to $\mathcal{D}_\fna$.
\end{proof}

Thus, we have the following, which is immediate from Lemma~\ref{lem:ds_right_d_lem}.
\begin{lemma}\label{lem:ds_right_d}
Let $J$ be a set of indices.
Assume we have invoked $\AddR(j)$ once for each $j \in J$, in any order.
Then, for every node $\fna$,
$\mathcal{D}_{\fna} = \set{ \evalc{ [\ell_\fna \btw j) } : j \in J, j\in [\ell_\fna, r_\fna) }$.
\end{lemma}

We now ensure that the method $\AddR$ correctly updates $\maxidx_{\fnb}[\trans]$ by the following lemma.
%

\begin{lemma}\label{lem:ds_right_idempotent_lem_2}
Let $j\in [|w|]$.
Then, the query $\AddR(j)$ updates $\maxidx_{\fnb}[\trans]$ by $\max(p,\maxidx_{\fnb}[\trans])$ for all idempotent node $\fnb$ with $j\in [\ell_{\fnb},r_{\fnb})$, where 
\begin{itemize}
    \item $\trans=e_{\fnb}\comp \evalc{ [\ell_\fna\btw j) }$, where $\fna$ is the unique child of $\fnb$ with $j\in [\ell_{\fna},r_{\fna})$, and
    \item $p$ is the integer such that $\fna$ is the $p$-th child of $\fnb$,
\end{itemize}
and does not update other idempotent nodes.
%
\end{lemma}
\begin{proof}
The loop of line~\ref{line:ds_right_while} iterates over all $(g,f)$ satisfying the condition of the lemma, and line~\ref{line:ds_right_updatemax} is processed when and only when $g$ is idempotent, $f$ is the $p$-th child of $g$, and $e_g\comp \delta=e_g\comp \evalc{ [\ell_\fna\btw j) }$.
\end{proof}

Now, we have the following.
\begin{lemma}\label{lem:ds_right_idem}
Let $J$ be a set of indices.
Assume we have invoked $\AddR(j)$ once for each $j \in J$, in any order.
Then, for each idempotent node $\fnb$ and $\trans \in \transm{C}$, we have
\[
\maxidx_{\fnb}[\trans] = \max\set{ t : 1 \leq t \leq q, \trans \in e_\fnb \comp \mathcal{D}_{\fna_t} }.
\]
where $q$ is the number of children of $\fnb$ and $\fna_t$ is the $t$-th child of $\fnb$.
\end{lemma}
\begin{proof}
From Lemma~\ref{lem:ds_right_idempotent_lem_2},
$\maxidx_{\fnb}[\trans]$ is equal to the maximum among integers $t$ such that there is an index $j\in J$ with
\begin{itemize}
    \item $j\in [\ell_{\fna_t}, r_{\fna_t})\subseteq [\ell_{\fnb}, r_{\fnb})$, where $\fna_t$ is the $t$-th child of $g$, and
    \item $\delta = e_g\comp \evalc{[\ell_{\fna_t}\btw j)}$.
\end{itemize}
%
From Lemma~\ref{lem:ds_right_d}, we have
\begin{align*}
    \set{e_{\fnb}\comp  \evalc{ [\ell_{\fna_t} \btw j) } : j \in J\cap [\ell_{\fna_t}, r_{\fna_t}) }
    =e_{\fnb}\comp \mathcal{D}_{\fna_t}.
\end{align*}
%
Thus, we have
\begin{align*}
    \maxidx_{\fnb}[\trans] 
    &= \max\set{ t : 1 \leq t \leq q,\ \exists j\in J\cap [\ell_{\fna_t},r_{\fna_t}),\ \trans = e_\fnb \comp \evalc{ [\ell_{\fna_t} \btw j) } }\\
    &= \max\set{ t : 1 \leq t \leq q, \trans \in e_\fnb \comp \mathcal{D}_{\fna_t} }.\qedhere
\end{align*}
\end{proof}

The following lemma proves the correctness of the implementation of the method $\GetR$ in Algorithm~\ref{alg:overview_right}.
\begin{lemma}\label{lem:ds_right_get}
Let $i \in [|w|]$. Then, $\GetR(i)$ returns the set $\set{ \evalc{ [i \btw j) } : j \in J,\ i \leq j }$.
\end{lemma}
\begin{proof}
\Call{Get$_\mathrm{R}$}{$i$} returns the union of
\begin{itemize}
    \item[(i)] $\mathcal{D}_\fna$ for the leaf node $\fna$ corresponding to the index $i$ (line~\ref{line:ds_right_lengthzero}),
    \item[(ii)] $\evalc{ [i \btw \ell_{\fna^+}) } \comp \mathcal{D}_{\fna^+}$ for each node $\fna$ with $i \in [\ell_\fna, r_\fna)$
    that is not the rightmost child of the parent of $\fna$, where $\fna^+$ is the immediate right sibling of $\fna$ (line~\ref{line:ds_right_addnext}), and
    \item[(iii)] $\set{ \evalc{ [i \btw \ell_{\fna^+}) } \comp \trans : \trans \in \transm{C},\ \maxidx_{\fnb}[\trans] > p+1}$
    for each node $\fna$ with $i \in [\ell_\fna, r_\fna)$ that is
    the non-rightmost $p$-th child of an idempotent node $\fnb$,
    where $\fna^+$ is the immediate right sibling of $\fna$, $(p+1)$-th child of $\fnb$ (line~\ref{line:ds_right_addidem}).
\end{itemize}
From now on, we rewrite these sets in a more manageable form.

\paragraph{(i).}
We have
\[
\mathcal{D}_\fna
= \set{ \evalc{ [\ell_\fna \btw j) } : j \in J \cap [\ell_\fna, r_\fna) }
= \set{ \evalc{ [i \btw j) } : j \in J,\ \underline{ j = i } }
\]
where the first equality is from Lemma~\ref{lem:ds_right_d} and the second equality follows from $[\ell_f,r_f)=\{i\}$.

\paragraph{(ii).}
For each node $\fna$ satisfying the condition of (ii), we have
\begin{align*}
  \evalc{ [i \btw \ell_{\fna^+}) }  \comp \mathcal{D}_{\fna^+}
   & =  \set{ \evalc{ [i \btw \ell_{\fna^+}) } \comp \evalc{ [\ell_{\fna^+} \btw j) } : j \in J \cap [\ell_{\fna^+}, r_{\fna^+}) } \\
   & =  \set{ \evalc{ [i \btw j) } : j \in J \cap \underline{ [\ell_{\fna^+}, r_{\fna^+}) } },
\end{align*}   
where the first equality is from Lemma~\ref{lem:ds_right_d}.

\paragraph{(iii).}
Let $\fnb$ be an idempotent node and $\fna_1,\ldots, \fna_{p}=\fna, \fna_{p+1}={\fna^+}, \ldots, \fna_q$ be its children such that $i \in [\ell_\fna, r_\fna)$.
Since Lemma~\ref{lem:ds_right_idem} ensures that $\maxidx_{\fnb}$ is correctly computed, we can apply Lemma~\ref{lemma:maxidx-is-crucial} to obtain
\begin{align*}
&\set{ \evalc{ [i \btw \ell_{\fna^+}) } \comp \trans : \trans \in \transm{C},\ \maxidx_{\fnb}[\trans] > p+1 }\\
&= \bigcup_{t=p+2}^{q} \evalc{ [i \btw \ell_{\fna_t}) }\comp \mathcal{D}_{\fna_t}\\
&= \bigcup_{t=p+2}^{q} \evalc{ [i \btw \ell_{\fna_t}) }\comp \set{ \evalc{ [\ell_{\fna_t} \btw j) } : j \in J \cap [\ell_{\fna_t}, r_{\fna_t}) }\\
&= \bigcup_{t=p+2}^{q} \set{ \evalc{ [i \btw j) } : j \in J \cap [\ell_{\fna_t}, r_{\fna_t})}\\
&= \set{ \evalc{ [i \btw j) } : j \in J \cap [\ell_{\fna_{p+2}}, r_{\fna_{q}})}\\
&= \set{ \evalc{ [i \btw j) } : j \in J \cap \underline{[r_{\fna^+}, r_{g})}}.
\end{align*}

Now, it suffices to prove that the union of intervals
\begin{itemize}
    \item[(I)] consisting only of an index $i$, 
    \item[(II)] $[\ell_{\fna^+}, r_{\fna^+})$ for each pair of nodes $(\fna, \fna^+)$ defined in (ii), and
    \item[(III)] $[r_{\fna^+}, r_{\fnb})$ for each tuple of nodes $(\fnb, \fna, \fna^+)$ defined in (iii)
\end{itemize}
coincides with the interval $[i \btw |w|]$.

\noindent{}($\subseteq [i \btw |w|]$):
Clearly, each interval is included in $[i \btw |w|]$.

\noindent{}($\supseteq [i \btw |w|]$):
We prove that for each $j \in [i \btw |w|]$, $j$ is in one of such intervals.
Let $\fnb$ be the lowest node of the factorization forest such that $i, j \in [\ell_\fnb, r_\fnb)$.

If $\fnb$ is a leaf node, we have $i = j$ and the interval (I) contains $j$.
Otherwise, let $\fna$ and $\fna'$ be the children of $\fnb$ such that $i \in [\ell_{\fna}, r_{\fna})$ and $ j\in [\ell_{\fna'}, r_{\fna'})$.
Since $\fnb$ is the lowest node with $i, j \in [\ell_\fnb, r_\fnb)$, we have $\fna \neq \fna'$.
Let $\fna^+$ be the node immediately at the right of $\fna$.

If $\fna' = \fna^+$, the interval $[\ell_{\fna^+}, r_{\fna^+})$ in (II) for the pair $(\fna, \fna^+)$ contains $j$.
Otherwise, $\fnb$ has at least three children, and thus, is an idempotent node.
Thus, the interval $[r_{\fna^+}, r_\fnb)$ in (III) for the tuple $(\fnb, \fna, \fna^+)$ contains $j$.
Therefore, the union of the intervals defined in~(I),~(II),~and~(III) coincides with $[i \btw |w|]$ and the lemma is proved.
\end{proof}

Clearly, each query in Algorithm~\ref{alg:ds_right} runs in time proportional to the height of the factorization forest and $|\transm{C}|$, both of which are constants.
Therefore, each query runs in constant time, and thus, Algorithm~\ref{alg:ds_right} is an implementation of \textsc{Right-Enumeration} data structure with the desired property.

%% file: onerewb/ffzu1.tex
\begin{tikzpicture}[x=0.95cm,y=1.05cm,
  leaf/.style={draw,rectangle,minimum width=0.95cm,minimum height=0.65cm,
               inner sep=1pt,align=center},
  box/.style={draw,rounded corners=4pt,fill=gray!12},
  guide/.style={thin},
]

\makearray{12} 

\putcell{1}{\small $1$}
\putcell{4}{\small $i-3$}
\putcell{5}{\small $i-2$}
\putcell{6}{\small $i-1$}
\putcell{7}{\small $i$}
\putcell{8}{\small $\cdots$}
\putcell{9}{\small $\cdots$}
\putcell{10}{\small $\cdots$}
\putcell{11}{\small $|w|{-}1$}
\putcell{12}{\small $|w|$}

\pgfmathtruncatemacro{\idx}{7}
\pgfmathtruncatemacro{\L}{max(1,\idx-2)}
\pgfmathtruncatemacro{\R}{\N-1}

\cover{\L}{\R}{.8}{nC}{$[i{-}2\,\btw\,|w|{-}1]$}
\cover{1}{4}{1.5}{nD}{$[1\,\btw\,i{-}3]$}

\northhit[-]{a5}{nC}
\northhit[-]{a6}{nC}
\northhit[-]{a7}{nC}
\northhit[-]{a11}{nC}

\northhitCover[-]{a1}{nD}
\northhitCover[-]{a4}{nD}

\cover{1}{12}{2.5}{root}{$[1\,\btw\,|w|]$}

\northhitCover[-]{nC}{root}
\northhitCover[-]{nD}{root}
\northhitCover[-]{a12}{root}

\end{tikzpicture}

%% file: onerewb/leftenumeration.tex
\newcommand{\dleft}{\mathcal{D}_{\text{left}}}
\newcommand{\GetNear}{\textsc{Get}_{\text{near}}}

\section{Algorithm Overview for Left-$B$ Transition Enumeration}
\label{new:section:left-B-overview}

\subsection{Overall Strategy}
Here, we solve \textsc{Left-$B$ Transition Enumeration} on a heavy-path $\hpath$ of the suffix tree $\stree_w$ of $w$ in $O(|\stree[\hpath]| \log |w|)$-time.
Our algorithm iterates over $v \in \pi$ in bottom-up order, and for each $v$ with $\nword{v}\in \lang{B}$, computes the following set of transitions for each $i \in \light_v$ such that $w[i + m_v \btw |w|] \in \lang{D}$:
\[
\dleft := \set{ \evalc{ [j + m_v \btw i) } : j \in \heavy_v,\ w[1 \btw j) \in \lang{A}, j + m_v \leq i }
\]
where $m_v = |\nword{v}|$ is the length of the associated string $\nword{v}$ of the vertex $v$.

To compute $\dleft$ in $O(\log |w|)$-time, we use a data structure, \leftenumds, defined on the ex{} $C$ and the suffix $\widehat{w}$ corresponding to the heavy-path $\pi$.
It provides queries \Call{Initialize$_\mathrm{L}$}{$\cdot$}, \Call{Add$_\mathrm{L}$}{$\cdot$}, \Call{Shrink$_\mathrm{L}$}{$\cdot$}, and \Call{Get$_\mathrm{L}$}{$\cdot$} defined as follows.

\begin{itembox}[l]{$\leftenumds$}
\begin{description}
    \item[Depending Global Data] \regex{} $C$, the entire string $w$, and factorization forest of $C$.
    \item[Shared Data] \mbox{}
    \begin{itemize}
      \item a set of indices $J$
      \item a suffix $\widehat{w}$
      \item a non-negative integer $m$
    \end{itemize}
\end{description}
Providing the following queries:
\begin{description}
    \item[\Call{Initialize$_\mathrm{L}$}{$\widehat{j}$}]: Set $\widehat{w}\gets w[\widehat{j} \btw |w|]$ and $m\leftarrow n-\widehat{j}$.
    \item[\Call{Add$_\mathrm{L}$}{$j$}]: $J\leftarrow J\cup \{j\}$. It is guaranteed that $w[j\btw j+m)$ is a prefix of $\widehat{w}$.
    \item[\Call{Shrink$_\mathrm{L}$}{$d$}]: Set $m\leftarrow m-d$. It is guaranteed that $d\geq 0$.
    \item[\Call{Get$_\mathrm{L}$}{$i$}]: Return the set $\set{ \evalc{ [j+m \btw i) } : j \in J,\ j+m \leq i }$.
\end{description}
\end{itembox}









\begin{algorithm}[t!]
\caption{Algorithm Overview for \textsc{Left-$B$} Transition Enumeration}\label{alg:overview_left}
    \Procedure{LeftBTransitionEnumeration($\pi$)}{
        $\mathit{enumerator} \gets \textsc{LeftEnumerator}$\;
        $\mathit{enumerator}$.\Call{Initialize$_\mathrm{L}$}{$\widehat{j}$}, where $\widehat{j}$ is the index corresponding to the leaf of $\pi$\;
        $\accum \gets \emptyset$\;
        \ForEach{$v \in \pi$ in bottom-up order}{\label{line:for-invariant}
            $m_v \gets | \nword{v} |$\;
            $\mathit{enumerator}$.\Call{Shrink$_\mathrm{L}$}{$m-m_v$}\;\label{line:left_shrink}
            \If{$\nword{v} \in \lang{B}$}{ \label{line:left-filterB}
                \tcp{this check is in $O(1)$ time using the factorization forest.}
                \ForEach{$i \in \light_v$}{
                  \If{$w[i+m_v \btw |w|] \in \lang{D}$}{ \label{line:left-filterD}
                    $\mathcal{D} \gets \mathit{enumerator}$.\Call{Get$_\mathrm{L}$}{$i$}\;\label{line:left_get_transitions}
                    $\accum \gets \accum \cup \mathcal{D}$\;\label{line:left_add_to_answer}
                  }
                }
            }
            \ForEach{$i \in \light_v$ with $w[1 \btw i) \in \lang{A}$}{\label{new:line:for-inv1}
                $\mathit{enumerator}$.\Call{Add$_\mathrm{L}$}{$i$}\;\label{line:left_add_indices}
            }
        }
        \textbf{return} $\accum$\;
    }
\end{algorithm}

We defer the implementation details of data structure $\leftenumds$, which processes each query in amortized $O(\log |w|)$ time, to the subsequent sections.
We now present our algorithm for \textsc{Left-$B$ Transition Enumeration}, as shown in Algorithm~\ref{alg:overview_left}.
As noted earlier~\ref{sec:leftb-and-rightb-tasks},
we begin by $J = \emptyset$ and iterate over $v \in \hpath$ in bottom-up order. For each $v\in \hpath$, we apply the following operations.
\begin{itemize}
    \item First, we set $m\leftarrow m_v$ using a \Call{Shrink$_\mathrm{L}$}{$\cdot$} query (line~\ref{line:left_shrink}).
    \item Then, assuming $w_v\in \lang{B}$, for each $i \in \light_v$ with $w[i+m \btw |w|]\in \lang{D}$, we compute the set $\set{ \evalc{ [j+m \btw i) } : j \in J,\ j+m \leq i }$ by calling $\GetL(i)$  (line~\ref{line:left_get_transitions}).
    \item Finally, we properly update $J$ by adding each $i \in \light_v$ with $w[1 \btw i) \in \lang{A}$ by calling $\AddL(i)$ (line~\ref{line:left_add_indices}).
\end{itemize}

The following lemma ensures that the algorithm correctly updates the set $J$.

\begin{lemma}[Loop Invariant]\label{new:lem:left_correct_J}
Every visiting the line~\ref{line:for-invariant} of Algorithm~\ref{alg:overview_left}, the following holds:
\[
\mathit{enumerator}.J = \set{ j \in \heavy_v : w[1 \btw j) \in \lang{A} }.
\]
\end{lemma}
\begin{proof}
We prove by induction.
If $v$ is the leaf of $\pi$, the statement is immediate from $\mathit{enumerator}.J = \heavy_v = \emptyset$.
Otherwise, let $u$ be the child of $v$ on $\pi$.
Since $\heavy_v = \heavy_u \cup \light_u$, we have
\[
\set{ j \in \heavy_v : w[1 \btw j) \in \lang{A} } = \set{ j \in \heavy_u : w[1 \btw j) \in \lang{A} } \cup \set{ j \in \light_v : w[1 \btw j) \in \lang{A} }.
\]
By induction hypothesis, the first term of the right-hand side was equal to $\mathit{enumerator}.J$ at the beginning of the loop for $u$.
The second term of the right-hand side is exactly the set that was added to $\mathit{enumerator}.J$ while the loop for $u$ in line~\ref{line:left_add_indices}. Thus, the lemma is proved.
\end{proof}

The following lemma ensures the correctness of Algorithm~\ref{alg:overview_left}.
\begin{lemma} 
Algorithm~\ref{alg:overview_left} correctly solves \textsc{Left-$B$ Transition Enumeration}.
More precisely, it returns the following set of transitions $\trans$ of $\transm{C}$:
\[
\left\{ \trans = \evalc{ [j+m_v \btw i) } :
\begin{array}{l}
v \in \hpath,\ i \in \light_v,\ j \in \heavy_v,\ j+m_v \leq i, \\
w[0 \btw j) \in \lang{A}, \\
w[j \btw j+m_v) = w[i \btw i+m_v) = \nword{v} \in \lang{B}, \\
w[i+m_v \btw |w|] \in \lang{D}
\end{array}
\right\}.
\]
\end{lemma}
\begin{proof}
From the definition of \Call{Get$_{\mathrm{L}}$}{$\cdot$}, line~\ref{line:left_get_transitions} adds transitions $\trans \in \transm{C}$ to $\accum$ if and only if there is a $j \in J$ with $j+m_v\leq i$ and $\evalc{ [j+m_v \btw i) } = \trans$.
By Lemma~\ref{new:lem:left_correct_J}, $j$ satisfies $w[1\btw j)\in \lang{A}$ and $j\in \heavy_v$.
Moreover, we have $w[j \btw j+m_v)=w[i \btw i+m_v)=\nword{v}$ because $j \in \heavy_v$.
Line~\ref{line:left-filterD} ensures $w[i+m_v\btw |w|]\in \lang{D}$.
Thus, all transitions the algorithm adds to $\accum$ satisfy the condition of this lemma.

Conversely, if a vertex $v$ and indices $i$ and $j$ satisfy the condition of the lemma, the algorithm adds $\trans=\evalc{ [j+m_v \btw i) }$ to $\accum$ at line~\ref{line:left_get_transitions} at the loop corresponding to $v$ and $i$.
\end{proof}

\paragraph{Time Complexity of Algorithm~\ref{alg:overview_left}.}

Now we analyze the time complexity.
It is clear that Algorithm~\ref{alg:overview_left} calls $O(|\stree[\pi]|)$ number of queries.
Thus, using an implementation of the \textsc{Left-Enumeration} data structure that processes each of four types of queries in (amortized) $O(\log |w|)$ time, we can bound the time complexity of Algorithm~\ref{alg:overview_left} by $O(|\stree[\pi]|\log  |w|)$.

\subsection{Data Structure Overview for $\leftenumds$}
\label{section:implementing-left-B}
To solve \textsc{Left-$B$ Transition Enumeration}, it suffices to implement $\leftenumds$.
Unfortunately, it does not seem easy to realize this data structure just by using factorization forests as we did in Section~\ref{sec:data_structures_right}, due to the presence of the \Call{Shrink$_\mathrm{L}$}{$\cdot$} queries.

To this end, we consider an auxiliary data structure, $\auxds$, that efficiently processes queries that are somewhat weaker than our purpose.
Then, we construct \textsc{LeftEnumerator} using $\auxds$.
The precise definition of $\auxds$ is as follows.

\begin{itembox}[l]{$\auxds$}
\begin{description}
    \item[Depending Global Data] \regex{} $C$, the entire string $w$, and factorization forest of $C$.
    \item[Shared Data] \mbox{}
    \begin{itemize}
    \item a set of \emph{intervals} $\mathcal{J} = \{ [j_1, k_1), [j_2, k_2), \ldots \}$
    \item a suffix $\widehat{w}$
    \item a non-negative integer $m$
    \end{itemize}
\end{description}
Providing the following queries or methods:
\begin{description}
    \item[\Call{Initialize$_\mathrm{A}$}{$\widehat{j}$}]: Set $\widehat{w}\leftarrow [\widehat{j},n)$ and $m\leftarrow n-\widehat{j}$.
    \item[\Call{Add$_\mathrm{A}$}{$j, k$}]: $\mathcal{J}\leftarrow \mathcal{J}\cup \{[j, k)\}$. It is guaranteed that $j+m \leq k$ and $w[j \btw k)$ is a prefix of $\widehat{w}$.
    \item[\Call{Shrink$_\mathrm{A}$}{$d$}]: Set $m\leftarrow m-d$. It is guaranteed that $d\geq 0$.
    \item[\Call{Get$_\mathrm{A}$}{$i$}]: Return the set
      $\set{ \transm{C}(\,w[j+m \btw i)\,) : [j, k) \in \mathcal{J},\ k \leq i }$. 
\end{description}
\end{itembox}

Here, we explain how $\auxds$ differs from $\leftenumds$.

\paragraph{First Difference: Interval.}
The first difference is that, whereas $\leftenumds$ maintains a set of indices $J$, the auxiliary data structure $\auxds$ maintains a set of intervals $\mathcal{J}$.
Due to this difference, the $\AddA(j, k)$ query takes an additional argument $k \geq j+m$ compared to the $\AddL(j)$ query, where $k$ is the index of the right endpoint of the interval.
Note that each interval $[j, k) \in \mathcal{J}$ always satisfies $j + m \leq k$, since this condition holds when that interval is added to $\mathcal{J}$ and $m$ never increases thereafter.

\paragraph{Second Difference: Outputs of \textsc{Get} queries.}
The second difference is what \Call{Get}{$\cdot$} queries return. 
Recall that \Call{Get$_\mathrm{L}$}{$i$} returns:
\[
\set{ \evalc{ [j + m \btw i) } : j \in J,\ j + m \leq i }.
\]
On the other hand, \Call{Get$_\mathrm{A}$}{$i$} returns:
\[
\set{ \evalc{ [j + m \btw i) } : [j, k) \in \mathcal{J},\ k \leq i }.
\]
Using the notation
\[
J_{\mathrm{L}} := \set{ j \in J : j+m \leq i },
\quad
J_{\mathrm{A}} := \set{ j : [j, k) \in \mathcal{J},\ k \leq i },
\]
we can write
\[
\GetL(i) = \set{  \evalc{ [j + m \btw i) } : j \in J_L },
\quad
\GetA(i) = \set{  \evalc{ [j + m \btw i) } : j \in J_A }.
\]
The definition of $J_{\mathrm{L}}$ depends on $m$, so whether $j\in J_{\mathrm{L}}$ may change as a result of $\ShrinkL$ queries. 
On the other hand, the definition of $J_{\mathrm{A}}$ does not depend on $m$, and thus whether $[j,k)\in J_{\mathrm{A}}$ remains invariant by $\ShrinkL$ queries.
This is the main reason we work with $\GetA(\cdot)$ instead of $\GetL(\cdot)$; $\GetA(\cdot)$ is easier to compute.

\paragraph{Simulating $\GetL(i)$ using $\GetA(i)$.}

The remaining part of this section focuses on how to simulate $\leftenumds$ using $\auxds$.
To do so, we need to simulate $\GetL(\cdot)$ using $\GetA(\cdot)$.
In our construction, we manage the following condition:
\[
J = \set{ j : [j, k) \in \mathcal{J} }.
\]
Particularly, since $j+m\leq k$ holds for all $[j,k)\in \mathcal{J}$, we have $J_{\mathrm{A}}\subseteq J_{\mathrm{L}}$.
Therefore, the remaining task is to compute
\newcommand{\Jnear}{J_{\mathrm{near}}}
\[
\GetNear(i) := \set{ \evalc{ [j + m \btw i) } : j \in J_{\mathrm{near}} }
\]
for some $J_{\mathrm{near}}$ with $J_{\mathrm{L}}\setminus J_{\mathrm{A}}\subseteq J_{\mathrm{near}} \subseteq J_{\mathrm{L}}$.
Then, we can compute $\GetL(i)$ by
\[
    \GetL(i) = \GetA(i) \cup \GetNear(i).
\]
We note that the definition of $J_{\mathrm{near}}$ has some flexibility; particularly, it may have non-empty intersection with $J_{\mathrm{A}}$.
Here, we adopt the definition
\[
    J_{\mathrm{near}}:= \{j\in J\colon j+m\leq i < j+2m\}.
\]
Clearly, we have $J_{\mathrm{near}}\subseteq J_{\mathrm{L}}$.

Let 
\[
    k^*_j:=\min_{k\colon [j,k)\in \mathcal{J}}k.
\]
Then, if we have $k^*_j\leq j+2m$ for all $j\in J$, we have
\[
    J_{\mathrm{L}}\setminus J_{\mathrm{A}}
    = \{j\colon j\in J\colon j+m\leq i < k^*_j\}
    \subseteq \{j\colon j\in J\colon j+m\leq i < j+2m\}
    = J_{\mathrm{near}}.
\]
Thus, the remaining tasks are:
\begin{itemize}
    \item[(i)] implementing $\auxds$,
    \item[(ii)] implementing $\GetNear(\cdot)$, and
    \item[(iii)] ensuring that $k^*_j\leq j+2m$ holds for every moment of the algorithm.
\end{itemize}
Here, we focus on (iii) and defer (i) and (ii) to Section~\ref{sec:auxdatastructure} and Section~\ref{section:near-transitions2}, respectively.

\paragraph{Ensuring (iii).}

The issue in ensuring (iii) is that $\ShrinkA(\cdot)$ may cause (iii) to fail by decreasing $m$.  
We overcome this issue by the following idea: if $\ShrinkA(\cdot)$ breaks (iii) for some $j\in J$, then we update $k^*_j$ to $j+m$ by adding to $\mathcal{J}$ the interval $[j,j+m)$ using new $m$.


Note that this re-registration scheme clearly preserves the condition 
\[
    J=\{j\colon [j,k)\in \mathcal{J}\}
\]    
because we add $[j,k)$ to $\mathcal{J}$ only when $j\in J$.
Moreover, once a $\ShrinkA(\cdot)$ query updates $k^*_j:=j+m$ by calling $\AddA(j,j+m)$, next update will happen only after $m$ is halved; it requires the condition $j+2m < k^*_j$.
Thus, we call at most $O(\log |w|)$ \Call{Add$_\mathrm{A}$}{$\cdot$} queries in total for each call of $\AddL(\cdot)$ query, and particularly, the re-registration scheme costs $O(\log |w|)$ to the total time complexity.




%


\begin{algorithm}
\caption{Implementation of $\leftenumds$ Using the Auxiliary Data Structure $\auxds$}
\label{alg:auxiliarytotarget_left}

\DefClass{LeftEnumerator}{
    \GlobalData{}{
        $w$: entire input string\;    
    }

    \Data{}{
        $m$: non-negative integer\;
        $\widehat{j}$: the starting index of the suffix\;
        $\que$: a queue\;
        $\mathit{aux}$: $\auxds$\;
    }
    \BlankLine

    \Method{\emph{\Call{Initialize$_\mathrm{L}$}{$\widehat{j}$}}}{
        \KwIn{$\widehat{j}\in [|w|]$}
        Initialize $\widehat{j}$\;
        $\mathit{aux} \gets$ \Call{Initialize$_{\mathrm{L}}$}{$\widehat{j}$}\;
        $\que \gets \text{empty queue}$\;
    }

    \Method{\emph{\Call{Add$_\mathrm{L}$}{$j$}}}{
        \KwIn{$j\in [|w|]$ such that $w[j \btw j+m)$ is a prefix of $\widehat{w} = w[\widehat{j} \btw |w|]$}
        $\mathit{aux}.\AddA(j, j + m)$\;
        \If{$m > 0$}{
            $\que.\mathrm{pushBack}((j, m))$; \tcp{add $(j, m)$ to the end of $\que$}
        }
    }

    \Method{\emph{\Call{Shrink$_\mathrm{L}$}{$d$}}}{
        \KwIn{$d\in [m]$}
        $m \gets m-d$\; \label{alg:shrink-modify-m}
        \While{$\que$ is not empty}{\label{line:left_ds:shrink:while-loop}
            $(j, m') \gets \que.\mathrm{front}()$; \tcp{load the first element of $\que$}
            \If{$2m < m'$}{ \label{alg:shrink-loop-cond}
                $\que.\mathrm{popFront}()$; \tcp{remove the first element $(= (j, m'))$ from $\que$} \label{alg:shrink:remove-from-Q}
                $\mathit{aux.}\AddA(j, j+m)$\; \label{alg:shrink-modify-J}
                \If{$m > 0$}{
                   $\que.\mathrm{pushBack}( (j, m) )$\;
                }
            }
            \Else{
                \textbf{break}\;
            }
        }
    }

    \Method{\emph{\Call{Get$_\mathrm{L}$}{$i$}}}{
        \KwIn{$i\in [|w|]$}
        $\mathcal{D}_1 \gets \mathit{aux.}\GetA(i)$\;
        $\mathcal{D}_2 \gets \Call{GetNearTransitions}{i, m}$\;
        \textbf{return} $\mathcal{D}_1 \cup \mathcal{D}_2$\;
    }
}
\end{algorithm}

\subsection{Implementing $\leftenumds$ using $\auxds$}
\label{section:implementing-leftenumds}

\paragraph{Implementation.}

Here, we formally implement the idea in Section~\ref{section:implementing-left-B}.
Our implementation is given in Algorithm~\ref{alg:auxiliarytotarget_left}.
It internally uses $\auxds$ and a procedure \Call{GetNearTransitions}{$i,m$}, which computes the set $\GetNear(i)$ for current $m$.
%
To manage the ``re-registered'' $\AddA(\cdot)$ queries, we use a queue $Q$.
As in Section~\ref{sec:data_structures_right}, we use class or module features in our pseudocode.
Specifically, we write
\begin{itemize}
    \item $Q.\mathrm{pushBack}(x)$ for the operation of adding $x$ to the last of the queue $Q$,
    \item $Q.\mathrm{popFront}()$ for the operation of removing the first element of the queue $Q$, and
    \item $Q.\mathrm{front}()$ for the first element of the queue $Q$.
\end{itemize}

\paragraph{Correctness proof.}

The following lemma is straightforward and specifies our invariants.
\begin{lemma}\label{lem:left_ds_invariants}
Let $J$ be a set of indices.
Assume we have invoked 
\begin{itemize}
    \item \Call{Initialize$_\mathrm{L}$}{$\cdot$} once at the beginning,
    \item $\AddL(j)$ once for each $j\in J$ in any order, and
    \item $\ShrinkL(\cdot)$ and $\GetL(\cdot)$ for any number of times.
\end{itemize}
Let $\mathcal{J}$ be the set of intervals $[j,k)$ that $\AddA(j,k)$ is called.
Then, at any moment, we have the following.
\begin{itemize}
    \item[\rm{(i)}] $J = \set{ j : [j, k) \in \mathcal{J} }$.
    \item[\rm{(ii)}] $j+m \leq k^*_j \leq j+2m$ holds for each $j\in J$, where $k^*_j=\min_{k\colon [j,k)\in \mathcal{J}}k$.
    \item[\rm{(iii)}]
    If $m=0$, we have $\que = \emptyset$. Otherwise, $J = \set{ j : (j, m') \in \que }$.
    \item[\rm{(iv)}] Every pair $(j, m') \in \que$ satisfies $m \leq m' < 2m$.
    \item[\rm{(v)}] The second elements of $\que$ is non-increasing: i.e., $m_1 \geq m_2 \geq \cdots \geq m_{|\que|}$ for $\que = [(j_1, m_1), (j_2, m_2), \ldots, (j_{|\que|}, m_{|\que|})]$
\end{itemize}
\end{lemma}
\begin{proof}
We prove by induction.
Since \Call{Initialize$_\mathrm{L}$}{$\cdot$} sets $J = \mathcal{J} = \que = \emptyset$, at the beginning of the algorithm, the conditions hold.
For the remaining three methods, we assume that the conditions are satisfied before they are invoked, and we show that they remain satisfied after the execution.

\paragraph{$\GetL(\cdot)$.}
All conditions are clearly preserved because \Call{Get$_{\mathrm{L}}$}{$\cdot$} updates no data.

\paragraph{$\AddL(j)$.}
The conditions~(i)~and~(ii) are preserved because $\AddL(j)$ just adds $j$ to $J$ and $[j, j+m)$ to $\mathcal{J}$ without updating $m$.
The conditions~(iii)~and~(iv) are preserved because $Q$ is not updated when $m=0$ and $(j, m)$ is added to $\que$ when $m>0$.
The condition~(v) is preserved because of the facts that $m_{|Q|}\geq m$ held before processing the query, which is from the condition~(iv) of the induction hypothesis, and $(j,m)$ is added to the last of $Q$.

\paragraph{$\ShrinkL(d)$.}
The condition~(i) is preserved because $\ShrinkL(d)$ adds an interval $[j,j+m)$ to $\mathcal{J}$ only when $(j,m')$ was already in $Q$ for some $m'$, which ensures $j\in J$ because of the condition~(iii) of the induction hypothesis.
For the remaining conditions, we branch into two cases: whether $m$ becomes $0$ at line~\ref{alg:shrink-modify-m} or not.

\subparagraph{$m$ becomes $0$.}
From the condition~(iii) of the induction hypothesis, we have $J=\{j\colon (j,m')\in Q\}$ before processing the query.
Moreover, from the condition~(iv) of the induction hypothesis, the condition $0 = 2m < m'$ at line~\ref{alg:shrink-loop-cond} holds for all $(j,m')\in Q$.
Thus, the method adds $[j,j)$ to $\mathcal{J}$ and sets $k^*_j=j$ for all $j\in J$, which makes the condition~(ii) satisfied.
Moreover, the method removes all elements of $Q$ and adds no element to $Q$, which makes the conditions~(iii),~(iv),~and~(v) satisfied.

\subparagraph{$m$ remains positive.}
From the condition~(v) of the induction hypothesis, the condition $2m < m'$ at line~\ref{alg:shrink-loop-cond} holds for an initial segment of $Q$ (in order from the front), and is not satisfied for the remaining elements. Thus, the method removes all elements $(j,m')$ with $2m < m'$ from $Q$ and instead adds $(j,m)$ at the last of $Q$.
That is, using the largest index $x$ of $Q$ with $2m < m_x$, $Q$ becomes 
\[
\que = [ (j_{x+1}, m_{x+1}), \ldots (j_{|\que|}, m_{|\que|}),  (j_1, m), (j_2, m), \ldots, (j_x, m) ].
\]
This sets $k^*_j:=m$ for each $j\in \{j_1,\dots, j_x\}$, which makes the condition~(ii) satisfied.
The condition~(iii) is preserved because the set $\{j\colon (j,m')\in Q\}$ does not change.
The condition~(iv) is clear from the definition of $x$.
The condition~(v) is from $m_{x+1}\geq \dots\geq m_{|\que|}\geq m$, which is from the conditions~(iv)~and~(v) of the induction hypothesis.
\end{proof}

Now, we prove the correctness of Algorithm~\ref{alg:auxiliarytotarget_left}.
\begin{lemma}
\Call{Get$_\mathrm{L}$}{$i$} returns the following set of transitions:
\[
\set{ \evalc{ [j+m \btw i) } : j \in J,\ j+m \leq i }.
\]
\end{lemma}
\begin{proof}
From the condition~(i) of Lemma~\ref{lem:left_ds_invariants}, we have
\begin{align*}
    \GetA(i) &= \left\{ \evalc{ [j+m \btw i) } : [j, k) \in \mathcal{J},\ k \leq i \right\}
    = \left\{ \evalc{ [j+m \btw i) } : j \in J,\ k^*_j \leq i \right\}.
\end{align*}
Moreover, from the condition~(ii) of Lemma~\ref{lem:left_ds_invariants}, we have
\begin{align*}
    \left\{ \evalc{ [j+m \btw i) } : j \in J,\ j+2m \leq i \right\}
    &\subseteq \left\{ \evalc{ [j+m \btw i) } : j \in J,\ k^*_j \leq i \right\}\\
    &\subseteq \left\{ \evalc{ [j+m \btw i) } : j \in J,\ j+m \leq i \right\}.
\end{align*}
Thus, we have
\begin{align*}
    \GetL(i) &= \GetA(i) \cup \Call{GetNearTransitions}{i,m}\\
    &= \left\{ \evalc{ [j+m \btw i) } : j \in J,\ k^*_j \leq i \right\}
    \cup \left\{ \evalc{ [j+m \btw i) } : j \in J,\ j+m \leq i < j+2m \right\}\\
    &= \left\{ \evalc{ [j+m \btw i) } : j \in J,\ j+m \leq i \right\}.\qedhere
\end{align*}
\end{proof}

\paragraph{On Time Complexity.}

Now we analyze the time complexity.
From the discussion in Sections~\ref{sec:auxdatastructure}~and~\ref{section:near-transitions2}, each query of $\auxds$ and the procedure \Call{GetNearTransitions}{$\cdot$} runs in constant time.
Thus, \Call{Initialize$_{\mathrm{L}}$}{$\cdot$} and \Call{Get$_{\mathrm{L}}$}{$\cdot$} runs in constant time.
To analyze the remaining queries, we define the potential $\Phi(Q)$ by $\sum_{(j,m')\in Q}(1+\log_2 m')$.
Then, each call of \Call{Add$_{\mathrm{L}}$}{$\cdot$} increases $\Phi(Q)$ by at most $1+\log_2 |w|\leq O(\log |w|)$, while each step of the loop in \Call{Shrink$_\mathrm{L}$}{$\cdot$} defined by line~\ref{line:left_ds:shrink:while-loop} decreases $\Phi(Q)$ by at least $1$.
Therefore, \Call{Add$_{\mathrm{L}}$}{$\cdot$} and \Call{Shrink$_{\mathrm{L}}$}{$\cdot$} runs in amortized $O(\log |w|)$-time and $O(1)$-time, respectively.

\section{Implementing $\auxds$}\label{sec:auxdatastructure}



In this section, we implement $\auxds$, the data structure introduced in Section~\ref{section:implementing-left-B}.
It supports queries $\ShrinkA$, $\AddA$, and $\GetA$ in $O(1)$ time.
Some of the ideas in this section are common to $\rightenumds$ in Section~\ref{sec:data_structures_right}, so we recommend reading that section first.

\begin{itembox}[l]{$\auxds$}
\begin{description}
    \item[Depending Global Data] \regex{} $C$, the entire string $w$, and factorization forest of $C$.
    \item[Shared Data] \mbox{}
    \begin{itemize}
    \item a set of \emph{intervals} $\mathcal{J} = \{ [j_1, k_1), [j_2, k_2), \ldots \}$
    \item a suffix $\widehat{w}$
    \item a non-negative integer $m$
    \end{itemize}
\end{description}
Providing the following queries or methods:
\begin{description}
    \item[\Call{Initialize$_\mathrm{A}$}{$\widehat{j}$}]: Set $\widehat{w}\leftarrow [\widehat{j},n)$ and $m\leftarrow n-\widehat{j}$.
    \item[\Call{Add$_\mathrm{A}$}{$j, k$}]: $\mathcal{J}\leftarrow \mathcal{J}\cup \{[j, k)\}$. It is guaranteed that $j+m \leq k$ and $w[j \btw k)$ is a prefix of $\widehat{w}$.
    \item[\Call{Shrink$_\mathrm{A}$}{$d$}]: Set $m\leftarrow m-d$. It is guaranteed that $d\geq 0$.
    \item[\Call{Get$_\mathrm{A}$}{$i$}]: Return the set
      $\set{ \transm{C}(\,w[j+m \btw i)\,) : [j, k) \in \mathcal{J},\ k \leq i }$. 
\end{description}
\end{itembox}

\subsection{Overview}


\paragraph{Differences from $\rightenumds$.}

Although $\auxds$ processes queries similar to $\rightenumds$, the following points are different:
\begin{itemize}
    \item The \Call{Get$_\mathrm{A}$}{$i$} query in $\auxds$ searches for indices $j$ to the left of a given index $i$: i.e., $[j, i)$, whereas \Call{Get$_\mathrm{R}$}{$i$} query in $\rightenumds$ searches for indices $j$ to the right of $i$: i.e., $[i, j)$.
    \item $\auxds$ maintains a set $\mathcal{J}$ of intervals $[j,k)$, not a set $J$ of indices $j$. Moreover, it is guaranteed that, for each $[j,k)\in \mathcal{J}$, 
    \begin{itemize}
        \item $j+m\leq k$ holds when it is added to $\mathcal{J}$, and
        \item $w[j\btw k)$ is a prefix of pre-determined string $\widehat{w}:=w[\widehat{j}\btw |w|]$.
    \end{itemize}
    \item $\auxds$ maintains an integer $m$, which decreases over time via \Call{Shrink$_\mathrm{A}$}{$\cdot$} method and influences the output of \Call{Get$_\mathrm{A}$}{$\cdot$} method.
\end{itemize}



        


    
    



\paragraph{Defining $\mathcal{D}_{\fna}$ and $\mini_{\fna}$.}

As in the case of $\rightenumds$, we utilize the factorization forest of $C$, where each node $f$ corresponds to the interval $[\ell_f,r_f)$.
We define the $\mathcal{D}_{\fna}$ and $\mini_f$, which correspond to $\mathcal{D}_{\fna}$ and $\mathrm{maxidx}_f$ in $\rightenumds$, respectively.
Precisely, for each node $\fna$, we define
\[
\mathcal{D}_{\fna} := \set{ \eval{[j + m \btw r_{\fna})} : [j, k) \in \mathcal{J},\ k \in [\ell_{\fna}, r_{\fna}) },
\]
and for each idempotent node $g$ and $\delta\in \transm{C}$, we define
\[
\mini_{\fnb}[\trans] := 
\min\set{ t : 1 \leq t \leq q,\ \delta \in \mathcal{D}_{f_t} \comp e_{\fnb}},
\]
where $q$ is the number of children of $g$ and $f_t$ is the $t$-th child of $g$.
Note that the definitions of these values depend on $m$; in particular, they change as a result of $\ShrinkA$ queries.

For convenience, we also introduce the following notation:
\[
\mathcal{J}_\fna := \set{\,[j, k) \in \mathcal{J} : k \in [\ell_{\fna}, r_{\fna}) }.
\]
In words, $\mathcal{J}_\fna$ is the set of intervals of $\mathcal{J}$ whose right end is covered by $\fna$.
Using this notation, we can rewrite $\mathcal{D}_{\fna}$ as follows:
\[
\mathcal{D}_{\fna} = \set{ \eval{[j + m \btw r_{\fna})} : [j, k) \in \mathcal{J}_\fna }.
\]

\paragraph{Lazy Update Strategy.}

If we manage $\mathcal{D}_f$ and $\mini_g$, then we can implement the $\GetA$ query in a way similar to the $\rightenumds$ (see Section~\ref{sec:aux_get} for a formal proof). 
However, it seems difficult to manage these values fully, because they depend on $m$, and every $\ShrinkA$ query requires updating all of them.

To resolve this issue, we adopt the so-called \emph{lazy update} strategy; we postpone updating the values of $\mathcal{D}_\fna$ and $\mini_{\fnb}[\cdot]$ until they are actually accessed.
Precisely, we introduce two new methods $\UpdateA(f)$ and $\UpdateIdemA(g)$, which updates the values of $\mathcal{D}_f$ and $\mini_g[\cdot]$, respectively, to their correct values. 
Whenever the algorithm refers to these values (in both $\AddA$ and $\GetA$ queries), the corresponding $\UpdateA(f)$ or $\UpdateIdemA(g)$ is called once immediately beforehand.

For each node $f$, let $m_f$ be the value of $m$ at the time when $\UpdateA(f)$ was last called (we explicitly manage this value $m_f$ together with the value described below). 
Since $m$ only decreases, we always have $m \leq m_f$. 
Instead of maintaining $\mathcal{D}_f$, we manage the following:
\[
\lazyD_{\fna} := \set{ \eval{ [j + \underline{\man_\fna} \btw r_\fna) } : [j, k) \in \mathcal{J}_\fna }.
\]
The only difference from $\mathcal{D}_f$ is the underlined part $m_{\fna}$, which was $m$ in $\mathcal{D}_f$.

Similarly, for each idempotent node $g$, let $m^{\mathrm{Idem}}_g$ be the value of $m$ at the time when $\UpdateIdemA(g)$ was last called (similarly, we explicitly manage this value).
Since $m$ only decreases, we always have $m\leq m^{\mathrm{Idem}}_g$.
Then, we manage the following:
\[
\lazymini_{\fnb}[\delta] :=
\min\left\{ t : 1 \leq t \leq q,\ \delta \in \underline{\{\evalc{[j+m^{\mathrm{Idem}}_g\btw r_{f_t})}\colon [j,k)\in \mathcal{J}_{f_t}\}} \comp e_{g} \right\},
\]
where $q$ is the number of children of $g$ and $f_t$ is the $t$-th child of $g$.
Note that the underlined part does not necessarily coincide with $\lazyD_{f_t}$ because it evaluates $w[j+m^{\mathrm{Idem}}_g\btw r_{f_t})$, not $w[j+m_{f_t}\btw r_{f_t})$.

Since the definition of $\lazyD_{\fna}$ and $\lazymini_{g}[\cdot]$ do not depend on $m$, these values do not change via $\ShrinkA$ queries.
Moreover, if $m_f=m$, $\lazyD_{f}$ coincides $\mathcal{D}_f$, and if $m^{\mathrm{Idem}}_g=m$, $\lazymini_{g}[\cdot]$ coincides $\mini_g[\cdot]$.


\newcommand{\inva}{\diamondsuit}
\newcommand{\baseJ}{\widehat{j}}
\newcommand{\localupd}[1]{\trans^{\mathrm{L}}_{#1}}

\todo{$m_f$ の定義がややこしいので update をふたつに分けた。要対応}

\paragraph{Implementing $\UpdateA$ and $\UpdateIdemA$.}
Now, we explain how to implement $\UpdateA$ and $\UpdateIdemA$.
We here use the invariant that, at any moment, for all nodes $f$ and $[j,k)\in \mathcal{J}_f$,
\[
    j + m_f\leq k, \quad w[\baseJ \btw \baseJ + m_{\fna}) = w[j \btw j + m_{\fna}),
\]
and for all idempotent node $g$ and $[j,k)\in \mathcal{J}_g$, 
\[
    j + m^{\mathrm{Idem}}_g\leq k, \quad w[\baseJ \btw \baseJ + m^{\mathrm{Idem}}_{g}) = w[j \btw j + m^{\mathrm{Idem}}_{g}).
\]
These invariants are from the precondition of $\AddA(j ,k)$, that is, $j+m\leq k$ and $w[j \btw k)$ is a prefix of $\widehat{w}=w[\baseJ \btw |w|]$.
Nevertheless, we defer the formal proofs of those invariants after we describe the precise implementation of $\auxds$.

The following lemma gives the update formula of $\lazyD_f$.
\begin{lemma}\label{lem:ds_aux_update_d}
Let $f$ be a node and assume 
\[
    j + m_f\leq k, \quad w[\baseJ \btw \baseJ + m_{\fna}) = w[j \btw j + m_{\fna})
\]
holds for all $[j,k)\in \mathcal{J}_f$.
Then, we have
\[
    \mathcal{D}_{\fna} = \evalc{ [ \baseJ + m \btw \baseJ + m_{\fna} ) } \comp \lazyD_{\fna}.
\]
\end{lemma}
\begin{proof}
We have
\begin{align*}
    \mathcal{D}_{\fna}
    & =  \set{ \eval{ [j+m \btw r_\fna) } : [j, k) \in \mathcal{J}_\fna } \\
    & =  \set{ \eval{ [j+m \btw j+m_{\fna}) } \comp \eval{ [j+m_\fna \btw r_\fna) } : [j, k) \in \mathcal{J}_\fna }\\
    & =  \evalcc{ w[\baseJ + m \btw \baseJ + m_f) } \comp \set{ \evalc{ [j+m_f \btw r_f) } : [j, k) \in \mathcal{J}_f} \\
    & =  \evalcc{ w[\baseJ + m \btw \baseJ + m_f) } \comp \lazyD_{\fna},
\end{align*}
where the second equality follows from $j+m\leq j+m_f\leq k\leq r_f$, and the third equality follows from $w[\baseJ \btw \baseJ + m_{\fna}) = w[j \btw j + m_{\fna})$.
\end{proof}

The following lemma gives the update formula of $\lazymini_g$.
\begin{lemma}\label{lem:ds_aux_update_idem}
Let $g$ be an idempotent node and assume 
\[
    j + m^{\mathrm{Idem}}_g\leq k, \quad w[\baseJ \btw \baseJ + m^{\mathrm{Idem}}_{g}) = w[j \btw j + m^{\mathrm{Idem}}_{g})
\]
holds for all $[j,k)\in \mathcal{J}_g$.
Then, for each $\delta\in \transm{C}$, we have
\[
    \mini_{\fnb}[\trans] = 
    \min_{\delta' \in \transm{C}} \left\{ \lazymini_{\fnb}[\delta'] : \delta = \evalc{ [\baseJ + m \btw \baseJ + m^{\mathrm{Idem}}_{\fnb} ) } \comp \delta' \right\}.
\]
\end{lemma}
\begin{proof}
We denote the number of children of $g$ by $q$ and the $t$-th child of $g$ by $f_t$.
Then, by the argument similar to the proof of Lemma~\ref{lem:ds_aux_update_d}, for each $t$, we have
\begin{align*}
    \mathcal{D}_{f_t}
    =\evalc{ [ \baseJ + m \btw \baseJ + m^{\mathrm{Idem}}_{g} ) } \comp \set{ \evalc{ [j+m^{\mathrm{Idem}}_g \btw r_{f_t}) } : [j, k) \in \mathcal{J}_{f_t}}.
\end{align*}

Thus, we have
\begin{align*}
    \mini_{\fnb}[\delta]
    &= \min\set{ t : 1 \leq t \leq q,\ \delta \in \mathcal{D}_{\fna_t} \comp e_{\fnb} }\\
    &= \min\Bigl\{ t : 1 \leq t \leq q,\\
    &\quad \quad \trans \in \{\evalc{[\widehat{j}+m\btw \widehat{j}+m^{\mathrm{Idem}}_g)}\comp \evalc{[j+m^{\mathrm{Idem}}_g\btw r_{f_t})}\comp e_g \colon [j,k)\in \mathcal{J}_{f_t}\}\Bigr\}\\
    &= \min\Bigl\{ t : 1 \leq t \leq q,\ \exists \delta'\in \{\evalc{[j+m^{\mathrm{Idem}}_g\btw r_{f_t})}\comp e_g \colon [j,k)\in \mathcal{J}_{f_t}\}\},\\
    &\quad \quad \trans = \evalc{[\widehat{j}+m\btw \widehat{j}+m^{\mathrm{Idem}}_g)}\comp \delta'\}\Bigr\}\\
    & = \min_{\delta' \in \transm{C}} \left\{ \lazymini_{\fnb}[\delta'] : \delta = \evalc{ [\baseJ + m \btw \baseJ + m^{\mathrm{Idem}}_{\fnb} ) } \comp \delta' \right\}.\qedhere
\end{align*}
\end{proof}
Clearly, these updates can be computed in constant time.

\subsection{Implementing $\auxds$}

Now, we give an implementation of $\auxds$.
Due to its length, we divide our implementation into three pseudocodes, Algorithms~\ref{alg:ds_aux:1},~\ref{alg:ds_aux:2},~and~\ref{alg:ds_aux:3}.

\paragraph{Implementing $\InitializeA$, $\ShrinkA$, $\UpdateA$, and $\UpdateIdemA$.}

Algorithm~\ref{alg:ds_aux:1} includes the methods $\InitializeA$, $\ShrinkA$, $\UpdateA$, and $\UpdateIdemA$.
\Call{Initialize$_\mathrm{A}$}{$\widehat{j}$} initializes $\widehat{j}$, $m$, $\lazyD$, and $\lazymini$, which seems to consume $O(|w|)$ time. As done in $\rightenumds$, we consider this runs in constant time by creating and initializing each node only when the data structure first accesses it.

\Call{Shrink$_\mathrm{A}$}{$d$} just reduces $m$ by $d$ and clearly runs in constant time.
\Call{Update$_\mathrm{A}$}{$f$} and \Call{UpdateIdem$_\mathrm{A}$}{$g$} updates $\lazyD_f$ and $\lazymini_f$, respectively, using Lemma~\ref{lem:ds_aux_update_d} and Lemma~\ref{lem:ds_aux_update_idem}, respectively, which also run in constant time.

\begin{algorithm}
\caption{Implementation of $\auxds$ (first part).}\label{alg:ds_aux:1}
\Method{\emph{\Call{Initialize$_\mathrm{A}$}{$\widehat{j}$}}}{
    \KwIn{$\widehat{j}\in [|w|]$}
    Initialize $\widehat{j}$\;
    $m \gets |w| - \widehat{j} + 1$\;
    For each node $f$, $m_f \gets m$ and $\lazyD_f \gets \emptyset$\;
    For each idempotent node $g$, $m^{\mathrm{Idem}}_g\gets m$ and for each transition $\delta\in \transm{C}$, $\lazymini_g[\delta] \gets \infty$\;
}
\Method{\emph{\Call{Shrink$_\mathrm{A}$}{$d$}}}{
    \KwIn{$0\leq d\leq m$}
    $m \gets m - d$\;
}
\Method{\emph{\Call{Update$_\mathrm{A}$}{$\fna$}}}{
    \KwIn{A node $\fna$}
    $\lazyD_f \gets \evalc{ [\baseJ + m \btw \baseJ + m_{\fna}) } \comp \lazyD_f$\; \label{line:ds_aux_update_d}
    $m_\fna \gets m$\;
}
\Method{\emph{\Call{UpdateIdem$_\mathrm{A}$}{$g$}}}{
    \KwIn{An idempotent node $g$}
    \ForEach{$\trans\in \transm{C}$}{
        $\lazymini_{\fnb}[\trans] \gets \displaystyle\min_{\delta' \in \transm{C}} \left\{ \lazymini_{\fnb}[\delta'] : \delta = \evalc{ [\baseJ + m \btw \baseJ + m^{\mathrm{Idem}}_{\fnb} ) } \comp \delta' \right\}$\;  \label{line:ds_aux_update_idem}
    }
    $m^{\mathrm{Idem}}_g \gets m$\;
}
\end{algorithm}

\begin{algorithm}
\caption{Implementation of $\auxds$ (second part).}\label{alg:ds_aux:2}
\Method{\emph{\Call{Add$_\mathrm{A}$}{$j,k$}}}{
    \KwIn{An interval $[j, k)$ such that $w[j \btw k)$ is a prefix of $\widehat{w} = w[\baseJ \btw |w|]$}
    $f\gets $ the leaf node corresponding to the index $k$\;
    \While{$\fna$ is not root}{\label{line:ds_aux_while}
        $\fnb \gets f.\mathrm{parent}$\;
        \Call{Update$_{\mathrm{A}}$}{$\fna$}\; \label{line:add-update-node}
        $\delta \gets \eval{ [j+m \btw r_\fna) }$\;
        $\lazyD_f.\mathrm{add}(\delta)$\;\label{line:ds_aux_add_d}

        \If{$\fnb$ is an idempotent node}{
            \Call{UpdateIdem$_{\mathrm{A}}$}{$\fnb$}\;\label{line:add-update_idempotent}
            $p \gets \fnb.\mathrm{childIdx}(\fna)$ \tcp*{ $\fna$ is the $p$-th child of $\fnb$}
            $\lazymini_g[\delta \comp e_\fnb] \gets \min(p, \lazymini_g[\delta \comp e_\fnb])$\; \label{line:ds_aux_updatemax}
        }
        $f \gets g$\;
    }
}
\end{algorithm}

\paragraph{Implementing $\AddA$.}

Algorithm~\ref{alg:ds_aux:2} implements the method \AddA{}.
As in previous algorithms, we adopt classes and modules; particularly, we write
\begin{itemize}
    \item $S.\mathrm{add}(x)$ for the operation of adding the element $x$ to the set $S$,
    \item $f.\mathrm{parent}$ to represent the parent of the node $f$,
    \item $g.\mathrm{numOfChildren}$ to represent the number of children of the node $g$,
    \item $g.\mathrm{childOf}(p)$ to represent the $p$-th child (from the left) of the node $g$ ($1$-origin), and
    \item $g.\mathrm{childIdx}(f)$ to represent, for two nodes $f$ and $g$ with $g=f.\mathrm{parent}$, an integer $p$ such that $f$ is the $p$-th child of $g$ (from the left).
\end{itemize}
Now we prove the correctness of Algorithm~\ref{alg:ds_aux:2}.
First, we prove the invariants used in Lemmas~\ref{lem:ds_aux_update_d}~and~\ref{lem:ds_aux_update_idem}.
It is clear that $\UpdateA$ and $\UpdateIdemA$ does not break those invariants because $m_f$ and $m^{\mathrm{Idem}}_g$ only decrease.
Moreover, we remark that the method $\GetA$, which we implement in Algorithm~\ref{alg:ds_aux:3}, does not change $\lazyD$ or $\lazymini$ other than by calling $\UpdateA$ or $\UpdateIdemA$, and thus, does not break the invariant.
Therefore, it suffices to prove the following two lemmas.
\begin{lemma}\label{lem:aux_add_invariant}
Assume that
\[
    j + m_f\leq k, \quad w[\baseJ \btw \baseJ + m_{\fna}) = w[j \btw j + m_{\fna})
\]
holds for all nodes $f$ and $[j,k)\in \mathcal{J}_{f}$.
Then, after calling $\AddA(\cdot)$, this invariant still holds.
\end{lemma}
\begin{proof}
Since $m_f$ does not increase, it suffices to prove that the invariant holds for the interval $[j,k)$ that is newly added to $\mathcal{J}$.
Since the loop of line~\ref{line:ds_aux_while} iterates over all nodes $f$ with $k\in [\ell_f,r_f)$, $\UpdateA(f)$ is called for each $f$ with $[j,k)\in \mathcal{J}_f$, and particularly, $m_f$ is updated to $m$.
Thus, from the precondition of $\AddA(j,k)$, the invariant holds.
\end{proof}

\begin{lemma}\label{lem:aux_add_invariant_idem}
Assume that
\[
    j + m^{\mathrm{Idem}}_g\leq k, \quad w[\baseJ \btw \baseJ + m^{\mathrm{Idem}}_{g}) = w[j \btw j + m^{\mathrm{Idem}}_{g})
\]
holds for all idempotent node $g$ and $[j,k)\in \mathcal{J}_g$.
Then, after calling $\AddA(\cdot)$, this invariant still holds.
\end{lemma}
\begin{proof}
Since $m^{\mathrm{Idem}}_f$ does not increase, it suffices to prove that the invariant holds for the interval $[j,k)$ that is newly added to $\mathcal{J}$.
Since the loop of line~\ref{line:ds_aux_while} iterates over all idempotent nodes $g$ with $k\in [\ell_g,r_g)$, $\UpdateIdemA(g)$ is called for each $g$ with $[j,k)\in \mathcal{J}_g$, and particularly, $m^{\mathrm{Idem}}_g$ is updated to $m$.
Thus, from the precondition of $\AddA(j,k)$, the invariant holds.
\end{proof}

Combining Lemmas~\ref{lem:ds_aux_update_d}~and~\ref{lem:aux_add_invariant} yields that, after calling $\UpdateA(f)$, $\lazyD_f$ coincides $\mathcal{D}_f$.
Moreover, combining Lemmas~\ref{lem:ds_aux_update_idem}~and~\ref{lem:aux_add_invariant_idem} yields that, after calling $\UpdateIdemA(g)$, $\lazymini_g$ coincides $\mini_g$.
These conditions continue to be satisfied until a $\ShrinkA$ query decreases $m$.
Therefore, hereafter, we identify
$\lazyD_f$ with $\mathcal{D}_f$ and
$\lazymini_g$ with $\mini_g$ on the nodes on which $\UpdateA$ and $\UpdateIdemA$, respectively, are called after the last $\ShrinkA$ query.

The statements and proofs of Lemmas~\ref{lem:ds_aux_d_lem}--\ref{lem:ds_aux_idem} are analogous to and almost same as those of Lemmas~\ref{lem:ds_right_d_lem}--\ref{lem:ds_right_idem}.

\begin{lemma}\label{lem:ds_aux_d_lem}
Let $j \in [|w|]$.
Then, the query $\AddA(j,k)$ adds $\evalc{ [j+m \btw r_f)}$ to $\mathcal{D}_{\fna}$ for all nodes $\fna$ with $k \in [\ell_\fna, r_\fna)$ (and only for those nodes).
\end{lemma}
\begin{proof}
The loop defined in line~\ref{line:ds_aux_while} iterates over all nodes $\fna$ with $k \in [\ell_\fna, r_\fna)$.
For each of such $\fna$, line~\ref{line:ds_aux_add_d} adds the transition $\evalc{ [j+m \btw r_f) }$ to $\mathcal{D}_\fna$.
\end{proof}

Thus, we have the following, which is immediate from Lemma~\ref{lem:ds_aux_d_lem}.
\begin{lemma}\label{lem:ds_aux_d}
Let $\mathcal{J}$ be a set of indices.
Assume we have invoked $\AddA(j,k)$ once for each $[j,k) \in \mathcal{J}$, in any order.
Then, for every node $\fna$ such that $\UpdateA(f)$ is called after the last $\ShrinkA$ query,
$\mathcal{D}_{\fna} = \set{ \evalc{ [j + m\btw r_\fna) } : [j,k) \in \mathcal{J}_f}$.
\end{lemma}

We now ensure that the method $\AddA$ correctly updates $\mini_{\fnb}[\trans]$ by the following lemma.

\begin{lemma}\label{lem:ds_aux_idempotent_lem_2}
The query $\AddA(j,k)$ updates $\mini_{\fnb}[\trans]$ by $\min(p,\mini_{\fnb}[\trans])$ for all idempotent node $\fnb$ with $k\in [\ell_{\fnb},r_{\fnb})$, where 
\begin{itemize}
    \item $\trans=\evalc{ [j+m \btw r_f) }\comp e_g$, where $\fna$ is the unique child of $\fnb$ with $k\in [\ell_{\fna},r_{\fna})$, and
    \item $p$ is the integer such that $\fna$ is the $p$-th child of $\fnb$,
\end{itemize}
and does not update other idempotent nodes.
\end{lemma}
\begin{proof}
The loop of line~\ref{line:ds_aux_while} iterates over all $(g,f)$ satisfying the condition of the lemma, and line~\ref{line:ds_aux_updatemax} is processed when and only when $g$ is idempotent, $f$ is the $p$-th child of $g$, and $\delta\comp e_g=\evalc{ [j+m\btw r_f) }\comp e_g$.
\end{proof}

Now, we have the following.
\begin{lemma}\label{lem:ds_aux_idem}
Let $\mathcal{J}$ be a set of intervals.
Assume we have invoked $\AddA(j,k)$ once for each $[j,k) \in \mathcal{J}$, in any order.
Then, for each idempotent node $\fnb$ and $\trans \in \transm{C}$, we have
\[
    \maxidx_{\fnb}[\trans] = \max\set{ t : 1 \leq t \leq q,\ \trans \in \mathcal{D}_{\fna_t}\comp e_g }.
\]
where $q$ is the number of children of $\fnb$ and $\fna_t$ is the $t$-th child of $\fnb$.
\end{lemma}
\begin{proof}
From Lemma~\ref{lem:ds_aux_idempotent_lem_2},
$\maxidx_{\fnb}[\trans]$ is equal to the maximum among integers $t$ such that there is an interval $[j,k)\in \mathcal{J}$ with
\begin{itemize}
    \item $k\in [\ell_{\fna_t}, r_{\fna_t})\subseteq [\ell_{\fnb}, r_{\fnb})$, where $\fna_t$ is the $t$-th child of $g$, and
    \item $\delta = \evalc{[j+m, r_{\fna_t})}\comp e_g$.
\end{itemize}
By definition of $\mathcal{D}_{f_t}$, we have
\begin{align*}
    \set{\evalc{ [j+m\btw r_{\fna_t}) }\comp e_g : [j,k) \in \mathcal{J}_{f_t} }
    =\mathcal{D}_{\fna_t}\comp e_g.
\end{align*}
Thus, we have
\begin{align*}
    \maxidx_{\fnb}[\trans] 
    &= \max\set{ t : 1 \leq t \leq q,\ \exists [j,k)\in \mathcal{J}_{f_t},\ \trans = \evalc{ [j+m, r_{\fna_t}) }\comp e_g }\\
    &= \max\set{ t : 1 \leq t \leq q, \trans \in \mathcal{D}_{\fna_t}\comp e_g }.\qedhere
\end{align*}
\end{proof}

\paragraph{Implementing $\GetA$.}\label{sec:aux_get}

\begin{algorithm}[t!]
\caption{Implementation of $\auxds$ (third part).}\label{alg:ds_aux:3}
\Method{\emph{\Call{Get$_\mathrm{A}$}{$i$}}}{
    \KwIn{$i\in [|w|]$}
    Let $\fna$ be the leaf node corresponding to the index $i$\;
    $\UpdateA(f)$\;
    $\accum \gets \lazyD_{\fna}$\;\label{line:ds_aux_lengthzero}
    \While{$f$ is not root}{\label{line:ds_aux_get_while}
        $g\gets f.\mathrm{parent}$\;
        $p\gets g.\mathrm{childIdx}(f)$\tcp*{$\fna$ is the $p$-th child of $\fnb$}
        \If(\tcp*[f]{IF: $\fna$ is not the leftmost child of $\fnb$}){$p>1$}{
            $\fna^- \gets \fnb.\text{childOf}(p-1)$ \tcp*{$\fna^-$ is the immediate left sibling of $\fna$}
            $p \gets g.\text{childIdx}(f^{-})$\;
            $\UpdateA(\fna^-)$\;
            \ForEach{$\trans \in \lazyD_{\fna^-}$}{
                $\accum.\text{add}\left(\trans \comp \evalc{ [r_{\fna^-} \btw i)}\right)$\;\label{line:ds_aux_addnext}
            }
            \If{$\fnb$ is an idempotent node}{
                $\UpdateIdemA(g)$\;
                \ForEach{$\trans \in \transm{C}$}{
                    \If{$p - 1 > \lazymini_{\fnb}[\delta]$}{
                        $\accum.\text{add}\left(\trans \comp \eval{ [r_{\fna^-} \btw i)}\right)$\;\label{line:ds_aux_addidem}
                    }
                }
            }
        }
        $f \gets g$\;
    }
    \textbf{return} $\accum$\;
}
\end{algorithm}

Now, we give an implementation of the $\GetA$ method in Algorithm~\ref{alg:ds_aux:3}.
Since the algorithm calls $\UpdateA(f^-)$ immediately before using the value of $\lazyD_{f^-}$ and $\UpdateIdemA(g)$ immediately before using the value of $\lazymini_g$, as done in the analysis of $\AddA$, we identify $\lazyD_{f^-}$ into $\mathcal{D}_f$ and $\lazymini_g$ into $\mini_g$.
The following lemma proves the correctness of the implementation, which is analogous to Lemma~\ref{lem:ds_right_get}.
\begin{lemma}\label{lemma:ds_aux_get}
Let $i\in [|w|]$. Then, \Call{Get$_\mathrm{A}$}{$i$} returns the set 
$\set{ \evalc{ [j+m \btw i) } : [j, k)\in \mathcal{J},\ k \leq i }$.
\end{lemma}
\begin{proof}
\Call{Get$_\mathrm{A}$}{$i$} returns the union of
\begin{itemize}
    \item[(i)] $\mathcal{D}_{\fna}$ for the leaf node $\fna$ corresponding to the index $i$ (line~\ref{line:ds_aux_lengthzero}),
    \item[(ii)] $\mathcal{D}_{\fna^-} \comp \eval{ [r_{\fna^-} \btw i)}$ for each node $\fna$ with $i \in [\ell_\fna, r_\fna)$ that is not the leftmost child of the parent of $\fna$, where $\fna^-$ is the immediate left sibling of $\fna$ (line~\ref{line:ds_aux_addnext}), and
    \item[(iii)] $\{ \trans \comp \evalc{ [r_{\fna^-} \btw i) } \colon \trans \in \transm{C},\ p-1> \mini_\fnb[\trans] \}$ for each node $\fna$ with $i \in [\ell_\fna, r_\fna)$ that is the non-leftmost $p$-th child of an idempotent node $\fnb$, where $\fna^-$ is the immediate left sibling of $\fna$, $(p-1)$-th child of $g$ (line~\ref{line:ds_aux_addidem}).
\end{itemize}
From now on, we rewrite these sets in a more manageable form.

\paragraph{(i).}
We have
\begin{align*}
    \mathcal{D}_f
    &= \set{ \evalc{ [j+m\btw r_f) } : [j,k) \in \mathcal{J}_f }
    = \set{ \evalc{ [j+m\btw r_f) } : [j,k) \in \mathcal{J},\ k\in [\ell_f,r_f)}\\
    &= \set{ \evalc{ [j+m \btw r_f) } : [j,k) \in \mathcal{J},\ \underline{ k = i } }
\end{align*}
where the first equality is from the definition of $\mathcal{D}_f$ and the third equality is from the fact that $f$ is the leaf.

\paragraph{(ii).}
For each node $\fna$ satisfying the condition of (ii), we have
\begin{align*}
  \mathcal{D}_{\fna^-}\comp \evalc{ [r_{\fna^-}\btw i) }
   & =  \set{\evalc{ [j+m\btw r_{\fna^-}) }\comp \evalc{ [r_{\fna^-}\btw i)} : [j,k) \in \mathcal{J}_{f^-}} \\
   & =  \set{ \evalc{ [j+m \btw i) } : [j,k)\in \mathcal{J},\ \underline{k\in [\ell_{\fna^-}, r_{\fna^-}) } },
\end{align*}   
where the first equality is from the definition of $\mathcal{D}_{f^-}$.

\paragraph{(iii).}
Let $\fnb$ be an idempotent node and $\fna_1,\ldots, \fna_{p-1}=\fna^{-}, \fna_{p}=\fna, \ldots, \fna_q$ be its children such that $i \in [\ell_\fna, r_\fna)$.
Then, we obtain
\begin{align*}
    &\set{\trans \comp  \evalc{[r_{\fna^-}\btw i) } : \trans \in \transm{C},\ \mini_{\fnb}[\trans] < p-1 }\\
    &= \set{\trans \comp  \evalc{[r_{\fna^-}\btw i) } : \trans \in \transm{C},\ \min\{t : 1\leq t\leq q,\ \delta \in \mathcal{D}_{f_t}\comp e_g\} < p-1 }\\
    &= \bigcup_{t=1}^{p-2} \mathcal{D}_{\fna_t}\comp e_g\comp \evalc{ [r_{\fna^-}\btw i) }\\
    &= \bigcup_{t=1}^{p-2} \mathcal{D}_{\fna_t}\comp \underbrace{e_g\comp\dots \comp e_g}_{p-1-t}\comp \evalc{ [r_{\fna^-}\btw i) }\\
    &= \bigcup_{t=1}^{p-2} \left\{\evalc{[j+m\btw r_{f_t})}\comp \evalc{[\ell_{f_{t+1}}\btw r_{f_{t+1}})}\comp \dots \comp \evalc{[\ell_{f_{p-1}}\btw r_{f_{p-1}})}\comp \evalc{ [\ell_{f_p}\btw i) }: [j,k)\in \mathcal{J}_{f_t}\right\}\\
    &= \bigcup_{t=1}^{p-2} \left\{\evalc{[j+m\btw i): [j,k)\in \mathcal{J}_{f_t} }\right\}\\
    &= \set{ \evalc{ [j+m \btw i) } : [j,k) \in \mathcal{J},\ \underline{k\in [\ell_{g}, \ell_{f^-})}}.
\end{align*}

Now, it suffices to prove that the union of intervals
\begin{itemize}
    \item[(I)] consisting only of an index $i$, 
    \item[(II)] $[\ell_{\fna^-}, r_{\fna^-})$ for each pair of nodes $(\fna, \fna^-)$ defined in (ii), and
    \item[(III)] $[\ell_{\fnb}, \ell_{\fna^-})$ for each tuple of nodes $(\fnb, \fna, \fna^-)$ defined in (iii)
\end{itemize}
coincides with the interval $[1, i]$.
We omit the proof, since it is identical to the corresponding part of the proof of Lemma~\ref{lem:ds_right_get}, except that the left and right are interchanged.
\end{proof}

\paragraph{On Time Complexity.}

Clearly, each query in Algorithms~\ref{alg:ds_aux:1},~\ref{alg:ds_aux:2},~and~\ref{alg:ds_aux:3} runs in time proportional to the depth of the factorization forest and $|\transm{C}|$, both of which are constants.
Therefore, as desired, each query runs in constant time.

%% file: onerewb/new-near-transitions.tex

\newcommand{\target}{\heavy^{i, m}_v}
\newcommand{\jmin}{j_0}
\newcommand{\jmax}{j_{\text{last}}}
\newcommand{\lightsym}{\sigma_{\!\light}}
\newcommand{\wB}{w_B}

\section{Enumerating Near Transitions}
\label{section:near-transitions2}

In this section, we construct \Call{GetNearTransitions}{$i,m$}, which we used in Section~\ref{new:section:left-B-overview} but whose construction was deferred.

\paragraph{Recalling the setting.}

We first recall our setting.
Let $\stree$ be the suffix tree of the string $w$, $v$ be a vertex of $\stree$, $\heavy_v$ and $\light_v$ be the set of indices in heavy and light subtree of $v$, respectively.
We assume that the substring $w_B:=\nword{v}$ corresponding to the root-$v$ path is in $\lang{B}$.
We denote $m:=|w_B|$.
Let $i \in \light_v$ be an index with $w[i + m \btw |w|]\in \lang{D}$.

The goal here is to give an $O(\log |w|)$-time algorithm to compute the set
\[
    \mathcal{D}:=\set{ \evalc{ [j+m \btw i) } : j \in J,\ j+m \leq i < j+2m },
\]
where 
\[
    J = \set{ j \in \heavy_v : w[1\btw j) \in \lang{A} }.
\]
Let
\[
    \target
    := \set{ j \in \heavy_v : j+m \leq i < j+2m }.
\]
Then, we can rewrite $\mathcal{D}$ as:
\[
    \mathcal{D} =\set{ \evalc{ [j+m \btw i)} : j \in \target, w[1 \btw j) \in \lang{A} }.
\]







\paragraph{Operations on $\heavy_v$.}
We manage $\heavy_v$ by a balanced binary search tree (e.g., an AVL tree), which can perform the following $O(\log |\heavy_v|)$-time~\cite[Chap.~13]{Cormen:2022}:
\begin{itemize}
    \item Adding a value $x$ to $\heavy_v$.
    \item Checking if there is a value $x$ in $\heavy_v$.
    \item Given a value $x$, computing a successor/predecessor of $x$, i.e., $\min\set{ y \in \heavy_v : x < y }$ or $\max\set{ y \in \heavy_v : y < x }$.
\end{itemize}
We can modify Algorithm~\ref{alg:auxiliarytotarget_left} so that it additionally maintains this search tree as follows. In the \Call{Initialize$_\mathrm{L}$}{$\cdot$} method, we initialize the search tree. Each time the $\AddL(j)$ method is called, we add $j$ into the search tree. This modification contributes to the time complexity of \textsc{Left-$B$ Transition Enumeration} by only $O(|\stree[\pi]|\log |w|)$, and thus does not affect the overall time complexity.

As in previous sections, we adopt classes or modules; particularly, we write $\heavy_v.\mathrm{succ}(x)$ and $\heavy_v.\mathrm{pred}(x)$ to represent the successor and predecessor of $x$ in $\heavy_v$, respectively.


\paragraph{Computing $\mathcal{D}$.}
Now we consider computing $\mathcal{D}$.
Assuming $\target$ is non-empty, we define the following two values:
\[
    j_0 := \min \target, \quad j_{\text{last}} := \max \target.
\]
We note that $\jmin$ and $\jmax$ can be computed in $O(\log |w|)$-time using $\heavy_v.\text{succ}(i-2m)$ and $\heavy_v.\text{pred}(i-m+1)$, as well as determining non-emptiness of $\target$.

If $|\target|\leq 2$, we can compute $\mathcal{D}$ just by brute-forcing on $\target$. Note that whether $|\target|\leq 2$ can be determined by the constant number of calls of successor operation.
Now, we assume $|\target| \geq 3$.
Let $\bar{w}:=w[j_0\btw  j_{\text{last}}+m)$.
The following lemma exploits the structure of $\target$ and $\bar{w}$.
\begin{lemma}\label{lem:near_periodic}
$\bar{w}$ is a periodic string: i.e., $\bar{w} = \theta^r$ for some prefix $\theta$ of $\bar{w}$ and rational $r > 1$.
Moreover, $\target$ is an arithmetic progression with common difference $d$, where $d$ is a period of $\bar{w}$.
\end{lemma}
\begin{proof}

Our proof is two-fold.
First, we determine the common difference $d$.
Next, we show the set $\set{ \jmin, \jmin + d, \jmin + 2d, \ldots, \jmax }$ coincides $\target$.

\noindent{}\textbf{Determining $d$.}
Take arbitrary $j\in \target \setminus \set{ \jmin, \jmax }$.
The positional relation among indices $j_0$, $j$, and $j_{\mathrm{last}}$ are illustrated as follows.
{
\newlength{\BoxW}\setlength{\BoxW}{3.5cm}
\newlength{\BoxH}\setlength{\BoxH}{6mm}
\newcommand{\TopY}{-1}   
\newcommand{\Drop}{.4}

\newcommand{\ArrowAt}[2]{%
  \draw[->,>=Stealth] (#1,\TopY) -- (#1,\TopY-\Drop);
  \node[above=.5mm] at (#1,\TopY) {\scriptsize $#2$};
}
\tikzset{wb/.style={draw,minimum width=\BoxW,minimum height=\BoxH,
                    inner sep=0pt,align=left}}

\[
\begin{tikzpicture}[x=8mm,y=\BoxH] 
  \ArrowAt{-1}{i-2m+1}
  \ArrowAt{3.3}{i-m}
  \ArrowAt{7.4}{i}

  \node[wb] (B1) at (1.4,-2) {$w_B$};
  \node[wb] (B2) at (2.4,-3.5) {$w_B$};
  \node[wb] (B3) at (4.4,-5) {$w_B$};

  
  \node at ($(B1.west)+(-0.08, -.8)$) {$\jmin$};
  \node at ($(B2.west)+(-0.08, -.8)$) {$j$};
  \node at ($(B3.west)+(-0.08, -.8)$) {$\jmax$};  
  
\end{tikzpicture}
\]}%

Since $j_0,j \in \target \subseteq \heavy_v$, we have $w[j_0 \btw j_0+m) = w[j \btw j+m) = w_B$.
Thus, $w_B[t]=w_B[t+(j-j_0)]$ holds for all $t\in \{1,\dots, m-(j-j_0)\}$, and particularly, $w_B$ has a period $(j-j_0)$.
Applying the same argument for $j$ and $j_{\text{last}}$ yields that $w_B$ has a period $(j_{\text{last}}-j)$.

By the constraint $i - 2m + 1 \leq \jmin < \jmax \leq i-m$, the following holds:
\[
    |w_B| = m > (i - m) - (i - 2m + 1) \geq j_{\text{last}}-j_0=(j_{\text{last}}-j) + (j-j_0).
\]
Thus, we can apply Fine and Wilf's Theorem~(Lemma~\ref{lemma:FineWilf}) to obtain that $w_B$ has a period 
\[
    d_j := \gcd(j-j_0,j_{\text{last}}-j) < \frac{m}{2},
\]
which is a divisor of $j_{\text{last}}-j_0$. 
Let $e_j := (\jmax - \jmin) / d_j$, which is an integer.

Let $\theta_j:=w_B[1\btw d_j]$.
Then, $w_B = (\theta_j)^r$ for some rational $r > 1$.
We also have $w[j_0 \btw j_\text{last}) =\theta_j^{e_j}$ because it is a prefix of $w_B$.
Thus, we have
\[
    \bar{w} = w[j_0\btw j_\text{last})\, w[j_{\text{last}}\btw j_{\text{last}}+m)
    = \theta_j^{e_j}\, \underbrace{w_B}_{\text{has period $d_j$}}
\]
also has a period $d_j$.
Since this holds for all $j\in \heavy_{i,m}\setminus \{j_0,j_{\text{last}}\}$, again from Fine and Wilf's Theorem, $\bar{w}$ has a period
\[
  d := \gcd \set{ d_j : j \in  \target \setminus \set{ \jmin, \jmax } },
\]
where the reason why we can apply the Theorem is that, for any two $j, j' \in \target \setminus \set{ \jmin, \jmax }$, $d_j + d_{j'} < m \leq |\bar{w}|$ holds.

\noindent{}\textbf{Proving $\target$ forming an arithmetic progression.}
Let $e=(j_{\text{last}}-j_0)/d$. 
To complete the proof, we prove 
\[
    \target = \set{ j_0+kd : k \in \set{ 0, 1, \ldots, e } }.
\]

\noindent{}\textbf{$(\subseteq)$:}
Let $j \in \target$.
If $j=j_0$, we have $j=j_0 + 0\cdot d$. Moreover, if $j=j_{\text{last}}$, we have $j=j_0+ed$.
Otherwise, take $d_j$ defined above, which divides $d$ and is a divisor of $j-j_0$.
It leads that there is an integer $k_j$ with $j = \jmin + k_j d$, which satisfies $0 < k_j < e$ because of $j_0 < j < j_{\text{last}}$.

\noindent{}\textbf{$(\supseteq)$:}
Conversely, let $k \in \set{ 0, 1, \ldots, e }$.
If $k=0$, we have $j_0+kd=j_0 \in \target$.
Moreover, if $k=e$, we have $j_0+kd=j_{\text{last}}\in \target$.
Assume otherwise.
Since $\bar{w} = w[j_0\btw j_{\text{last}}+m)$ has a period $d$, 
we have $w[j_0 + kd\btw j_0+kd+m)=w[j_0\btw j_0+m)=w_B$.
Thus, $j_0+kd$ is either in $\heavy_v$ or $\light_v$.
Moreover, again from periodicity, we have $w[j_0+kd+m]=w[j_0+m]$; that is, the letter immediately after the substring $w[j_0 + kd\btw j_0+kd+m)$ is same as that of the substring $w[j_0\btw j_0+m)$.
Thus, combining with the fact that $j_0\in \heavy_v$, we have $j_0+kd\in \heavy_v$, and thus, $j_0+kd\in \target$.
\end{proof}

From Lemma~\ref{lem:near_periodic}, we can write $\target = \set{ j_0 + kd : k \in \set{0, 1, \ldots, e } }$ using a tuple $(j_0, d, e)$.
This tuple $(\jmin, d, e)$ can be computed in $O(\log |w|)$-time as follows.
First, we compute the smallest element $j_0$, the second smallest element $j_1$, and the largest element $j_{\text{last}}$ of $\target$ in $O(\log |w|)$-time using the successor/predecessor queries on $\heavy_v$.
Then, we have $d=j_1-j_0$ and $e=(j_{\text{last}}-j_0)/(j_1-j_0)$.

\def\suba{\theta_1}
\def\subb{\theta_2}

On $(\jmin, d, e)$, we define two substrings $\suba := w[j_0 \btw j_0+d)$, and $\subb := w[j_0+m \btw j_0+m+d)$.
Since $\bar{w}$ has a period $d$ that divides $j_{\mathrm{last}}-j_0$, we have
\[
    \bar{w} = \suba^{e}\ w_B = w_B\ \subb^e.
\]

Let $w_1 := w[1 \btw \jmin)$ and $w_2 := w[\jmax + m  \btw i)$.
We define
\[
\mathcal{S}_{A}  :=  \set{ k \geq 0 : w_1\, \suba^k \in \lang{A}},
\]
and for each $\delta\in \transm{C}$,
\[
    \mathcal{S}_{C, \delta} :=  \set{ k \geq 0 : \Delta_C(\subb^k \: w_2) = \delta }.
\]
From Lemma~\ref{lem:wtw_is_ult_periodic}, $\mathcal{S}_{A}$ and $\mathcal{S}_{C,\delta}$ are ultimately periodic.
Moreover, for every transition $\delta$ of $C$,
the following set
\[
    \mathcal{S}_{\delta} := \mathcal{S}_{A} \underbrace{+}_{\text{Minkowski Sum}} \mathcal{S}_{C,\delta}
\]
is also ultimately periodic because of Lemma~\ref{lem:ult_operations}.
Furthermore, if we write $\mathcal{S}_{\delta}=(\mu,M,\lambda,T)$, $\mu$ and $\lambda$ are bounded by a constant that depends only on $A$ and $C$.
Now, we have the following.
\begin{lemma}\label{lem:S_delta_represents_D}
A transition $\delta\in \transm{C}$ is in $\mathcal{D}$ if and only if $e\in \mathcal{S}_{\delta}$.
\end{lemma}
\begin{proof}

\noindent{}\textbf{($\Rightarrow$):}
Assume $\delta\in \mathcal{D}$.
Then, there is an index $j \in \target$ such that $w[1 \btw j) \in \lang{A}$ and $\evalc{[j+m \btw i)} = \delta$.
From Lemma~\ref{lem:near_periodic}, $j = j_0+kd$ holds for some $k \in \set{0, \dots, e }$, and thus,
\[
    w[1\btw j) =  w_1\, w[j_0 \btw j) = w_1\, \suba^k,
\]
which implies $k\in \mathcal{S}_A$.
Moreover, we have 
\[
    w[j+m\btw i) = w[j+m \btw \jmax+m)\, w_2 = \subb^{e-k}\, w_2,
\]
and thus, $e-k\in \mathcal{S}_{C,\delta}$.
Therefore, from the definition of the Minkowski sum, we have $e=k+(e-k)\in \mathcal{S}_A+\mathcal{S}_{C,\delta}=\mathcal{S}_{\delta}$.

\noindent{}\textbf{($\Leftarrow$):}
Let $e\in \mathcal{S}_{\delta}$.
Then, there is an integer $k\in \{0,\dots, e\}$ with $k\in \mathcal{S}_A$ and $e-k\in \mathcal{S}_{C,\delta}$.
Let $j=j_0+kd$.
Then, we have
\[
    w[1\btw j) = w_1\, w[j_0\btw j) = w_1\, w[j_0 \btw j_0+kd) = w_1\, \theta^k \in \lang{A}.
\]
Moreover, we have
\[
    \evalc{[j+m\btw i)}=\Delta_C(w[j+m\btw j_{\text{last}}+m)\ w_2) = \Delta_C(\subb^{e-k}\, w_2) = \delta,
\]
where the last equality is from $e-k \in \mathcal{S}_{C,\delta}$.
Therefore, we have $\delta \in \mathcal{D}$ and the lemma is proved.
\end{proof}

\paragraph{On Time Complexity.}

Since $\mathcal{S}_{\delta}$ depends only on $A$ and $C$, we can compute it in constant time, as well as determining if $e\in \mathcal{S}_{\delta}$. Computing $e$ takes $O(\log |w|)$ time.
Thus, from Lemma~\ref{lem:S_delta_represents_D}, the desired set $\mathcal{D}$ can be computed in $O(\log |w|)$ time.

%% file: onerewb/ultimately-periodic.tex

\section{Proof of Lemma for Closure Properties of Ultimately Periodic Sets}\label{appendix:closure-of-UP}

We prove Lemma~\ref{lem:ult_operations} in Section~\ref{section:tools}.

\PropertiesOfUltimatelyPeriodic*

\begin{proof}
Let $P=(\mu, M, \lambda, T)$ and $Q=(\mu', M', \lambda', T')$.
Moreover, let $\bar{\lambda}$ be the least common multiple of $\lambda$ and $\lambda'$.

\textbf{Union:} It suffices to prove that for $i\geq \max(\mu,\mu')$, $i\in P\cup Q$ if and only if $i+\bar{\lambda}\in P\cup Q$.
Assume $i\in P\cup Q$. 
Then, we have either $i\in P$ and $i\in Q$. 
In the case where $i\in P$, we have $i\geq \max(\mu,\mu')\geq \mu$, and thus, $(i\bmod \lambda)\in T$. 
Thus, we have $((i+\bar{\lambda})\bmod \lambda)\in T$, which implies $i+\bar{\lambda}\in P\subseteq P\cup Q$ because $i+\bar{\lambda}\geq \mu$.
The same argument also works for the case of $i\in Q$.
Thus, we have $i+\bar{\lambda}\in P\cup Q$.
The inverse direction is proved in the same way.

\textbf{Minkowski Sum:} It suffices to prove that for $i\geq \mu+\mu'+\bar{\lambda}$, $i\in P\cup Q$ holds if and only if $i+\bar{\lambda}\in P + Q$.
Assume $i\in P + Q$. 
Then, there is a pair of integers $j\in P$ and $k\in Q$ with $j+k=i$.
Since $j+k=i\geq \mu + \mu'$, we have either $(j\bmod \lambda)\in T$ or $(k\bmod \lambda')\in T'$. 
If $(j\bmod \lambda)\in T$, we have $((j+\bar{\lambda})\bmod \lambda)\in T$ and thus, $j+\bar{\lambda}\in P$. 
Thus, $i+\bar{\lambda}=(j+\bar{\lambda})+k\in P+Q$.
The same argument also holds for the case $(k\bmod \lambda')\in T'$.
Therefore, we have $i+\bar{\lambda}\in P+Q$.

For the inverse direction, assume $i+\bar{\lambda}\in P+Q$.
Then, there is a pair of integers $j'\in P$ and $k'\in Q$ with $j'+k'=i+\bar{\lambda}$.
Since $j'+k'-2\bar{\lambda}=i-\bar{\lambda}\geq \mu + \mu'$, we have either $j'-\bar{\lambda}\geq \mu$ or $k'-\bar{\lambda}\geq \mu'$.
If the former holds, since $((j'-\bar{\lambda})\bmod \lambda)\in T$, we have $i=(j'-\bar{\lambda})+k'\in P+Q$.
The same argument also holds for the latter case.
Therefore, we have $i\in P+Q$.

\textbf{$e$-Multiplication:} The set $e\cdot P$ is equal to $(e\mu, e\cdot M, e\lambda, e\cdot T)$.
\end{proof}

%% file: onerewb/new-normalizing-1useREWB.tex
\section{Normalizing 1-use REWB: Proof of Theorem~{\ref{theorem:abcbd-singleuse}}}

\label{appendix:1use1rewb-to-abcbd}

\newcommand{\template}{\mathbb{T}}
\newcommand{\elim}{\ensuremath{\mathit{elim}}}
\newcommand{\zero}{\ensuremath{\mathit{zero}}}
\newcommand{\one}{\ensuremath{\mathit{one}}}
\newcommand{\defeq}{{:=}}
\newcommand{\exprN}{K}
\newcommand{\Vars}{\mathcal{V}}

\newcommand{\REWB}{\rewb}
\newcommand{\REGEX}{\regex}

In this section, we prove our Theorem~\ref{theorem:abcbd-singleuse} for 1-use \rewb.

\Singleuse*

\paragraph{Structure of This Appendix.}

\newcommand{\noref}{\mathit{noref}}
\newcommand{\last}{\mathit{last}}

We prove Theorem~\ref{theorem:abcbd-singleuse} by combining Lemmas~\ref{lemma:appendix:normalizing step1}~and~\ref{lemma:appendix:normalizing step2}.
After introducing the required definitions through Sections~\ref{appendix:normalizing:basic definitions}~and~\ref{appendix:normalizing:overview},
we prove Lemma~\ref{lemma:appendix:normalizing step1} in Section~\ref{appendix:lift references}.
\begin{restatable}[Normalizing First Step]{lemma}{NormalizingFirstStep}
\label{lemma:appendix:normalizing step1}
Let $E$ be a $1$-use \rewb.
Then, we can decompose $E$ into the following form:
\[
(\alpha_1\,\readv{x}\,D_1) + (\alpha_2\,\readv{x}\,D_2) + \cdots + (\alpha_{K'}\,\readv{x}\,D_{K'}) + \Phi
\]
where each $\alpha_i$ is a reference-free (i.e., not including $\readv{y}$ for any $y$) and $D_i$ and $\Phi$ are \regex{} rather than \rewb.
In particular, the following holds:
\[
K' \leq |E|,\ 
|\alpha_i| = O(|E|),\ 
|D_j| = O(|E|),\ 
|\Phi| = O(|E|).
\]
\end{restatable}

After proving this lemma,
in Section~\ref{appendix:restricted grammar},
we transform $\alpha_i$ into a corresponding expression generated by the following restricted grammar:
\[
\mathcal{H} ::= \epsilon \mid \sigma \mid \mathcal{H}\,\mathcal{H} \mid \mathcal{H}\,+\,\mathcal{H} \mid \mathcal{H}^* \mid [ \mathbf{REGEX} ]_x.
\]

Then we prove the second phase lemma in Section~\ref{appendix:normalizing:lift binding}.
\begin{restatable}[Normalizing Second Step]{lemma}{NormalizingSecondStep}
\label{lemma:appendix:normalizing step2}

Let $\alpha$ be an expression of the restricted grammar $\mathcal{H}$.
Then, we can decompose $\alpha$ into the following form:
\[
(A_1\ [B_1]_x\ C_1) +
(A_2\ [B_2]_x\ C_2) + \cdots +
(A_{K''}\ [B_{K''}]_x\ C_{K''}) + \Psi,
\]
where $A_i, B_i, C_i$, and $\Psi$ are \regex{es}.

In particular, the following holds for each expression size:
\[
K'' \leq |E|,\ 
|A_i| = O(|E|^2),\ 
|B_i| \leq |E|,\ 
|C_i| = O(|E|^2),\ 
|\Psi| \leq |E|.
\]
\end{restatable}

Combining Lemmas~\ref{lemma:appendix:normalizing step1}~and~\ref{lemma:appendix:normalizing step2} immediately leads to our Theorem~\ref{theorem:abcbd-singleuse}.

\subsection{Basic Definitions}

\label{appendix:normalizing:basic definitions}

We first recall and introduce required definitions.

\paragraph{REWB.}

Let $\Sigma$ be a finite alphabet and $\mathcal{V}$ be a finite set of variables.
The set of regular expressions with backreferences (\rewb) extends standard regular expressions (\regex) and is defined by the following grammar:
\[
\begin{array}{lcl}
E & ::= & \epsilon \mid \sigma \mid \emptyset \\
  & \mid & E \cdot E \mid E + E \mid (E)^* \\
  & \mid & [ E ]_x \mid \readv{x}
\end{array}
\]
where $\sigma \in \Sigma$ and $x \in \mathcal{V}$.

\noindent{}\textbf{Notation Remark:} In this appendix, for ease of readability
\begin{itemize}
\item We sometimes write $[\,E\,]_x$ instead of $(\,E\,)_x$ to distinguish them from parentheses.
\item We sometimes simply write $E_1 E_2$ instead of $E_1 \cdot E_2$ by omitting the symbol $\cdot$.
\end{itemize}

\noindent{}\textbf{Term: Binding and Reference.}
We call each term of the form $[\cdots]_x$ a \emph{binding} and one of the form $\readv{x}$ a \emph{reference}.

\paragraph{Size of Expression: $| E |$.}

\newcommand{\size}{\mathit{size}}

Here we define the size of expressions $E$ as the total number of nodes of the abstract syntax tree of $E$.
Formally, $|E|$ is defined by the size function $\size(E)$ as follows:
\[
\begin{array}{lcl}
\size(t) & :=&  1, \quad \text{ (for $t \in \set{ \epsilon, \sigma, \emptyset, \readv{x}}$)}, \\
\size(E_1 \diamond E_2) & := & \size(E_1) + \size(E_2) + 1, \quad \text{ (for $\diamond \in \set{\ \cdot \ , \ + \ }$)} \\
\size(\,(E)^*\,) & := & \size(E) + 1, \\
\size(\,[E]_x\,) & := & \size(E) + 1.
\end{array}
\]

\paragraph{REGEX.}
We define the fragment \regex{} by the following grammar:
\[
E_R ::= \epsilon \mid \sigma \mid \emptyset \mid E_R E_R \mid E_R + E_R \mid (E_R)^*.
\]

\paragraph{Semantics and Language.}

Next, we define the semantics of \rewb{} via the semantic function $\sem{E}_{\lambda}$.
For a \rewb{} $E$ and a valuation $\lambda$ from $\mathcal{V}$ to $\Sigma^*$,
$\sem{E}_\lambda$ returns a set of pairs $\tuple{w, \lambda'}$, 
each consisting of a matched string and an updated valuation:
\[
\begin{array}{lcl}
\sem{ \epsilon }_{\lambda} & := & \set{\,\tuple{\epsilon, \lambda}\,}, \\
\sem{ \sigma }_{\lambda} & := & \set{\,\tuple{\sigma, \lambda}\,}, \\[3pt]
\sem{ \emptyset }_{\lambda} & := & \emptyset, \\[3pt]
\sem{ E_1 E_2 }_{\lambda} & := & \set{ \tuple{w_1 w_2, \nu} : \tuple{w_1, \mu} \in \sem{E_1}_{\lambda}, \tuple{w_2, \nu} \in \sem{ E_2 }_{\mu} }, \\[3pt]
\sem{ E_1 + E_2 }_{\lambda} & := & \sem{ E_1 }_{\lambda} \cup \sem{ E_2 }_{\lambda}, \\[3pt]
\sem{ E^* }_{\lambda} & := & \bigcup_{k\ge 0} \sem{E^k}_\lambda, \qquad \text{where}\ E^0 := \epsilon,\; E^{k+1}:=E^k E, \\[3pt]
\sem{ [E]_x }_{\lambda} & := & \set{ \tuple{w, \mu[x := w]} : \tuple{w, \mu} \in \sem{E}_{\lambda} }, \\[3pt]
\sem{ \readv{x} }_{\lambda} & := & \set{ \tuple{\lambda(x), \lambda} }.
\end{array}
\]
This function is well-defined by induction on the height of Kleene stars.

Let $\iota$ be the initialized valuation that maps every variable to the empty string $\epsilon$: i.e., $\iota(v) = \epsilon$ for all $v \in \mathcal{V}$.
We simply write $\sem{E}$ to denote $\sem{E}_{\iota}$.
We define the language of a \rewb{} $E$, $\lang{E}$, as follows:
\[
\lang{E} := \set{ w : \tuple{w, \lambda} \in \sem{E} }.
\]

\paragraph{Pre-processing: $\emptyset$-free \rewb.}

As a preprocessing step,
we eliminate $\emptyset$ from a given \rewb{} by repeatedly applying  the following language-preserving rewrites:
\[
\begin{array}{l}
E + \emptyset \Mapsto E,\ \emptyset + E \Mapsto E, \qquad
E\,\emptyset \Mapsto \emptyset,\ \emptyset\,E \Mapsto \emptyset, \qquad
\emptyset^* \Mapsto \epsilon,\ 
[\emptyset]_x \Mapsto \emptyset.
\end{array}
\]
Observe that each rule strictly decreases the size (the number of symbols) of the expression.
Therefore, this rewriting process always terminates.

We say that a \rewb{} $E$ is $\emptyset$-free
if $E$ contains no occurrence of the constant $\emptyset$.
More formally, such expressions are generated by the following grammar:
\[
E_{-\emptyset} ::= \epsilon \mid \sigma \mid E_{-\emptyset}\ E_{-\emptyset} \mid E_{-\emptyset} + E_{-\emptyset} \mid \bigl(E_{-\emptyset}\bigr)^*  \mid \bigl[ E_{-\emptyset} \bigr]_x \mid \readv{x}
\]

\begin{proposition}
For every \rewb{} $E$, there exists a \rewb{} $G$
such that $\lang{E} = \lang{G}$
and either $G$ is $\emptyset$-free or $G = \emptyset$.
\end{proposition}

\noindent{}\textbf{Remark:} Hereafter, we assume that the input \rewb{} is $\emptyset$-free and not equal to $\emptyset$.

\paragraph{1-use REWB.}

Let $E$ be a ($\emptyset$-free) \rewb{}.
We call $E$ \emph{1-use} if, in every possible computation of $E$,
at most one backreference operation $\readv{x}$ (for some variable $x$) is performed.
For example,
\begin{itemize}
\item $[(a + b)^*]_x \# \readv{x}$ is a 1-use \rewb.
\item $[\Sigma^*]_x \readv{x} \readv{x}$ is not a 1-use \rewb{}.
\item $[a^*]_x [b^*]_y \# (\readv{x} + \readv{y})$ is a 1-use \rewb.
\item $[\Sigma^*]_x [\Sigma^*_y] \readv{x} \readv{y}$ is not a 1-use \rewb{}.
\item $[\Sigma^*]_x (\readv{x} + a)(b + \readv{x})$ is not a 1-use \rewb{}.
\item $[\Sigma^*]_x (a \readv{x} b)^*$ is not a 1-use \rewb{}.
\end{itemize}

To formally define the notion of 1-use,
we introduce the notion of \emph{template strings}.

\paragraph{Template String.}

For a given \rewb{} expression $E$,
we can see $E$ as a language generator over the alphabet
$\Gamma := \Sigma \cup \set{\,[,\, ]_x, \readv{x} : x \in \mathcal{V} }$ by treating special symbols $[$, $]_x$, and $\readv{x}$ as usual letters.
We write $\template(E) \subseteq \Gamma^*$ to denote the set of template strings generated by $E$.
Every element of $\template(E)$ is called a \emph{template string}.

$\template(E)$ is formally defined as follows:
We define the set of template strings $T(E)\subseteq \Gamma^*$ by structural recursion:
\[
\begin{array}{lcl}
\template(\epsilon) & := & \{\epsilon\}, \\
\template(\sigma) & := & \{\sigma\}, \\
\template(E_1 E_2) & := & T(E_1) T(E_2), \\
\template(E_1 + E_2) & := & T(E_1) \cup T(E_2), \\
\template(E^*) & := & (\template(E))^*, \\
\template([E]_x) & := & \set{ [\,w \,]_x : w \in \template(E) }, \\
\template(\readv{x}) & := & \{ \readv{x} \}.
\end{array}
\]

For our example expression $E_4$, $\template(E_4)$ contains the following template strings for $\Sigma = \set{ a, b }$:
\[
\template(E_4) \ni
[ bab ]_x\,a\readv{x}b\,a\readv{x}b,~~
[ abba ]_x\,a\readv{x}b\,a\readv{x}b\,a\readv{x}b,~~
\ldots
\]

By the definition, the following properties trivially hold.
\begin{proposition}
Let $E$ be an $\emptyset$-free expression.
If $E$ contains $\readv{x}$, then there exists $s \in \template(E)$ such that $s$ contains $\readv{x}$.
Similarly, if $E$ has $[\cdots]_x$, then there exists $s \in \template(E)$ such that $s$ contains $[\cdots]_x$.
\end{proposition}

\noindent{}\textbf{Remark.} Each template string can be evaluated by using $\sem{\bullet}$. For example,
\[
\sem{ a [bc]_x d } = \set{ \tuple{abcd, x \mapsto bc } }.
\]

From the definition and the remark, the following property of $\template(\bullet)$ holds.
\begin{proposition}
Let $E$ be a \rewb{}. Then, the following holds:
\[
\sem{E} =\sem{\template(E)}
\]
where $\sem{\template(E)} = \bigcup_{s \in \template(E)} \sem{s}$.
\end{proposition}

For a template string $s$, we write $\#_{\readv{\mathcal{V}}}(s)$ to denote the total number of occurrences of reference symbols $\readv{z}$ for $z \in \mathcal{V}$.
Let $\#_{\readv{\mathcal{V}}}(E) := \max\set{ \#_{\readv{\mathcal{V}}}(s) : s \in \template(E) }$
where $\max \emptyset = 0$.
Now we define 1-use \rewb{} as follows:
\begin{center}
\rewb{} $E$ is 1-use 
if and only if 
$\#_{\readv{\mathcal{V}}}(E) \leq 1$.
\end{center}

\subsection{Normalization Overview}

\label{appendix:normalizing:overview}

We aim to normalize a given 1-use \rewb{} $E$ into the following form:
\[
(A_1 \, [B_1]_{x} \, C_1 \, \readv{x} \, D_1) +
(A_2 \, [B_2]_{x} \, C_2 \, \readv{x} \, D_2) + \cdots +
(A_{\exprN} \, [B_{\exprN}]_{x}\,  C_{\exprN} \readv{x} \, D_{\exprN}) + E'
\]
where $A_i$, $B_i$, $C_i$, $D_i$, and $E'$ are \regex{es}.
Here $E'$ represents for the 0-use (pure \regex) part.

The crucial restriction is that, in each summand,
the only binding is $[B_i]_x$ and the only reference is $\readv{x}$, and both occur at the top level.
In particular, the subexpressions $A_i, B_i, C_i,$ and $D_i$ are pure \regex{es} (i.e., they contain neither bindings nor references).

To obtain our goal form, we need to solve the following tasks:
\begin{enumerate}
\item Lift nested reference subterms $\readv{x}$ to the top-level in Section~\ref{appendix:lift references}.
\item Then, lift nested binding subterms $[ \cdots ]_x$ to the top-level in Section~\ref{appendix:normalizing:lift binding}.
\end{enumerate}

\subsection{Lifting References $\readv{x}$ and Lemma~\ref{lemma:appendix:normalizing step1}}

\label{appendix:lift references}

On lifting $\readv{x}$ to the top-level,
we first consider two patterns $( \cdots \readv{x} \cdots)^*$ and $[ \cdots \readv{x} \cdots]_{y}$.

\subsubsection{Handling Simple Patterns}

\paragraph{Pattern ``references in Kleene stars'': $( \cdots \readv{x} \cdots)^*$.}

Since our input $E$ is $1$-use, any references $\readv{x}$ never appear inside Kleene stars.
For instance, an expression $[\Sigma^*]_x (a \readv{x} b)^*$ uses the variable $x$ finitely but unboundedly many times.

\paragraph{Pattern ``references in bindings'': $[ \cdots \readv{x} \cdots]_{y}$.}

We rewrite and simplify a given 1-use \rewb{} $E$ so that any references $\readv{x}$ never appear inside bindings $[\cdots]_y$.
As an example, consider an expression $[\Sigma^*]_x [a\,\readv{x}\,b]_y$.
Since we use the variable $x$ in the subterm $a\,\readv{x}\,b$,
any subterms executed after this part are not allowed to use a backreference $\readv{y}$ for any variable $y$.
In other words, we can remove all bindings executed after $a\,\readv{x}\,b$.
In fact, for this expression,
we can safely rewrite it as follows:
\[
[\Sigma^*]_x [a \,\readv{x}\, b]_y \leadsto [\Sigma^*]_x a\,\readv{x}\,b
\]
by removing $[\_]_y$.

\newcommand{\unnest}{\mathit{unnest}}

This rewriting is accomplished by the following translation function:
\[
\begin{array}{lcl}
\unnest(\,[E]_x\,) & := & \begin{cases}
\unnest(E) & \text{If some references appear in $E$}, \\
[E]_x & \text{Otherwise}, \\
\end{cases}
\\[5pt]
\unnest( E F ) & := & \unnest(E) \unnest(F), \\
\unnest( E + F ) & := & \unnest(E) + \unnest(F), \\
\unnest( E^*) & := & E^*, \\
\unnest( t ) & := & t \quad (t \in \set{ \epsilon, \sigma} \cup \set{ \readv{z} : z \in \mathcal{V}})
\end{array}
\]

The following properties clearly hold by the definition of $\unnest$.
\begin{proposition}
\[
\begin{array}{ll}
(1) & \lang{E} = \lang{\unnest(E)}, \\
(2) & |\unnest(E)| \leq |E|.
\end{array}
\]
\end{proposition}

After applying the function $\unnest$, it suffices to consider 1-use \rewb{s} generated by the following restricted grammar $\mathcal{F}$ of \rewb:
\[
\mathcal{F} ::= \epsilon \mid \sigma \mid \mathcal{F} \, \mathcal{F} \mid \mathcal{F} + \mathcal{F} \mid (F)^* \mid [ F ]_x \mid \readv{x}
\]
where $F$ is restricted to expressions that are \emph{reference}-free \rewb{s}.

\subsubsection{Tagged References and Lifting References Covered by Unions}

After applying $\unnest$, our lifting phase for references is not yet complete because, in general, references may not appear at the top-level but rather be nested within unions.
For instance, consider the following expression:
\[
E := a + \readv{x} + (b\,\readv{y} c) + ((d^* \readv{x}) + (\readv{x} e^*)) f^*.
\]
In this example, the third and fourth reference occurrences are not at the top-level.

To lift such references,
as a first step, we annotate a designated reference occurrence with a tag $\rho$ as $\readv{x}^{\rho}$.
This tag serves purely as a syntactic marker to distinguish this occurrence from other references.
Applying this to our example $E$,
we use $\rho$ to generate the following set of tagged expressions:
\[
\begin{array}{ll}
E^\rho_1 :& a + \readv{x}^{\rho} + (b\,\readv{y} c) + ((d^* \readv{x}) + (\readv{x} e^*)) f^*, \\
E^\rho_2 :& a + \readv{x} + (b\,\readv{y}^{\rho} c) + ((d^* \readv{x}) + (\readv{x} e^*)) f^*, \\
E^\rho_3 :& a + \readv{x} + (b\,\readv{y} c) + ((d^* \readv{x}^{\rho}) + (\readv{x} e^*)) f^*, \\
E^\rho_4 :& a + \readv{x} + (b\,\readv{y} c) + ((d^* \readv{x}) + (\readv{x}^{\rho} e^*)) f^*.
\end{array}
\]

\newcommand{\OccRef}{\mathit{Occ}_{\mathrm{ref}}}

We formally define the $\rho$-tagging,
denoted by $\Xi_\rho(E)$, as a function that maps an expression $E$ to a set of expressions where exactly one reference occurrence is tagged with $\rho$.
Let $E$ be an expression and let $\OccRef(E)$ be the list
of all \emph{occurrences} of reference symbols in $E$, ordered from left to right in the concrete syntax.
Write $\OccRef(E) = (o_1,o_2,\dots,o_N)$.
For each $i \in \set{1, 2, \ldots, N}$, define $E_i^{\rho}$ to be the expression obtained from $E$ by replacing
\emph{exactly} the occurrence $o_i$ of a reference symbol (say $\backslash x$) with its tagged form
$\backslash x^{\rho}$, leaving all other symbols unchanged.
Then, we define $\Xi_{\rho}(E)$ as follows:
\[
  \Xi_{\rho}(E) := \set{ E_1^{\rho}, E_2^{\rho}, \dots, E_N^{\rho} }
\]
Applying this function to our example expression $E$, we obtain the set $\Xi_{\rho}(E) = \set{ E^\rho_1, E^\rho_2, E^\rho_3, E^\rho_4 }$.

We naturally extend $\sem{\cdot}$ and $\template(\cdot)$ to accommodate the tag $\rho$ as follows:
\[
\begin{array}{lcl}
\sem{\, \readv{x}^{\rho}\,} & := & \sem{\,\readv{x}\,}, \\
\template(\,\readv{x}^{\rho}\,) & := & \set{\,\readv{x}^{\rho}\,}.
\end{array}
\]

\newcommand{\ZeroRef}{\mathcal{Z}_{\text{ref}}}
\newcommand{\RefRestr}{\restriction_{\text{ref}} \rho}

To decompose the set of tagged template strings, we define two notations:
\[
\begin{array}{lcl}
\ZeroRef(E) & := & \set{ w \in \template(E) : \text{$w$ does not contain any reference symbols }}, \\
\template(E) \RefRestr 
  & := & \set{ w \in \template(E) : \text{$w$ contains $\rho$} } = \template(E) \setminus \ZeroRef(E)
\end{array}
\]

Since $\ZeroRef(E)$ and $\template(E) \restriction_{\text{ref}} \rho$ are disjoint sets of $\template(E)$, the following basic property clearly holds.
\begin{proposition}
\label{appendix:tag:decompose}
Let $E$ be an expression and 
assume $\Xi_\rho(E) = E^{\rho}_1, E^{\rho}_2, \ldots, E^{\rho}_N$ where $N := \OccRef(E)$ is the total number of occurrences of references in $E$.
Then, the following holds:
\[
\sem{\,\template(E)\,} = 
\sem{\,\ZeroRef(E)\,} \cup
\bigcup^{N}_{i = 1} \bigl(\sem{\,\template(E_i) \RefRestr\, }\,\bigr)
\]
\end{proposition}

\paragraph{Computing $\mathcal{Z}$.}

Let $E$ an expression of the grammar $\mathcal{F}$.
We define the function $\noref(E)$ that constructs an expression computing $\ZeroRef(E)$.

\[
\begin{array}{lcl}
\noref(\epsilon) & := & \epsilon, \\
\noref(\sigma) & := & \sigma, \\
\noref(E_1 E_2) & := & \noref(E_1) \noref(E_2), \\
\noref(E_1 + E_2) & := & \noref(E_1) + \noref(E_2), \\
\noref(E^*) & := & E^*, \\
\noref([E]_x) & := & \noref(E), \\
\noref(\readv{x}) & := & \emptyset,
\end{array}
\]

By the definition, the following basic properties clearly hold.
\begin{lemma}
\[
\begin{array}{ll}
(1) & \ZeroRef(E) = \template( \noref(E) ), \\
(2) & |\noref(E)| \leq |E|.
\end{array}
\]
\end{lemma}

\paragraph{Computing $\template(E) \RefRestr$.}

\newcommand{\appear}{\mathit{appear}}

Let $E$ be a \rewb{} with a tagged reference.
We define the function $\appear_\rho(E)$ that computes an expression computing $\template(E) \RefRestr$:
 \[
\begin{array}{lcl}
\appear_\rho(\readv{x}^\rho) & := & \readv{x}^{\rho}, \\
\appear_\rho(E_1 E_2) & := & 
\begin{cases}
    \appear_{\rho}(E_1)\ E_2 & \text{If $E_1$ contains $\rho$}, \\
    E_1\ \appear_{\rho}(E_2) & \text{Else if $E_2$ contains $\rho$}, \\
    \emptyset & \text{Otherwise},
\end{cases} \\
\appear_\rho(E_1 + E_2) & := &
\begin{cases}
    \appear_{\rho}(E_1) & \text{If $E_1$ contains $\rho$}, \\
    \appear_{\rho}(E_2) & \text{Else if $E_2$ contains $\rho$}, \\
    \emptyset & \text{Otherwise},
\end{cases} \\
\appear_\rho( t ) & := & \emptyset, \qquad
\bigl(t = \epsilon, \sigma, \readv{x}, E^*, [E]_x\bigr).
\end{array}
\]
\textbf{Remark:} In our setting,
since $E^*$ and $[E]_x$ never contain $\rho$, we set $\appear_\rho(t) = \emptyset$ for such arguments.
By the definition, the following basic property clearly holds.
\begin{lemma}
\[
\template(E) \RefRestr = \template(\appear_\rho(E)).
\]
\end{lemma}
Also, the following property for the form and size holds.
\begin{proposition}
\[
\appear_\rho(E) = \alpha\,\readv{x}^{\rho}\,\beta
\]
where $\alpha$ and $\beta$ are reference-free \rewb{s}.
Furthermore, $|\appear_{\rho}(E)| \leq |E|$ holds.
In particular, the following holds:
\[
|\alpha| \leq |E|, \qquad |\beta| \leq |E|.
\]
\end{proposition}

\paragraph{Summary of This Section.}

From a given 1-use \rewb{} $E$, we can construct the following:
\[
(\alpha_1\,\readv{x_1}\,\beta_1) + (\alpha_2\,\readv{x_2}\,\beta_2) + \cdots + (\alpha_{K'}\,\readv{x_{K'}}\,\beta_{K'}) + \noref(E)
\]
where $\alpha_i$ and $\beta_i$ are reference-free \rewb{s}.
Also, the following holds:
\[
K' \leq |E|,\ 
|\alpha_i| = O(|E|),\ 
|\beta_j| = O(|E|),\ 
|\noref(E)| = O(|E|).
\]

This clearly corresponds our Lemma~\ref{lemma:appendix:normalizing step1}.
\NormalizingFirstStep*

\subsection{Removing Unnecessary Variables and only Keeping $x$}
\newcommand{\keep}{\mathit{keep}}    

\label{appendix:restricted grammar}

For each subexpression $\alpha_i\ \readv{x_i}\ \beta_i$, we perform the following two translations:
\begin{enumerate}
\item We replace all bindings of the form $[F]_y$ (for all variables $y$) with $F$ in $\beta_i$.
\begin{itemize}
    \item This replacing does not alter the expression's behaviour because, on 1-use \rewb{s}, any use references of the form $\readv{x}$ do not appear in $\beta_i$.
    \item To formally state this rewrite, we use the following function and rewrite $\beta_i \leadsto \elim(\beta_i)$.
\[
\begin{array}{lcl}
\elim(t) & \defeq & t \quad (t \in \set{ \epsilon, \sigma }), \\
\elim(\: E_1 E_2 \:) & \defeq & \elim(\, E_1 \,) \elim(\, E_2 \,), \\
\elim(\: E_1 + E_2 \:) & \defeq & \elim(\,E_1\,) + \elim(\,E_2\,), \\
\elim( (E)^* ) & \defeq & \elim(E)^*, \\
\elim(\:[E]_y\:) & \defeq &  \elim(\,E\,)
\end{array}
\]    
\end{itemize}
\item We replace all subterms of the form $[F]_y$ ($x_i \neq y$) with $F$ in $\alpha_i$.
\begin{itemize}
    \item This replacing does not alter the expression's behaviour because we only use the variable $x_i$ and thus modifications on $y$ are irrelevant.
    \item To formally state this rewrite, we use the following function and rewrite $\alpha_i \leadsto \keep_{x_i}(\alpha_i)$.
\[
\begin{array}{lcl}
\keep_{x}(t) & \defeq & t \quad (t \in \set{ \epsilon, \sigma }), \\[3pt]
\keep_{x}(\: E_1 E_2 \:) & \defeq & \keep_{x}(\,E_1\,) \keep_{x}(\, E_2 \,), \\[3pt]
\keep_{x}(\: E_1 + E_2 \:) & \defeq & \keep_{x}(\,E_1\,) + \keep_{x}(\,E_2\,), \\[3pt]
\keep_{x}( (E)^* ) & \defeq & \keep_{x}(E)^*, \\
\keep[(\:[E]_y\:) & \defeq &  \begin{cases}
    [\,\keep_x(E)\,]_x & \text{If $y = x$}, \\
    \keep_x(E) & \text{Otherwise}.
\end{cases}
\end{array}
\]        
\end{itemize}
\end{enumerate}

\paragraph{Grammar for this Phase.}

We first decompose $E$ by Proposition~\ref{appendix:tag:decompose} as follows:
\[
E \Mapsto \noref(E), \appear_{\rho}(E_1), \appear_{\rho}(E_2), \ldots, \appear_{\rho}(E_N),
\]
and for each $\appear_\rho(E_i) = \alpha_i\ \readv{x_i}\ \beta$, we rewrite it as follows:
\[
\alpha_i\ \readv{x_i}\ \beta \Mapsto \keep_{x_i}(\alpha_i)\ \readv{x_i}\ \elim(\beta).
\]

Each expression of the form $\keep_x(\alpha)$ is generated by the following restricted grammar $\mathcal{G}$:
\[
\mathcal{G} ::= \epsilon \mid \sigma \mid \mathcal{G} \mathcal{G} \mid \mathcal{G} + \mathcal{G} \mid \mathcal{G}^* \mid [ \,\mathcal{G}\, ]_x.
\]

\noindent\textbf{Remark:} If $[\cdots]_y$ appears in $\keep_x(\alpha)$, then $y = x$ must hold.

Hereafter, we translate an expression $\alpha$ of $\mathcal{G}$ into the following language-equivalent form:
\[
\alpha \Mapsto A\:[B]_x\:C
\]
where $A$, $B$, and $C$ are pure REGEX.

\subsection{Lift Bindings to the Top level and Lemma~\ref{lemma:appendix:normalizing step2}}

\label{appendix:normalizing:lift binding}


In moving all occurrences of binding to the top-level,
we handle two patterns: $[ \cdots [F]_x \cdots]_x$ and $( \cdots [F]_x \cdots)^*$.


\subsubsection{%
  \texorpdfstring{Pattern $[ \cdots [F]_x \cdots]_x$}%
                 {Pattern [... [F]_x ...]_x}
}

For instance, consider the expression $[ a \, [\Sigma^*]_x \, b]_x$.
Since only the last modification to the variable $x$ is effective, we can rewrite the expression to $[ a \, \Sigma^* \, b]_x$ without changing its behaviour.

\[
\begin{array}{lcl}
\unnest_2(\,[E]_x\,) & := & [\,\elim(E)\,]_x, \\
\unnest_2( E F ) & := & \unnest_2(E) \unnest_2(F), \\
\unnest_2( E + F ) & := & \unnest_2(E) + \unnest_2(F), \\
\unnest_2( E^*) & := & (\unnest_2(E))^*, \\
\unnest_2( t ) & := & t \quad (t \in \set{ \epsilon, \sigma})
\end{array}
\]

By the definition, the following basic properties trivially hold.
\begin{proposition}
The following properties hold for $\unnest_2$.
\[
\begin{array}{ll}
(1) & \lang{E} = \lang{\unnest_2(E)}, \\
(2) & |\unnest_2(E)| \leq |E|.
\end{array}
\]
\end{proposition}

After this translation, we can further restrict our grammar to the following form:
\[
\mathcal{H} ::= \epsilon \mid \sigma \mid \mathcal{H} \mathcal{H} \mid \mathcal{H} + \mathcal{H} \mid \mathcal{H}^* \mid [ \mathbf{REGEX} ]_x.
\]

\subsubsection{%
  \texorpdfstring{Pattern $(\cdots [F]_x \cdots)^*$}%
                 {Pattern (... [F]_x ...)*}
}

Handling this pattern is the most elaborate part in our normalization procedure.

\subsubsection*{Tagged Binding.}

We annotate a designated binding occurrence with a tag $\rho$ by writing it as $[E]_x^{\rho}$.
The tag serves only as a syntactic marker to distinguish this occurrence from other bindings.

For instance, let us consider the following expression:
\[
E := (a + [b^*]_x + (c\,[d^*]_x e)^*\,)^* .
\]
To handle each binding subterm separately, we use $\rho$ as follows and generate two tagged expressions from the above one:
\[
\begin{array}{ll}
E^{\rho}_1 : & (a + [b^*]^\rho_x + (c\,[d^*]_x e)^*\,)^*\,, \\
E^{\rho}_2 : & (a + [b^*]_x + (c\,[d^*]^{\rho}_x\,e)^*\,)^*\,.
\end{array}
\]

\newcommand{\OccBind}{\mathit{Occ}_{\text{bind}}}

Similar to tagging references, we write $\OccBind(E)$ to denote the total number of bind occurrences in $E$ and $\Xi'(E)$ to denote the set of expressions where exactly one binding occurrence is tagged with $\rho$.

We naturally extend $\sem{\cdot}$, $\template(\cdot)$, and $\elim$ to accommodate the tag $\rho$ as follows:
\[
\begin{array}{lcl}
\sem{\, [E]^{\rho}_x\,} & := & \sem{\,[E]_x\,}, \\
\template(\,[E]^{\rho}_x\,) & := & [ \, \template(E) \, ]^{\rho}_x = \set{\ [ \, w \, ]^{\rho}_x : w \in \template(E) }, \\
\elim(\,[E]^{\rho}_x\,) & := & \elim(\,E\,).
\end{array}
\]

\newcommand{\BindRestr}{\restriction_{\text{bind}} \rho}

For a tagged expression $E$, we write $\template(E) \restriction \rho$ to denote the following set:
\[
\template(E) \BindRestr := \set{ w \in \template(E) :
  \text{$w$ contains a substring of the form $[\cdots]^\rho_x$ at least once }
\ }.
\]

\subsubsection*{Decomposition by the last tag}

For a template string $s$,
we write $\last_\rho(s)$ if both the following conditions hold:
\begin{itemize}
    \item $s$ contain some binding occurrences.
    \item Furthermore, its last binding occurrence is of the form $[ \cdots ]^{\rho}_x$.
\end{itemize}

\newcommand{\ZeroBind}{\mathcal{Z}^{\rho}_{\text{bind}}}

Now we write $\ZeroBind(\alpha)$ to denote the set of template strings generated by $\alpha$ that contains no binding symbols:
\[
\ZeroBind(\alpha) := \template(E) \setminus (\template(E) \BindRestr) = \set{ w \in \template(\alpha) : \text{$w$ does not contain any binding occurrnces} }.
\]

Now we define the following notation to denote the set of template strings generated by $E$ whose last modifier is tagged by $\rho$:
\[
\template(E) \Uparrow \rho := 
\set{ s \in \template(E) \BindRestr  : \last_\rho(s) }.
\]

As with Proposition~\ref{appendix:tag:decompose}, the following decomposition clearly holds.
\begin{proposition}
Let $E$ be an expression of the grammar $\mathcal{H}$.
Assume $\Xi'(E) = E^{\rho}_1, E^{\rho}_2, \ldots, E^{\rho}_N$ where $N := \OccBind(E)$ is the total number of occurrences of binding in $E$.
Then, the following holds:
\[
\sem{\,\template(E)\,} = \sem{\,\ZeroBind(E)\,} \cup \bigcup^{N}_{i = 1} \bigl(\sem{\,\template(E_i) \Uparrow \rho\,}\bigr)
\]
\end{proposition}

Hereafter, for each $E_i$, we compute an expression of the form $A_i\ [B_i]^{\rho}_x\ C_i$ such that
\begin{itemize}
    \item $\sem{ \template(E_i) \Uparrow \rho } = \sem{ A_i\  [B_i]^{\rho}_x\ C_i }$; and
    \item $A_i, B_i, C_i$ are pure \regex{es}.
\end{itemize}

\paragraph{Computing $\mathcal{Z}$.}

\newcommand{\nobind}{\mathit{nobind}}

Let $E$ be an expression of $\mathcal{H}$.
We define the function $\nobind(E)$ that computes an expression computing $\mathcal{Z}(E)$.

\[
\begin{array}{lcl}
\nobind(\epsilon) & \defeq & \epsilon, \\
\nobind(\sigma) & \defeq & \sigma, \\
\nobind(\: E_1 E_2 \:) & \defeq & \nobind(\, E_1 \,) \nobind(\, E_2 \,), \\
\nobind(\: E_1 + E_2 \:) & \defeq & \nobind(\,E_1\,) + \nobind(\,E_2\,), \\
\nobind( (E)^* ) & \defeq & \nobind(E)^*, \\
\nobind(\:[E]_x\:) & \defeq &  \emptyset, \\
\nobind(\:[E]^{\rho}_x\:) & \defeq &  \emptyset,
\end{array} 
\]

By the definition, the following properties clearly hold for $\nobind$.
\begin{proposition}
\[
\begin{array}{ll}
(1) & \sem{\ZeroBind(E)} = \sem{\nobind(E)}, \\
(2) & |\nobind(E)| \leq |E| .
\end{array}
\]
\end{proposition}

\subsubsection*{Strategy for Representing the Part $\sem{\template(E_i) \Uparrow \rho}$}

To explain our idea, let us consider the following expression:
\[
(ab + c [\Sigma^*]_x d + e [\Sigma^*]_x f)^*
\]
where $\Sigma = \set{a, b, c, d}$.
By tagging, we obtain the following two expressions:
\[
E^\rho_1:\ (ab + c [\Sigma^*]^{\rho}_x d + e [\Sigma^*]_x f)^*, \qquad
E^\rho_2:\ (ab + c [\Sigma^*]_x d + e [\Sigma^*]^{\rho}_x f)^*.
\]
We here consider $E^\rho_1$ and write $E$ to denote the subterm $(ab + c [\Sigma^*]^{\rho}_x d + e [\Sigma^*]_x f)$.

Our key idea is to unfold Kleene stars $E^*$ so that the following holds:
\[
\begin{array}{lcl}
\template(\,E^*\,) \Uparrow \rho & = &
\text{repeating $E$} \\
& \cdot & \text{running $E$ guaranteed to visit   the $[\cdots]^{\rho}_x$ part exactly once} \\
& \cdot & \text{repeating $E$ while never visiting any $[\cdots]_x$ subterms}~(\text{corresponds to}~\nobind(E)^*).
\end{array}
\]
For our example, we unfold $E^*$ as follows
\[
E^* \ \cdot \ \uline{c[\Sigma^*]^{\rho}_x d}\ \cdot \ \uuline{(ab)^*}\ .
\]
Using the single-underlined and double-underlined parts,
we enforce that $\rho$ appears as the last occurrence of binding across the entire computation of $E^*$.

Furthermore, in this form, we can remove $[\cdots]_x$ from the $E^*$ because the single-underlined part overwrites the variable $x$.
Namely, it suffices to consider the following expression:
\[
E' :=
\elim(E)^* \ \cdot \ {c[\Sigma^*]^{\rho}_x d}\ \cdot \ {(ab)^*}
~~=~~
\bigl(ab + c\,\Sigma^*\,d + e\,\Sigma^*\,f\bigr)^*
\:c\,[\Sigma^*]^{\rho}_x\,d\ (ab)^*~.
\]
The transformed expression fits into our goal form $A [B]_x C$.

To compare $\template(E^*)$ and $\template(E')$,
we consider the following template string generated by $E^*$:
\[
w~~:=~~ab\:c\,[ac]^{\rho}_x\,d\:e\,[aaaa]_x\,f\:c\,[a]^{\rho}_x\,d\:ababab \in \template(E^*).
\]
Since this string contains $\rho$ as the last binding occurrence, it is represented in our transformed version $E'$ by discarding binding symbols before the last $\rho$ as follows:
\[
w'~~:=~~ab\:c\,ac\,d\:e\,aaaa\,f\:c\,[a]^{\rho}_x\,d\:ababab \in \template(E').
\]
Furthermore, $\sem{w} = \sem{w'}$ clearly holds.

\subsubsection*{Formalizing Our Translation}

To compute $\sem{ \template(E) } \Uparrow \rho$ on the basis of the above our idea, we introduce the following function $\last_\rho$ so that $\sem{ \template(E) } \Uparrow \rho = \sem{ \last_\rho(E)}$:
\[
\begin{array}{lcl}
\last_\rho(\epsilon) & := & \emptyset, \\
\last_\rho(\sigma) & := & \emptyset, \\
\last_\rho(E_1 E_2) & := & 
\begin{cases}
    \last_{\rho}(E_1)\ \nobind(E_2) & \text{If $E_1$ contains $\rho$}, \\
    \elim(E_1)\ \last_{\rho}(E_2) & \text{Otherwise (since $\rho$ is unique and in $E$, it must be in $E_2$)},
\end{cases} \\[15pt]
\last_\rho(E_1 + E_2) & := &
\begin{cases}
    \last_{\rho}(E_1) & \text{If $E_1$ contains $\rho$}, \\
    \last_{\rho}(E_2) & \text{Otherwise: $E_2$ contains $\rho$},
\end{cases} \\[15pt]
\last_\rho(\,(E)^*\,) & := & \elim(E)^*\ \last_\rho(E)\ \nobind(E)^*, \\
\last_\rho([E]_x) & := & \emptyset, \\
\last_\rho([E]^{\rho}_x) & := & [E]^{\rho}_x.
\end{array}
\]

\noindent{}\textbf{Remark:}
The argument expressions $E$ of $\last_\rho$ satisfy the following condition:
\begin{enumerate}
\item Since $E \in \mathcal{H}$,
the expression is \emph{reference-free}, i.e., it contains no reference/usage symbols of the form $\readv{z}$;
\item and moreover, every binding occurrence binds the same variable $x$, i.e., every binder is of the form $[\cdot]_x$ or $[\cdot]^\rho_x$
\item Also, since $E \in \mathcal{H}$, every binding occurrence is of the form $[R]_x$ (and, for tagged expressions, $[R]^\rho_x$) with $R$ a pure \regex. In particular, binders are never nested, i.e., no binder occurs within the scope of another binder.
\item Since $E$ is an element of the set of tagged expressions $\Xi'$, the tag $\rho$ is syntactically unique, i.e., there is exactly one tagged binding occurrence of the form $[\cdot]_x^{\rho}$ in the expression.
\end{enumerate}

Then, the following property holds for $\last_\rho$.
\begin{lemma}
Let $E$ be an expression that satisfies the conditions in the above remark. Then, the following holds:
\[
\sem{\template(E) \Uparrow \rho} = \sem{\last_\rho(E)}.
\]
\end{lemma}

In particular, the generated expression $\last_\rho(E)$ satisfies the following condition.
\begin{proposition}
Let $E$ be an expression that satisfies the conditions in the above remark.
Then, $\last_\rho(E)$ evaluates to an expression of the form $A\,[B]^{\rho}_x\,C$ where $A$, $B$, and $C$ are pure \regex.
Furthermore,
\[
|A| = O(|E|^2), \quad |B| \leq |E|, \quad |C| = O(|E|^2).
\]
\end{proposition}
\begin{proof}
By definition, $\elim$ and $\nobind$ erase all binding constructs and thus always produce \regex{es}.
Moreover, $\last_\rho$ never introduces new binders: untagged binders are discarded
($\last_\rho([F]_x) = \emptyset$),
while the unique tagged binder is preserved
($\last_\rho([F]^\rho_x)=[F]^\rho_x$).
Since binders are not nested and $\rho$ is syntactically unique, the expression $\last_\rho(E)$ contains exactly one binder,
namely the tagged one, and hence can be written as $A[B]^\rho_xC$ with $A,B,C$ pure \regex{es}.

To show the size bound, we use the following trivial properties, which are immediately shown by the definition of $\elim$ and $\nobind$:
\[
\text{for every expression $E$}, \quad |\elim(E)| \leq |E|, \quad |\nobind(E)| \leq |E|.
\]

To perform our induction, we consider a spine of subexpressions starting from $E$ as follows.
Consider the syntax tree of the input expression $E$ and follow the unique path from the root to the unique tagged binder.
This yields a chain of subexpressions
\[
E = E_0 \succ E_1 \succ \cdots \succ E_h = [B]^\rho_x,
\]
where $X \succ Y$ means that $Y$ is a proper subexpression of $X$.
Thus, $E_{i+1}$ is the unique immediate subexpression of $E_i$ that still contains the tag $\rho$. Here, the existence and uniqueness follow from the syntactic uniqueness of $\rho$.

\medskip
\noindent\textbf{Induction on the length of the spine.}
We need to consider four cases:
$E_i = E_{i+1} + F$, $E_i = E_{i+1}F$, $E_i = F E_{i+1}$, and $E_i = (E_{i+1})^*$.
For these cases,
each step from $E_{i+1}$ to $E_i$ modifies the prefix/suffix
by concatenating at most one extra piece whose size is $O(|E_i|)$:
either $\elim(F) \bullet$, or $\bullet \nobind(F)$, or $\elim(E_{i+1})^*\ \bullet\ \nobind(E_{i+1})^*$.
Consequently,
\[
|A_i| \le |A_{i+1}| + O(|E_i|),\qquad
|C_i| \le |C_{i+1}| + O(|E_i|).
\]
Unrolling from $i = h-1, h-2, \ldots, 0$,
we have
\[
|A_0| = O\!\left(\sum_{i=0}^{h-1} |E_i|\right),\qquad
|C_0| = O\!\left(\sum_{i=0}^{h-1} |E_i|\right).
\]

Finally, along a strict subexpression chain, we have
$|E_0| > |E_1| > \cdots > |E_h|\ge 1$,
hence
\[
\sum_{i=0}^{h-1} |E_i|
\;\le\; |E| + (|E|-1) + \cdots + 1
\;=\; O(|E|^2).
\]
Therefore $|A|=|A_0|=O(|E|^2)$ and $|C|=|C_0|=O(|E|^2)$, and we already noted $|B|\le |E|$.
\end{proof}

Now, summarizing all the objectives up to this point, we immediately prove Lemma~\ref{lemma:appendix:normalizing step2}.

\NormalizingSecondStep*

%% file: onerewb/hidden-constant.tex
\section{Analyzing Hidden Constants}\label{sec:hidden_constant}

In this section, we analyze the hidden constants of our theorems.
We mainly consider the constant of Theorem~\ref{thm:ABCBD}.
We first obtain a bound of $2^{O((|A|+|B|+|C|+|D|)^4)}$ by using the fact that any \regex{} $E$ admits a transition monoid of size $2^{O(|E|^2)}$.
We then refine this bound to $2^{O((|A|+|B|+|C|+|D|)^2)}$ through a more fine-grained analysis combined with additional data structures.

\subsection{First Bound for Theorem~\ref{thm:ABCBD}: $2^{O((|A|+|B|+|C|+|D|)^4)}$}

We first observe that the time complexity of all parts of our algorithm depends polynomially on the size of the transition monoids we use.

\paragraph{Interval evaluation by factorization forest.}
In our algorithm, we repeatedly perform the operations of evaluating $\Delta_E(w[i \btw j))$ for $E\in \set{ A,B,C,D }$.
By Lemma~\ref{lemma:constant-eval-by-ff}, it takes $O(|\transm{E}|)$ time, which only increases the time complexity polynomially in $|\transm{E}|$.
Thus, from now on, we focus on the remaining part and ignore the complexity of evaluating $\Delta_E(w[i,j))$.

\paragraph{Section~\ref{section:XYYZ problem}.}
The time complexity of this part depends polynomially on the size of the representation of the ultimately periodic set $\mathcal{S}=\mathcal{S}_X + (2\cdot \mathcal{S}_Y) + \mathcal{S}_Z$.
By a similar argument to that in the proof of Lemma~\ref{lem:wtw_is_ult_periodic}, we can state that the parameters $\mu$ and $\lambda$ of $\mathcal{S}_X$, $\mathcal{S}_Y$, and $\mathcal{S}_Z$ are bounded by $|\transm{X}|$, $|\transm{Y}|$, and $|\transm{Z}|$, respectively.
Thus, the parameters of $\mathcal{S}$ are bounded by a polynomial in $|\transm{X}|+|\transm{Y}|+|\transm{Z}|$.

\paragraph{Section~\ref{section:branchingABCBD-to-twosubproblems}.}
It is easy to see that this part affects the hidden constant only by $|\mathcal{K}_L|+|\mathcal{K}_R|$ in Algorithm~\ref{new:alg:overview_all}, which are bounded by $O(|\transm{C}|)$.

\paragraph{Section~\ref{section:solve:rightbenum}.}
This part contributes a factor polynomial in $|\Lambda|$, $|\textsc{Accum}|$, and $|\transm{D}|$
to the hidden constant in Algorithm~\ref{alg:overview_right}. 
Since $|\Lambda|\leq |\transm{D}|$, this part affects the hidden constant by a polynomial in $|\transm{C}|+|\transm{D}|$.

\paragraph{Section~\ref{sec:data_structures_right}.}
It is not difficult to see that each method in this part runs in time proportional to the depth of the factorization forest $\fforest$ and the size of $\transm{C}$.
Since the depth of $\fforest$ is $O(|\transm{C}|)$, each method runs in $O(|\transm{C}|^2)$ time.

\paragraph{Section~\ref{new:section:left-B-overview}.}
It is easy to see that this part affects the hidden constant only by $|\mathcal{D}|$ in Algorithm~\ref{alg:overview_left} and $|\mathcal{D}_1|+|\mathcal{D}_2|$ in Algorithm~\ref{alg:auxiliarytotarget_left}, which are bounded by $O(|\transm{C}|)$.

\paragraph{Section~\ref{sec:auxdatastructure}.}
By the same argument as in Section~\ref{sec:data_structures_right},  each method runs in $O(|\transm{C}|^2)$ time.

\paragraph{Section~\ref{section:near-transitions2}.}
The time complexity of this part depends polynomially on $|\transm{C}|$ and the size of the representation of the ultimately periodic set $\mathcal{S}_{\delta}=\mathcal{S}_A + \mathcal{S}_{C,\delta}$.
By a similar argument to that in the proof of Lemma~\ref{lem:wtw_is_ult_periodic}, we can state that the parameters $\mu$ and $\lambda$ of $\mathcal{S}_A$ and $\mathcal{S}_{C,\delta}$ are bounded by $|\transm{A}|$ and $|\transm{C}|$, respectively.
Thus, the parameters of $\mathcal{S}$ are bounded by a polynomial in $|\transm{A}|+|\transm{C}|$.

\paragraph{Overhead of normalization.}
We have established that the hidden constant depends polynomially on $|\transm{A}|+|\transm{B}|+|\transm{C}|+|\transm{D}|$, where \regex{es} $A,B,C,D$ are those defined in Sections~\ref{section:sanitize}~and~\ref{section:two_subproblems}.
We complete the analysis by seeing that each of those \regex{es}, which we refer to as $A',B',C',D'$ to distinguish them from the original \regex{es}, has length linear in the size of the original \regex{es} $A,B,C,D$.

The operations in Section~\ref{section:sanitize} multiply the length of each \regex{} by $O(\log |\Sigma|)$, which we consider to be an absolute constant\footnote{In fact, we can avoid the dependence on $\log |\Sigma|$ here by not binarizing the \regex{es}. This requires additional care in later sections, as the suffix tree is no longer binary; nevertheless, it is still possible to obtain the same bound for general $|\Sigma|$. For the sake of simplicity of presentation, we refrained from implementing this approach.}. 
Moreover, in Section~\ref{section:two_subproblems}, we defined
\[
    \lang{B'} := (\lang{B}\setminus \{\epsilon\})\cdot \left(\lang{C^{q^0_C \to \set{ q_C }}} \cap \lang{D^{q^0_D \to \set{ q_D }}}\right)
\]
for some $q_C\in Q_C$ and $q_D\in Q_D$, and
\[
    \lang{A'} := \lang{A}, \qquad \lang{C'} := \lang{C^{\set{ q_C }\to F_C}}, \qquad \lang{D'} := \lang{D^{\set{ q_D }\to F_D}}, 
\]
for some $q_C\in Q_C$ and $q_D\in Q_D$.

Clearly, the NFAs for $A'$, $C'$, and $D'$ have the same number of vertices as the NFAs for $A$, $C$, and $D$, respectively.
Thus, we can construct transition monoids $\transm{A'}$, $\transm{C'}$, and $\transm{D}'$ of size $2^{O(|A|^2)}$, $2^{O(|C|^2)}$, and $2^{O(|D|^2)}$, respectively.
For $B'$, straightforward construction yields a transition monoid $\transm{B'}$ of size $2^{O((|B|+|C||D|)^2)}$ because, using the standard construction, the NFA for $B'$ has size $O(|B|+|C||D|)$.
Therefore, $|\transm{A'}|+|\transm{B'}|+|\transm{C'}|+|\transm{D'}|\leq 2^{O(|A|^2)}+2^{O((|B|+|C||D|)^2)}+2^{O(|C|^2)}+2^{O(|D|^2)}\leq 2^{O((|A|+|B|+|C|+|D|)^4)}$.
Thus, the overall hidden constant is $2^{O((|A|+|B|+|C|+|D|)^4)}$.

\subsection{Refined Bound for Theorem~\ref{thm:ABCBD}: $2^{O((|A|+|B|+|C|+|D|)^2)}$}

The main obstacle to achieving a running time of $2^{O((|A|+|B|+|C|+|D|)^2)}$ is that the size of the NFA for $B'$ is quadratic.
If this size could be reduced to linear in the size of original \regex{es},
the desired bound would follow immediately; however, this appears to be difficult.

Instead, we observe that the algorithm accesses the transition monoid $\transm{B'}$ only in very restricted situations.
More precisely, the only operations we perform with respect to $B'$ are:
\begin{itemize}
    \item the construction of an ultimately periodic set $\mathcal{S}$ in Section~\ref{section:XYYZ problem} (note that $Y=B'$), and
    \item deciding whether $w[i\btw j)\in \lang{B'}$ (Algorithms~\ref{alg:overview_right} and~\ref{alg:overview_left}).
\end{itemize}

\paragraph{Bounding the Size of $\mathcal{S}$ (Section~\ref{section:XYYZ problem}).}
In Section~\ref{section:XYYZ problem}, we need to decide whether
$e \in \mathcal{S} = \mathcal{S}_X + (2\cdot \mathcal{S}_Y) + \mathcal{S}_Z$.
Previously, we bounded the size of the representation of $\mathcal{S}_Y$ by
$|\transm{B'}| \le 2^{O((|B|+|C||D|)^2)}$, which resulted in a large hidden constant.
However, we can show that the size of
\[
\mathcal{S}_Y = \{ j \colon \theta^j \in \lang{Y} \}
\]
is in fact only $2^{O(|B|+|C||D|)}$.

To see this, let
$N_{B'} = (Q_B, \Sigma, \Delta_B, \mathcal{E}_B, q^0_B, F_B)$
be an $\epsilon$-NFA recognizing $\lang{B'}$.
For each $j \ge 0$, let $R^j \subseteq Q_B$ denote the set of states reachable from $q^0_B$ by reading $\theta^j$.
Then, by a similar argument to that in the proof of Lemma~\ref{lem:wtw_is_ult_periodic}, $R^{j+\lambda}=R^{j}$ holds for all $j\geq \mu$,
where $\mu,\lambda\leq |Q_B|$. Consequently, the ultimately periodic set
\[
\mathcal{S}_Y = \{ j \colon R^j \cap F_B \neq \emptyset \}
\]
admits a representation of size $2^{|Q_B|} \le 2^{O(|B|+|C||D|)}$, and thus, $\mathcal{S}$ admits a representation of size $2^{O((|A|+|B|+|C|+|D|)^2)}$.

\paragraph{Reducing the Hidden Constant of Interval Evaluation.}

By Lemma~\ref{lemma:constant-eval-by-ff}, interval evaluation over $\transm{B'}$ runs in
$O(|\transm{B'}|) \le 2^{O((|B|+|C||D|)^2)}$ time, which again yields a large hidden constant.
However, this part is not a bottleneck with respect to $|w|$, since the algorithm performs interval evaluation over $\transm{B'}$ only $O(|w|\log |w|)$ times in total; $O(|w|)$ times in the branching ABCBD problem and $O(|w|\log |w|)$ times in the XYYZ problem when computing the set $\mathcal{S}_Y$.\todo{complexity of construction}

We therefore replace this procedure with a data structure that answers each query in
$O(\log |w|)$ time while having a significantly smaller hidden constant.
We employ a segment tree, which can be viewed as a factorization forest without idempotent nodes, but with height $O(\log |w|)$ independent of the monoid size.

Let $N_{B'} = (Q_B, \Sigma, \Delta_B, \mathcal{E}_B, q^0_B, F_B)$
be an $\epsilon$-NFA recognizing $\lang{B'}$.
Each node of the segment tree corresponding to an interval $[\ell,r)$ stores a binary matrix
$M \in 2^{Q_B \times Q_B}$, where $M_{q_1,q_2}=1$ if and only if there exists a path from $q_1$ to $q_2$ labeled by $w[\ell\btw r)$.
Matrix multiplication is defined via $(\lor,\land)$-convolution.

For any interval query $[i,j)$, by the standard argument, we can compute a binary matrix
$M^{i,j} \in 2^{Q_B \times Q_B}$ satisfying
$M^{i,j}_{q_1,q_2}=1$ if and only if there exists a path from $q_1$ to $q_2$ labeled by $w[i\btw j)$,
in time $O(|Q_B|^3 \log |w|)$.
Finally, we can decide whether $w[i\btw j) \in \lang{B'}$ by checking whether
there exists a state $q_2 \in F_B$ such that $M^{i,j}_{q^0_B,q_2}=1$.
Thus, we can implement this part in a way that it contributes only $O((|B|+|C||D|)^3|w|\log |w|)$ to the total time complexity.


\subsection{Bound for Theorem~\ref{thm:oneuseonerewb}}

Observe that, in the construction in Appendix~\ref{appendix:1use1rewb-to-abcbd}, the lengths $|A|$, $|B|$, $|C|$, $|D|$ are at most $O(|E|^2)$, $O(|E|)$, $O(|E|^2)$, and $O(|E|)$, respectively, for the input \rewb{} $E$.
Combining with the discussion in the previous part, the hidden constant of Theorem~\ref{thm:oneuseonerewb} is $O(2^{|E|^4})$.


\subsection{Bound for Theorem~\ref{thm:XYYYYYZ}}

By a similar argument to the one above,
we obtain a hidden constant of $2^{O((|X|+|Y|+|Z|)^2)}+2^{O(|X|+|Y|+|Z|)} \cdot \alpha$ for our $O(|w|\log |w|)$-time algorithm in Theorem~\ref{thm:XYYYYYZ}. 
The term $2^{O((|X|+|Y|+|Z|)^2)}$ arises from the computation of the set $\mathcal{S}_X$, $\mathcal{S}_Y$, and $\mathcal{S}_Z$,
where each computation requires $O(|w|\log |w|)$ interval evaluations on the factorization forest.
$2^{O(|X|+|Y|+|Z|)} \cdot \alpha$ part is from the computation of the representation of the set $\mathcal{S}$.

We see that if we allow an additional $O(\log |w|)$ factor on the dependence on $|w|$, we can reduce the hidden constant to a polynomial.
As discussed above, for interval evaluation, we adopt a segment tree instead of a factorization forest; this works in $O(poly(|X|+|Y|+|Z|)\log |w|)$ time.

The remaining task is to determine whether $e\in \mathcal{S}=\mathcal{S}_X+(\alpha \cdot \mathcal{S}_Y)+\mathcal{S}_Z$.
Instead of computing the representation of $\mathcal{S}$, we decide whether $e\in \mathcal{S}$ directly.
For an $\epsilon$-NFA $N = (Q, \Sigma, \Delta, \mathcal{E}, q^0, F)$ and a string $\theta$, let $N(\theta):=(Q,\{\theta\},\Delta(\theta),\emptyset, q^0, F)$ be an NFA over the unary alphabet $\{\theta\}$, where $\Delta(\theta)$ is defined by 
\[
    \{(p,\theta, q)\colon \text{$q$ can be reached from $p$ by reading $\theta$ in $N$}\}.
\]
Let $N_X=(Q_X, \Sigma, \Delta_X, \mathcal{E}_X, q^0_X, F_X)$ be the $\epsilon$-NFA for $X$. Define $N_Y$ and $N_Z$ similarly.
We define 
\[
    \mathcal{N}'_X:=\{ N_{X}^{q\to F_X}(\theta)\colon \text{$q$ is reachable from $q^0_X$ by reading $w_1$} \}.
\]

Let $N'_Y$ be the NFA obtained from $N_{Y}^{q^0_Y\to F_Y}(\theta)$ by subdividing each edge into $\alpha$ segments, and
let $N'_Z:=N_{Z}^{q^0_Z\to F}(\theta)$, where $F$ is the set of states $q$ such that $F_Z$ is reachable from $q$ by reading $w_2$.

Then, for each $N'_X\in \mathcal{N}'_X$, we construct an $\epsilon$-NFA $N_{\mathrm{all}}$ by connecting all accepting states of $N'_X$ to the initial state of $N'_Y$, and all accepting states of $N'_Y$ to the initial state of $N'_Z$.
The initial state and the accepting states of $N_{\mathrm{all}}$ are the initial state of $N'_X$ and the accepting states of $N'_Z$, respectively.
Now, we can claim that $e$ is in $\mathcal{S}$ if and only if $N_{\mathrm{all}}$ accepts the length-$e$ string, or equivalently, there is a walk of length $e+2$ from the initial state to an accepting state of $N_{\mathrm{all}}$.
This can be determined in $O(|N_{\mathrm{all}}|^3\log e)$ time by standard matrix powering.
Thus, the whole algorithm runs in $O(poly(|X|+|Y|+|Z|)|w|\log^2 |w|)$ time.




%% file: onerewb/bib.bib
@inproceedings{Bringmann:2024,
  author = {Bringmann, Karl and Gr{\o}nlund, Allan and K\"{u}nnemann, Marvin and Larsen, Kasper Green},
  booktitle = {ITCS 2024},
  title = {{The NFA Acceptance Hypothesis: Non-Combinatorial and Dynamic Lower Bounds}},
  year = {2024},
  pages = {22:1--22:25},
  publisher = {Schloss Dagstuhl},
  series = {LIPIcs},
  volume = {287},
  timestamp = {Tue, 14 Oct 2025 01:00:00 +0200},
  biburl = {https://dblp.org/rec/conf/innovations/BringmannGKL24.bib},
  bibsource = {dblp computer science bibliography, https://dblp.org},
  doi = {10.4230/LIPIcs.ITCS.2024.22},
  _bib2doi_selected = {dblp:/rec/conf/innovations/BringmannGKL24.bib},
  _bib2doi_confirmed = {true},
}

@inproceedings{Backurs:2016,
  author = {Backurs, Arturs and Indyk, Piotr},
  booktitle = {FOCS},
  title = {{ Which Regular Expression Patterns Are Hard to Match? }},
  year = {2016},
  pages = {457-466},
  publisher = {IEEE Computer Society},
  timestamp = {Thu, 23 Mar 2023 00:00:00 +0100},
  biburl = {https://dblp.org/rec/conf/focs/BackursI16.bib},
  bibsource = {dblp computer science bibliography, https://dblp.org},
  doi = {10.1109/FOCS.2016.56},
  _bib2doi_selected = {dblp:/rec/conf/focs/BackursI16.bib},
  _bib2doi_confirmed = {true},
}

@article{Downey:1995:1,
  author = {Downey, Rod G. and Fellows, Michael R.},
  title = {Fixed-Parameter Tractability and Completeness I: Basic Results},
  journal = {SIAM Journal on Computing},
  volume = {24},
  number = {4},
  pages = {873--921},
  year = {1995},
  timestamp = {Wed, 14 Nov 2018 00:00:00 +0100},
  biburl = {https://dblp.org/rec/journals/siamcomp/DowneyF95.bib},
  bibsource = {dblp computer science bibliography, https://dblp.org},
  doi = {10.1137/S0097539792228228},
  _bib2doi_selected = {dblp:/rec/journals/siamcomp/DowneyF95.bib},
  _bib2doi_confirmed = {true},
}

@book{Gusfield:Book,
  author = {Dan Gusfield},
  title = {Algorithms on Strings, Trees, and Sequences - Computer Science and Computational Biology},
  year = {1997},
  isbn = {0-521-58519-8},
  publisher = {Cambridge University Press},
  timestamp = {Mon, 29 Jul 2019 01:00:00 +0200},
  biburl = {https://dblp.org/rec/books/cu/Gusfield1997.bib},
  bibsource = {dblp computer science bibliography, https://dblp.org},
  doi = {10.1017/cbo9780511574931},
  _bib2doi_selected = {dblp:/rec/books/cu/Gusfield1997.bib},
  _bib2doi_confirmed = {true},
  _bib2doi_finished = {true},
}

@book{Jewels:Book,
  author = {Maxime Crochemore and Wojciech Rytter},
  title = {Jewels of stringology},
  publisher = {World Scientific},
  year = {2002},
  isbn = {978-981-02-4782-9},
  timestamp = {Tue, 13 Jun 2017 01:00:00 +0200},
  biburl = {https://dblp.org/rec/books/daglib/0020111.bib},
  bibsource = {dblp computer science bibliography, https://dblp.org},
  doi = {10.1142/4838},
  _bib2doi_selected = {dblp:/rec/books/daglib/0020111.bib},
  _bib2doi_confirmed = {true},
}

@article{Crochemore:1981,
  title = {An optimal algorithm for computing the repetitions in a word},
  author = {Max Crochemore},
  journal = {Information Processing Letters},
  volume = {12},
  number = {5},
  pages = {244--250},
  year = {1981},
  timestamp = {Sun, 02 Jun 2019 01:00:00 +0200},
  biburl = {https://dblp.org/rec/journals/ipl/Crochemore81.bib},
  bibsource = {dblp computer science bibliography, https://dblp.org},
  doi = {10.1016/0020-0190(81)90024-7},
  _bib2doi_selected = {dblp:/rec/journals/ipl/Crochemore81.bib},
  _bib2doi_confirmed = {true},
}

@book{Crochemore:2007,
  title = {Algorithms on Strings},
  author = {Crochemore, Maxime and Hancart, Christophe and Lecroq, Thierry},
  place = {Cambridge},
  publisher = {Cambridge University Press},
  year = {2007},
  doi = {10.1017/CBO9780511546853},
  timestamp = {Wed, 23 Mar 2011 00:00:00 +0100},
  biburl = {https://dblp.org/rec/books/daglib/0020103.bib},
  bibsource = {dblp computer science bibliography, https://dblp.org},
  isbn = {978-0-521-84899-2},
  _bib2doi_selected = {dblp:/rec/books/daglib/0020103.bib},
  _bib2doi_confirmed = {true},
}

@book{Sakarovitch:2009,
  title = {Elements of Automata Theory},
  publisher = {Cambridge University Press},
  author = {Sakarovitch, Jacques},
  year = {2009},
  doi = {10.1017/CBO9781139195218},
  timestamp = {Wed, 09 Feb 2011 00:00:00 +0100},
  biburl = {https://dblp.org/rec/books/daglib/0023547.bib},
  bibsource = {dblp computer science bibliography, https://dblp.org},
  isbn = {978-0-521-84425-3},
  _bib2doi_selected = {dblp:/rec/books/daglib/0023547.bib},
  _bib2doi_confirmed = {true},
}

@article{Thompson:1968,
  author = {Ken Thompson},
  title = {Programming Techniques: Regular expression search algorithm},
  year = {1968},
  issue_date = {June 1968},
  publisher = {ACM},
  volume = {11},
  number = {6},
  journal = {Commun. {ACM}},
  month = {jun},
  pages = {419--422},
  timestamp = {Wed, 14 Nov 2018 00:00:00 +0100},
  biburl = {https://dblp.org/rec/journals/cacm/Thompson68.bib},
  bibsource = {dblp computer science bibliography, https://dblp.org},
  doi = {10.1145/363347.363387},
  _bib2doi_selected = {dblp:/rec/journals/cacm/Thompson68.bib},
  _bib2doi_confirmed = {true},
  _bib2doi_finished = {true},
}

@article{Bojanczyk:2012,
  author = {Boja\'{n}czyk, Miko\l{}aj},
  title = {Algorithms for regular languages that use algebra},
  year = {2012},
  issue_date = {June 2012},
  publisher = {ACM},
  volume = {41},
  number = {2},
  journal = {SIGMOD Rec.},
  month = {aug},
  pages = {5--14},
  timestamp = {Fri, 06 Mar 2020 00:00:00 +0100},
  biburl = {https://dblp.org/rec/journals/sigmod/Bojanczyk12.bib},
  bibsource = {dblp computer science bibliography, https://dblp.org},
  doi = {10.1145/2350036.2350038},
  _bib2doi_selected = {dblp:/rec/journals/sigmod/Bojanczyk12.bib},
  _bib2doi_confirmed = {true},
}

@incollection{Colcombet:2021,
  author = {Thomas Colcombet},
  editor = {Jean{-}{\'{E}}ric Pin},
  title = {The factorisation forest theorem},
  booktitle = {Handbook of Automata Theory},
  pages = {653--693},
  publisher = {European Mathematical Society Publishing House, Z{\"{u}}rich, Switzerland},
  year = {2021},
  timestamp = {Mon, 11 Apr 2022 01:00:00 +0200},
  biburl = {https://dblp.org/rec/books/ems/21/Colcombet21.bib},
  bibsource = {dblp computer science bibliography, https://dblp.org},
  doi = {10.4171/Automata-1/18},
  _bib2doi_selected = {dblp:/rec/books/ems/21/Colcombet21.bib},
  _bib2doi_confirmed = {true},
}

@article{Bojanczyk:2020,
  author = {Boja\'{n}czyk, Miko\l{}aj},
  journal = {CoRR},
  title = {Languages recognised by finite semigroups, and their generalisations to objects such as trees and graphs, with an emphasis on definability in monadic second-order logic},
  year = {2020},
  volume = {abs/2008.11635},
  archiveprefix = {arxiv},
  doi = {10.48550/arxiv.2008.11635},
  eprint = {2008.11635},
  timestamp = {Tue, 15 Sep 2020 01:00:00 +0200},
  biburl = {https://dblp.org/rec/journals/corr/abs-2008-11635.bib},
  bibsource = {dblp computer science bibliography, https://dblp.org},
  _bib2doi_selected = {dblp:/rec/journals/corr/abs-2008-11635.bib},
  _bib2doi_confirmed = {true},
}

@article{Simon:1990,
  title = {Factorization forests of finite height},
  author = {Imre Simon},
  journal = {Theoretical Computer Science},
  volume = {72},
  number = {1},
  pages = {65--94},
  year = {1990},
  timestamp = {Wed, 17 Feb 2021 00:00:00 +0100},
  biburl = {https://dblp.org/rec/journals/tcs/Simon90.bib},
  bibsource = {dblp computer science bibliography, https://dblp.org},
  doi = {10.1016/0304-3975(90)90047-L},
  _bib2doi_selected = {dblp:/rec/journals/tcs/Simon90.bib},
  _bib2doi_confirmed = {true},
}

@inproceedings{Crosby:2003:1,
  author = {Scott A. Crosby and Dan S. Wallach},
  title = {Denial of Service via Algorithmic Complexity Attacks},
  booktitle = {12th USENIX Security Symposium (USENIX Security 03)},
  year = {2003},
  publisher = {USENIX Association},
  timestamp = {Mon, 01 Feb 2021 00:00:00 +0100},
  biburl = {https://dblp.org/rec/conf/uss/CrosbyW03.bib},
  bibsource = {dblp computer science bibliography, https://dblp.org},
  url = {https://www.usenix.org/conference/12th-usenix-security-symposium/denial-service-algorithmic-complexity-attacks},
  _bib2doi_selected = {dblp:/rec/conf/uss/CrosbyW03.bib},
  _bib2doi_confirmed = {true},
}

@conference{Crosby:2003:2,
  author = {Scott A. Crosby},
  title = {Denial of Service through Regular Expressions},
  year = {2003},
  publisher = {USENIX Association},
  url = {https://www.usenix.org/conference/12th-usenix-security-symposium/denial-service-through-regular-expressions},
  _bib2doi_selected = {dblp:/rec/conf/uss/CrosbyW03.bib},
  _bib2doi_finished = {true},
}

@misc{Goyvaerts:2021,
  author = {Jan Goyvaerts},
  title = {{Runaway Regular Expressions: Catastrophic Backtracking}},
  howpublished = {\url{https://www.regular-expressions.info/catastrophic.html}},
  year = {2021},
  _bib2doi_finished = {true},
}

@book{Aho:1986,
  author = {Alfred V. Aho and Ravi Sethi and Jeffrey D. Ullman},
  title = {Compilers: Principles, Techniques, and Tools},
  publisher = {Addison-Wesley},
  year = {1986},
  timestamp = {Fri, 17 Jul 2020 01:00:00 +0200},
  biburl = {https://dblp.org/rec/books/aw/AhoSU86.bib},
  bibsource = {dblp computer science bibliography, https://dblp.org},
  isbn = {0-201-10088-6},
  url = {https://www.worldcat.org/oclc/12285707},
  _bib2doi_selected = {dblp:/rec/books/aw/AhoSU86.bib},
  _bib2doi_confirmed = {true},
  _bib2doi_finished = {true},
}

@misc{Cox:2007,
  author = {Russ Cox},
  title = {Regular Expression Matching Can Be Simple And Fast},
  howpublished = {\url{https://swtch.com/~rsc/regexp/regexp1.html}},
  year = {2007},
  month = {January},
  _bib2doi_finished = {true},
}

@misc{Cox:2010,
  author = {Cox, Russ},
  title = {Regular Expression Matching in the Wild},
  howpublished = {\url{https://swtch.com/\textasciitilde rsc/regexp/regexp3.html}},
  year = {2010},
  month = {mar},
  _bib2doi_finished = {true},
}

@inproceedings{Davis:2021,
  author = {Davis, James C. and Servant, Francisco and Lee, Dongyoon},
  booktitle = {SP'21},
  title = {Using Selective Memoization to Defeat Regular Expression Denial of Service (ReDoS)},
  year = {2021},
  pages = {1--17},
  timestamp = {Mon, 03 Mar 2025 00:00:00 +0100},
  biburl = {https://dblp.org/rec/conf/sp/DavisSL21.bib},
  bibsource = {dblp computer science bibliography, https://dblp.org},
  doi = {10.1109/SP40001.2021.00032},
  _bib2doi_selected = {dblp:/rec/conf/sp/DavisSL21.bib},
  _bib2doi_confirmed = {true},
}

@misc{CWE1333,
  title = {{CWE-1333: Inefficient Regular Expression Complexity}},
  author = {{MITRE}},
  howpublished = {\url{https://cwe.mitre.org/data/definitions/1333.html}},
  note = {CWE version 4.17. Page last updated: 2025-04-03.},
  year = {2021},
  _bib2doi_finished = {true},
}

@book{Jeffrey:2006,
  author = {Friedl, Jeffrey},
  title = {Mastering Regular Expressions},
  year = {2006},
  publisher = {O'Reilly},
  timestamp = {Wed, 25 Jan 2023 00:00:00 +0100},
  biburl = {https://dblp.org/rec/books/daglib/0016809.bib},
  bibsource = {dblp computer science bibliography, https://dblp.org},
  isbn = {978-0-596-52812-6},
  url = {https://www.oreilly.com/library/view/mastering-regular-expressions/0596528124/},
  _bib2doi_selected = {dblp:/rec/books/daglib/0016809.bib},
  _bib2doi_confirmed = {true},
  _bib2doi_finished = {true},
}

@inproceedings{Patrascu:2010,
  author = {P{\u{a}}tra{\c{s}}cu, Mihai and Williams, Ryan},
  title = {On the possibility of faster SAT algorithms},
  year = {2010},
  publisher = {SIAM},
  booktitle = {SODA '10},
  pages = {1065--1075},
  timestamp = {Tue, 02 Feb 2021 00:00:00 +0100},
  biburl = {https://dblp.org/rec/conf/soda/PatrascuW10.bib},
  bibsource = {dblp computer science bibliography, https://dblp.org},
  doi = {10.1137/1.9781611973075.86},
  _bib2doi_selected = {dblp:/rec/conf/soda/PatrascuW10.bib},
  _bib2doi_confirmed = {true},
}

@book{Cygan:2015,
  author = {Marek Cygan and Fedor V. Fomin and Lukasz Kowalik and Daniel Lokshtanov and D{\'{a}}niel Marx and Marcin Pilipczuk and Michal Pilipczuk and Saket Saurabh},
  title = {Parameterized Algorithms},
  publisher = {Springer},
  year = {2015},
  timestamp = {Sun, 25 Oct 2020 01:00:00 +0200},
  biburl = {https://dblp.org/rec/books/sp/CyganFKLMPPS15.bib},
  bibsource = {dblp computer science bibliography, https://dblp.org},
  doi = {10.1007/978-3-319-21275-3},
  _bib2doi_selected = {dblp:/rec/books/sp/CyganFKLMPPS15.bib},
  _bib2doi_confirmed = {true},
}

@misc{Abboud:2014,
  author = {Amir Abboud and Virginia Vassilevska Williams},
  title = {Popular conjectures imply strong lower bounds for dynamic problems},
  year = {2014},
  archiveprefix = {arXiv},
  doi = {10.48550/arXiv.1402.0054},
  eprint = {1402.0054},
  eprinttype = {arxiv},
  timestamp = {Mon, 13 Aug 2018 01:00:00 +0200},
  biburl = {https://dblp.org/rec/journals/corr/AbboudW14.bib},
  bibsource = {dblp computer science bibliography, https://dblp.org},
  _bib2doi_selected = {dblp:/rec/journals/corr/AbboudW14.bib},
  _bib2doi_confirmed = {true},
}

@inproceedings{Abboud:2014:focs,
  author = {Abboud, Amir and Williams, Virginia Vassilevska},
  title = {Popular Conjectures Imply Strong Lower Bounds for Dynamic Problems},
  year = {2014},
  publisher = {IEEE Computer Society},
  booktitle = {FOCS '14},
  pages = {434--443},
  timestamp = {Thu, 23 Mar 2023 00:00:00 +0100},
  biburl = {https://dblp.org/rec/conf/focs/AbboudW14.bib},
  bibsource = {dblp computer science bibliography, https://dblp.org},
  doi = {10.1109/FOCS.2014.53},
  _bib2doi_selected = {dblp:/rec/conf/focs/AbboudW14.bib},
  _bib2doi_confirmed = {true},
}

@incollection{Mateescu:1997,
  author = {Alexandru Mateescu and Arto Salomaa},
  editor = {Grzegorz Rozenberg and Arto Salomaa},
  title = {Aspects of Classical Language Theory},
  booktitle = {{Handbook of Formal Languages, Volume 1: Word, Language, Grammar}},
  pages = {175--251},
  publisher = {Springer},
  year = {1997},
  timestamp = {Tue, 06 Aug 2019 01:00:00 +0200},
  biburl = {https://dblp.org/rec/reference/hfl/MateescuS97a.bib},
  bibsource = {dblp computer science bibliography, https://dblp.org},
  doi = {10.1007/978-3-642-59136-5_4},
  _bib2doi_selected = {dblp:/rec/reference/hfl/MateescuS97a.bib},
  _bib2doi_confirmed = {true},
}

@inproceedings{Angluin:1979,
  author = {Angluin, Dana},
  title = {Finding patterns common to a set of strings (Extended Abstract)},
  year = {1979},
  publisher = {ACM},
  booktitle = {STOC '79},
  pages = {130--141},
  timestamp = {Tue, 06 Nov 2018 00:00:00 +0100},
  biburl = {https://dblp.org/rec/conf/stoc/Angluin79.bib},
  bibsource = {dblp computer science bibliography, https://dblp.org},
  doi = {10.1145/800135.804406},
  _bib2doi_selected = {dblp:/rec/conf/stoc/Angluin79.bib},
  _bib2doi_confirmed = {true},
}

@article{Angluin:1980,
  title = {Finding patterns common to a set of strings},
  journal = {JCSS},
  volume = {21},
  number = {1},
  pages = {46--62},
  year = {1980},
  author = {Dana Angluin},
  timestamp = {Tue, 16 Feb 2021 00:00:00 +0100},
  biburl = {https://dblp.org/rec/journals/jcss/Angluin80.bib},
  bibsource = {dblp computer science bibliography, https://dblp.org},
  doi = {10.1016/0022-0000(80)90041-0},
  _bib2doi_selected = {dblp:/rec/journals/jcss/Angluin80.bib},
  _bib2doi_confirmed = {true},
}

@article{Jiang:1994,
  author = {Tao Jiang and Efim Kinber and Arto Salomaa and Kai Salomaa and Sheng Yu},
  title = {Pattern languages with and without erasing},
  journal = {International Journal of Computer Mathematics},
  volume = {50},
  number = {3--4},
  pages = {147--163},
  year = {1994},
  publisher = {Taylor \& Francis},
  doi = {10.1080/00207169408804252},
  _bib2doi_finished = {true},
}

@article{Fernau:2020,
  author = {Fernau, Henning and Manea, Florin and Merca{\c{s}}, Robert and Schmid, Markus L.},
  title = {Pattern Matching with Variables: Efficient Algorithms and Complexity Results},
  year = {2020},
  issue_date = {March 2020},
  publisher = {ACM},
  volume = {12},
  number = {1},
  journal = {ACM Trans. Comput. Theory},
  month = {feb},
  articleno = {6},
  timestamp = {Sun, 19 Jan 2025 00:00:00 +0100},
  biburl = {https://dblp.org/rec/journals/toct/FernauMMS20.bib},
  bibsource = {dblp computer science bibliography, https://dblp.org},
  doi = {10.1145/3369935},
  _bib2doi_selected = {dblp:/rec/journals/toct/FernauMMS20.bib},
  _bib2doi_confirmed = {true},
}

@inproceedings{Fernau:2015,
  author = {Fernau, Henning and Manea, Florin and Merca{\c{s}}, Robert and Schmid, Markus L.},
  title = {{Pattern Matching with Variables: Fast Algorithms and New Hardness Results}},
  booktitle = {STACS 2015},
  pages = {302--315},
  series = {LIPIcs},
  year = {2015},
  volume = {30},
  publisher = {Schloss Dagstuhl},
  timestamp = {Tue, 21 Mar 2023 00:00:00 +0100},
  biburl = {https://dblp.org/rec/conf/stacs/FernauMMS15.bib},
  bibsource = {dblp computer science bibliography, https://dblp.org},
  doi = {10.4230/LIPIcs.STACS.2015.302},
  _bib2doi_selected = {dblp:/rec/conf/stacs/FernauMMS15.bib},
  _bib2doi_confirmed = {true},
}

@inproceedings{Stephan:2012,
  author = {Stephan, Frank and Yoshinaka, Ryo and Zeugmann, Thomas},
  title = {On the Parameterised Complexity of Learning Patterns},
  booktitle = {Computer and Information Sciences II},
  year = {2012},
  publisher = {Springer},
  pages = {277--281},
  timestamp = {Fri, 26 May 2017 01:00:00 +0200},
  biburl = {https://dblp.org/rec/conf/iscis/StephanYZ11.bib},
  bibsource = {dblp computer science bibliography, https://dblp.org},
  doi = {10.1007/978-1-4471-2155-8_35},
  _bib2doi_selected = {dblp:/rec/conf/iscis/StephanYZ11.bib},
  _bib2doi_confirmed = {true},
}

@article{Fernau:2016,
  title = {On the Parameterised Complexity of String Morphism Problems},
  author = {Henning Fernau and Markus L. Schmid and Yngve Villanger},
  journal = {Theory of Computing Systems},
  year = {2016},
  volume = {59},
  number = {1},
  pages = {24--51},
  timestamp = {Tue, 21 Mar 2023 00:00:00 +0100},
  biburl = {https://dblp.org/rec/journals/mst/FernauSV16.bib},
  bibsource = {dblp computer science bibliography, https://dblp.org},
  doi = {10.1007/s00224-015-9635-3},
  _bib2doi_selected = {dblp:/rec/journals/mst/FernauSV16.bib},
  _bib2doi_confirmed = {true},
}

@article{Ehrenfeucht:1979,
  title = {Finding a homomorphism between two words in NP-complete},
  journal = {Information Processing Letters},
  volume = {9},
  number = {2},
  pages = {86--88},
  year = {1979},
  author = {Andrzej Ehrenfeucht and Grzegorz Rozenberg},
  timestamp = {Fri, 26 May 2017 01:00:00 +0200},
  biburl = {https://dblp.org/rec/journals/ipl/EhrenfeuchtR79a.bib},
  bibsource = {dblp computer science bibliography, https://dblp.org},
  doi = {10.1016/0020-0190(79)90135-2},
  _bib2doi_selected = {dblp:/rec/journals/ipl/EhrenfeuchtR79a.bib},
  _bib2doi_confirmed = {true},
}

@book{Sedgewick:2011,
  author = {Robert Sedgewick and Kevin Wayne},
  title = {Algorithms, 4th Edition},
  publisher = {Addison-Wesley},
  year = {2011},
  timestamp = {Tue, 06 Nov 2012 00:00:00 +0100},
  biburl = {https://dblp.org/rec/books/daglib/0029345.bib},
  bibsource = {dblp computer science bibliography, https://dblp.org},
  isbn = {978-0-321-57351-3},
  url = {https://algs4.cs.princeton.edu/home/},
  _bib2doi_selected = {dblp:/rec/books/daglib/0029345.bib},
  _bib2doi_confirmed = {true},
}

@article{Ginsburg:1966,
  author = {Seymour Ginsburg and Edwin H. Spanier},
  title = {Semigroups, {Presburger} Formulas, and Languages},
  journal = {Pacific Journal of Mathematics},
  volume = {16},
  number = {2},
  pages = {285--296},
  year = {1966},
  doi = {10.2140/pjm.1966.16.285},
  _bib2doi_finished = {true},
}

@incollection{Aho:1991,
  author = {Alfred V. Aho},
  title = {Algorithms for finding patterns in strings},
  year = {1990},
  publisher = {MIT Press},
  booktitle = {Handbook of Theoretical Computer Science, Volume {A:} Algorithms and Complexity},
  pages = {255--300},
  timestamp = {Sat, 03 Aug 2019 01:00:00 +0200},
  biburl = {https://dblp.org/rec/books/el/leeuwen90/Aho90.bib},
  bibsource = {dblp computer science bibliography, https://dblp.org},
  editor = {Jan van Leeuwen},
  doi = {10.1016/B978-0-444-88071-0.50010-2},
  _bib2doi_selected = {dblp:/rec/books/el/leeuwen90/Aho90.bib},
  _bib2doi_confirmed = {true},
  _bib2doi_finished = {true},
}

@inproceedings{Nogami:2025,
  author = {Taisei Nogami and Tachio Terauchi},
  title = {Efficient Matching of Some Fundamental Regular Expressions with Backreferences},
  booktitle = {MFCS 2025},
  series = {LIPIcs},
  volume = {345},
  pages = {81:1--81:19},
  publisher = {Schloss Dagstuhl},
  year = {2025},
  timestamp = {Sat, 15 Nov 2025 00:00:00 +0100},
  biburl = {https://dblp.org/rec/conf/mfcs/NogamiT25.bib},
  bibsource = {dblp computer science bibliography, https://dblp.org},
  doi = {10.4230/LIPIcs.MFCS.2025.81},
  _bib2doi_selected = {dblp:/rec/conf/mfcs/NogamiT25.bib},
  _bib2doi_confirmed = {true},
}

@inproceedings{Bille:2024,
  author = {Philip Bille and Inge Li G{\o}rtz},
  title = {Sparse Regular Expression Matching},
  booktitle = {SODA 2024},
  pages = {3354--3375},
  publisher = {{SIAM}},
  year = {2024},
  timestamp = {Tue, 07 May 2024 01:00:00 +0200},
  biburl = {https://dblp.org/rec/conf/soda/BilleG24.bib},
  bibsource = {dblp computer science bibliography, https://dblp.org},
  doi = {10.1137/1.9781611977912.120},
  _bib2doi_selected = {dblp:/rec/conf/soda/BilleG24.bib},
  _bib2doi_confirmed = {true},
}

@inproceedings{Shinohara:1983,
  author = {Shinohara, Takeshi},
  title = {Polynomial time inference of extended regular pattern languages},
  booktitle = {RIMS Symposia on Software Science and Engineering},
  year = {1983},
  publisher = {Springer},
  pages = {115--127},
  timestamp = {Tue, 14 May 2019 10:00:53 +0200},
  biburl = {https://dblp.org/rec/conf/rims/Shinohara82.bib},
  bibsource = {dblp computer science bibliography, https://dblp.org},
  doi = {10.1007/3-540-11980-9_19},
  _bib2doi_selected = {dblp:/rec/conf/rims/Shinohara82.bib},
  _bib2doi_confirmed = {true},
}

@incollection{Choffrut:1997,
  author = {Christian Choffrut and Juhani Karhum{\"{a}}ki},
  editor = {Grzegorz Rozenberg and Arto Salomaa},
  title = {Combinatorics of Words},
  booktitle = {Handbook of Formal Languages, Volume 1: Word, Language, Grammar},
  pages = {329--438},
  publisher = {Springer},
  year = {1997},
  timestamp = {Tue, 06 Aug 2019 01:00:00 +0200},
  biburl = {https://dblp.org/rec/reference/hfl/ChoffrutK97.bib},
  bibsource = {dblp computer science bibliography, https://dblp.org},
  doi = {10.1007/978-3-642-59136-5_6},
  _bib2doi_selected = {dblp:/rec/reference/hfl/ChoffrutK97.bib},
  _bib2doi_confirmed = {true},
}

@article{FineWilf:1965,
  author = {N. J. Fine and H. S. Wilf},
  journal = {Proceedings of the American Mathematical Society},
  number = {1},
  pages = {109--114},
  publisher = {American Mathematical Society},
  title = {Uniqueness Theorems for Periodic Functions},
  volume = {16},
  year = {1965},
  doi = {10.2307/2034009},
  _bib2doi_finished = {true},
}

@book{Cormen:2022,
  title = {Introduction to Algorithms, fourth edition},
  author = {Thomas H. Cormen and Charles E. Leiserson and Ronald L. Rivest and Clifford Stein},
  year = {2022},
  publisher = {MIT Press},
  url = {https://mitpress.mit.edu/9780262046305/introduction-to-algorithms/},
  isbn = {978-0-262-04630-5},
  _bib2doi_selected = {dblp:/rec/books/daglib/0023376.bib},
  _bib2doi_confirmed = {true},
  _bib2doi_finished = {true},
}

@inproceedings{uezato:2024,
  author = {Uezato, Yuya},
  title = {{Regular Expressions with Backreferences and Lookaheads Capture NLOG}},
  booktitle = {ICALP 2024},
  pages = {155:1--155:20},
  series = {LIPIcs},
  year = {2024},
  volume = {297},
  publisher = {Schloss Dagstuhl},
  timestamp = {Fri, 04 Jul 2025 01:00:00 +0200},
  biburl = {https://dblp.org/rec/conf/icalp/Uezato24.bib},
  bibsource = {dblp computer science bibliography, https://dblp.org},
  doi = {10.4230/LIPIcs.ICALP.2024.155},
  _bib2doi_selected = {dblp:/rec/conf/icalp/Uezato24.bib},
  _bib2doi_confirmed = {true},
}

@article{FreydenbergerS19,
  author = {Dominik D. Freydenberger and Markus L. Schmid},
  title = {Deterministic regular expressions with back-references},
  journal = {J. Comput. Syst. Sci.},
  volume = {105},
  pages = {1--39},
  year = {2019},
  url = {https://doi.org/10.1016/j.jcss.2019.04.001},
  doi = {10.1016/j.jcss.2019.04.001},
  timestamp = {Sun, 19 Jan 2025 00:00:00 +0100},
  biburl = {https://dblp.org/rec/journals/jcss/FreydenbergerS19.bib},
  bibsource = {dblp computer science bibliography, https://dblp.org},
  _bib2doi_selected = {dblp:/rec/journals/jcss/FreydenbergerS19.bib},
  _bib2doi_confirmed = {true},
}

@article{Schmid24,
  author = {Markus L. Schmid},
  title = {Regular Expressions with Backreferences: Polynomial-Time Matching Techniques},
  journal = {J. Autom. Lang. Comb.},
  volume = {29},
  number = {2-4},
  pages = {321--357},
  year = {2024},
  url = {https://doi.org/10.25596/jalc-2024-321},
  doi = {10.25596/jalc-2024-321},
  timestamp = {Thu, 09 Oct 2025 01:00:00 +0200},
  biburl = {https://dblp.org/rec/journals/jalc/Schmid24.bib},
  bibsource = {dblp computer science bibliography, https://dblp.org},
  _bib2doi_selected = {dblp:/rec/journals/jalc/Schmid24.bib},
  _bib2doi_confirmed = {true},
}

@article{BerglundM23,
  author = {Martin Berglund and Brink van der Merwe},
  title = {Re-examining regular expressions with backreferences},
  journal = {Theor. Comput. Sci.},
  volume = {940},
  number = {Part},
  pages = {66--80},
  year = {2023},
  url = {https://doi.org/10.1016/j.tcs.2022.10.041},
  doi = {10.1016/j.tcs.2022.10.041},
  timestamp = {Wed, 07 Dec 2022 00:00:00 +0100},
  biburl = {https://dblp.org/rec/journals/tcs/BerglundM23.bib},
  bibsource = {dblp computer science bibliography, https://dblp.org},
  _bib2doi_selected = {dblp:/rec/journals/tcs/BerglundM23.bib},
  _bib2doi_confirmed = {true},
}

@article{Jez:2016:onevariable,
  title = {One-Variable Word Equations in Linear Time},
  author = {Artur Je\.{z}},
  journal = {Algorithmica},
  year = {2016},
  volume = {74},
  number = {1},
  pages = {1--48},
  isbn = {1432-0541},
  doi = {10.1007/s00453-014-9931-3},
  timestamp = {Sat, 19 Oct 2019 01:00:00 +0200},
  biburl = {https://dblp.org/rec/journals/algorithmica/Jez16.bib},
  bibsource = {dblp computer science bibliography, https://dblp.org},
  _bib2doi_selected = {dblp:/rec/journals/algorithmica/Jez16.bib},
  _bib2doi_confirmed = {true},
}

@article{Ibarra:1995,
  title = {A note on parsing pattern languages},
  journal = {Pattern Recognit. Lett.},
  volume = {16},
  number = {2},
  pages = {179--182},
  year = {1995},
  doi = {10.1016/0167-8655(94)00091-G},
  author = {Oscar H. Ibarra and Ting{-}Chuen Pong and Stephen M. Sohn},
  timestamp = {Sat, 22 Feb 2020 00:00:00 +0100},
  biburl = {https://dblp.org/rec/journals/prl/IbarraPS95.bib},
  bibsource = {dblp computer science bibliography, https://dblp.org},
  _bib2doi_old_doi = {https://doi.org/10.1016/0167-8655(94)00091-G},
  _bib2doi_selected = {dblp:/rec/journals/prl/IbarraPS95.bib},
  _bib2doi_confirmed = {true},
  _bib2doi_finished = {true},
}

@article{Plandowsk:2004,
  author = {Plandowski, Wojciech},
  title = {Satisfiability of word equations with constants is in PSPACE},
  year = {2004},
  issue_date = {May 2004},
  publisher = {ACM},
  volume = {51},
  number = {3},
  doi = {10.1145/990308.990312},
  journal = {J. ACM},
  month = {may},
  pages = {483--496},
  timestamp = {Tue, 06 Nov 2018 00:00:00 +0100},
  biburl = {https://dblp.org/rec/journals/jacm/Plandowski04.bib},
  bibsource = {dblp computer science bibliography, https://dblp.org},
  _bib2doi_selected = {dblp:/rec/journals/jacm/Plandowski04.bib},
  _bib2doi_confirmed = {true},
}

@article{Jez:2016,
  author = {Je\.{z}, Artur},
  title = {Recompression: A Simple and Powerful Technique for Word Equations},
  year = {2016},
  issue_date = {March 2016},
  publisher = {ACM},
  volume = {63},
  number = {1},
  issn = {0004-5411},
  doi = {10.1145/2743014},
  journal = {J. ACM},
  month = {feb},
  articleno = {4},
  numpages = {51},
  timestamp = {Sat, 19 Oct 2019 01:00:00 +0200},
  biburl = {https://dblp.org/rec/journals/jacm/Jez16.bib},
  bibsource = {dblp computer science bibliography, https://dblp.org},
  _bib2doi_selected = {dblp:/rec/journals/jacm/Jez16.bib},
  _bib2doi_confirmed = {true},
}

@inproceedings{Jez:2020,
  author = {Je\.{z}, Artur},
  title = {{Solving Word Equations (And Other Unification Problems) by Recompression}},
  booktitle = {28th {EACSL} Annual Conference on Computer Science Logic, {CSL} 2020, Barcelona, Spain, January 13-16, 2020},
  pages = {3:1--3:17},
  series = {LIPIcs},
  year = {2020},
  volume = {152},
  publisher = {Schloss Dagstuhl - Leibniz-Zentrum f{\"{u}}r Informatik},
  doi = {10.4230/LIPIcs.CSL.2020.3},
  timestamp = {Wed, 15 Jan 2020 00:00:00 +0100},
  biburl = {https://dblp.org/rec/conf/csl/Jez20.bib},
  bibsource = {dblp computer science bibliography, https://dblp.org},
  editor = {Maribel Fern{\'{a}}ndez and Anca Muscholl},
  _bib2doi_selected = {dblp:/rec/conf/csl/Jez20.bib},
  _bib2doi_confirmed = {true},
  _bib2doi_finished = {true},
}

@article{FernauS15,
  author = {Henning Fernau and Markus L. Schmid},
  title = {Pattern matching with variables: {A} multivariate complexity analysis},
  journal = {Inf. Comput.},
  volume = {242},
  pages = {287--305},
  year = {2015},
  url = {https://doi.org/10.1016/j.ic.2015.03.006},
  doi = {10.1016/j.ic.2015.03.006},
  timestamp = {Tue, 21 Mar 2023 00:00:00 +0100},
  biburl = {https://dblp.org/rec/journals/iandc/FernauS15.bib},
  bibsource = {dblp computer science bibliography, https://dblp.org},
  _bib2doi_selected = {dblp:/rec/journals/iandc/FernauS15.bib},
  _bib2doi_confirmed = {true},
}

@inproceedings{Dabrowski:2004,
  author = {D{\k{a}}browski, Robert and Plandowski, Wojtek},
  title = {Solving Two-Variable Word Equations (Extended Abstract)},
  booktitle = {Automata, Languages and Programming: 31st International Colloquium, {ICALP} 2004, Turku, Finland, July 12-16, 2004. Proceedings},
  year = {2004},
  publisher = {Springer},
  pages = {408--419},
  timestamp = {Tue, 23 May 2017 01:00:00 +0200},
  biburl = {https://dblp.org/rec/conf/icalp/DabrowskiP04.bib},
  bibsource = {dblp computer science bibliography, https://dblp.org},
  doi = {10.1007/978-3-540-27836-8_36},
  volume = {3142},
  editor = {Josep D{\'{\i}}az and Juhani Karhum{\"{a}}ki and Arto Lepist{\"{o}} and Donald Sannella},
  series = {Lecture Notes in Computer Science},
  _bib2doi_selected = {dblp:/rec/conf/icalp/DabrowskiP04.bib},
  _bib2doi_confirmed = {true},
  _bib2doi_finished = {true},
}

@article{Dabrowski:2011,
  author = {D{\k{a}}browski, Robert and Plandowski, Wojciech},
  title = {On Word Equations in One Variable},
  journal = {Algorithmica},
  year = {2011},
  month = {Aug},
  day = {01},
  volume = {60},
  number = {4},
  pages = {819--828},
  doi = {10.1007/s00453-009-9375-3},
  timestamp = {Wed, 17 May 2017 01:00:00 +0200},
  biburl = {https://dblp.org/rec/journals/algorithmica/DabrowskiP11.bib},
  bibsource = {dblp computer science bibliography, https://dblp.org},
  _bib2doi_selected = {dblp:/rec/journals/algorithmica/DabrowskiP11.bib},
  _bib2doi_confirmed = {true},
}

@inproceedings{Chattopadhyay:2025,
  author = {Chattopadhyay, Agnishom and Li, Angela W. and Mamouras, Konstantinos},
  title = {Verified and Efficient Matching of Regular Expressions with Lookaround},
  year = {2025},
  publisher = {ACM},
  doi = {10.1145/3703595.3705884},
  booktitle = {Proceedings of the 14th ACM SIGPLAN International Conference on Certified Programs and Proofs},
  pages = {198--213},
  series = {CPP '25},
  timestamp = {Sat, 25 Jan 2025 00:00:00 +0100},
  biburl = {https://dblp.org/rec/conf/cpp/ChattopadhyayLM25.bib},
  bibsource = {dblp computer science bibliography, https://dblp.org},
  _bib2doi_selected = {dblp:/rec/conf/cpp/ChattopadhyayLM25.bib},
  _bib2doi_confirmed = {true},
}

@article{Mamouras:2024,
  author = {Mamouras, Konstantinos and Chattopadhyay, Agnishom},
  title = {Efficient Matching of Regular Expressions with Lookaround Assertions},
  year = {2024},
  issue_date = {January 2024},
  publisher = {ACM},
  volume = {8},
  number = {POPL},
  doi = {10.1145/3632934},
  journal = {Proc. ACM Program. Lang.},
  month = {jan},
  articleno = {92},
  numpages = {31},
  timestamp = {Sat, 10 Feb 2024 00:00:00 +0100},
  biburl = {https://dblp.org/rec/journals/pacmpl/MamourasC24.bib},
  bibsource = {dblp computer science bibliography, https://dblp.org},
  _bib2doi_selected = {dblp:/rec/journals/pacmpl/MamourasC24.bib},
  _bib2doi_confirmed = {true},
}

@article{Barriere:2024,
  author = {Barri\`{e}re, Aur\`{e}le and Pit-Claudel, Cl\'{e}ment},
  title = {Linear Matching of JavaScript Regular Expressions},
  year = {2024},
  issue_date = {June 2024},
  publisher = {Association for Computing Machinery},
  address = {New York, NY, USA},
  volume = {8},
  number = {PLDI},
  doi = {10.1145/3656431},
  journal = {Proc. ACM Program. Lang.},
  month = {jun},
  articleno = {201},
  numpages = {25},
  keywords = {Regex, Automata, JavaScript},
  timestamp = {Sun, 19 Jan 2025 00:00:00 +0100},
  biburl = {https://dblp.org/rec/journals/pacmpl/BarriereP24.bib},
  bibsource = {dblp computer science bibliography, https://dblp.org},
  _bib2doi_selected = {dblp:/rec/journals/pacmpl/BarriereP24.bib},
  _bib2doi_confirmed = {true},
}

@inproceedings{Fujinami:2024,
  author = {Hiroya Fujinami and Ichiro Hasuo},
  editor = {Stephanie Weirich},
  title = {Efficient Matching with Memoization for Regexes with Look-around and Atomic Grouping},
  booktitle = {33rd European Symposium on Programming, {ESOP} 2024},
  series = {Lecture Notes in Computer Science},
  volume = {14577},
  pages = {90--118},
  publisher = {Springer},
  year = {2024},
  doi = {10.1007/978-3-031-57267-8\_4},
  timestamp = {Sat, 08 Jun 2024 01:00:00 +0200},
  biburl = {https://dblp.org/rec/conf/esop/FujinamiH24.bib},
  bibsource = {dblp computer science bibliography, https://dblp.org},
  _bib2doi_selected = {dblp:/rec/conf/esop/FujinamiH24.bib},
  _bib2doi_confirmed = {true},
}

@inproceedings{Schulz90,
  author = {Klaus U. Schulz},
  title = {Makanin's Algorithm for Word Equations - Two Improvements and a Generalization},
  booktitle = {Word Equations and Related Topics, First International Workshop, {IWWERT} '90, T{\"{u}}bingen, Germany, October 1-3, 1990, Proceedings},
  series = {Lecture Notes in Computer Science},
  volume = {572},
  pages = {85--150},
  publisher = {Springer},
  year = {1990},
  url = {https://doi.org/10.1007/3-540-55124-7\_4},
  doi = {10.1007/3-540-55124-7\_4},
  timestamp = {Wed, 17 May 2017 01:00:00 +0200},
  biburl = {https://dblp.org/rec/conf/iwwert/Schulz90.bib},
  bibsource = {dblp computer science bibliography, https://dblp.org},
  _bib2doi_selected = {dblp:/rec/conf/iwwert/Schulz90.bib},
  _bib2doi_confirmed = {true},
}

@article{DiekertGH05,
  author = {Volker Diekert and Claudio Gutierrez and Christian Hagenah},
  title = {The existential theory of equations with rational constraints in free groups is PSPACE-complete},
  journal = {Inf. Comput.},
  volume = {202},
  number = {2},
  pages = {105--140},
  year = {2005},
  url = {https://doi.org/10.1016/j.ic.2005.04.002},
  doi = {10.1016/j.ic.2005.04.002},
  timestamp = {Thu, 10 Nov 2022 00:00:00 +0100},
  biburl = {https://dblp.org/rec/journals/iandc/DiekertGH05.bib},
  bibsource = {dblp computer science bibliography, https://dblp.org},
  _bib2doi_selected = {dblp:/rec/journals/iandc/DiekertGH05.bib},
  _bib2doi_confirmed = {true},
}

@article{Ladner:1984,
  author = {Ladner, Richard E. and Lipton, Richard J. and Stockmeyer, Larry J.},
  title = {Alternating Pushdown and Stack Automata},
  journal = {SIAM Journal on Computing},
  volume = {13},
  number = {1},
  pages = {135--155},
  year = {1984},
  doi = {10.1137/0213010},
  timestamp = {Mon, 05 Feb 2024 00:00:00 +0100},
  biburl = {https://dblp.org/rec/journals/siamcomp/LadnerLS84.bib},
  bibsource = {dblp computer science bibliography, https://dblp.org},
  _bib2doi_selected = {dblp:/rec/journals/siamcomp/LadnerLS84.bib},
  _bib2doi_confirmed = {true},
}

@article{Chandra:1981,
  author = {Chandra, Ashok K. and Kozen, Dexter C. and Stockmeyer, Larry J.},
  title = {Alternation},
  year = {1981},
  issue_date = {Jan. 1981},
  publisher = {ACM},
  volume = {28},
  number = {1},
  doi = {10.1145/322234.322243},
  journal = {J. ACM},
  month = {jan},
  pages = {114--133},
  numpages = {20},
  timestamp = {Thu, 14 Oct 2021 01:00:00 +0200},
  biburl = {https://dblp.org/rec/journals/jacm/ChandraKS81.bib},
  bibsource = {dblp computer science bibliography, https://dblp.org},
  _bib2doi_selected = {dblp:/rec/journals/jacm/ChandraKS81.bib},
  _bib2doi_confirmed = {true},
}

@inproceedings{Fischer:1971,
  author = {Fischer, M. J. and Meyer, A. R.},
  booktitle = {Annual Symposium on Switching and Automata Theory},
  title = {Boolean matrix multiplication and transitive closure},
  year = {1971},
  pages = {129--131},
  doi = {10.1109/SWAT.1971.4},
  timestamp = {Thu, 23 Mar 2023 00:00:00 +0100},
  biburl = {https://dblp.org/rec/conf/focs/FischerM71.bib},
  bibsource = {dblp computer science bibliography, https://dblp.org},
  _bib2doi_selected = {dblp:/rec/conf/focs/FischerM71.bib},
  _bib2doi_confirmed = {true},
}

@article{Munro:1971,
  title = {Efficient determination of the transitive closure of a directed graph},
  journal = {Inf. Process. Lett.},
  volume = {1},
  number = {2},
  pages = {56--58},
  year = {1971},
  doi = {10.1016/0020-0190(71)90006-8},
  author = {J. Ian Munro},
  timestamp = {Wed, 14 Nov 2018 00:00:00 +0100},
  biburl = {https://dblp.org/rec/journals/ipl/Munro71.bib},
  bibsource = {dblp computer science bibliography, https://dblp.org},
  _bib2doi_old_doi = {https://doi.org/10.1016/0020-0190(71)90006-8},
  _bib2doi_selected = {dblp:/rec/journals/ipl/Munro71.bib},
  _bib2doi_confirmed = {true},
  _bib2doi_finished = {true},
}

@article{Furman:1970,
  title = {Application of a method of rapid multiplication of matrices to the problem of finding the transitive closure of a graph},
  author = {Furman, M.E.},
  journal = {Doklady Akademii Nauk},
  volume = {194},
  number = {3},
  pages = {524--524},
  year = {1970},
  url = {https://www.mathnet.ru/eng/dan35686},
  _bib2doi_finished = {true},
}

@techreport{Holzer:2023,
  title = {Comments on Monoids Induced by {NFAs}},
  author = {Holzer, Markus},
  institution = {Institut f{\"u}r Informatik, Universit{\"a}t Giessen},
  year = {2023},
  month = {March},
  number = {2301},
  type = {IFIG Research Report},
  url = {https://jlupub.ub.uni-giessen.de/items/c469513c-e738-48ae-900b-397efc96b042},
  doi = {10.22029/jlupub-14955},
  _bib2doi_finished = {true},
}

@incollection{Pin:1997,
  author = {Pin, Jean-{\'{E}}ric},
  editor = {Grzegorz Rozenberg and Arto Salomaa},
  title = {Syntactic Semigroups},
  booktitle = {Handbook of Formal Languages, Volume 1: Word, Language, Grammar},
  pages = {679--746},
  publisher = {Springer},
  year = {1997},
  doi = {10.1007/978-3-642-59136-5\_10},
  timestamp = {Tue, 06 Aug 2019 01:00:00 +0200},
  biburl = {https://dblp.org/rec/reference/hfl/Pin97.bib},
  bibsource = {dblp computer science bibliography, https://dblp.org},
  _bib2doi_selected = {dblp:/rec/reference/hfl/Pin97.bib},
  _bib2doi_confirmed = {true},
}

@article{Pin:1995,
  title = {Finite semigroups and recognizable languages: an introduction},
  author = {Pin, Jean-{\'{E}}ric},
  journal = {NATO Advanced Study Institute},
  pages = {1--32},
  year = {1995},
  publisher = {Kluwer academic publishers},
  booktitle = {Semigroups, Formal Languages and Groups},
  url = {https://www.irif.fr/~jep/Resumes/York1.html},
  _bib2doi_finished = {true},
}

@inproceedings{Kufleitner:2008,
  author = {Kufleitner, Manfred},
  editor = {Ochma{\'{n}}ski, Edward and Tyszkiewicz, Jerzy},
  title = {The Height of Factorization Forests},
  booktitle = {MFCS},
  year = {2008},
  publisher = {Springer},
  pages = {443--454},
  doi = {10.1007/978-3-540-85238-4\_36},
  timestamp = {Fri, 02 Nov 2018 00:00:00 +0100},
  biburl = {https://dblp.org/rec/conf/mfcs/Kufleitner08.bib},
  bibsource = {dblp computer science bibliography, https://dblp.org},
  _bib2doi_selected = {dblp:/rec/conf/mfcs/Kufleitner08.bib},
  _bib2doi_confirmed = {true},
}

@inproceedings{Bojanczyk:2010,
  author = {Boja{\'{n}}czyk, Miko{\l}aj and Parys, Pawe{\l}},
  editor = {Abramsky, Samson and Gavoille, Cyril and Kirchner, Claude and Meyer auf der Heide, Friedhelm and Spirakis, Paul G.},
  title = {Efficient Evaluation of Nondeterministic Automata Using Factorization Forests},
  booktitle = {ICALP},
  year = {2010},
  publisher = {Springer},
  pages = {515--526},
  doi = {10.1007/978-3-642-14165-2\_44},
  timestamp = {Tue, 23 May 2017 01:00:00 +0200},
  biburl = {https://dblp.org/rec/conf/icalp/BojanczykP10.bib},
  bibsource = {dblp computer science bibliography, https://dblp.org},
  _bib2doi_selected = {dblp:/rec/conf/icalp/BojanczykP10.bib},
  _bib2doi_confirmed = {true},
}

@misc{Kufleitner:2007,
  title = {A Proof of the Factorization Forest Theorem},
  author = {Manfred Kufleitner},
  year = {2007},
  eprint = {0710.5130},
  archiveprefix = {arXiv},
  primaryclass = {cs.LO},
  url = {https://arxiv.org/abs/0710.5130},
  doi = {10.48550/arXiv.0710.5130},
  timestamp = {Mon, 13 Aug 2018 01:00:00 +0200},
  biburl = {https://dblp.org/rec/journals/corr/abs-0710-5130.bib},
  bibsource = {dblp computer science bibliography, https://dblp.org},
  _bib2doi_selected = {dblp:/rec/journals/corr/abs-0710-5130.bib},
  _bib2doi_confirmed = {true},
}

@inproceedings{Williams:2019,
  title = {On some fine-grained questions in algorithms and complexity},
  doi = {10.1142/9789813272880_0188},
  booktitle = {Proceedings of the International Congress of Mathematicians (ICM 2018)},
  publisher = {WORLD SCIENTIFIC},
  author = {Williams, Virginia Vassilevska},
  year = {2019},
  _bib2doi_finished = {true},
}

@article{Duraj:2019,
  author = {Lech Duraj and Marvin K{\"{u}}nnemann and Adam Polak},
  title = {Tight Conditional Lower Bounds for Longest Common Increasing Subsequence},
  journal = {Algorithmica},
  volume = {81},
  number = {10},
  pages = {3968--3992},
  year = {2019},
  doi = {10.1007/S00453-018-0485-7},
  timestamp = {Thu, 31 Oct 2019 00:00:00 +0100},
  biburl = {https://dblp.org/rec/journals/algorithmica/DurajKP19.bib},
  bibsource = {dblp computer science bibliography, https://dblp.org},
  _bib2doi_selected = {dblp:/rec/journals/algorithmica/DurajKP19.bib},
  _bib2doi_confirmed = {true},
}

@article{Williams:2005,
  title = {A new algorithm for optimal 2-constraint satisfaction and its implications},
  journal = {Theor. Comput. Sci.},
  volume = {348},
  number = {2-3},
  pages = {357--365},
  year = {2005},
  doi = {10.1016/j.tcs.2005.09.023},
  author = {Ryan Williams},
  timestamp = {Wed, 17 Feb 2021 00:00:00 +0100},
  biburl = {https://dblp.org/rec/journals/tcs/Williams05.bib},
  bibsource = {dblp computer science bibliography, https://dblp.org},
  _bib2doi_old_doi = {https://doi.org/10.1016/j.tcs.2005.09.023},
  _bib2doi_selected = {dblp:/rec/journals/tcs/Williams05.bib},
  _bib2doi_confirmed = {true},
  _bib2doi_finished = {true},
}

@article{Hansen:2021,
  title = {Tight bounds for reachability problems on one-counter and pushdown systems},
  journal = {Inf. Process. Lett.},
  volume = {171},
  pages = {106135},
  year = {2021},
  doi = {10.1016/j.ipl.2021.106135},
  author = {Jakob Cetti Hansen and Adam Husted Kjelstr{\o}m and Andreas Pavlogiannis},
  timestamp = {Thu, 14 Oct 2021 01:00:00 +0200},
  biburl = {https://dblp.org/rec/journals/ipl/HansenKP21.bib},
  bibsource = {dblp computer science bibliography, https://dblp.org},
  _bib2doi_old_doi = {https://doi.org/10.1016/j.ipl.2021.106135},
  _bib2doi_selected = {dblp:/rec/journals/ipl/HansenKP21.bib},
  _bib2doi_confirmed = {true},
  _bib2doi_finished = {true},
}

@misc{Nogami:2026:arxiv,
  title = {Hardness of Regular Expression Matching with Extensions},
  author = {Taisei Nogami and Tachio Terauchi},
  year = {2026},
  eprint = {2601.03020},
  archiveprefix = {arXiv},
  primaryclass = {cs.DS},
  url = {https://arxiv.org/abs/2601.03020},
  doi = {10.48550/arXiv.2601.03020},
  _bib2doi_finished = {true},
}

@misc{simple-markdown-redos-patch,
  author       = {Buckles, Aria},
  title        = {{simple-markdown}: inlineCode: Fix {ReDoS} \& improve escape semantics (commit 89797fe)},
  howpublished = {GitHub commit},
  url          = {https://github.com/ariabuckles/simple-markdown/commit/89797fe},
  year         = {2019},
}

@misc{marked-redos-patch,
  author       = {{UziTech}},
  title        = {{marked}: fix inline code regex (commit 3271b8b)},
  howpublished = {GitHub commit},
  url          = {https://github.com/markedjs/marked/commit/3271b8b},
  year         = {2018},
}

@Misc{Cloudflare:2019,
  author       = {John Graham-Cumming},
  howpublished = {Cloudflare Blog. Archived at \url{https://web.archive.org/web/20190712160002/https://blog.cloudflare.com/details-of-the-cloudflare-outage-on-july-2-2019/}},
  month        = {July},
  title        = {{Details of the Cloudflare outage on July 2, 2019}},
  year         = {2019},
  url          = {https://blog.cloudflare.com/details-of-the-cloudflare-outage-on-july-2-2019/},
}

@Misc{stackstatus:2016,
  author       = {{Stack Exchange}},
  howpublished = {Stack Exchange Network Status Blog. Archived at \url{https://web.archive.org/web/20180801005940/http://stackstatus.net/post/147710624694/outage-postmortem-july-20-2016}},
  month        = {July},
  title        = {{Outage Postmortem - July 20, 2016}},
  year         = {2016},
  url          = {http://stackstatus.net/post/147710624694/outage-postmortem-july-20-2016},
}

@misc{playCloudflare:2019,
  author = {{Anonymous Author(s)}},
  title = {Checking Quadratic Behaviors of REGEX of Cloudflare Outage},
  howpublished = {\url{https://regex101.com/r/hidvB3/1/debugger}},
  year = {2025},
}

@misc{playstackstatus:2016,
  author = {{Anonymous Author(s)}},
  title = {Checking Quadratic Behaviors of REGEX of Stack Overflow Outage},
  howpublished = {\url{https://regex101.com/r/pR6DbC/1/debugger}},
  year = {2025},
}

@Misc{CVE-2021-21240,
  author       = {{National Vulnerability Database (NVD)}},
  howpublished = {\url{https://nvd.nist.gov/vuln/detail/CVE-2021-21240}},
  note         = {GitHub Advisory:},
  title        = {{CVE-2021-21240}},
  year         = {2021},
  url          = {https://github.com/advisories/GHSA-93xj-8mrv-444m},
}

@Misc{CVE-2023-6159,
  author       = {{National Vulnerability Database (NVD)}},
  howpublished = {\url{https://nvd.nist.gov/vuln/detail/cve-2023-6159}},
  note         = {GitLab Security Release:},
  title        = {{CVE-2023-6159}},
  year         = {2023},
  url          = {https://about.gitlab.com/releases/2024/01/25/critical-security-release-gitlab-16-8-1-released/},
}

@Misc{CVE-2023-39663,
  author       = {{National Vulnerability Database (NVD)}},
  howpublished = {\url{https://nvd.nist.gov/vuln/detail/CVE-2023-39663}},
  note         = {GitHub Advisory:},
  title        = {{CVE-2023-39663}},
  year         = {2023},
  url          = {https://github.com/advisories/GHSA-v638-q856-grg8},
}

@Misc{CVE-2024-28865,
  author       = {{National Vulnerability Database (NVD)}},
  howpublished = {\url{https://nvd.nist.gov/vuln/detail/CVE-2024-28865}},
  note         = {GitHub Advisory:},
  title        = {{CVE-2024-28865}},
  year         = {2024},
  url          = {https://github.com/advisories/GHSA-wj85-w4f4-xh8h},
}

@Misc{CVE-2024-26142,
  author       = {{National Vulnerability Database (NVD)}},
  howpublished = {\url{https://nvd.nist.gov/vuln/detail/CVE-2024-26142}},
  note         = {GitHub Advisory:},
  title        = {{CVE-2024-26142}},
  year         = {2024},
  url          = {https://github.com/advisories/GHSA-jjhx-jhvp-74wq},
}

@Misc{CVE-2025-25200,
  author       = {{National Vulnerability Database (NVD)}},
  howpublished = {\url{https://nvd.nist.gov/vuln/detail/CVE-2025-25200}},
  note         = {GitHub Advisory:},
  title        = {{CVE-2025-25200}},
  year         = {2025},
  url          = {https://github.com/advisories/GHSA-593f-38f6-jp5m},
}

@Misc{CVE-2025-29907,
  author       = {{National Vulnerability Database (NVD)}},
  howpublished = {\url{https://nvd.nist.gov/vuln/detail/CVE-2025-29907}},
  note         = {GitHub Advisory:},
  title        = {{CVE-2025-29907}},
  year         = {2025},
  url          = {https://github.com/advisories/GHSA-w532-jxjh-hjhj},
}

@Misc{CVE-2025-5197,
  author       = {{National Vulnerability Database (NVD)}},
  howpublished = {\url{https://nvd.nist.gov/vuln/detail/CVE-2025-5197}},
  note         = {GitHub Advisory:},
  title        = {{CVE-2025-5197}},
  year         = {2025},
  url          = {https://github.com/advisories/GHSA-9356-575x-2w9m},
}

@misc{Goyvaerts:lookaround,
  author       = {Goyvaerts, Jan},
  title        = {{Lookahead and Lookbehind Zero-Length Assertions}},
  howpublished = {\url{https://www.regular-expressions.info/lookaround.html}},
  year         = {2025},
}

@Misc{CVE-2017-16114,
  author       = {{National Vulnerability Database (NVD)}},
  howpublished = {\url{https://nvd.nist.gov/vuln/detail/cve-2023-6159}},
  note         = {GitLab Security Release:},
  title        = {{CVE-2017-16114}},
  year         = {2017},
  url          = {https://nvd.nist.gov/vuln/detail/CVE-2017-16114},
}

@Misc{CVE-2019-25103,
  author       = {{National Vulnerability Database (NVD)}},
  howpublished = {\url{https://nvd.nist.gov/vuln/detail/cve-2023-6159}},
  note         = {GitLab Security Release:},
  title        = {{CVE-2019-25103}},
  year         = {2019},
  url          = {https://nvd.nist.gov/vuln/detail/CVE-2019-25103},
}

@misc{Goyvaerts:atomic,
  author       = {Goyvaerts, Jan},
  title        = {Atomic Grouping},
  howpublished = {\url{https://www.regular-expressions.info/atomic.html}},
  year         = {2025},
}

@misc{Goyvaerts:kleenestar,
  author       = {Goyvaerts, Jan},
  title        = {Repetition with Star and Plus},
  howpublished = {\url{https://www.regular-expressions.info/repeat.html}},
  year         = {2025},
}

@article{BilleGVW12,
  author  = {Bille, Philip and G{\o}rtz, Inge Li and Vildh{\o}j, Hjalte Wedel and Wind, David Kofoed},
  title   = {String Matching with Variable Length Gaps},
  journal = {Theoretical Computer Science},
  volume  = {443},
  pages   = {25--34},
  year    = {2012},
  doi     = {10.1016/j.tcs.2012.03.029}
}

@article{NavarroRaffinot03,
  author  = {Navarro, Gonzalo and Raffinot, Mathieu},
  title   = {Fast and Simple Character Classes and Bounded Gaps Pattern Matching, with Applications to Protein Searching},
  journal = {Journal of Computational Biology},
  volume  = {10},
  number  = {6},
  pages   = {903--923},
  year    = {2003},
  doi     = {10.1089/106652703322756140}
}

@article{FredrikssonGrabowski08,
  author  = {Fredriksson, Kimmo and Grabowski, Szymon},
  title   = {Efficient Algorithms for Pattern Matching with General Gaps, Character Classes, and Transposition Invariance},
  journal = {Information Retrieval},
  volume  = {11},
  pages   = {335--357},
  year    = {2008},
  doi     = {10.1007/s10791-008-9054-z}
}

@inproceedings{HaapasaloSSSS11,
  author    = {Haapasalo, Tuukka and Silvasti, Panu and Sippu, Seppo and Soisalon-Soininen, Eljas},
  title     = {Online Dictionary Matching with Variable-Length Gaps},
  booktitle = {Experimental Algorithms},
  series    = {Lecture Notes in Computer Science},
  volume    = {6630},
  pages     = {76--87},
  publisher = {Springer},
  year      = {2011},
  doi       = {10.1007/978-3-642-20662-7_7}
}

@article{HonLSTTY18,
  author  = {Hon, Wing-Kai and Lam, Tak-Wah and Shah, Rahul and Thankachan, Sharma V. and Ting, Hing-Fung and Yang, Yilin},
  title   = {Dictionary Matching with a Bounded Gap in Pattern or in Text},
  journal = {Algorithmica},
  volume  = {80},
  pages   = {698--713},
  year    = {2018},
  doi     = {10.1007/s00453-017-0288-2}
}

@article{LevyS22,
  author  = {Levy, Avivit and Shalom, B. Riva},
  title   = {A Comparative Study of Dictionary Matching with Gaps: Limitations, Techniques and Challenges},
  journal = {Algorithmica},
  volume  = {84},
  pages   = {590--638},
  year    = {2022},
  doi     = {10.1007/s00453-021-00851-6}
}

@article{AmirKLPPS19,
  author  = {Amir, Amihood and Kopelowitz, Tsvi and Levy, Avivit and Pettie, Seth and Porat, Ely and Shalom, B. Riva},
  title   = {Mind the Gap! Online Dictionary Matching with One Gap},
  journal = {Algorithmica},
  volume  = {81},
  pages   = {2123--2157},
  year    = {2019},
  doi     = {10.1007/s00453-018-0526-2}
}

@article{GajentaanO95,
  author  = {Gajentaan, Anka and Overmars, Mark H.},
  title   = {On a Class of ${O}(n^2)$ Problems in Computational Geometry},
  journal = {Computational Geometry},
  volume  = {5},
  number  = {3},
  pages   = {165--185},
  year    = {1995},
  doi     = {10.1016/0925-7721(95)00022-2}
}

@inproceedings{Patrascu10,
  author    = {P{\u{a}}tra{\c{s}}cu, Mihai},
  title     = {Towards Polynomial Lower Bounds for Dynamic Problems},
  booktitle = {STOC 2010},
  pages     = {603--610},
  year      = {2010},
  publisher = {ACM},
  doi       = {10.1145/1806689.1806772}
}

@inproceedings{KopelowitzPP16,
  author    = {Kopelowitz, Tsvi and Pettie, Seth and Porat, Ely},
  title     = {Higher Lower Bounds from the {3SUM} Conjecture},
  booktitle = {SODA 2016},
  pages     = {1272--1287},
  year      = {2016},
  publisher = {SIAM},
  doi       = {10.1137/1.9781611974331.ch89}
}
